%% file: OptInitState.tex
\UseRawInputEncoding
\documentclass[letterpaper,onecolumn,11pt,accepted=2026-02-14]{quantumarticle}
\pdfoutput=1

\usepackage{silence}
\usepackage{amssymb,amsmath,amsthm,amsfonts,dsfont}
\usepackage[numbers,comma,sort&compress]{natbib}
\usepackage[pagebackref,colorlinks,hypertexnames=false]{hyperref}
\hypersetup{
	pdfstartview={FitH},
	pdfnewwindow=true,
	colorlinks=true,
	linkcolor=blue,
	citecolor=blue,
	filecolor=blue,
	urlcolor=blue}
\usepackage[capitalise]{cleveref}
\usepackage[all]{hypcap}
\usepackage{url}
\usepackage{graphicx}
\usepackage[font=small]{caption}
\usepackage[subrefformat=parens,labelformat=parens]{subcaption}
\usepackage{float}
\usepackage{footnote}
\usepackage{enumitem}
\newgeometry{margin=1in}

\newcommand{\abs}[1]{\left\lvert#1\right\rvert}
\newcommand{\norm}[1]{\left\lVert#1\right\rVert}
\DeclareMathOperator{\poly}{poly}
\DeclareMathOperator{\polylog}{polylog}

\newtheorem{theorem}{Theorem}
\newtheorem{lemma}{Lemma}
\newtheorem{definition}[lemma]{Definition}
\newtheorem{proposition}[lemma]{Proposition}
\newtheorem{corollary}[lemma]{Corollary}
\theoremstyle{remark}
\newtheorem*{remark}{Remark}
\theoremstyle{plain}
\newtheorem*{lemma*}{Lemma}

\newcommand{\eq}[1]{\cref{eq:#1}}
\newcommand{\thm}[1]{\hyperref[thm:#1]{Theorem~\ref*{thm:#1}}}
\newcommand{\defn}[1]{\hyperref[defn:#1]{Definition~\ref*{defn:#1}}}
\newcommand{\lem}[1]{\hyperref[lem:#1]{Lemma~\ref*{lem:#1}}}
\newcommand{\prop}[1]{\hyperref[prop:#1]{Proposition~\ref*{prop:#1}}}
\newcommand{\fig}[1]{\hyperref[fig:#1]{Figure~\ref*{fig:#1}}}
\newcommand{\tab}[1]{\hyperref[tab:#1]{Table~\ref*{tab:#1}}}
\renewcommand{\sec}[1]{\hyperref[sec:#1]{Section~\ref*{sec:#1}}}
\newcommand{\append}[1]{\hyperref[append:#1]{Appendix~\ref*{append:#1}}}
\newcommand{\cor}[1]{\hyperref[cor:#1]{Corollary~\ref*{cor:#1}}}

\newcommand{\ket}[1]{|#1\rangle}
\newcommand{\bra}[1]{\langle#1|}
\newcommand{\ketbra}[2]{\ket{#1}\!\bra{#2}}

\newcommand\blfootnote[1]{%
	\begingroup
	\renewcommand\thefootnote{}\footnote{#1}%
	\addtocounter{footnote}{-1}%
	\endgroup
}

\usepackage{mathtools}

\DeclareFontFamily{U}{matha}{\hyphenchar\font45}
\DeclareFontShape{U}{matha}{m}{n}{
	<5> <6> <7> <8> <9> <10> gen * matha
	<10.95> matha10 <12> <14.4> <17.28> <20.74> <24.88> matha12
}{}
\DeclareSymbolFont{matha}{U}{matha}{m}{n}
\DeclareFontFamily{U}{mathx}{\hyphenchar\font45}
\DeclareFontShape{U}{mathx}{m}{n}{
	<5> <6> <7> <8> <9> <10>
	<10.95> <12> <14.4> <17.28> <20.74> <24.88>
	mathx10
}{}
\DeclareSymbolFont{mathx}{U}{mathx}{m}{n}
\DeclareMathSymbol{\obot}         {2}{matha}{"6B}
\DeclareMathSymbol{\bigobot}       {1}{mathx}{"CB}

\usepackage{nicematrix}
\usepackage{tikz}

\makeatletter
\def\newmaketag{%
  \def\maketag@@@##1{\hbox{\m@th\normalfont\normalsize##1}}%
  }
\makeatother

\usepackage{etoolbox}
\apptocmd{\sloppy}{\hbadness 10000\relax}{}{}

\usepackage{multirow} 
\usepackage{makecell}

\makeatletter
\newcommand{\thickhline}{%
    \noalign {\ifnum 0=`}\fi \hrule height 1.5pt
    \futurelet \reserved@a \@xhline
}
\newcolumntype{"}{@{\hskip\tabcolsep\vrule width 1pt\hskip\tabcolsep}}
\makeatother

\usepackage{colortbl}
\definecolor{Gray}{gray}{0.75}

\usetikzlibrary{positioning}

\usepackage{changepage}

\usepackage[hyperpageref]{backref}
\renewcommand*{\backrefalt}[4]{%
\ifcase #1 %
No citations.%
\or
(Cited on page #2).%
\else
(Cited on pages #2).%
\fi
}

\usepackage{etoolbox}
\makeatletter
\patchcmd\NAT@citexnum{\let\NAT@last@num\NAT@num}{\MakeLinkTarget[cite]{}\Hy@backout{\@citeb\@extra@b@citeb}\let\NAT@last@num\NAT@num}{}{\fail}
\makeatother

\title{Quantum linear system algorithm with optimal queries to initial state preparation}
\author[aff1,aff2]{Guang Hao Low}
\orcid{0000-0002-6934-1052}
\author[aff1,aff3]{Yuan Su}
\orcid{0000-0003-1144-3563}
\affiliation[aff1]{Azure Quantum, Microsoft, Redmond, WA 98052, USA}
\affiliation[aff2]{Google Quantum AI, Venice, CA 90291, USA}
\affiliation[aff3]{AWS Center for Quantum Computing, Pasadena, CA 91106, USA}
\begin{document}
\maketitle
\blfootnote{Paper was written while both authors were affiliated with Microsoft.}
\vspace{-0.5cm}

\begin{abstract}
\begin{adjustwidth}{-20pt}{-20pt}
The quantum linear system problem provides one of the most enticing sources of exponential quantum speedups, and its resolution underlies other interesting quantum algorithms for differential equations and eigenvalue processing. 
The goal is to produce a state proportional to the solution $A^{-1}|b\rangle$ of a linear system, by querying an oracle $O_A$ that block encodes the coefficient matrix and an oracle $O_b$ that prepares the initial state. 

We present a quantum linear system algorithm with query complexity $\mathbf{\Theta}\left(1/\sqrt{p}\right)$ to $O_b$ that is optimal, and query complexity 
$\mathbf{O}\left(\kappa\log\left(1/p\right)\left(\log\log\left(1/p\right)+\log\left(1/\epsilon\right)\right)\right)$ to $O_A$ that is nearly optimal in all parameters including the condition number $\kappa=\|A\|\|A^{-1}\|$, success amplitude $\sqrt{p}=\|A^{-1}|b\rangle\|/\|A^{-1}\|$, and accuracy $\epsilon$. 
In various applications to solving differential equations, preparing ground states of operators with real spectra, estimating and transforming eigenvalues of non-normal matrices, we can further improve the dependence on $p$ to nearly match or outperform best previous results based on other methods.
As $\kappa$ can be arbitrarily larger than $1/\sqrt{p}$, our algorithm contrasts with 
recent results that have $\mathbf{O}\left(\kappa\log\left(1/\epsilon\right)\right)$ complexity to both oracles, which, while optimal in $O_A$, is highly suboptimal in $O_b$.

We achieve this using a new Variable Time Amplitude Amplification algorithm with Tunable thresholds (Tunable VTAA), which fully characterizes generic nested amplitude amplifications, eliminates redundant nestings, and is of independent interest. With an optimized schedule of thresholds, we prove that the complexity of Tunable VTAA scales with $\ell_{\frac{2}{3}}$-quasinorm of the input cost, improving over the $\ell_1$-norm result of Ambainis and the more common $\ell_2$-norm scaling.
Specialized to the quantum linear system problem, we construct a discretized inverse state, for which a deterministic amplification schedule exists. This leads to a substantially simplified VTAA with an optimal initial state preparation cost, even when the value of $p$ is not known {\emph{a priori}}.

We also introduce a \emph{block preconditioning} scheme that can artificially boost $\sqrt{p}$ in generic situations, in contrast to previous negative preconditioning results focusing on reducing $\kappa$. This further reduces the cost of initial state preparation in linear-system-based differential equation solvers, ground state preparators and eigenvalue processors. Additionally, block preconditioning furnishes a particularly simple quantum linear system algorithm with optimal $\mathbf{O}\left(\kappa\log\left(\frac{1}{\epsilon}\right)\right)$ queries to $O_A$ using $|b\rangle$ itself as the preconditioner. It also realizes a block-encoded eigenvalue transformer with $\mathbf{O}(n)$ scaling in degree of the target polynomial, compared to the best existing result of $\mathbf{O}\left(n^{1.5}\right)$.
\end{adjustwidth}
\end{abstract}
%%%%%%%%%%%%%%%%%%%%%%%%%%%%%%%%%%%%%%%%%%%%%%%%%%%%%%%%%%%%%%%%%%%%%%%%%%%%%%
\newpage
{
	\thispagestyle{empty}
	\renewcommand{\baselinestretch}{0.95}\normalsize
	\clearpage\tableofcontents
	\renewcommand{\baselinestretch}{1.}\normalsize
	\thispagestyle{empty}
}
\newpage

%%%%%%%%%%%%%%%%%%%%%%%%%%%%%%%%%%%%%%%%%%%%%%%%%%%%%%%%%%%%%%%%%%%%%%%%%%%%%%
\section{Introduction}
\label{sec:intro}
\input{intro.tex}

%%%%%%%%%%%%%%%%%%%%%%%%%%%%%%%%%%%%%%%%%%%%%%%%%%%%%%%%%%%%%%%%%%%%%%%%%%%%%%
\section{Preliminaries}
\label{sec:prelim}
\input{prelim.tex}

%%%%%%%%%%%%%%%%%%%%%%%%%%%%%%%%%%%%%%%%%%%%%%%%%%%%%%%%%%%%%%%%%%%%%%%%%%%%%%
\section{Tunable variable time amplitude amplification}
\label{sec:tunable}
\input{tunable.tex}

%%%%%%%%%%%%%%%%%%%%%%%%%%%%%%%%%%%%%%%%%%%%%%%%%%%%%%%%%%%%%%%%%%%%%%%%%%%%%%
\section{Discretized inverse state}
\label{sec:dinv}
\input{dinv.tex}

%%%%%%%%%%%%%%%%%%%%%%%%%%%%%%%%%%%%%%%%%%%%%%%%%%%%%%%%%%%%%%%%%%%%%%%%%%%%%%
\section{Solving linear systems}
\label{sec:lin}
\input{lin.tex}

%%%%%%%%%%%%%%%%%%%%%%%%%%%%%%%%%%%%%%%%%%%%%%%%%%%%%%%%%%%%%%%%%%%%%%%%%%%%%%
\section{Block preconditioning}
\label{sec:precond}
\input{precond.tex}

%%%%%%%%%%%%%%%%%%%%%%%%%%%%%%%%%%%%%%%%%%%%%%%%%%%%%%%%%%%%%%%%%%%%%%%%%%%%%%
\section{Discussion}
\label{sec:discuss}
\input{discuss.tex}

%%%%%%%%%%%%%%%%%%%%%%%%%%%%%%%%%%%%%%%%%%%%%%%%%%%%%%%%%%%%%%%%%%%%%%%%%%%%%%
\section*{Acknowledgements}
We thank Rolando Somma for helpful comments on an earlier draft.

%%%%%%%%%%%%%%%%%%%%%%%%%%%%%%%%%%%%%%%%%%%%%%%%%%%%%%%%%%%%%%%%%%%%%%%%%%%%%%
\appendix
\section{Axiomatic definition of variable time amplification}
\label{append:axiom}
\input{axiom.tex}

%%%%%%%%%%%%%%%%%%%%%%%%%%%%%%%%%%%%%%%%%%%%%%%%%%%%%%%%%%%%%%%%%%%%%%%%%%%%%%
\section{Tight bounds on the Dirichlet kernel}
\label{append:dirichlet}
\input{dirichlet.tex}

%%%%%%%%%%%%%%%%%%%%%%%%%%%%%%%%%%%%%%%%%%%%%%%%%%%%%%%%%%%%%%%%%%%%%%%%%%%%%%
\section{Hermitian qubitization}
\label{append:qubitization}
\input{qubitization.tex}

%%%%%%%%%%%%%%%%%%%%%%%%%%%%%%%%%%%%%%%%%%%%%%%%%%%%%%%%%%%%%%%%%%%%%%%%%%%%%%
\section{Gapped phase estimation with branch marking}
\label{append:gpe_bm}
\input{gpe_bm.tex}

%%%%%%%%%%%%%%%%%%%%%%%%%%%%%%%%%%%%%%%%%%%%%%%%%%%%%%%%%%%%%%%%%%%%%%%%%%%%%%
\section{Multiplicative approximation of success probabilities}
\label{append:multi}
\input{multi.tex}

%%%%%%%%%%%%%%%%%%%%%%%%%%%%%%%%%%%%%%%%%%%%%%%%%%%%%%%%%%%%%%%%%%%%%%%%%%%%%%
\clearpage
\bibliographystyle{myhamsplain2}
\bibliography{OptInitState.bib}

\end{document}

%% file: intro.tex
%%%%%%%%%%%%%%%%%%%%%%%%%%%%%%%%%%%%%%%%%%%%%%%%%%%%%%%%%%%%%%%%%%%%%%%%%%%%%%
\subsection{Quantum linear system algorithms}
\label{sec:intro_qlsa}

The problem of solving large systems of linear equations provides one of the most enticing sources of exponential speedups for quantum computers. Since the output of any quantum circuit can be simulated by inverting an appropriately chosen matrix, the quantum linear system problem is BQP-complete and captures the full potential of quantum computing---this was proved in the seminal work by Harrow, Hassidim and Lloyd~\cite{Harrow2009}, where the first quantum algorithm for sparse systems of linear equations was developed. As linear equations are ubiquitous in science, quantum linear system algorithms have found broad applications such as computing electromagnetic scattering~\cite{Scherer2017}, estimating electrical resistance of networks~\cite{Wang2017Resistence}, solving differential equations~\cite{Berry2017Differential,BerryCosta22,Krovi2023improvedquantum,Fang2023timemarchingbased,AnChildsLin23}, optimization, and more recently, processing eigenvalues of non-normal matrices~\cite{QEVP}.
While identifying exponential speedups in end-to-end scientific applications of quantum linear system solvers remains challenging, physically-motivated candidates such as computing Green's functions of quantum many-body systems~\cite{2021Yupreconditioned} may have well-characterized complexities.

Quantum linear system solvers also provide a natural reduction for designing more advanced quantum algorithms. One class of such algorithms aim to efficiently apply functions $f(A)$ to high-dimensional matrices accessed by a quantum computer through block encoding oracles. When the input is Hermitian, the target function can be readily implemented by the so-called Quantum Singular Value Transformation (QSVT)~\cite{Gilyen2018singular}, which since its conception has unified many existing quantum algorithms and uncovered new ones~\cite{2021MartynGrand}. The general case where $A$ is non-normal is relevant for various other applications such as solving differential equations and simulating transcorrelated quantum chemistry~\cite{McArdle20}, but the problem appears to be considerably more difficult.
One prior approach~\cite{takahira2021contour,Takahira2020QuantumCauchy,Fang2023timemarchingbased,2021Yupreconditioned} uses the Cauchy integral formula: $f(A)=\frac{1}{2\pi i}\int_{\mathcal{C}}\mathrm{d}z\ f(z)(zI-A)^{-1}$ for $\mathcal{C}$ a contour enclosing all eigenvalues of $A$. By discretizing the integral, $f(A)$ can be approximated by a linear combination of shifted matrix inverses. 
Alternative method employs a rational matrix generating function to create a history state encoding a polynomial basis of the input matrix in quantum superposition, which can be post-processed to yield the desired $f(A)$.
This technique was initially developed to create monomial basis for solving differential equations~\cite{Berry2017Differential}, but is more recently extended to handle Faber polynomials~\cite{QEVP} for approximating more general functions over the complex plane. In any event, these advanced algorithms all rely on quantum linear system solvers as a subroutine and would directly benefit from any enhancements to the solvers.

With an eye toward these applications, we now take a closer look at the cost of quantum linear system algorithms. Here, the goal is to produce a state proportional to the solution $A^{-1}\ket{b}$ of a linear system, using oracular queries to the coefficient matrix $A$ and the initial state $\ket{b}$.
The query complexity of quantum linear system algorithms has gradually improved over a long series of work. The seminal work of Harrow, Hassidim and Lloyd showed that $\mathbf{O}\left(\poly\left(\kappa,\frac{1}{\epsilon}\right)\right)$ queries suffice to produce the solution state with accuracy $\epsilon$~\cite{Harrow2009}, where
\begin{equation}
    \kappa=\norm{A}\norm{A^{-1}}
\end{equation}
is the spectral condition number of the target linear system. The scaling in $\kappa$ was later improved to almost linear by Ambainis' Variable Time Amplitude Amplification (VTAA) algorithm~\cite{Ambainis2012VTAA}. Subsequently, the error scaling is improved exponentially to $\mathbf{O}\left(\polylog\left(\frac{1}{\epsilon}\right)\right)$ by applying a linear combination of quantum walks within VTAA~\cite{Childs2015LinearSystems}. Recent work achieves a similar scaling in $\kappa$ and $\epsilon$ using alternative methods based on the adiabatic evolution and eigenstate filtering~\cite{Subasi2019QLSPadiabatic,An2022QSLP,Lin2020QLSPfiltering}. This ultimately culminated in an algorithm with $\mathbf{O}\left(\kappa\log\left(\frac{1}{\epsilon}\right)\right)$ scaling and constant success probability~\cite{Costa2021linearsystems}---this was claimed to be optimal, citing an unpublished lower bound of Harrow and Kothari, confirmed in~\cite[Appendix A]{Costa2023constant}. Even more recently, a simpler algorithm with the same $\mathbf{O}\left(\kappa\log\left(\frac{1}{\epsilon}\right)\right)$ scaling was designed by performing kernel reflection on an augmented linear system~\cite{Dalzell2024shortcut}.

There are two distinct query oracles involved in the quantum linear system problem: the first one is $O_A$ which block encodes the matrix $A$ to be inverted, and the second one is $O_b$ which prepares the initial state $\ket{b}$ from a standard reference state, chosen to be the computational basis state $\ket{0}$ without loss of generality. The query complexities cited above are all the worst-case combined cost where the two oracles $O_A$ and $O_b$ are treated on equal footing. However, the cost to $O_b$ can in fact be much lower than that to $O_A$, as we will now argue.

Consider the approach where we first construct a block encoding of $\frac{A^{-1}}{2\norm{A^{-1}}}$ (we have used $\norm{A^{-1}}$ in place of its known upper bound for presentational purpose). This can be achieved using QSVT with accuracy $\epsilon$ by making $\mathbf{O}\left(\kappa\log\left(\frac{1}{\epsilon}\right)\right)$ queries to $O_A$~\cite[Corollary 69]{Gilyen2018singular}. Applying it to the initial state then produces the unnormalized state $\frac{A^{-1}\ket{b}}{2\norm{A^{-1}}}$ with accuracy $\epsilon$ and success probability close to $\frac{p_{\text{succ}}}{4}$, consuming one query to $O_b$, where
\begin{equation}
    p_{\text{succ}}=\frac{\norm{A^{-1}\ket{b}}^2}{\norm{A^{-1}}^2},\qquad
    1\leq \frac{1}{\sqrt{p_{\text{succ}}}}=\frac{\norm{A^{-1}}}{\norm{A^{-1}\ket{b}}}\leq \kappa.
\end{equation}
This probability can be boosted close to unity with $\mathbf{O}\left(\frac{1}{\sqrt{p_{\text{succ}}}}\right)$ rounds of amplitude amplification, but the error of solution state is also amplified accordingly. So to achieve an accuracy $\epsilon$ in $\frac{A^{-1}\ket{b}}{\norm{A^{-1}\ket{b}}}$, it suffices to query the oracles a number of times scaling like
\begin{equation}
    \mathbf{O}\left(\frac{1}{\sqrt{p_{\text{succ}}}}\mathbf{Cost}(O_b)
    +\frac{\kappa}{\sqrt{p_{\text{succ}}}}\log\left(\frac{1}{\sqrt{p_{\text{succ}}}\epsilon}\right)\mathbf{Cost}\left(O_A\right)\right).
\end{equation}
Thus the number of queries to $O_b$ attains a strictly linear scaling with $\frac{1}{\sqrt{p_{\text{succ}}}}$, surpassing almost all previous results quoted in~\tab{compare} and can be much smaller than $\kappa$ in practice.
In fact, this scaling is already achievable by the algorithm of~\cite{Harrow2009} as is observed in~\cite{PRXQuantum.2.010315}.
Note also that the approach works even without a prior knowledge of $p_{\text{succ}}$, in which case we simply use the success probability lower bound $\alpha_{p_{\text{succ}}}\leq p_{\text{succ}}$.
Of course, this makes $\mathbf{O}\left(\kappa^2\log\left(\frac{\kappa}{\epsilon}\right)\right)$ queries to $O_A$ when $\frac{1}{\sqrt{p_{\text{succ}}}}\approx\kappa$, suffering from a quadratic slowdown in the worst case.
Methods based on the adiabatic theorem and kernel reflection technique have the scaling
\begin{equation}
    \mathbf{O}\left(\kappa\log\left(\frac{1}{\epsilon}\right)\mathbf{Cost}(O_b)
    +\kappa\log\left(\frac{1}{\epsilon}\right)\mathbf{Cost}(O_A)\right),
\end{equation}
where the complexity of initial state preparation becomes significantly worse.

To date, the \emph{one} and \emph{only one} method offering a close-to-optimal complexity simultaneously to both $O_b$ and $O_A$ is based on VTAA. The state-of-the-art VTAA algorithm makes $\mathbf{O}\left(\frac{1}{\sqrt{p_{\text{succ}}}}\log(\kappa)\right)$ queries to $O_b$, under the nontrivial assumption that the probabilities with which the algorithm potentially succeeds at intermediate stages are all estimated to a constant multiplicative accuracy~\cite{Chakraborty2018BlockEncoding}. Without this prior knowledge, the query complexity worsens to $\mathbf{O}\left(\frac{1}{\sqrt{p_{\text{succ}}}}\log^3(\kappa)\log\log(\kappa)\right)$, and the algorithm becomes considerably more complicated (although this is a one-time cost to pay only when VTAA is compiled into a quantum circuit).

Suppose one were willing to accept the logarithmic slowdown of VTAA and its substantial setup cost---would that improve the complexity of initial state preparation for existing quantum differential equation solvers, eigenvalue estimators and transformers? In solving differential equations, the norm $\norm{A^{-1}}\sim t$ scales close to linearly with the evolution time $t$, whereas $\norm{A^{-1}\ket{b}}\sim\sqrt{t}$. So the dependence on inverse success amplitude scales like $\sim\sqrt{t}$, but there exist methods with cost independent of the evolution time~\cite{Fang2023timemarchingbased,AnChildsLin23}. In estimating eigenvalues of non-normal matrices, the norm $\norm{A^{-1}}\sim \frac{1}{\epsilon}$ scales linearly with the inverse accuracy, whereas $\norm{A^{-1}\ket{b}}\sim\frac{1}{\sqrt{\epsilon}}$. So the dependence on inverse success amplitude scales like $\sim\frac{1}{\sqrt{\epsilon}}$, but again there exist methods with complexity independent of the inverse accuracy~\cite{Zhang2024Nonnormal} (albeit using a different input model). In all cases, one finds oneself in the intermediate regime where
\begin{equation}
    1\ll \frac{1}{\sqrt{p_{\text{succ}}}}=\frac{\norm{A^{-1}}}{\norm{A^{-1}\ket{b}}}
    \ll \kappa=\norm{A}\norm{A^{-1}},
\end{equation}
with a high query complexity of initial state preparation. This becomes problematic when the initial state itself is created by an expensive quantum subroutine, as could arise in natural applications to ground state preparation~\cite{berry2024rapid}, differential equations, and eigenvalue processing.

%%%%%%%%%%%%%%%%%%%%%%%%%%%%%%%%%%%%%%%%%%%%%%%%%%%%%%%%%%%%%%%%%%%%%%%%%%%%%%
\begin{table}[t]
    \centering
\resizebox{\textwidth}{!}{
    \begin{tabular}{cccccc}
\thickhline 
\multirow{1}{*}{Year} & \multirow{1}{*}{Algorithm} & \multirow{1}{*}{Primary innovation} & \multicolumn{1}{c}{Queries to $O_{b}$} & \multirow{1}{*}{Queries to $O_{A}$} \tabularnewline
\thickhline 
2008 & \cite{Harrow2009} & \makecell{Phase estimation\\+ Hamiltonian simulation} &\cellcolor{Gray} $\mathbf{O}(\frac{1}{\sqrt{p_{\text{succ}}}})$ & $\mathbf{O}\left(\frac{\kappa}{\sqrt{p_{\text{succ}}}\epsilon^{2}}\right)$ \tabularnewline
\hline 
2012 & \cite{Ambainis2012VTAA} & VTAA & $\mathbf{O}\left(\kappa\polylog\left(\frac{\kappa}{\epsilon}\right)\right)$ & $\mathbf{O}\left(\frac{\kappa}{\epsilon^{3}}\log^{3}\left(\frac{\kappa}{\epsilon}\right)\log^{2}\left(\frac{1}{\epsilon}\right)\right)$ \tabularnewline
\hline 
2017 & \cite{Childs2015LinearSystems} & Linear combination of quantum walks & $\mathbf{O}(\frac{1}{\sqrt{p_{\text{succ}}}}\log\left(\frac{\kappa}{\epsilon}\right))$ & $\mathbf{O}\left(\frac{\kappa}{\sqrt{p_{\text{succ}}}}\mathrm{polylog}\left(\frac{\kappa}{\epsilon}\right)\right)$ \tabularnewline
\hline
2017 & \cite{Childs2015LinearSystems} & \makecell{Gapped phase estimation\\+ Quantum walk + VTAA} & $\mathbf{O}(\kappa\polylog\left(\frac{\kappa}{\epsilon}\right))$ & $\mathbf{O}\left(\kappa\mathrm{polylog}\left(\frac{\kappa}{\epsilon}\right)\right)$ \tabularnewline
\hline 
2018 & \cite{Subasi2019QLSPadiabatic} & Randomized abiabatic evolution & $\mathbf{O}\left(\frac{\kappa}{\epsilon}\log(\kappa)\right)$ & $\mathbf{O}\left(\frac{\kappa}{\epsilon}\log(\kappa)\right)$ \tabularnewline
\hline 
2018 & \cite{Chakraborty2018BlockEncoding} & \cellcolor{Gray!50}Detailed analysis of VTAA & \multicolumn{1}{c}{$\mathbf{O}\left(\frac{1}{\sqrt{p_{\text{succ}}}}\log(\kappa)\right)$} & $\mathbf{O}\left(\kappa\log(\kappa)\log^{2}\left(\frac{\kappa}{\epsilon}\right)\right)$ \tabularnewline
\hline 
2019 & \cite{An2022QSLP} & Continuous time adiabatic evolution & $\mathbf{O}\left(\kappa\mathrm{polylog}\left(\frac{\kappa}{\epsilon}\right)\right)$ & $\mathbf{O}\left(\kappa\mathrm{polylog}\left(\frac{\kappa}{\epsilon}\right)\right)$ \tabularnewline
\hline 
2019 & \cite{Lin2020QLSPfiltering} & Eigenstate filtering & $\mathbf{O}\left(\kappa\log\left(\frac{\kappa}{\epsilon}\right)\right)$ & $\mathbf{O}\left(\kappa\log\left(\frac{\kappa}{\epsilon}\right)\right)$ \tabularnewline
\hline 
2022 & \cite{Costa2021linearsystems} & Discrete time adiabatic evolution & $\mathbf{O}\left(\kappa\log\left(\frac{1}{\epsilon}\right)\right)$ & \cellcolor{Gray} $\Theta\left(\kappa\log\left(\frac{1}{\epsilon}\right)\right)$ \tabularnewline
\hline 
2023 & \cite{Chakraborty2023VTAAQLSP} &\cellcolor{Gray!50}{QSVT-based GPE + VTAA}  & $\mathbf{O}\left(\frac{1}{\sqrt{p_{\text{succ}}}}\log(\kappa)\right)$ & $\mathbf{O}\left(\kappa\log(\kappa)\log\left(\frac{\kappa}{\epsilon}\right)\right)$ \tabularnewline
\hline 
2024 & \cite{Dalzell2024shortcut} & \makecell{Kernel reflection}  & $\mathbf{O}\left(\kappa\log\left(\frac{1}{\epsilon}\right)\right)$ & \cellcolor{Gray} $\Theta\left(\kappa\log\left(\frac{1}{\epsilon}\right)\right)$ \tabularnewline
\thickhline 
2024 & This work &\cellcolor{Gray!50} \colorbox{Gray!50}{\makecell{Tunable VTAA\\+ Discretized inverse state}} & \cellcolor{Gray}$\Theta\left(\frac{1}{\sqrt{p_{\text{succ}}}}\right)$ & $\mathbf{O}\left(\kappa\log\left(\frac{1}{\sqrt{p_{\text{succ}}}}\right)\log\left(\frac{\log(1/\sqrt{p_{\text{succ}}})}{\epsilon}\right)\right)$ \tabularnewline
\hline 
2024 & This work & \makecell{Block preconditioning}  & $\mathbf{O}\left(\kappa\log\left(\frac{1}{\epsilon}\right)\right)$ & \cellcolor{Gray} $\Theta\left(\kappa\log\left(\frac{1}{\epsilon}\right)\right)$ \tabularnewline
\thickhline 
    \end{tabular}
    }
    \caption{Complexity comparison of the new and previous methods for the quantum linear system problem, with optimal query complexities shaded dark gray. The present work achieves for the first time an optimal query cost of initial state preparation, among innovations shaded light gray with nearly optimal queries to both oracles, leading to a $\mathbf{\Theta}\left(\log(\kappa)\right)$ acceleration when $p_{\text{succ}}$ is known and a $\mathbf{\Theta}\left(\log^3(\kappa)\log\log(\kappa)\right)$ speedup when $p_{\text{succ}}$ is unknown.
    The block preconditioning technique further boosts $p_{\text{succ}}$ close to $1$ in applications to solving differential equations, preparing ground states of non-Hermitian operators and processing eigenvalues of non-normal matrices, and furnishes an extremely simple quantum algorithm with an optimal coefficient block encoding cost in the last row.
    }
    \label{tab:compare}
\end{table}
%%%%%%%%%%%%%%%%%%%%%%%%%%%%%%%%%%%%%%%%%%%%%%%%%%%%%%%%%%%%%%%%%%%%%%%%%%%%%%

%%%%%%%%%%%%%%%%%%%%%%%%%%%%%%%%%%%%%%%%%%%%%%%%%%%%%%%%%%%%%%%%%%%%%%%%%%%%%%
\subsection{Main result}
\label{sec:intro_result}
We develop faster quantum linear system algorithms, with a focus on improving the query complexity of initial state preparation. Our main result includes the following.
\begin{enumerate}[label=(\roman*)]
    \item \label{enum:opt_init} We develop a quantum linear system algorithm that makes $\Theta\left(\frac{1}{\sqrt{p_{\text{succ}}}}\right)$ queries to $O_b$ and\\ $\mathbf{O}\left(\kappa\log\left(\frac{1}{\sqrt{p_{\text{succ}}}}\right)\left(\log\log\left(\frac{1}{\sqrt{p_{\text{succ}}}}\right)+\log\left(\frac{1}{\epsilon}\right)\right)\right)$ queries to $O_A$. The strictly linear scaling with inverse success amplitude holds even when a multiplicative estimate of the solution norm is unavailable, where $p_{\text{succ}}$ is replaced by its lower bound, in contrast to all previous approaches where finding such an estimate incurs at least a polylogarithmic overhead. The optimal lower bound to $O_b$ is proven through a reduction to Grover search.
    \item \label{enum:precond} We show how one can artificially boost $p_{\text{succ}}$ close to $1$. We demonstrate that this leads to differential equation solvers, ground state preparators and eigenvalue processors with improved complexities of initial state preparation, matching or outperforming the state of the art.
\end{enumerate}
Result \ref{enum:opt_init} is realized by a substantially simplified VTAA that uses a deterministic amplification schedule, whose prerequisites are fulfilled by a discretized inverse state we design, giving a  $\mathbf{\Theta}\left(\log(\kappa)\right)$ improvement in the cost of initial state preparation when $p_{\text{succ}}$ is known and a $\mathbf{\Theta}\left(\log^3(\kappa)\log\log(\kappa)\right)$ speedup when $p_{\text{succ}}$ is unknown.
Result \ref{enum:precond} is based on a block preconditioning technique that amplifies a subspace flagged by the initial ancilla state using no oracular queries.

Additionally, we find that the block preconditioning technique can be more broadly applied to improve other complexity scalings of quantum linear system solvers.
Specifically, by choosing initial state $\ket{b}$ itself as the preconditioner, we develop a straightforward quantum linear system algorithm that makes $\mathbf{O}\left(\kappa\log\left(\frac{1}{\epsilon}\right)\right)$ queries to both $O_A$ and $O_b$. Our method appears to be conceptually even simpler than the recent kernel-reflection approach~\cite{Dalzell2024shortcut}.
We also develop a block-encoded quantum eigenvalue transformation algorithm with $\mathbf{O}(n)$ scaling in degree of the target polynomial, by applying preconditioning within the block-encoded version of quantum linear system algorithm, whereas the best previous result was $\mathbf{O}(n^{1.5})$~\cite[Section 5.2]{QEVP}.
We illustrate our results in \fig{diagram} and tabulate them in \tab{compare} and \tab{precond} against previous results.

%%%%%%%%%%%%%%%%%%%%%%%%%%%%%%%%%%%%%%%%%%%%%%%%%%%%%%%%%%%%%%%%%%%%%%%%%%%%%%
\begin{figure}[t]
	\centering
\includegraphics[width=0.8\textwidth]{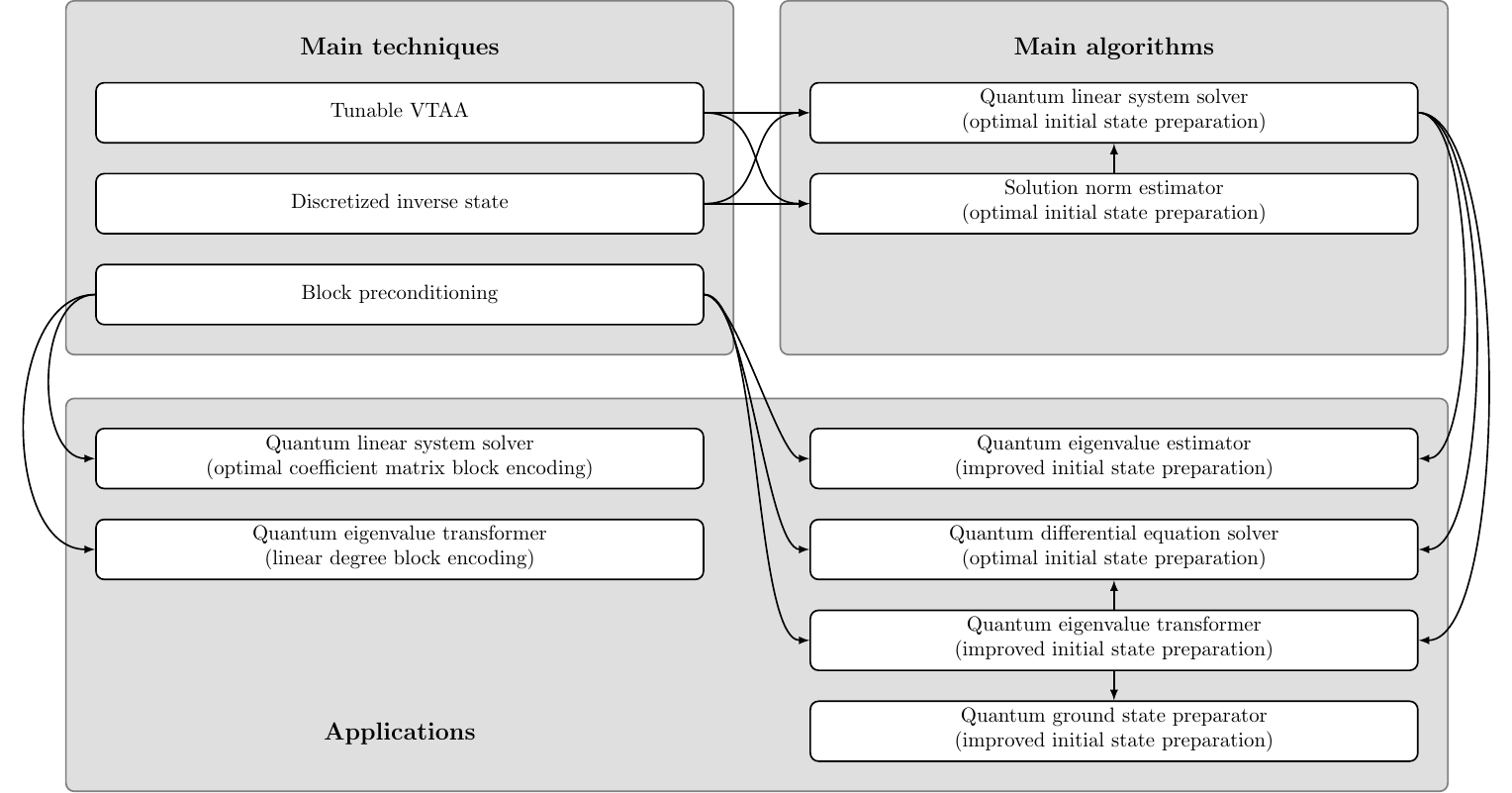}
\caption{A diagrammatic illustration of the main result and its applications.}
\label{fig:diagram}
\end{figure}
%%%%%%%%%%%%%%%%%%%%%%%%%%%%%%%%%%%%%%%%%%%%%%%%%%%%%%%%%%%%%%%%%%%%%%%%%%%%%%

%%%%%%%%%%%%%%%%%%%%%%%%%%%%%%%%%%%%%%%%%%%%%%%%%%%%%%%%%%%%%%%%%%%%%%%%%%%%%%
\subsection{Tunable variable time amplitude amplification}
\label{sec:intro_tunable}
Our first main technical contribution is the analysis of a generic version of VTAA in \sec{tunable}, with adjustable amplification thresholds that can be tuned/optimized for different input algorithms and initial states. We name this method \emph{Tunable VTAA}.

Before discussing Tunable VTAA in detail, let us first review the VTAA framework introduced by Ambainis~\cite{Ambainis2012VTAA}. The goal here is to construct a quantum state by applying a sequence of \emph{input algorithms} $A_1,\ldots,A_m$ on a starting state $\ket{\psi_0}$ incrementally, where the desired state resides in a subspace flagged by the \emph{flag projection} $\overline{\Pi_b}=I-\Pi_b$. If we apply the quantum algorithms naively, then the state is prepared with probability
\begin{equation}
    p_{\text{succ}}=\norm{\overline{\Pi_b}A_m\cdots A_1\ket{\psi_0}}^2.
\end{equation}
To prepare it with a probability close to unity, we perform $\mathbf{O}\left(\frac{1}{\sqrt{p_{\text{succ}}}}\right)$ rounds of amplitude amplification, leading to a total cost of
\begin{equation}
\label{eq:naive_amp_cost}
    \mathbf{O}\left(\frac{1}{\sqrt{p_{\text{succ}}}}\mathbf{Cost}\left(\ket{\psi_0}\right)+\frac{1}{\sqrt{p_{\text{succ}}}}\sum_{j=1}^m\mathbf{Cost}\left(A_j\right)\right).
\end{equation}

This is a worst-case scenario as all input algorithms $A_1,\ldots,A_m$ are treated on equal footing and get invoked the same number of times in the final stage. VTAA avoids this worst-case cost by performing amplitude amplification at intermediate stages. Formally, we have $m+1$ orthogonal projections all commuting with $\Pi_b$ and partially ordered according to their positive semidefinitess as $0=\Pi_0\leq\Pi_1\leq\cdots\leq\Pi_m=I$. These \emph{clock projections} are used to capture that quantum algorithms can terminate early at or before intermediate stages, i.e., we have $A_j\Pi_{j-1}=\Pi_{j-1}$ for all $j=1,\ldots,m$. Then instead of deferring the amplification to the very end, we perform $r_j$ rounds of amplitude amplification at stage $j$ toward the potentially good subspace flagged by $\overline{\Pi_j\Pi_b}=I-\Pi_j\Pi_b$. This gives \emph{amplified algorithms} $\widetilde A_j
    =\left(-\left(I-2A_j\widetilde A_{j-1}\ketbra{\psi_0}{\psi_0}\widetilde A_{j-1}^\dagger A_j^\dagger\right)\left(I-2\overline{\Pi_j\Pi_b}\right)\right)^{r_j}A_j\widetilde A_{j-1}$ with $\widetilde A_0=I$.
We call $2r_j+1$ the \emph{amplification schedule} or \emph{amplification step numbers}.

To ensure that $\widetilde A_m$ produces the desired state with a sufficiently large probability, one needs to choose $r_j$ carefully to avoid overshoot and undershoot. In the work of~\cite{Ambainis2012VTAA} and many of its followups, the schedule is selected as follows:
set $r_j$ to be the smallest nonnegative integer satisfying $(2r_j+1)\norm{\overline{\Pi_j\Pi_b}A_j\widetilde A_{j-1}\ket{\psi_0}}
        \geq\frac{1}{3\sqrt{m}}$.
Then using the state-of-the-art VTAA analysis~\cite[Lemma 10]{Chakraborty2023VTAAQLSP}, this gives an algorithm with query complexity
\begin{equation}
\label{eq:vtaa_cost}
    \mathbf{O}\left(\frac{\sqrt{m}}{\sqrt{p_{\text{succ}}}}\mathbf{Cost}\left(\ket{\psi_0}\right)
    +\frac{\sqrt{m}}{\sqrt{p_{\text{succ}}}}\sum_{j=1}^m\norm{\overline{\Pi_{j-1}\Pi_b}A_{j-1}\cdots A_1\ket{\psi_0}}\mathbf{Cost}\left(A_j\right)\right).
\end{equation}
Compared to the naive complexity \eq{naive_amp_cost}, VTAA makes fewer queries to the input algorithms by balancing the summation of cost.
Note the above discussion implicitly assumes that the amplitudes $\norm{\overline{\Pi_j\Pi_b}A_j\widetilde A_{j-1}\ket{\psi_0}}$ are known a prior. This assumption is only for presentational purpose: asymptotic scaling of the complexity remains the same if we instead have a constant multiplicative estimate of the amplitudes. Otherwise, we need to perform amplitude estimation to compute them to a constant multiplicative accuracy, incurring a substantial overhead---although this only needs to be done once during compilation.

In contrast, our Tunable VTAA uses the following schedule:
set $r_j$ to be the smallest nonnegative integer satisfying $(2r_j+1)\norm{\overline{\Pi_j\Pi_b}A_j\widetilde A_{j-1}\ket{\psi_0}}
        \geq\frac{1}{3}\sqrt{\alpha_j}$.
Here $\alpha_j$ are \emph{amplification thresholds} satisfying the technical conditions $0\leq\alpha_j\leq1$ and $\sum_{j=1}^m\alpha_j=\mathbf{O}(1)$.
Obviously, Tunable VTAA offers more flexibility as $\alpha_j$ can now be tuned/optimized for different input algorithms and initial states, although it may not be immediately apparent how much more powerful this is over previous results. We prove the following.

\begin{enumerate}[label=(\roman*)]
\item \label{enum:universality} 
Tunable VTAA consists of nested amplitude amplifications with  $\norm{\overline{\Pi_j\Pi_b}A_j\widetilde A_{j-1}\ket{\psi_0}}\leq\frac{1}{2r_j+1}$ and constant \emph{loss factor} $\prod_{j=1}^m\frac{\norm{\overline{\Pi_j\Pi_b}\widetilde A_j\ket{\psi_0}}}{(2r_j+1)\norm{\overline{\Pi_j\Pi_b}A_j\widetilde A_{j-1}\ket{\psi_0}}}=\mathbf{\Omega}(1)$.
Conversely, any nested amplitude amplifications with no overshoot $\norm{\overline{\Pi_j\Pi_b}A_j\widetilde A_{j-1}\ket{\psi_0}}\leq\frac{1}{3(2r_j+1)}$ and constant loss factor $\prod_{j=1}^m\frac{\norm{\overline{\Pi_j\Pi_b}\widetilde A_j\ket{\psi_0}}}{(2r_j+1)\norm{\overline{\Pi_j\Pi_b}A_j\widetilde A_{j-1}\ket{\psi_0}}}=\mathbf{\Omega}(1)$ is a Tunable VTAA.
\item \label{enum:nontrivial} Nontrivial amplifications $r_j\geq1$ in Tunable VTAA happen only at $l$ stages $1\leq s_1\leq\cdots\leq s_{l}\leq m$, where
\begin{equation}
    l=\mathbf{O}\left(\log\left(\frac{1}{\sqrt{p_{\text{succ}}}}\right)\right).
\end{equation}
Under the convention $s_{l+1}=m+1$ and $A_{m+1}=I$, Tunable VTAA has query complexity
\begin{small}
\begin{equation}
\newmaketag
\hspace{-1.cm}
    \mathbf{O}\left(\frac{1}{\sqrt{p_{\text{succ}}}}\mathbf{Cost}(A_{s_1}\cdots A_1\ket{\psi_0})
    +\frac{1}{\sqrt{p_{\text{succ}}}}\sum_{v=2}^{l+1}\frac{1}{\sqrt{\alpha_{s_{v-1}}}}
    \norm{\overline{\Pi_{s_{v-1}}\Pi_b}A_{s_{v-1}}\cdots A_1\ket{\psi_0}}\mathbf{Cost}\left(A_{s_v}\cdots A_{s_{v-1}+1}\right)\right).
\end{equation}
\end{small}%
\item \label{enum:quasinorm} Pre-merging trivial stages and using thresholds
\begin{equation}
    \alpha_{s_{v-1}}\propto
    \left(\norm{\overline{\Pi_{s_{v-1}}\Pi_b}A_{s_{v-1}}\cdots A_1\ket{\psi_0}}\mathbf{Cost}\left(A_{s_v}\cdots A_{s_{v-1}+1}\right)\right)^{\frac{2}{3}},
\end{equation}
Tunable VTAA attains the \emph{$\ell_{\frac{2}{3}}$-quasinorm} scaling
\begin{small}
\begin{equation}
\newmaketag
\hspace{-1.1cm}
    \mathbf{O}\left(\frac{1}{\sqrt{p_{\text{succ}}}}\mathbf{Cost}(A_{s_1}\cdots A_1\ket{\psi_0})
    +\frac{1}{\sqrt{p_{\text{succ}}}}
    \left(\sum_{v=2}^{l+1}\left(\norm{\overline{\Pi_{s_{v-1}}\Pi_b}A_{s_{v-1}}\cdots A_1\ket{\psi_0}}\mathbf{Cost}\left(A_{s_v}\cdots A_{s_{v-1}+1}\right)\right)^{\frac{2}{3}}\right)^{\frac{3}{2}}\right).
\end{equation}
\end{small}%
\end{enumerate}

Let us interpret this result. Property \ref{enum:universality} shows that Tunable VTAA is universal, in the sense that it captures the full power of a generic nested amplitude amplification. Stated differently, to improve the complexity of a nested amplification, it suffices to optimize over the threshold values. Such an exhaustive optimization can be computationally demanding, but one may find suboptimal thresholds that have analytic forms. In any event, it is this universality that motivates us to examine Tunable VTAA in greater detail. 
Property \ref{enum:nontrivial} demonstrates that the number of nontrivial amplification stages is at most logarithmic in the inverse success amplitude. Thus in the regime where $\log\left(\frac{1}{\sqrt{p_{\text{succ}}}}\right)\ll m$, majority of the input algorithms can be pre-merged and VTAA can be significantly simplified. 
Our query complexity then follows from a tightened analysis of VTAA with a $\sqrt{m}$ factor improvement, taking into account the fact that trivial stages can be pre-merged. To compile this into a quantum algorithm, one would naively estimate all $\norm{\overline{\Pi_j\Pi_b}A_j\widetilde A_{j-1}\ket{\psi_0}}$ using amplitude estimation and compare with threshold values to determine the schedule, but this is not really necessary if the thresholds $\alpha_j$ are chosen analytically, as we will demonstrate for the quantum linear system problem. Property \ref{enum:quasinorm} follows by minimizing the cost of Tunable VTAA under the constraints $0\leq\alpha_{s_{v-1}}\leq1$ and $\sum_{v=2}^{l+1}\alpha_{s_{v-1}}=\mathbf{O}(1)$. This $\ell_{\frac{2}{3}}$-quasinorm scaling is lower than the $\ell_1$- or $\ell_2$-norm result from previous work as
\begin{equation}
\begin{aligned}
    &\left(\sum_{v=2}^{l+1}\left(\norm{\overline{\Pi_{s_{v-1}}\Pi_b}A_{s_{v-1}}\cdots A_1\ket{\psi_0}}\mathbf{Cost}\left(A_{s_v}\cdots A_{s_{v-1}+1}\right)\right)^{\frac{2}{3}}\right)^{\frac{3}{2}}\\
    &\leq\sqrt{l}\sum_{v=2}^{l+1}\norm{\overline{\Pi_{s_{v-1}}\Pi_b}A_{s_{v-1}}\cdots A_1\ket{\psi_0}}\mathbf{Cost}\left(A_{s_v}\cdots A_{s_{v-1}+1}\right).
\end{aligned}
\end{equation}
Again, to achieve this ``ultimate performance'' of VTAA with a quantum algorithm, one needs to estimate all potentially good amplitudes $\norm{\overline{\Pi_{j-1}\Pi_b}A_{j-1}\cdots A_1\ket{\psi_0}}$ with logarithmic overhead, so this is more advantageous when the algorithm may be repeatedly invoked after compilation.

%%%%%%%%%%%%%%%%%%%%%%%%%%%%%%%%%%%%%%%%%%%%%%%%%%%%%%%%%%%%%%%%%%%%%%%%%%%%%%
\subsection{Discretized inverse state}
\label{sec:intro_dinv}
For the quantum linear system problem, our goal is to produce a normalized version of the solution $A^{-1}\ket{b}$.
Here, the coefficient matrix is block encoded as $A/\alpha_A$ with normalization factor $\alpha_A\geq\norm{A}$, and a norm upper bound on its inverse $\alpha_{A^{-1}}\geq\norm{A^{-1}}$ is known a prior, giving $\kappa=\alpha_A\alpha_{A^{-1}}$. 
By considering the Hermitian dilation $\ketbra{0}{1}\otimes A+\ketbra{1}{0}\otimes A^\dagger$, we may without loss of generality assume that $A$ itself is Hermitian and has the spectral decomposition $A=\sum_u\lambda_u\ketbra{\phi_u}{\phi_u}$, where $\frac{1}{\alpha_{A^{-1}}}\leq\abs{\lambda_u}\leq\alpha_A$. Meanwhile, the initial state can be expanded in the eigenbasis as $\ket{b}=\sum_u\gamma_u\ket{\phi_u}$, so a direct application of QSVT produces an unnormalized state close to
\begin{equation}
    \frac{A^{-1}}{2\alpha_{A^{-1}}}\ket{b}
    =\sum_u\frac{1}{2\alpha_{A^{-1}}\lambda_u}\gamma_u\ket{\phi_u},\qquad
    \frac{\norm{A^{-1}\ket{b}}}{\alpha_{A^{-1}}}
    =\frac{1}{\alpha_{A^{-1}}}\sqrt{\sum_u\abs{\frac{\gamma_u}{\lambda_u}}^2}
    =\sqrt{p_{\text{succ}}}.
\end{equation}

This can be improved by VTAA, which operates on an enlarged Hilbert space and produces an intermediate state of the form~\cite{Childs2015LinearSystems}
\begin{equation}
\begin{aligned}
    \psi_{\text{inv}}&=
    \sum_{k=0}^{m-1}\sum_{\abs{\frac{\lambda_u}{\alpha_A}}\in\big[\frac{1}{2^{k+1}},\frac{1}{2^k}\big)}
    \frac{1}{\alpha_{A^{-1}}\lambda_u}
    \left(\zeta_{k+1,u}\ket{k}+\zeta_{k,u}\ket{k-1}\right)\gamma_u\ket{\phi_u},\\
    \norm{\psi_{\text{inv}}}&=\frac{1}{\alpha_{A^{-1}}}\sqrt{\sum_{k=0}^{m-1}\sum_{\abs{\frac{\lambda_u}{\alpha_A}}\in\big[\frac{1}{2^{k+1}},\frac{1}{2^k}\big)}
    \left(\abs{\zeta_{k+1,u}}^2+\abs{\zeta_{k,u}}^2\right)
    \abs{\frac{\gamma_u}{\lambda_{u}}}^2}
    =\sqrt{p_{\text{succ,inv}}}.
\end{aligned}
\end{equation}
Here, the ancilla register can be thought of as a ``clock register'' holding $m=\mathbf{Ceil}\left(\log_2\left(\alpha_A\alpha_{A^{-1}}\right)\right)=\mathbf{\Theta}(\log(\kappa))$ values, and $\abs{\zeta_{k+1,u}}^2+\abs{\zeta_{k,u}}^2\lessapprox1$ are close to being normalized for all $u$, implying $p_{\text{succ,inv}}=\mathbf{\Theta}(p_{\text{succ}})$. This intermediate state can be created using VTAA with a constant success probability. Specifically, one chooses the clock projections to be $\Pi_j=\sum_{k=0}^{j-1}\ketbra{k}{k}$ for $j=1,\ldots,m$, 
satisfying
$0\leq\Pi_1\leq\Pi_2\leq\cdots\leq\Pi_m=I$. 
Within the VTAA stage $k$, one performs a procedure called Gapped Phase Estimation (GPE) to approximately check whether the eigenvalues are in $\abs{\frac{\lambda_u}{\alpha_A}}\in\big[\frac{1}{2^{k}},1\big)$ with a certain confidence level, inverting the matrix if the condition is satisfied and terminating the current branch of computation afterward. Finally, the clock register can be uncomputed by running GPE in reverse (without VTAA), resulting in the desired $A^{-1}\ket{b}$.

Our second main technical contribution, to be detailed in \sec{dinv}, is the replacement of above state by the \emph{discretized inverse state}
\begin{equation}
\begin{aligned}
    \psi_{\text{d-inv}}&=
    \sum_{k=0}^{m-1}\sum_{\abs{\frac{\lambda_u}{\alpha_A}}\in\big[\frac{1}{\rho^{k+1}},\frac{1}{\rho^k}\big)}
    \left(\zeta_{k+1,u}\frac{\rho^{k+1}}{\rho^m}\ket{k}+\zeta_{k,u}\frac{\rho^k}{\rho^m}\ket{k-1}\right)\gamma_u\ket{\phi_u},\\
    \norm{\psi_{\text{d-inv}}}&=\sqrt{\sum_{k=0}^{m-1}\sum_{\abs{\frac{\lambda_u}{\alpha_A}}\in\big[\frac{1}{\rho^{k+1}},\frac{1}{\rho^k}\big)}
    \left(\abs{\zeta_{k+1,u}\frac{\rho^{k+1}}{\rho^m}}^2+\abs{\zeta_{k,u}\frac{\rho^k}{\rho^m}}^2\right)\abs{\gamma_u}^2}=\sqrt{p_{\text{succ,d-inv}}},
\end{aligned}
\end{equation}
where $\rho$ is an odd integer which can be chosen as $\rho=3$ without loss of generality and $m=\mathbf{Ceil}\left(\log_3\left(\alpha_A\alpha_{A^{-1}}\right)\right)=\mathbf{\Theta}(\log(\kappa))$. 
Note that we have switched from a base-$2$ partition of the eigenvalues to an odd-number base partition. 
This change plays the central role in achieving the optimal initial state preparation cost while substantially simplifying the structure of VTAA. We will return to this point momentarily.
The second change is we replace the eigenvalue inversion $\frac{1}{\lambda_u}$ by the discrete values $\frac{3^{k+1}}{3^m}$ and $\frac{3^k}{3^m}$, which can be introduced using only controlled single-qubit rotations. This modification frees the expensive matrix inversion from any nested amplitude amplifications, further simplifying the algorithm. Otherwise, we still perform GPE in each stage $k$ which has a cost of $\mathbf{O}\left(3^k\log\left(\frac{1}{\epsilon_{\text{gpe}}}\right)\right)$.
Since the eigenvalues are discretized with constant multiplicative accuracy, we have $p_{\text{succ,d-inv}}=\mathbf{\Theta}(p_{\text{succ}})$.

We now describe a simple VTAA schedule. We first recall from Property \ref{enum:nontrivial} that nontrivial amplifications happen only at $\mathbf{O}\left(\log\left(\frac{1}{\sqrt{p_{\text{succ}}}}\right)\right)$ stages of VTAA, and majority of the algorithms can be pre-merged. In our case, the complexity of GPE $\sim3^k$ increases exponentially in $k$, so intuitively one should pre-merge the initial stages with lower cost. We rigorize this idea by showing that the asymptotic complexity of VTAA remains unaffected after pre-merging the first $m-l$ stages, as long as $l=\mathbf{\Omega}\left(\log_3\left(\frac{1}{\sqrt{p_{\text{succ}}}}\right)\right)$. 
For the quantum linear system problem, the $\ell_{\frac{2}{3}}$-quasinorm from Property \ref{enum:quasinorm} can be difficult to evaluate directly, so we relax it to the $\ell_2$-norm, which is attained by setting the amplification thresholds $\alpha_k\propto 9^{k}\norm{\overline{\Pi_k\Pi_b}A_k\cdots A_1\ket{\psi_0}}^2$. These thresholds can be normalized to $\alpha_k=\mathbf{\Theta}\left(\frac{9^{k}\norm{\overline{\Pi_k\Pi_b}A_k\cdots A_1\ket{\psi_0}}^2}{9^m p_{\text{succ}}}\right)$ to fulfill the requirement $\sum_k\alpha_k=\mathbf{O}(1)$. With the $\ell_2$-norm cost, Tunable VTAA makes $\mathbf{O}\left(l3^m\log\left(\frac{1}{\epsilon_{\text{gpe}}}\right)\right)$ queries to the block encoding oracle $O_A$, which can be minimized by reducing the value of $l$. But recall that we need it sufficiently large due to pre-merging, so we should choose $l=\mathbf{\Theta}\left(\log_3\left(\frac{1}{\sqrt{p_{\text{succ}}}}\right)\right)$. This implies $p_{\text{succ}}=\mathbf{\Theta}\left(9^{-l}\right)$ and hence the thresholds
\begin{equation}
\label{eq:deterministic_threshold}
    \alpha_k=
    \begin{cases}
        \mathbf{\Theta}\left(9^{k-m+l}\norm{\overline{\Pi_k\Pi_b}A_k\cdots A_1\ket{\psi_0}}^2\right),\quad&k=m-l+1,\ldots,m,\\
        0,&k=1,\ldots,m-l.
    \end{cases}
\end{equation}

Property \ref{enum:universality} provides the theoretical guarantee that the above Tunable VTAA can be translated into a nested amplitude amplification.
To realize this,
one would naively perform amplitude estimation to determine all thresholds $\alpha_k$ to a constant multiplicative accuracy, which introduces polylogarithmic overhead and ruins the optimal scaling of initial state preparation. But this is unnecessary overkill. Suppose we have adjusted the constant factor so that $\alpha_k=c^29^{k-m+l}\norm{\overline{\Pi_k\Pi_b}A_k\cdots A_1\ket{\psi_0}}^2$ for some $c\gtrapprox 1$ (say $c=1.001$).
Then at the first nontrivial stage $k=m-l+1$, we would compare $\norm{\overline{\Pi_{m-l+1}\Pi_b}A_{m-l+1}\cdots A_1\ket{\psi_0}}$ with $\frac{1}{3}\sqrt{\alpha_{m-l+1}}=c\norm{\overline{\Pi_{m-l+1}\Pi_b}A_{m-l+1}\cdots A_1\ket{\psi_0}}$. We find that our current amplitude falls just short of the threshold, concluding that we need $2r_{m-l+1}+1=3$ steps of amplification. After this stage, we have approximately amplified the amplitude by $3$. At the next stage $k=m-l+2$, our current amplitude $\norm{\overline{\Pi_{m-l+2}\Pi_b}A_{m-l+2}\widetilde A_{m-l+1}\ket{\psi_0}}\lessapprox3\norm{\overline{\Pi_{m-l+2}\Pi_b}A_{m-l+2}\cdots A_1\ket{\psi_0}}$ is slightly below $\frac{1}{3}\sqrt{\alpha_{m-l+2}}=c3\norm{\overline{\Pi_{m-l+2}\Pi_b}A_{m-l+2}\cdots A_1\ket{\psi_0}}$, so we conclude once more we need $2r_{m-l+2}+1=3$ steps of amplification, which boosts the amplitude again by $3$ up to a loss factor.
Proceeding in this way, we obtain the following simple schedule
\begin{equation}
\label{eq:deterministic_schedule}
    2r_k+1=
    \begin{cases}
        3,\quad &k=m-l+1,\ldots,m,\\
        1,&k=1,\ldots,m-l.
    \end{cases}
\end{equation}
Note that this \emph{deterministic schedule} is derived entirely from the definition of Tunable VTAA, and the derivation requires \emph{no} amplitude estimation.
The amplification schedule \eq{deterministic_schedule} is reminiscent of another one used in the recent study of quantum search~\cite{SimpleSearch23} (see also~\cite{BoundedErrorInputs,NestedSearch,chakraborty_et_al:LIPIcs.STACS.2022.20}). This is not a coincidence. For the quantum search algorithm of~\cite{SimpleSearch23}, the target cost function is $\ell_2$-norm of the input cost, which grows exponentially in the stage number like $\sim 3^k$. Following an analysis similar to above, one can show that Tunable VTAA reproduces the search algorithm developed in~\cite{SimpleSearch23}.

With the deterministic schedule, GPE is invoked at most $m3^l$ times. To achieve accuracy $\epsilon$, we choose $\epsilon_{\text{gpe}}=\frac{\epsilon}{m3^l}$, giving the cost
$\mathbf{O}\left(\frac{1}{\sqrt{p_{\text{succ}}}}\mathbf{Cost}(O_b)
    +\kappa\log\left(\frac{1}{\sqrt{p_{\text{succ}}}}\right)\log\left(\frac{\log(\kappa)}{\sqrt{p_{\text{succ}}}\epsilon}\right)\mathbf{Cost}(O_A)\right)$.
By using a non-uniform error distribution in the first $m-l$ stages and projecting error onto the potentially good subspaces, we further improve this result to
\begin{equation}
    \mathbf{O}\left(\frac{1}{\sqrt{p_{\text{succ}}}}\mathbf{Cost}(O_b)
    +\kappa\log\left(\frac{1}{\sqrt{p_{\text{succ}}}}\right)\log\left(\frac{\log\left(\frac{1}{\sqrt{p_{\text{succ}}}}\right)}{\epsilon}\right)\mathbf{Cost}(O_A)\right).
\end{equation}
This achieves an optimal cost of initial state preparation with a substantially simplified VTAA schedule \eq{deterministic_schedule}. We summarize this algorithm in \thm{qls_opt_init} of \sec{lin} and prove the optimality in \thm{qls_lower}.

Note that our above discussion assumes that we have a constant multiplicative approximation of $\sqrt{p_{\text{succ}}}$, or equivalently, the \emph{solution norm} $\norm{A^{-1}\ket{b}}$. If this information is not available, then we develop a solution norm estimation algorithm in \thm{sol_est} that estimates $\norm{A^{-1}\ket{b}}$ to a constant multiplicative accuracy, by running VTAA with pre-merging parameter $l=0,1,2,\ldots$ that terminates at $l=\mathbf{\Theta}\left(\log_3\left(\frac{1}{\sqrt{p_{\text{succ}}}}\right)\right)$ with high probability. This costs $\mathbf{O}\left(\frac{1}{\sqrt{\alpha_{p_{\text{succ}}}}}\right)$ queries to the initial state preparation oracle with $\alpha_{p_{\text{succ}}}$ a lower bound on the success probability, again achieving the strictly linear scaling.

%%%%%%%%%%%%%%%%%%%%%%%%%%%%%%%%%%%%%%%%%%%%%%%%%%%%%%%%%%%%%%%%%%%%%%%%%%%%%%
\subsection{Block preconditioning}
\label{sec:intro_precond}
As aforementioned, when applying the quantum linear system solvers, we are often in the intermediate regime where $1\ll \frac{1}{\sqrt{p_{\text{succ}}}}=\frac{\alpha_{A^{-1}}}{\norm{A^{-1}\ket{b}}}
    \ll \kappa=\alpha_{A}\alpha_{A^{-1}}$. 
Thus, the cost of querying $O_b$ (which depends on $\frac{1}{\sqrt{p_{\text{succ}}}}$) can be quite significant even with our optimal scaling algorithm. We now describe the third main contribution, a simple block preconditioning technique, that artifically boosts the solution norm $\norm{A^{-1}\ket{b}}$ and thereby reduces $\frac{1}{\sqrt{p_{\text{succ}}}}$ and hence the complexity of initial state preparation. In fact, for solving differential equations, preparing ground states and processing eigenvalues, the initial state query cost can even be reduced close to constant (in the evolution time, inverse spectral gap and degree of the input polynomial). See \sec{precond} for further details.

%%%%%%%%%%%%%%%%%%%%%%%%%%%%%%%%%%%%%%%%%%%%%%%%%%%%%%%%%%%%%%%%%%%%%%%%%%%%%%
\begin{table}[t]
    \centering
\resizebox{\textwidth}{!}{
    \begin{tabular}{cccccc}
    \thickhline
    Application & Improvement & Preconditioner & Scaling without preconditioning & Scaling with preconditioning\\
    \thickhline
    \makecell{Linear system solver\\\thm{qls_opt_coeff}} & Input matrix & Initial state
    & $\frac{\kappa}{\sqrt{p_{\text{succ}}}}\log\left(\frac{1}{\sqrt{p_{\text{succ}}}\epsilon}\right)$
    & $\kappa\log\left(\frac{1}{\epsilon}\right)$
    \tabularnewline
\hline 
    \makecell{Differential equation solver\\\thm{diff_eq_init}} & Initial state & Ancilla state
    & $\frac{t}{\norm{e^{tA}\ket{b}}}\log\left(\frac{1}{\norm{e^{tA}\ket{b}}\epsilon}\right)\log\left(\frac{t}{\epsilon}\right)$
    & $\frac{1}{\norm{e^{tA}\ket{b}}}$
    \tabularnewline
\hline 
    \makecell{Eigenvalue estimator\\\thm{qeve_init}} & Initial state & Ancilla state 
    & $\frac{\kappa_S}{\epsilon}$
    & $\kappa_S$
    \tabularnewline
\hline 
    \makecell{Eigenvalue transformer\\\thm{qevt_init}} & Initial state & Ancilla state
    & $\frac{\kappa_S^2n\log\left(n\right)}{\norm{p\left(\frac{A}{\alpha_A}\right)\ket{\psi}}}
    \log\left(
    \frac{\kappa_S\log(n)}{\norm{p\left(\frac{A}{\alpha_A}\right)\ket{\psi}}\epsilon}\right)$
    & $\frac{\kappa_S\log(n)}{\norm{p\left(\frac{A}{\alpha_A}\right)\ket{\psi}}}$
    \tabularnewline
\hline 
    \makecell{Ground state preparator\\\thm{gs_init}} & Initial state & Ancilla state
    & $\frac{\kappa_S^2}{|\gamma_0|\delta_A}\log^2\left(\frac{\kappa_S}{|\gamma_0|\epsilon}\right)$
    & $\frac{\kappa_S}{|\gamma_0|}\log\left(\frac{1}{\delta_A}
        \log\left(\frac{\kappa_S}{|\gamma_0|\epsilon}\right)\right)$
    \tabularnewline
\hline 
    \makecell{Block-encoded\\eigenvalue transformer\\\thm{qevt_blk}} & Input matrix & Ancilla state
    & $\frac{n^{1.5}\kappa_S^2}{\norm{p\left(\frac{A}{\alpha_A}\right)}}\log\left(\frac{\sqrt{n}\kappa_S}{\norm{p\left(\frac{A}{\alpha_A}\right)}\epsilon}\right)\log\left(\frac{1}{\epsilon}\right)$
    & $\frac{n\kappa_S^2}{\norm{p\left(\frac{A}{\alpha_A}\right)}}\log\left(\frac{\kappa_S}{\norm{p\left(\frac{A}{\alpha_A}\right)}\epsilon}\right)\log\left(\frac{1}{\epsilon}\right)$
    \tabularnewline
    \thickhline
    \end{tabular}
}
\caption{Summary of applications of the block preconditioning technique. Using the initial ancilla state as the preconditioner, one can reduce the cost of initial state preparation in linear-system-based differential equation solvers, ground state preparators and eigenvalue processors, nearly matching or outperforming the state of the art. By choosing the initial state itself as the preconditioner and applying QSVT, one gets an extremely simple quantum linear system solver with optimal queries to the coefficient matrix.
See the corresponding theorem statements for definitions of the parameters.
}
\label{tab:precond}
\end{table}
%%%%%%%%%%%%%%%%%%%%%%%%%%%%%%%%%%%%%%%%%%%%%%%%%%%%%%%%%%%%%%%%%%%%%%%%%%%%%%

Our block preconditioning technique targets at a subspace containing the initial state. Specifically, suppose that there is an orthogonal projection $\Pi_{\text{cond}}$ whose image includes the initial state
\begin{equation}
    \Pi_{\text{cond}}\ket{b}=\ket{b}.
\end{equation}
Then we introduce the scaling operator
\begin{equation}
    S=s\Pi_{\text{cond}}+\left(I-\Pi_{\text{cond}}\right),\qquad
    S^{-1}=\frac{1}{s}\Pi_{\text{cond}}+\left(I-\Pi_{\text{cond}}\right),
\end{equation}
where $0<s<1$ is a scaling parameter that can be tuned for specific applications.
Assuming $\Pi_{\text{cond}}$ is accessible through an oracle reflecting along the image space $-\Pi_{\text{cond}}+\left(I-\Pi_{\text{cond}}\right)
    =I-2\Pi_{\text{cond}}$,
we can block encode $S$ with normalization factor $1$ through the linear combination
$S=\frac{1-s}{2}\left(I-2\Pi_{\text{cond}}\right)
    +\frac{1+s}{2}
    I$.
This assumption is justified because in all our applications, we can efficiently implement unitaries $V_{\text{cond}}$ such that $V_{\text{cond}}^\dagger\left(I-2\Pi_{\text{cond}}\right)V_{\text{cond}}$ are reflections in the computational basis.

Now instead of $Ax=b$, we consider the preconditioned linear system $SAx=Sb$,
which yields the original solution $\frac{(SA)^{-1}Sb}{\norm{(SA)^{-1}Sb}}=\frac{A^{-1}\ket{b}}{\norm{A^{-1}\ket{b}}}$ when solved exactly.
Here, the preconditioned coefficient matrix can be block encoded as $(SA)/\alpha_A$ with the same normalization as before, and the initial state $\frac{Sb}{\norm{Sb}}=\ket{b}$ also remains the same, whereas the inverse matrix has a norm bounded by $\norm{(SA)^{-1}}
    =\norm{A^{-1}S^{-1}}
    % =\norm{\frac{1}{s}A^{-1}\Pi_{\text{cond}}+A^{-1}\left(I-\Pi_{\text{cond}}\right)}
    \leq\sqrt{\frac{\norm{A^{-1}\Pi_{\text{cond}}}^2}{s^2}+\norm{A^{-1}\left(I-\Pi_{\text{cond}}\right)}^2}
    \leq\sqrt{\frac{\norm{A^{-1}\Pi_{\text{cond}}}^2}{s^2\norm{A^{-1}}^2}+1}\norm{A^{-1}}$.
However, the solution norm now becomes
\begin{equation}
    \norm{(SA)^{-1}\ket{b}}
    =\norm{A^{-1}S^{-1}\ket{b}}
    =\frac{1}{s}\norm{A^{-1}\ket{b}}.
\end{equation}
\addtocounter{equation}{1}%
So asymptotically the condition number stays the same as long as $s=\mathbf{\Omega}\left(\frac{\norm{A^{-1}\Pi_{\text{cond}}}}{\norm{A^{-1}}}\right)$,
but the solution norm will be artificially boosted by a factor of $\frac{1}{s}$ thanks to block preconditioning.

Preconditioning is a classical subject in numerical linear algebra, the goal of which is to transform linear system problems to improve their algorithmic solvability. Among the various preconditioning methods, the one closest to our work is the so-called matrix scaling (also known as matrix balancing or matrix equilibration); see the review article~\cite{Idel16} and the references therein. Using that language, our method can be interpreted as follows. When partitioned with respect to the orthogonal decomposition $\mathbf{Im}\left(\Pi_{\text{cond}}\right)\obot\mathbf{Im}\left(I-\Pi_{\text{cond}}\right)$, operators $A^{-1}$ and $S^{-1}$ take the form
\begin{equation}
    A^{-1}=
    \begin{bmatrix}
    \left.A^{-1}\right|_{\mathbf{Im}\left(\Pi_{\text{cond}}\right)} & \left.A^{-1}\right|_{\mathbf{Im}\left(I-\Pi_{\text{cond}}\right)}
    \end{bmatrix},\qquad
    S^{-1}=
    \begin{bmatrix}
        \frac{1}{s}\left.I\right|_{\mathbf{Im}\left(\Pi_{\text{cond}}\right)} & 0\\
        0 & \left.I\right|_{\mathbf{Im}\left(I-\Pi_{\text{cond}}\right)}
    \end{bmatrix}.
\end{equation}
Here, the block $\left.A^{-1}\right|_{\mathbf{Im}\left(\Pi_{\text{cond}}\right)}$ is of most interest to us, because it completely determines the action of $A^{-1}$ on a subspace containing the initial state $\ket{b}$. The issue is that $\norm{\left.A^{-1}\right|_{\mathbf{Im}\left(\Pi_{\text{cond}}\right)}}\ll\norm{A^{-1}}$ can often happen in practice, so we are inverting a coefficient matrix with a large condition number, while the initial state actually resides in a subspace with a relatively small condition number. To remedy this issue, we use the scaling matrix to boost the block $\frac{1}{s}\norm{\left.A^{-1}\right|_{\mathbf{Im}\left(\Pi_{\text{cond}}\right)}}\approx\norm{A^{-1}}$ without significantly increasing the overall condition number, leading to faster quantum linear system solvers; hence the name \emph{block preconditioning}. It is worth mentioning that preconditioning has also been explored in previous quantum work such as~\cite{2021Yupreconditioned,Preconditioned13,Circulant18}, although their focus is on reducing the condition number for problems with specially structured coefficients, as opposed to increasing the solution norm of a generic linear system considered here. 

As an immediate application, let us consider what if we choose the initial state itself as the preconditioner $\Pi_{\text{cond}}=\ketbra{b}{b}$. Then, the coefficient matrix
\begin{equation}
    SA=\left(s\ketbra{b}{b}+\left(I-\ketbra{b}{b}\right)\right)A
\end{equation}
can be block encoded with normalization factor $\alpha_A$, using $1$ query to $O_A$ and $2$ queries to $O_b$ and its inverse. Given a constant multiplicative approximation of $\norm{A^{-1}\ket{b}}$, we set
\begin{equation}
    s=\mathbf{\Theta}\left(\frac{\norm{A^{-1}\ket{b}}}{\alpha_{A^{-1}}}\right)=\mathbf{\Theta}\left(\sqrt{p_{\text{succ}}}\right).
\end{equation}
This implies the scaling of solution norm $\norm{(SA)^{-1}\ket{b}}
    =\frac{1}{s}\norm{A^{-1}\ket{b}}
    =\mathbf{\Theta}\left(\alpha_{A^{-1}}\right)$,
\addtocounter{equation}{1}%
and norm bound on the inverse coefficient matrix $\norm{(SA)^{-1}}
    \leq\sqrt{\frac{\norm{A^{-1}\ket{b}}^2}{s^2\norm{A^{-1}}^2}+1}\norm{A^{-1}}
    =\mathbf{O}\left(\alpha_{A^{-1}}\right)$.
\addtocounter{equation}{1}%
Therefore, by applying QSVT to the preconditioned problem, we obtain an extremely simple quantum linear system algorithm with complexity
\begin{equation}
    \mathbf{O}\left(\kappa\log\left(\frac{1}{\epsilon}\right)\mathbf{Cost}(O_b)
    +\kappa\log\left(\frac{1}{\epsilon}\right)\mathbf{Cost}(O_A)\right).
\end{equation}
This appears to be conceptually even simpler than the recent kernel-reflection method, which instead solves a padded linear system over an expanded Hilbert space~\cite{Dalzell2024shortcut}. Of course, one needs to estimate the solution norm to a constant multiplicative accuracy before running either approach. This can be achieved using either the method of~\cite{Dalzell2024shortcut} which makes optimal queries to the coefficient matrix, or our Tunable VTAA which makes optimal queries to the initial state preparation and nearly optimal uses of the coefficient matrix.

Observe that the preconditioned coefficient matrix $SA$ can be block encoded without querying $O_b$, if the initial state is in tensor product $\ket{0}\ket{b}$ and we choose the ancilla state $\ket{0}\otimes I$ as the preconditioner. We utilize this observation to reduce the cost of initial state preparation in solving differential equations, preparing ground states of operators with real spectra and processing eigenvalues of non-normal matrices, nearly matching or outperforming the state of the art. This refutes the widely held belief that linear-system-based differential equation solvers and eigenvalue estimators necessarily require more queries to the initial state oracle than their alternatives~\cite{Fang2023timemarchingbased,AnChildsLin23,Zhang2024Nonnormal}. Of independent interest, we also get a block-encoded eigenvalue transformation algorithm with $\mathbf{O}(n)$ scaling in degree of the target polynomial, whereas the best previous result was $\mathbf{O}(n^{1.5})$~\cite[Section 5.2]{QEVP}. See \tab{precond} for a summary of applications of the block preconditioning technique.

We summarize in~\sec{prelim} preliminaries required to understand our algorithms, and include in~\sec{discuss} a brief recap of the main result and a collection of questions for future work.

%% file: prelim.tex
This section contains preliminaries useful for understanding our quantum linear system algorithms and their applications. We provide a list of notation and terminology in \sec{prelim_notation} to be used throughout the paper. We then review known results on VTAA in \sec{prelim_amp}. Unlike prior work, we introduce an axiomatic formulation of VTAA which is not bonded by any concrete circuit implementation and offers more flexibility. We refer the reader to~\cite[Section 2.4]{QEVP} for further preliminaries about the block encoding framework, within which our algorithms are developed.

%%%%%%%%%%%%%%%%%%%%%%%%%%%%%%%%%%%%%%%%%%%%%%%%%%%%%%%%%%%%%%%%%%%%%%%%%%%%%%
\subsection{Notation and terminology}
\label{sec:prelim_notation}
We use lowercase Latin and Greek alphabets to denote scalars, vectors and functions. For instance, we often use $\epsilon$ for the accuracy of algorithm, $\kappa$ for a known upper bound on the condition number, $p_{\text{succ}}$ for success probability of the input algorithm of VTAA or the direct inversion of input matrix with block encoding, $m$ for the number of VTAA stages, and $r$ for the number of amplification steps. We also represent known upper/lower bounds of quantities using $\alpha$, writing $\alpha_A\geq\norm{A}$ for a known upper bound on the spectral norm of input matrix $A$ for block encoding, $\alpha_{A^{-1}}\geq\norm{A^{-1}}$ for a known upper bound on the norm of $A^{-1}$ leading to $\kappa=\alpha_A\alpha_{A^{-1}}$, $\alpha_{p_{\text{succ}}}\leq p_{\text{succ}}$ a known lower bound on the success probability, and $\alpha_j$ amplification thresholds associated with a Tunable VTAA. We employ standard notations for number sets, using $\mathbb{Z}$ for integers ($\mathbb{Z}_{\geq 0}$ for nonnegative integers), $\mathbb{R}$ for real numbers, and $\mathbb{C}$ for complex numbers.

We user uppercase Latin and Greek alphabets to represent matrices and operators. For example, we write $A$ for the coefficient matrix of a linear system, $O_A$ for its block encoding unitary, $O_b$ for the unitary preparing the initial state, $A_j$ ($j=1,\ldots,m$) for the input algorithms of VTAA, $\Pi$ for an orthogonal projection ($\overline{\Pi}=I-\Pi$ for its complement), $S$ for the scaling operator used for preconditioning, and $I,X,Y,Z$ for the identity and Pauli matrices. We use $\norm{A}$ to denote the spectral norm of $A$, i.e., its largest singular value.
For a matrix $A$ and scalar $\alpha_A\geq 0$, we say $A/\alpha_A$ can be block encoded if there exist isometries $G_1$, $G_2$ and unitary $U$ such that $\frac{A}{\alpha_A}=G_2^\dagger UG_1$. Such a \emph{block encoding} is mathematically feasible if and only if $\alpha_A\geq\norm{A}$, but additional normalization factors may be introduced when the block encoding is implemented by a quantum circuit. Moreover, the block encoding is Hermitian if the two isometries $G_1=G_2$ are identical and $U$ is a Hermitian unitary. This arises when we block encode the Hermitian dilation of a matrix as
\begin{equation*}
    \frac{1}{\alpha_A}
    \begin{bmatrix}
        0 & A\\
        A^\dagger & 0
    \end{bmatrix}
    =\begin{bmatrix}
        G_2^\dagger & 0\\
        0 & G_1^\dagger
    \end{bmatrix}
    \begin{bmatrix}
        0 & U\\
        U^\dagger & 0
    \end{bmatrix}
    \begin{bmatrix}
        G_2 & 0\\
        0 & G_1
    \end{bmatrix}.
\end{equation*}

We use boldface symbols to denote functions and operations having specific meanings. We write $\mathbf{Cost}(\cdot)$ for the (query) cost of implementing an operator, $\mathbf{Ceil}(\cdot)$/$\mathbf{Floor}(\cdot)$ for the nearest integer rounded up/down, $\mathbf{Ker}(\cdot)$/$\mathbf{Im}(\cdot)$ for the kernel/image of a linear operator, $\mathbf{P}(\cdot)$ for the probability of an event, and $\mathbf{T}_j(\cdot)$/$\mathbf{U}_j(\cdot)$ for Chebyshev polynomials of the first/second kind (writing its rescaled version as $\widetilde{\mathbf{T}}_j(\cdot)$ with $\widetilde{\mathbf{T}}_0=\frac{1}{2}$).
If $\Pi$ is an orthogonal projection, $\mathbf{Im}(\Pi)\obot\mathbf{Im}(\overline{\Pi})$ is an orthogonal decomposition of the underlying space and we may use it to bound the spectral norm of an operator as
\begin{equation}
    \norm{A}=\sqrt{\norm{AA^\dagger}}
    \leq\sqrt{\norm{A\Pi A^\dagger}+\norm{A\overline{\Pi}A^\dagger}}
    =\sqrt{\norm{A\Pi}^2+\norm{A\overline{\Pi}}^2}.
\end{equation}
We write $\mathbf{O}(\cdot)$ to mean asymptotically less than, $\mathbf{\Omega}(\cdot)$ to mean asymptotically more than, and $\mathbf{\Theta}(\cdot)$ to represent quantities having the same asymptotic scaling.
We let a summation be zero and a product by one if their lower limits exceed upper limits.

In analyzing the performance of Tunable VTAA, we will need to evaluate
\begin{equation}
    \norm{\beta}_{p}=\left(\sum_{j=1}^n\abs{\beta_j}^p\right)^{\frac{1}{p}},\qquad
    p\in(0,+\infty]
\end{equation}
for vectors $\begin{bmatrix}
    \beta_1 & \cdots & \beta_n
\end{bmatrix}\in\mathbb{C}^n$. When $1\leq p\leq\infty$, $\norm{\cdot}_p$ is the vector $\ell_p$-norm and its properties are summarized in standard references such as~\cite[Chapter 5]{horn2012matrix}. In particular, the case $p=2$ corresponds to the Euclidean norm and we may drop the subscript since it is the same as spectral norm when we equate vectors with column matrices. However if $0<p<1$, $\norm{\cdot}_p$ is no longer a vector norm and its properties may be less familiar. In this case, $\norm{\cdot}_p$ is known as a \emph{quasinorm}~\cite{Rosenzweig,Conrad}, which satisfies the following defining properties.
\begin{enumerate}[label=(\roman*)]
\item \label{enum:positivity} $\norm{\beta}_p\geq0$ for all $\beta\in\mathbb{C}^n$, with $\norm{\beta}_p=0$ if and only if $\beta=0$.
\item \label{enum:homogeneity} $\norm{c\beta}_p=\abs{c}\norm{\beta}_p$ for all $\beta\in\mathbb{C}^n$ and $c\in\mathbb{C}$.
\item \label{enum:modified_triangle} $\norm{\beta+\gamma}_p\leq2^{\frac{1}{p}-1}\left(\norm{\beta}_p+\norm{\gamma}_p\right)$ for all $\beta,\gamma\in\mathbb{C}^n$.
\end{enumerate}
The factor $2^{\frac{1}{p}-1}$ in Property \ref{enum:modified_triangle} is the best one can get for a modified triangle inequality $\norm{\beta+\gamma}_p\leq c\left(\norm{\beta}_p+\norm{\gamma}_p\right)$ with some universal constant $c\geq0$ for all $\beta,\gamma\in\mathbb{C}^n$. In fact, if $\beta$ and $\gamma$ are entrywise nonnegative, it holds the reverse Minkowski's inequality $\norm{\beta+\gamma}_p\geq\norm{\beta}_p+\norm{\gamma}_p$. Additionally, when $0<p<q\leq\infty$,
\begin{equation}
    \norm{\beta}_q\leq\norm{\beta}_p\leq n^{\frac{1}{p}-\frac{1}{q}}\norm{\beta}_q,
\end{equation}
and $n^{\frac{1}{p}-\frac{1}{q}}$ is again the best constant one can hope for.

We use the Dirac notation $\ket{\psi}$ for a vector only when it is normalized with respect to the Euclidean norm $\norm{\ket{\psi}}=1$. For general nonzero vectors $\psi$ and $\phi$, we define
\begin{equation}
    \psi\propto\phi
    \quad\Leftrightarrow\quad
    \frac{\psi}{\norm{\psi}}=\frac{\phi}{\norm{\phi}}
    \quad\Leftrightarrow\quad
    \psi=c\phi\ \exists c>0.
\end{equation}
That is, $\psi$ is proportional to $\phi$ if and only if they agree up to a positive rescaling, whereas $\ket{\psi}\propto\ket{\phi}\Leftrightarrow \ket{\psi}=\ket{\phi}$ holds for unit vectors. It can be verified that proportionality is reflexive, symmetric, and transitive. Moreover, 
\begin{equation}
    \psi\propto\phi
    \quad\Rightarrow\quad
    B\psi\propto B\phi
\end{equation}
for an operator $B$ whenever the composition makes sense.

Finally, we say $u>0$ is a $c$-\emph{multiplicative approximation} of $v>0$ if
\begin{equation}
    \frac{1}{c}\leq\frac{u}{v}\leq c
\end{equation}
for some $c\geq 1$. It is apparent that this relation is reflexive, symmetric, and invariant under reciprocal: any $u>0$ is a $1$-multiplicative approximation of itself; if $u>0$ is a $c$-multiplicative approximation of $v>0$, then $v$ is a $c$-multiplicative approximation of $u$, and $\frac{1}{u}$ is a $c$-multiplicative approximation of $\frac{1}{v}$. Moreover, if $u_1$, $u_2$ are $c_1$-, $c_2$-multiplicative approximations of $v_1$, $v_2$ respectively, their product $u_1u_2$ is a ($c_1c_2$)-multiplicative approximation of $v_1v_2$. When analyzing VTAA, we often consider the case where $c$ is constant. But there may also be scenarios where we want $1-\epsilon\leq\frac{u}{v}\leq\frac{1}{1-\epsilon}$ for some small $0<\epsilon<1$. This then leads to a simple product rule:
\begin{equation}
    1-\epsilon_1\leq\frac{u_1}{v_1}\leq\frac{1}{1-\epsilon_1},\quad
    1-\epsilon_2\leq\frac{u_2}{v_2}\leq\frac{1}{1-\epsilon_2}
    \quad\Rightarrow\quad
    1-(\epsilon_1+\epsilon_2)\leq\frac{u_1u_2}{v_1v_2}\leq\frac{1}{1-(\epsilon_1+\epsilon_2)}.
\end{equation}
By a possible rescaling of $\epsilon$, we can also consider alternative definitions such as $\frac{1}{1+\epsilon}\leq\frac{u}{v}\leq1+\epsilon$ and $1-\epsilon\leq\frac{u}{v}\leq1+\epsilon$, which relate more closely to additive approximations as
\begin{equation}
    1-\epsilon\leq\frac{u}{v}\leq1+\epsilon
    \quad\Leftrightarrow\quad
    \abs{u-v}\leq v\epsilon.
\end{equation}

%%%%%%%%%%%%%%%%%%%%%%%%%%%%%%%%%%%%%%%%%%%%%%%%%%%%%%%%%%%%%%%%%%%%%%%%%%%%%%
\subsection{Variable time amplitude amplification}
\label{sec:prelim_amp}
We now formally introduce variable time amplitude amplification, and review known results about this framework~\cite{Ambainis2012VTAA,Chakraborty2018BlockEncoding}, using an axiomatic formulation not bonded by its specific circuit implementation. We assume throughout this subsection that an underlying Hilbert space $\mathcal{H}$ has been fixed on which all operators act.

\begin{definition}[Variable time algorithm and amplification]
\label{defn:vta}
A variable time quantum algorithm is a $3$-tuple $\left(\{\Pi_j\}_{j=0}^m,\Pi_b,\{A_j\}_{j=0}^m\right)$ satisfying the following axioms.
\begin{enumerate}[label=(\roman*)]
    \item \label{enum:vta_clock} $\Pi_j$ are orthogonal projections partially ordered as $0=\Pi_0\leq\Pi_1\leq\cdots\leq\Pi_m=I$.
    \item \label{enum:vta_flag} $\Pi_b$ is an orthogonal projection  commuting with all $\Pi_j$: $\Pi_b\Pi_j=\Pi_j\Pi_b$ for $j=0,\ldots, m$.
    \item \label{enum:vta_input} $A_j$ are unitaries such that $A_j\Pi_{j-1}=\Pi_{j-1}$ for all $j=1,\ldots, m$, and $A_0=I$.

\end{enumerate}
A variable time amplification algorithm is a $5$-tuple $\left(\{\Pi_j\}_{j=0}^m,\Pi_b,\{A_j\}_{j=0}^m,\{\widetilde{A}_j\}_{j=0}^m,\ket{\psi_0}\right)$
that additionally satisfies
\begin{enumerate}[label=(\roman*)]
\setcounter{enumi}{3}
    \item \label{enum:vta_amp} $\widetilde A_j$ are unitaries such that $\frac{\overline{\Pi_j\Pi_b}\widetilde{A}_j\ket{\psi_0}}{\norm{\overline{\Pi_j\Pi_b}\widetilde{A}_j\ket{\psi_0}}}
    =\frac{\overline{\Pi_j\Pi_b}A_j\widetilde{A}_{j-1}\ket{\psi_0}}{\norm{\overline{\Pi_j\Pi_b}A_j\widetilde{A}_{j-1}\ket{\psi_0}}}$ for all $j=1,\ldots, m$, and $\widetilde A_0=I$.
\end{enumerate}
\end{definition}

The above definition deserves a few remarks. We call $\Pi_j$ ($j=0,\ldots,m$) the \emph{clock projections}. As $\Pi_j$ are Hermitian, they can be partially ordered according to the positive semidefinitess, and Axiom~\ref{enum:vta_clock} requires that $\Pi_j\leq\Pi_{j+1}$ hold true for all $j=0,\ldots,m-1$. In \append{axiom_clock}, we present a number of equivalent characterizations of the partial ordering $\Pi_j\leq\Pi_{j+1}$, showing in particular that $\Pi_{j}=\Pi_{j+1}\Pi_j=\Pi_{j}\Pi_{j+1}$.
As an immediate corollary, we have
\begin{equation}
    \Pi_j=\Pi_k\Pi_j=\Pi_j\Pi_k,\qquad 0\leq j\leq k\leq m
\end{equation}
and $\Pi_k-\Pi_j$ are themselves orthogonal projections. In a variable time quantum algorithm, $\mathbf{Im}(\Pi_j)$ represents the part of space in which the algorithm stops running before or at stage $j$. The space $\mathbf{Im}(\Pi_j)$ monotonically increases with $j$, echoing the fact that more branches of the quantum algorithm will halt as the computation proceeds toward completion.

We call $\Pi_b$ and its complement $\overline{\Pi_b}=I-\Pi_b$ the \emph{flag projections}. In the context of a variable time quantum algorithm, $\mathbf{Im}(\overline{\Pi_b})$ is the part of space where the desired output state resides, corresponding to the case of success. Axiom \ref{enum:vta_flag} requires that $\Pi_b$ commutes with all $\Pi_j$, so these projections may be simultaneously measured in a quantum computation. In \append{axiom_flag}, we tabulate the meaning of different outcomes from the simultaneous measurement of $\{\Pi_j-\Pi_{j-1}\}_{j=1}^m$ and $\{\Pi_b,I-\Pi_b\}$. Of particular interest is the outcome corresponding to $\overline{\Pi_j\Pi_b}$, which represents that the computation can potentially succeed at stage $j$. These potentially good projections form a monotonically decreasing sequence
\begin{equation}
    I=\overline{\Pi_0\Pi_b}\geq\overline{\Pi_1\Pi_b}\geq\cdots\geq\overline{\Pi_m\Pi_b}=\overline{\Pi_b},
\end{equation}
and they satisfy
\begin{equation}
    \overline{\Pi_k\Pi_b}=\overline{\Pi_j\Pi_b}\cdot\overline{\Pi_k\Pi_b}=\overline{\Pi_k\Pi_b}\cdot\overline{\Pi_j\Pi_b},\qquad
    0\leq j\leq k\leq m.
\end{equation}

We call $A_j$ the \emph{input algorithms}, which act trivially on $\mathbf{Im}(\Pi_{j-1})$ as is required by Axiom \ref{enum:vta_input}. In \append{axiom_input}, we present a number of equivalent characterizations of the property $A_j\Pi_{j-1}=\Pi_{j-1}$, finding that $A_j=\Pi_{j-1}+\overline{\Pi_{j-1}}A_j\overline{\Pi_{j-1}}$ are necessarily controlled unitaries controlled by $\overline{\Pi_{j-1}}$. This implies
\begin{equation}
    A_j\Pi_l=\Pi_l=\Pi_lA_j,\qquad A_j\overline{\Pi_l\Pi_b}=\overline{\Pi_l\Pi_b}A_j,\qquad
    0\leq l<j\leq m.
\end{equation}
As an immediate consequence of this property, we have that the potentially good amplitudes
\begin{equation}
\label{eq:pg_amp}
    \norm{\overline{\Pi_{j}\Pi_b}A_{j}\cdots A_1\ket{\psi_0}}=\norm{\overline{\Pi_{j}\Pi_b}A_{m}\cdots A_1\ket{\psi_0}},\qquad
    j=0,\ldots,m
\end{equation}
form a monotonic sequence decreasing from $1$ to $\sqrt{p_{\text{succ}}}$:
\begin{equation}
    1=\norm{\overline{\Pi_{0}\Pi_b}\ket{\psi_0}}\geq\norm{\overline{\Pi_{1}\Pi_b}A_1\ket{\psi_0}}\geq\cdots\geq\norm{\overline{\Pi_{m}\Pi_b}A_{m}\cdots A_1\ket{\psi_0}}=\norm{\overline{\Pi_b}A_{m}\cdots A_1\ket{\psi_0}}=\sqrt{p_{\text{succ}}}.
\end{equation}

Finally, we call $\widetilde A_j$ the \emph{amplified algorithms}. At each stage $j$, we would like the amplification to preserve the part of state potentially leading to success, while suppressing the part that has already resulted in failure. This requirement is captured by Axiom \ref{enum:vta_amp}, which can be succinctly represented as $\overline{\Pi_j\Pi_b}\widetilde{A}_j\ket{\psi_0}\propto\overline{\Pi_j\Pi_b}A_j\widetilde{A}_{j-1}\ket{\psi_0}$. We then prove in \append{axiom_amp} that
\begin{equation}
\label{eq:trans_amp}
    \frac{\norm{\overline{\Pi_{h}\Pi_b}A_{h}\cdots A_{j+1}\widetilde{A}_{j}\ket{\psi_0}}}{\norm{\overline{\Pi_{k}\Pi_b}A_{k}\cdots A_{j+1}\widetilde{A}_{j}\ket{\psi_0}}}
    =\frac{\norm{\overline{\Pi_{h}\Pi_b}A_{h}\cdots A_{l+1}\widetilde{A}_{l}\ket{\psi_0}}}{\norm{\overline{\Pi_{k}\Pi_b}A_{k}\cdots A_{l+1}\widetilde{A}_{l}\ket{\psi_0}}},\qquad
    0\leq l,j\leq k,h\leq m.
\end{equation}
That is, the transition of potentially good amplitudes remains the same, regardless of whether we start with the pre- or post-amplified algorithms.

Until this point, we have not specified the amplification algorithms $\widetilde A_j$ other than the requirement that they should preserve the potentially good outcomes. In what follows, we consider the case where $\widetilde A_j$ are constructed by the standard amplitude amplification toward the potentially good subspaces $\widetilde A_j
    =\left(-\left(I-2A_j\widetilde A_{j-1}\ketbra{\psi_0}{\psi_0}\widetilde A_{j-1}^\dagger A_j^\dagger\right)\left(I-2\overline{\Pi_j\Pi_b}\right)\right)^{r_j}A_j\widetilde A_{j-1}$.
See~\append{qubitization_amp} for a brief review of amplitude amplification.
Here, $r_j$ are nonnegative integers and we call $2r_j+1$ the \emph{amplification schedules} or \emph{amplification step numbers}. The query complexity of $\widetilde A_j$ then satisfies the recurrence
\begin{equation}
    \mathbf{Cost}\left(\widetilde A_j\ket{\psi_0}\right)=(2r_j+1)\left(\mathbf{Cost}( A_j)+\mathbf{Cost}\left(\widetilde A_{j-1}\ket{\psi_0}\right)\right),\qquad
    j=1,\ldots,m.
\end{equation}
For a given choice of $r_j$, we can unwrap the recursion to obtain a \emph{variable time nested amplitude amplification} $\widetilde A_m$, which has a query cost of
\begin{equation}
\label{eq:nest_complexity}
    \mathbf{Cost}\left(\widetilde A_m\ket{\psi_0}\right)
    =\mathbf{Cost}(\ket{\psi_0})\prod_{k=1}^m(2r_k+1)
    +\sum_{j=1}^m\mathbf{Cost}(A_j)\prod_{k=j}^m(2r_k+1).
\end{equation}
Formally:
\begin{definition}[Variable time nested amplitude amplification]
\label{defn:nested}
A variable time nested amplitude amplification is a $5$-tuple $\left(\{\Pi_j\}_{j=0}^m,\Pi_b,\{A_j\}_{j=0}^m,\{r_j\}_{j=1}^m,\ket{\psi_0}\right)$ that additionally satisfies
\begin{enumerate}[label=(\roman*$'$)]
\setcounter{enumi}{3}
    \item \label{enum:nested_amp} $r_j$ are nonnegative integers for $j=1,\ldots,m$, which define
\begin{equation}
    \widetilde A_j=
    \begin{cases}
        \left(-\left(I-2A_j\widetilde A_{j-1}\ketbra{\psi_0}{\psi_0}\widetilde A_{j-1}^\dagger A_j^\dagger\right)\left(I-2\overline{\Pi_j\Pi_b}\right)\right)^{r_j}A_j\widetilde A_{j-1},\quad&j=1,\ldots,m,\\
        I,&j=0.
    \end{cases}
\end{equation}
\end{enumerate}
\end{definition}

When $(2r_j+1)\norm{\overline{\Pi_j\Pi_b}A_j\widetilde A_{j-1}\ket{\psi_0}}\leq1$, we have $(2r_j+1)\arcsin\left(\norm{\overline{\Pi_j\Pi_b}A_j\widetilde A_{j-1}\ket{\psi_0}}\right)\leq\frac{\pi}{2}$ and there is no over amplification/overshoot. Consequently, the pre- and post-amplified amplitudes of stage $j$ are related by the analytic expression
\begin{equation}
    \norm{\overline{\Pi_j\Pi_b}\widetilde A_j\ket{\psi_0}}
    =\sin\left((2r_j+1)\arcsin\left(\norm{\overline{\Pi_j\Pi_b}A_j\widetilde A_{j-1}\ket{\psi_0}}\right)\right).
\end{equation}
Note we have almost amplified the amplitude by $2r_j+1$, up to the \emph{loss factor}
\begin{equation}
    \frac{\sin\left((2r_j+1)\arcsin\left(\norm{\overline{\Pi_j\Pi_b}A_j\widetilde A_{j-1}\ket{\psi_0}}\right)\right)}{(2r_j+1)\norm{\overline{\Pi_j\Pi_b}A_j\widetilde A_{j-1}\ket{\psi_0}}}.
\end{equation}
This is always $\leq1$, and one should maintain it $\approx 1$ so as to avoid unnecessary waste of oracular queries. See~\append{dirichlet} for a tight bound on this quantity.
Carrying out this analysis to all the stages, we have the total loss factor
\begin{equation}
    \prod_{j=1}^m\left(\frac{\sin\left((2r_j+1)\arcsin\left(\norm{\overline{\Pi_j\Pi_b}A_j\widetilde A_{j-1}\ket{\psi_0}}\right)\right)}{(2r_j+1)\norm{\overline{\Pi_j\Pi_b}A_j\widetilde A_{j-1}\ket{\psi_0}}}\right).
\end{equation}
See \fig{vtaa_transition} for an illustration of the amplitude transitions in a nested amplitude amplification.

%%%%%%%%%%%%%%%%%%%%%%%%%%%%%%%%%%%%%%%%%%%%%%%%%%%%%%%%%%%%%%%%%%%%%%%%%%%%%%
\begin{figure}[t]
    \centering
\includegraphics[width=1.\textwidth]{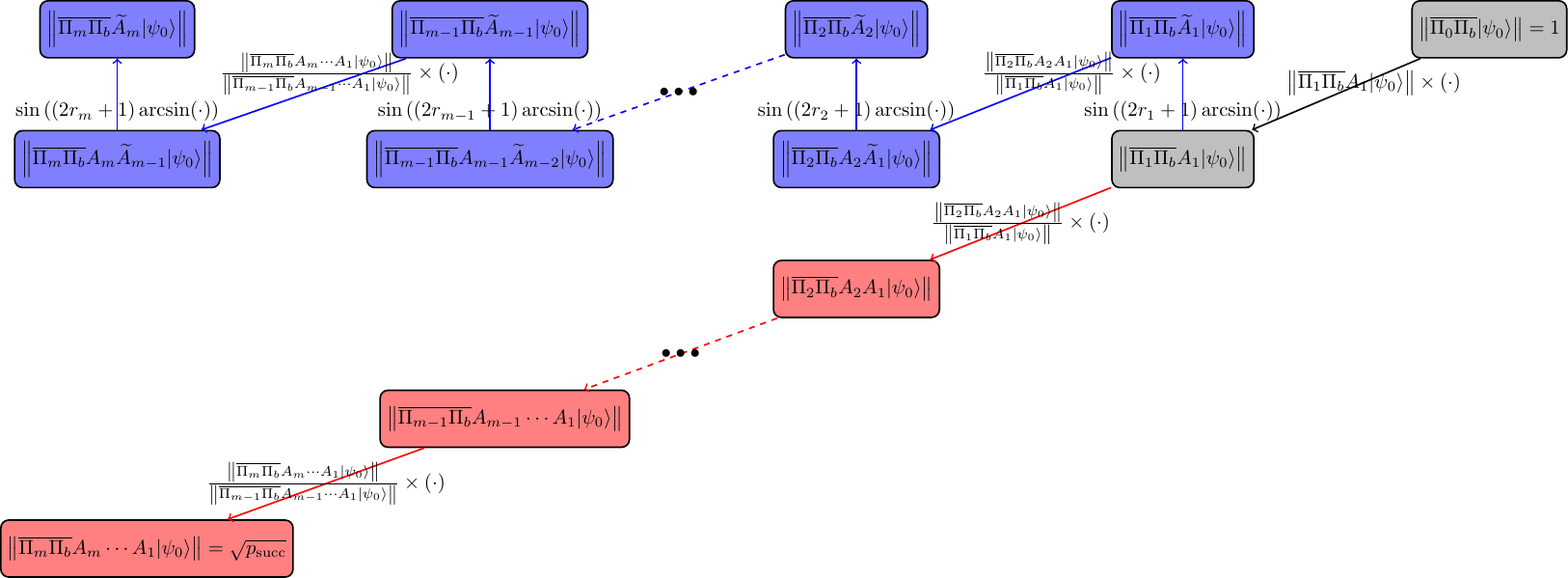}
    \caption{Illustration of the transition of amplitudes in a VTAA/nested amplitude amplification. The flow of algorithms is represented by arrows, with the input branch colored in red and amplified one in blue. Assuming there is no over amplification, for any pair of connected nodes, the amplitude from the bottom node is less than or equal to that from the top node.}
    \label{fig:vtaa_transition}
\end{figure}
%%%%%%%%%%%%%%%%%%%%%%%%%%%%%%%%%%%%%%%%%%%%%%%%%%%%%%%%%%%%%%%%%%%%%%%%%%%%%%

In Ambainis' \emph{Variable Time Amplitude Amplification} (VTAA)~\cite{Ambainis2012VTAA}, the schedules are determined as follows: we choose $r_j$ as the smallest nonnegative integer satisfying $(2r_j+1)\norm{\overline{\Pi_j\Pi_b}A_j\widetilde A_{j-1}\ket{\psi_0}}\geq\frac{1}{3\sqrt{m}}$. This choice of $r_j$ necessarily means that $(2r_j+1)\norm{\overline{\Pi_j\Pi_b}A_j\widetilde A_{j-1}\ket{\psi_0}}\leq\frac{1}{\sqrt{m}}$, and one can show via an induction that there is no over amplification at all stages. On the other hand, we have $\sum_{j=1}^m(2r_j+1)^2\norm{\overline{\Pi_j\Pi_b}A_j\widetilde A_{j-1}\ket{\psi_0}}^2\leq1$, which yields $\prod_{j=1}^m\left(\frac{\sin\left((2r_j+1)\arcsin\left(\norm{\overline{\Pi_j\Pi_b}A_j\widetilde A_{j-1}\ket{\psi_0}}\right)\right)}{(2r_j+1)\norm{\overline{\Pi_j\Pi_b}A_j\widetilde A_{j-1}\ket{\psi_0}}}\right)=\mathbf{\Omega}(1)$. The conclusion is that Ambainis' VTAA is a nested amplitude amplification with no overshoot and constant loss factor. Note our above definition of $r_j$ uses the precise value of the amplitude $\norm{\overline{\Pi_j\Pi_b}A_j\widetilde A_{j-1}\ket{\psi_0}}$. This is only for presentational purpose. In practice, one would estimate it to a constant muliplicative precision by running VTAA up to stage $j-1$~\cite[Theorem 23]{Chakraborty2018BlockEncoding}. Then one compares the upper bound with the threshold value to determine the amplification schedule, which introduces only a constant factor overhead.

Using the state-of-the-art analysis, VTAA has the query complexity
\begin{equation}
\label{eq:vtaa_cost2}
    \mathbf{O}\left(\frac{\sqrt{m}}{\sqrt{p_{\text{succ}}}}\mathbf{Cost}\left(\ket{\psi_0}\right)
    +\frac{\sqrt{m}}{\sqrt{p_{\text{succ}}}}\sum_{j=1}^m\norm{\overline{\Pi_{j-1}\Pi_b}A_{j-1}\cdots A_1\ket{\psi_0}}\mathbf{Cost}\left(A_j\right)\right).
\end{equation}
Here, the cost of input algorithms has the $\ell_1$-norm scaling with an additional $\sqrt{m}$ factor, which can be relaxed to the $\ell_2$-norm scaling introducing another $\sqrt{m}$ factor. This complexity analysis can be further improved---we present the tightened analysis in \sec{tunable} for Tunable VTAA which subsumes Ambainis' scheme as a special case. In any event, the above complexity is achievable only when all the amplification schedules are pre-determined during compilation. Otherwise, we need to compile the algorithm by frequently invoking amplitude estimations. This makes VTAA substantially more complicated, and introduces polylogarithmic factors to the cost of both input algorithms and the initial state preparation.

%% file: tunable.tex
We now introduce \emph{Tunable Variable Time Amplitude Amplification} (Tunable VTAA) and analyze its performance. Specifically, we formally define Tunable VTAA in \sec{tunable_universal} and show its equivalence to a generic nested amplitude amplification algorithm reflecting toward the potentially good subspaces. We then analyze the query complexity of Tunable VTAA in \sec{tunable_query}. We present an optimized amplification schedule in \sec{tunable_optimal} under which Tunable VTAA has the lowest possible query cost.

For presentational purpose, we assume that values of all norms in our discussion are known a prior. We will comment on how this assumption can be relaxed at the end of \sec{tunable_optimal}.

%%%%%%%%%%%%%%%%%%%%%%%%%%%%%%%%%%%%%%%%%%%%%%%%%%%%%%%%%%%%%%%%%%%%%%%%%%%%%%
\subsection{Definition and universal property}
\label{sec:tunable_universal}
We begin with the formal definition of Tunable VTAA.
\begin{definition}[Tunable variable time amplitude amplification]
\label{defn:tunable}
A tunable variable time amplitude amplification is a $5$-tuple $\left(\{\Pi_j\}_{j=0}^m,\Pi_b,\{A_j\}_{j=0}^m,\{\alpha_j\}_{j=0}^m,\ket{\psi_0}\right)$ that additionally satisfies
\begin{enumerate}[label=(\roman*$''$)]
\setcounter{enumi}{3}
    \item \label{enum:tunable_amp} $\alpha_j$ are nonnegative real numbers for $j=1,\ldots,m$, and $\alpha_0=1$, which define
\begin{equation}
\begin{aligned}
    r_j&=\min\left\{r\in\mathbb Z_{\geq0}\ \big|\ (2r+1)\norm{\overline{\Pi_j\Pi_b}A_j\widetilde A_{j-1}\ket{\psi_0}}\geq\frac{1}{3}\sqrt{\alpha_j}\right\},\qquad j=1,\ldots,m\\
    \widetilde A_j
    &=
    \begin{cases}
        \left(-\left(I-2A_j\widetilde A_{j-1}\ketbra{\psi_0}{\psi_0}\widetilde A_{j-1}^\dagger A_j^\dagger\right)\left(I-2\overline{\Pi_j\Pi_b}\right)\right)^{r_j}A_j\widetilde A_{j-1},\quad&j=1,\ldots,m,\\
        I,&j=0.
    \end{cases}
\end{aligned}
\end{equation}
\end{enumerate}
\end{definition}
Compared with the standard VTAA, we have introduced the \emph{amplification thresholds} $\alpha_j$ in Tunable VTAA. These thresholds determine the amplification schedules and amplified algorithms via the specified recurrence. That is, the schedule $r_j$ is the smallest nonnegative integer satisfying $(2r_j+1)\norm{\overline{\Pi_j\Pi_b}A_j\widetilde A_{j-1}\ket{\psi_0}}\geq\frac{1}{3}\sqrt{\alpha_j}$; it has the following closed-form representation
    \begin{equation}
        r_j=\max\left\{\mathbf{Ceil}\left(\frac{\sqrt{\alpha_j}}{6\norm{\overline{\Pi_j\Pi_b}A_j\widetilde A_{j-1}\ket{\psi_0}}}-\frac{1}{2}\right),0\right\}.
    \end{equation}
When it happens $\norm{\overline{\Pi_j\Pi_b}A_j\widetilde A_{j-1}\ket{\psi_0}}\geq\frac{1}{3}\sqrt{\alpha_j}$, we choose $r_j=0$. In this case, we have a trivial step $\widetilde A_j=A_j\widetilde A_{j-1}$. Otherwise if $0<\norm{\overline{\Pi_j\Pi_b}A_j\widetilde A_{j-1}\ket{\psi_0}}<\frac{1}{3}\sqrt{\alpha_j}$, we have a nontrivial amplification with $r_j\geq1$. $\widetilde A_j$ is then obtained by performing $2r_j+1$ steps of amplitude amplification toward the potentially good subspace $\mathbf{Im}\left(\overline{\Pi_j\Pi_b}\right)$.

Apparently, Tunable VTAA offers more flexibility than the VTAA of Ambainis, due to the introduction of tunable threshold values $\alpha_j$ that can be optimized for different input algorithms and initial states, though it may not be immediately clear how much more powerful this modification is. In the following, we show that Tunable VTAA is universal in the sense that it is equivalent to a generic nested amplitude amplification with no overshoot and constant loss factor. This provides a compelling reason for us to examine Tunable VTAA in greater detail.

In the forward direction, let us suppose we have a Tunable VTAA with amplification thresholds $0\leq\alpha_j\leq1$ for all $j=1,\ldots,m$ and $\sum_{j=1}^m\alpha_j=\mathbf{O}(1)$. Then, in case we have a nontrivial amplification with $r_j\geq1$ and $0<\norm{\overline{\Pi_j\Pi_b}A_j\widetilde A_{j-1}\ket{\psi_0}}<\frac{1}{3}\sqrt{\alpha_j}$,
\begin{equation}
\begin{aligned}
    (2r_j+1)\norm{\overline{\Pi_j\Pi_b}A_j\widetilde A_{j-1}\ket{\psi_0}}
    &<\left(2\left(\frac{\sqrt{\alpha_j}}{6\norm{\overline{\Pi_j\Pi_b}A_j\widetilde A_{j-1}\ket{\psi_0}}}+\frac{1}{2}\right)+1\right)\norm{\overline{\Pi_j\Pi_b}A_j\widetilde A_{j-1}\ket{\psi_0}}\\
    &=\frac{\sqrt{\alpha_j}}{3}+2\norm{\overline{\Pi_j\Pi_b}A_j\widetilde A_{j-1}\ket{\psi_0}}\\
    &<\sqrt{\alpha_j}\leq 1,
\end{aligned}
\end{equation}
so there is no over amplification. If on the other hand $r_j=0$, then it trivially holds $(2\cdot0+1)\norm{\overline{\Pi_j\Pi_b}A_j\widetilde A_{j-1}\ket{\psi_0}}\leq 1$. We thus conclude that there is no overshoot in all the stages $j=1,\ldots,m$.
Now we lower bound the loss factor. Applying \prop{dirichlet} from \append{dirichlet} only to the nontrivial stages,
\begin{equation}
\begin{aligned}
    &\prod_{r_j\geq1}\left(\frac{\sin\left((2r_j+1)\arcsin\left(\norm{\overline{\Pi_j\Pi_b}A_j\widetilde A_{j-1}\ket{\psi_0}}\right)\right)}{(2r_j+1)\norm{\overline{\Pi_j\Pi_b}A_j\widetilde A_{j-1}\ket{\psi_0}}}\right)\\
    &\geq\prod_{r_j\geq1}\left(1-\frac{1}{6}(2r_j+1)^2\norm{\overline{\Pi_j\Pi_b}A_j\widetilde A_{j-1}\ket{\psi_0}}^2\right)
    \geq\exp\left(\sum_{r_j\geq1}\ln\left(1-\frac{\alpha_j}{6}\right)\right)\\
    &\geq\exp\left(\sum_{r_j\geq1}\ln\left(\frac{5}{6}\right)\alpha_j\right)
    =\left(\frac{5}{6}\right)^{\sum\limits_{r_j\geq1}\alpha_j}
    \geq\left(\frac{5}{6}\right)^{\sum_{j=1}^m\alpha_j}
    =\mathbf{\Omega}(1),
\end{aligned}
\end{equation}
where in the third inequality we have used the fact that $\ln(1-x)\geq\ln(1-c)\frac{x}{c}$ for $0<x\leq c<1$. Hence, the total loss factor of the nested amplitude amplification is at least constant.

Conversely, suppose we have a nested amplitude amplification with schedules $r_j$, no overshoot $\norm{\overline{\Pi_j\Pi_b}A_j\widetilde A_{j-1}\ket{\psi_0}}\leq\frac{1}{3(2r_j+1)}$, and constant total loss factor $\prod_{j=1}^m\frac{\norm{\overline{\Pi_j\Pi_b}\widetilde A_j\ket{\psi_0}}}{(2r_j+1)\norm{\overline{\Pi_j\Pi_b}A_j\widetilde A_{j-1}\ket{\psi_0}}}=\mathbf{\Omega}(1)$. Then, let us choose the following thresholds
\begin{equation}
    \alpha_j=
    \begin{cases}
        9(2r_j+1)^2\norm{\overline{\Pi_j\Pi_b}A_j\widetilde A_{j-1}\ket{\psi_0}}^2,\qquad& r_j\geq1,\\
        0,&r_j=0.
    \end{cases}
\end{equation}
Clearly, $\alpha_j\leq1$ which follows from the no overshoot condition. It is also verifiable that Tunable VTAA with thresholds $\alpha_j$ implements exactly a nested amplitude amplification with schedules $r_j$. Moreover, we upper bound the loss factor by applying \prop{dirichlet} to all the nontrivial stages
\begin{equation}
\begin{aligned}
    &\prod_{r_j\geq1}\left(\frac{\sin\left((2r_j+1)\arcsin\left(\norm{\overline{\Pi_j\Pi_b}A_j\widetilde A_{j-1}\ket{\psi_0}}\right)\right)}{(2r_j+1)\norm{\overline{\Pi_j\Pi_b}A_j\widetilde A_{j-1}\ket{\psi_0}}}\right)\\
    &\leq\prod_{r_j\geq1}\left(1-\frac{4\pi-8}{\pi^3}(2r_j+1)^2\norm{\overline{\Pi_j\Pi_b}A_j\widetilde A_{j-1}\ket{\psi_0}}^2\right)
    =\exp\left(\sum_{r_j\geq1}\ln\left(1-\frac{4\pi-8}{9\pi^3}\alpha_j\right)\right)\\
    &\leq\exp\left(-\sum_{r_j\geq1}\frac{4\pi-8}{9\pi^3}\alpha_j\right)
    =\frac{1}{\left(\exp\left(\frac{4\pi-8}{9\pi^3}\right)\right)^{\sum_{r_j\geq1}\alpha_j}}
    =\frac{1}{\left(\exp\left(\frac{4\pi-8}{9\pi^3}\right)\right)^{\sum_{j=1}^m\alpha_j}},
\end{aligned}
\end{equation}
where in the second inequality we use the fact that $\ln(1-x)\leq-x$ for $x<1$.
The assumption that the total loss factor is constant then implies that $\sum_{j=1}^m\alpha_j=\mathbf{O}(1)$. We have thus established:

\begin{proposition}[Universality of Tunable VTAA]
\label{prop:tunable_universal}
The following correspondence holds between Tunable VTAA and variable time nested amplitude amplification.
\begin{enumerate}[label=(\roman*)]
\item A Tunable VTAA $\left(\{\Pi_j\}_{j=0}^m,\Pi_b,\{A_j\}_{j=0}^m,\{\alpha_j\}_{j=0}^m,\ket{\psi_0}\right)$ according to \defn{tunable} where $\alpha_j\leq1$ and $\sum_{j=1}^m\alpha_j=\mathbf{O}(1)$, is a variable time nested amplitude amplification according to \defn{nested} where  $\norm{\overline{\Pi_j\Pi_b}A_j\widetilde A_{j-1}\ket{\psi_0}}\leq\frac{1}{2r_j+1}$ and $\prod_{j=1}^m\frac{\norm{\overline{\Pi_j\Pi_b}\widetilde A_j\ket{\psi_0}}}{(2r_j+1)\norm{\overline{\Pi_j\Pi_b}A_j\widetilde A_{j-1}\ket{\psi_0}}}=\mathbf{\Omega}(1)$.
\item A variable time nested amplitude amplification $\left(\{\Pi_j\}_{j=0}^m,\Pi_b,\{A_j\}_{j=0}^m,\{r_j\}_{j=1}^m,\ket{\psi_0}\right)$ according to \defn{nested} where $\norm{\overline{\Pi_j\Pi_b}A_j\widetilde A_{j-1}\ket{\psi_0}}\leq\frac{1}{3(2r_j+1)}$ and  $\prod_{j=1}^m\frac{\norm{\overline{\Pi_j\Pi_b}\widetilde A_j\ket{\psi_0}}}{(2r_j+1)\norm{\overline{\Pi_j\Pi_b}A_j\widetilde A_{j-1}\ket{\psi_0}}}=\mathbf{\Omega}(1)$, is some Tunable VTAA according to \defn{tunable} where $\alpha_j\leq1$ and $\sum_{j=1}^m\alpha_j=\mathbf{O}(1)$.
\end{enumerate}
\end{proposition}

%%%%%%%%%%%%%%%%%%%%%%%%%%%%%%%%%%%%%%%%%%%%%%%%%%%%%%%%%%%%%%%%%%%%%%%%%%%%%%
\subsection{Query complexity}
\label{sec:tunable_query}
We now analyze the performance of Tunable VTAA, taking into account the fact that only a few stages have nontrivial amplitude amplifications. 

Specifically, suppose there are nontrivial amplifications at stages $1\leq s_1\leq\cdots\leq s_{l}\leq m$, with the corresponding amplification step numbers $r_{s_1},\ldots,r_{s_{l}}\geq 1$. Then, the cost of nested amplitude amplification previously stated in \eq{nest_complexity} should be revised to
\begin{equation}
\label{eq:tunable_complexity}
    \mathbf{Cost}(\widetilde A_m\ket{\psi_0})
    =\mathbf{Cost}(\ket{\psi_0})\prod_{u=1}^{l}(2r_{s_u}+1)
    +\sum_{v=1}^{l+1}\mathbf{Cost}(A_{s_v}\cdots A_{s_{v-1}+1})\prod_{u=v}^{l}(2r_{s_u}+1),
\end{equation}
where we adopt the convention that $s_0=0$, $s_{l+1}=m+1$ and $A_{m+1}=I$. That is, the query cost is multiplied by $2r_{s_u}+1$ only at the $l$ nontrivial stages, and we merge the cost of input algorithms $\mathbf{Cost}(A_{k})$ for $k=s_{v-1}+1,\ldots,s_v$, associated with the same query product $\prod_{u=v}^{l}(2r_{s_u}+1)$.

Let us consider a general query product from stage $s_v$ to $s_w$ and re-express it as
\begin{equation}
\label{eq:query_product}
\begin{aligned}
    &\prod_{u=v}^{w}(2r_{s_u}+1)
    =\prod_{k=s_v}^{s_{w}}(2r_k+1)\\
    &=\prod_{k=s_v}^{s_{w}}\left(\frac{1}{\frac{\norm{\overline{\Pi_k\Pi_b}\widetilde A_k\ket{\psi_0}}}{(2r_k+1)\norm{\overline{\Pi_k\Pi_b}A_k\widetilde A_{k-1}\ket{\psi_0}}}}\right)
    \prod_{k=s_v}^{s_{w}}\left(\frac{\norm{\overline{\Pi_k\Pi_b}\widetilde A_k\ket{\psi_0}}}{\norm{\overline{\Pi_k\Pi_b}A_k\widetilde A_{k-1}\ket{\psi_0}}}\right)\\
    &=\prod_{k=s_v}^{s_{w}}\left(\frac{1}{\frac{\norm{\overline{\Pi_k\Pi_b}\widetilde A_k\ket{\psi_0}}}{(2r_k+1)\norm{\overline{\Pi_k\Pi_b}A_k\widetilde A_{k-1}\ket{\psi_0}}}}\right)
    \frac{\norm{\overline{\Pi_{s_{w}}\Pi_b}\widetilde A_{s_{w}}\ket{\psi_0}}}{\norm{\overline{\Pi_{s_v-1}\Pi_b}\widetilde{A}_{s_v-1}\ket{\psi_0}}}
    \frac{\norm{\overline{\Pi_{s_v-1}\Pi_b}A_{s_v-1}\cdots A_1\ket{\psi_0}}}{\norm{\overline{\Pi_{{s_{w}}}\Pi_b}A_{{s_{w}}}\cdots A_1\ket{\psi_0}}}\\
    &=\prod_{u=v}^{w}\left(\frac{1}{\frac{\norm{\overline{\Pi_{s_u}\Pi_b}\widetilde A_{s_u}\ket{\psi_0}}}{(2r_{s_u}+1)\norm{\overline{\Pi_{s_u}\Pi_b}A_{s_u}\widetilde A_{s_u-1}\ket{\psi_0}}}}\right)
    \frac{\norm{\overline{\Pi_{s_{w}}\Pi_b}\widetilde A_{s_{w}}\ket{\psi_0}}}{\norm{\overline{\Pi_{s_{v-1}}\Pi_b}\widetilde{A}_{s_{v-1}}\ket{\psi_0}}}
    \frac{\norm{\overline{\Pi_{s_{v-1}}\Pi_b}A_{s_{v-1}}\cdots A_1\ket{\psi_0}}}{\norm{\overline{\Pi_{s_{w}}\Pi_b}A_{s_{w}}\cdots A_1\ket{\psi_0}}}.
\end{aligned}
\end{equation}
Here, the first equality uses the observation that for $s_v\leq k\leq s_w$, either we have a nontrivial stage with $k=s_u$ for $v\leq u\leq w$, or we have a trivial stage with $r_k=0$. The second equality is a direct rewriting. The third equality can be reasoned using \eq{trans_amp} in a similar way as~\cite[Lemma 16]{Chakraborty2018BlockEncoding}:
\begin{equation}
\label{eq:amp_product}
\begin{aligned}
    \prod_{k=s_v}^{s_{w}}\left(\frac{\norm{\overline{\Pi_k\Pi_b}\widetilde A_k\ket{\psi_0}}}{\norm{\overline{\Pi_k\Pi_b}A_k\widetilde A_{k-1}\ket{\psi_0}}}\right)
    &=\prod_{k=s_v}^{s_{w}}\left(\frac{\norm{\overline{\Pi_k\Pi_b}\widetilde A_k\ket{\psi_0}}}{\norm{\overline{\Pi_{k-1}\Pi_b}\widetilde A_{k-1}\ket{\psi_0}}}
    \frac{\norm{\overline{\Pi_{k-1}\Pi_b}\widetilde A_{k-1}\ket{\psi_0}}}{\norm{\overline{\Pi_k\Pi_b}A_k\widetilde A_{k-1}\ket{\psi_0}}}\right)\\
    &=\prod_{k=s_v}^{s_{w}}\left(\frac{\norm{\overline{\Pi_k\Pi_b}\widetilde A_k\ket{\psi_0}}}{\norm{\overline{\Pi_{k-1}\Pi_b}\widetilde A_{k-1}\ket{\psi_0}}}
    \frac{\norm{\overline{\Pi_{k-1}\Pi_b}A_{k-1}\cdots A_1\ket{\psi_0}}}{\norm{\overline{\Pi_k\Pi_b}A_k\cdots A_1\ket{\psi_0}}}\right)\\
    &=\frac{\norm{\overline{\Pi_{s_{w}}\Pi_b}\widetilde A_{s_{w}}\ket{\psi_0}}}{\norm{\overline{\Pi_{s_v-1}\Pi_b}\widetilde{A}_{s_v-1}\ket{\psi_0}}}
    \frac{\norm{\overline{\Pi_{s_v-1}\Pi_b}A_{s_v-1}\cdots A_1\ket{\psi_0}}}{\norm{\overline{\Pi_{{s_{w}}}\Pi_b}A_{{s_{w}}}\cdots A_1\ket{\psi_0}}}.
\end{aligned}
\end{equation}

The justification of the last equality requires some more efforts. In the first factor, we have changed the multiplication variable from $k$ back to $u$, using again the observation that for $s_v\leq k\leq s_w$, either we have a nontrivial stage with $k=s_u$ for $v\leq u\leq w$, or we have a trivial stage with $r_k=0$ and loss factor $\frac{\norm{\overline{\Pi_k\Pi_b}\widetilde A_k\ket{\psi_0}}}{(2r_k+1)\norm{\overline{\Pi_k\Pi_b}A_k\widetilde A_{k-1}\ket{\psi_0}}}=1$. 
Now recall that $s_{v-1}+1,\ldots,s_v-1$ are all trivial stages by definition. 
Thus, invoking \eq{trans_amp},
\begin{small}
\begin{equation}
\newmaketag
    \frac{\norm{\overline{\Pi_{s_v-1}\Pi_b}A_{s_v-1}\cdots A_1\ket{\psi_0}}}{\norm{\overline{\Pi_{s_v-1}\Pi_b}\widetilde{A}_{s_v-1}\ket{\psi_0}}}
    =\frac{\norm{\overline{\Pi_{s_v-1}\Pi_b}A_{s_v-1}\cdots A_1\ket{\psi_0}}}{\norm{\overline{\Pi_{s_v-1}\Pi_b}{A}_{s_v-1}\cdots A_{s_{v-1}+1}\widetilde{A}_{s_{v-1}}\ket{\psi_0}}}
    =\frac{\norm{\overline{\Pi_{s_{v-1}}\Pi_b}A_{s_{v-1}}\cdots A_1\ket{\psi_0}}}{\norm{\overline{\Pi_{s_{v-1}}\Pi_b}\widetilde{A}_{s_{v-1}}\ket{\psi_0}}}.
\end{equation}
\end{small}%
This proves the claimed representation of query product.
As this representation may be of independent interest, let us encapsulate it into a lemma.

\begin{lemma}[Representation of query product]
\label{lem:rep_query}
Let $\left(\{\Pi_j\}_{j=0}^m,\Pi_b,\{A_j\}_{j=0}^m,\{r_j\}_{j=1}^m,\ket{\psi_0}\right)$ be a variable time nested amplitude amplification according to \defn{nested}. Suppose that $r_k\geq 1$ happens only at $l$ stages $1\leq s_1\leq\cdots\leq s_{l}\leq m$.
Then,
\begin{small}
\begin{equation}
\newmaketag
\begin{aligned}
    \prod_{u=v}^{w}(2r_{s_u}+1)
    =\prod_{u=v}^{w}\left(\frac{1}{\frac{\norm{\overline{\Pi_{s_u}\Pi_b}\widetilde A_{s_u}\ket{\psi_0}}}{(2r_{s_u}+1)\norm{\overline{\Pi_{s_u}\Pi_b}A_{s_u}\widetilde A_{s_u-1}\ket{\psi_0}}}}\right)
    \frac{\norm{\overline{\Pi_{s_{w}}\Pi_b}\widetilde A_{s_{w}}\ket{\psi_0}}}{\norm{\overline{\Pi_{s_{v-1}}\Pi_b}\widetilde{A}_{s_{v-1}}\ket{\psi_0}}}
    \frac{\norm{\overline{\Pi_{s_{v-1}}\Pi_b}A_{s_{v-1}}\cdots A_1\ket{\psi_0}}}{\norm{\overline{\Pi_{s_{w}}\Pi_b}A_{s_{w}}\cdots A_1\ket{\psi_0}}}
\end{aligned}
\end{equation}
\end{small}%
for $1\leq v\leq w\leq l$, under the convention that $s_0=0$.
\end{lemma}

We now apply this representation to the query products in \eq{tunable_complexity}, all of which end with $w=l$.
By definition, $s_l$ is the last nontrivial stage so $s_{l}+1,\ldots,m$ are all trivial. 
Invoking \eq{trans_amp} once more,
\begin{equation}
    \frac{\norm{\overline{\Pi_{s_{l}}\Pi_b}A_{s_{l}}\cdots A_1\ket{\psi_0}}}{\norm{\overline{\Pi_{s_{l}}\Pi_b}\widetilde{A}_{s_{l}}\ket{\psi_0}}}
    =\frac{\norm{\overline{\Pi_{m}\Pi_b}A_{m}\cdots A_1\ket{\psi_0}}}{\norm{\overline{\Pi_{m}\Pi_b}A_m\cdots A_{s_{l}+1}\widetilde{A}_{s_{l}}\ket{\psi_0}}}
    =\frac{\norm{\overline{\Pi_{m}\Pi_b}A_{m}\cdots A_1\ket{\psi_0}}}{\norm{\overline{\Pi_{m}\Pi_b}\widetilde{A}_{m}\ket{\psi_0}}},
\end{equation}
giving
\begin{small}
\begin{equation}
\newmaketag
\begin{aligned}
    \prod_{u=v}^{l}(2r_{s_u}+1)
    =\prod_{u=v}^{l}\left(\frac{1}{\frac{\norm{\overline{\Pi_{s_u}\Pi_b}\widetilde A_{s_u}\ket{\psi_0}}}{(2r_{s_u}+1)\norm{\overline{\Pi_{s_u}\Pi_b}A_{s_u}\widetilde A_{s_u-1}\ket{\psi_0}}}}\right)
    \frac{\norm{\overline{\Pi_{m}\Pi_b}\widetilde A_{m}\ket{\psi_0}}}{\norm{\overline{\Pi_{s_{v-1}}\Pi_b}\widetilde{A}_{s_{v-1}}\ket{\psi_0}}}
    \frac{\norm{\overline{\Pi_{s_{v-1}}\Pi_b}A_{s_{v-1}}\cdots A_1\ket{\psi_0}}}{\norm{\overline{\Pi_{m}\Pi_b}A_{m}\cdots A_1\ket{\psi_0}}}.
\end{aligned}
\end{equation}
\end{small}%
To proceed, note that the first factor is the inverse loss factor, which can be bounded by $\left(\frac{6}{5}\right)^{\sum_{j=1}^m\alpha_j}$. Then, we have the amplitude of post-amplified algorithm trivially bounded by $\norm{\overline{\Pi_{m}\Pi_b}\widetilde A_{m}\ket{\psi_0}}\leq1$. In practice, we would also adjust the last stage of VTAA so that the amplitude is at least constant. Next, $\norm{\overline{\Pi_{m}\Pi_b}A_{m}\cdots A_1\ket{\psi_0}}=\sqrt{p_{\text{succ}}}$ by definition. We now lower bound $\norm{\overline{\Pi_{s_{v-1}}\Pi_b}\widetilde{A}_{s_{v-1}}\ket{\psi_0}}$:
\begin{equation}
\begin{aligned}
    \norm{\overline{\Pi_{s_{v-1}}\Pi_b}\widetilde{A}_{s_{v-1}}\ket{\psi_0}}
    &=\sin\left((2r_{s_{v-1}}+1)\arcsin\left(\norm{\overline{\Pi_{s_{v-1}}\Pi_b}A_{s_{v-1}}\widetilde A_{{s_{v-1}}-1}\ket{\psi_0}}\right)\right)\\
    &\geq\left(\frac{5}{6}\right)^{\alpha_{s_{v-1}}}(2r_{s_{v-1}}+1)\norm{\overline{\Pi_{s_{v-1}}\Pi_b}A_{s_{v-1}}\widetilde A_{{s_{v-1}}-1}\ket{\psi_0}}\\
    &=\mathbf{\Omega}\left(\sqrt{\alpha_{s_{v-1}}}\right).
\end{aligned}
\end{equation}
This analysis holds for all query products starting at $v=1,\ldots,l$. In particular, for the cost of initial state preparation, we have $v=1$, $\norm{\overline{\Pi_0\Pi_b}\widetilde A_0\ket{\psi_0}}=\norm{\overline{\Pi_0\Pi_b}A_0\ket{\psi_0}}=1$.
Altogether, this gives the query complexity of Tunable VTAA
\begin{equation}
\begin{aligned}
    \mathbf{Cost}(\widetilde A_m\ket{\psi_0})
    &=\mathbf{O}\Bigg(\frac{1}{\sqrt{p_{\text{succ}}}}\mathbf{Cost}(A_{s_1}\cdots A_{1}\ket{\psi_0})\\
    &\qquad\quad+\frac{1}{\sqrt{p_{\text{succ}}}}\sum_{v=2}^{l+1}\frac{1}{\sqrt{\alpha_{s_{v-1}}}}
    \norm{\overline{\Pi_{s_{v-1}}\Pi_b}A_{s_{v-1}}\cdots A_1\ket{\psi_0}}\mathbf{Cost}\left(A_{s_v}\cdots A_{s_{v-1}+1}\right)\Bigg).
\end{aligned}
\end{equation}

Note that VTAA would have the same cost if we were to pre-merge the input algorithms $A_{s_v}\cdots A_{s_{v-1}+1}$ during compilation, resulting in a ($l+1$)-stage algorithm. We now claim that $l=\mathbf{O}\left(\log\left(\frac{1}{\sqrt{p_{\text{succ}}}}\right)\right)$. Indeed, this follows from the observation that
\begin{equation}
    3^l
    \leq\prod_{u=1}^{l}(2r_{s_u}+1)
    \leq\left(\frac{6}{5}\right)^{\sum_{j=1}^m\alpha_j}
    \frac{1}{\sqrt{p_{\text{succ}}}}
    =\mathbf{O}\left(\frac{1}{\sqrt{p_{\text{succ}}}}\right)
    \quad\Rightarrow\quad
    l=\mathbf{O}\left(\log_3\left(\frac{1}{\sqrt{p_{\text{succ}}}}\right)\right).
\end{equation}
Therefore, in the interesting regime where $\log\left(\frac{1}{\sqrt{p_{\text{succ}}}}\right)\ll m$, majority of the algorithms can be pre-merged yielding a significantly simplified VTAA.

%%%%%%%%%%%%%%%%%%%%%%%%%%%%%%%%%%%%%%%%%%%%%%%%%%%%%%%%%%%%%%%%%%%%%%%%%%%%%%
\subsection{\texorpdfstring{$\ell_{\frac{2}{3}}$}{l2/3}-quasinorm scaling}
\label{sec:tunable_optimal}
Having established the universal property and query complexity of Tunable VTAA, we now explore its limitation. In particular, we will show that complexity of Tunable VTAA can be optimized to scale with $\ell_{\frac{2}{3}}$-quasinorm of the input cost.

To this end, we note that the query cost of initial state preparation and algorithm $A_{s_1}\cdots A_1$ has the scaling $\mathbf{O}\left(\frac{1}{\sqrt{p_{\text{succ}}}}\right)$ independent of the amplification schedules, so it suffices to optimize query complexity of the remaining input algorithms. Up to a rescaling, this reduces to solving the following problem:
\begin{equation}
\begin{aligned}
    &\mathrm{minimize} \quad && \sum_{v=2}^{l+1}\frac{1}{\sqrt{\alpha_{s_{v-1}}}}
    \norm{\overline{\Pi_{s_{v-1}}\Pi_b}A_{s_{v-1}}\cdots A_1\ket{\psi_0}}\mathbf{Cost}\left(A_{s_v}\cdots A_{s_{v-1}+1}\right)\\
    &\textrm{subject to}\quad &&\sum_{v=2}^{l+1}\alpha_{s_{v-1}}=1,\\
    & &&\alpha_{s_{v-1}}>0,\quad v=2,\ldots,l+1.
\end{aligned}
\end{equation}

This optimization problem can be readily solved by the weighted mean inequality. Specifically, given positive real numbers $\{w_v\}_{v=1}^l$ and $\{x_v\}_{v=1}^l$, the inequality asserts that the weighted harmonic mean is always upper bounded by the weighted quadratic mean~\cite[Problem 8.3]{steele2004cauchy}:
\begin{equation}
    \frac{\sum\limits_{v=1}^{l}w_v}{\sum\limits_{v=1}^{l}w_v\frac{1}{x_v}}
    \leq\sqrt{\frac{\sum\limits_{v=1}^{l}w_vx_v^2}{\sum\limits_{v=1}^{l}w_v}}
    \quad\Leftrightarrow\quad
    \sum\limits_{v=1}^{l}w_v\frac{1}{x_v}\geq\frac{\left(\sum\limits_{v=1}^{l}w_v\right)^{\frac{3}{2}}}{\sqrt{\sum\limits_{v=1}^{l}w_vx_v^2}}
\end{equation}
with equality if and only if $x_1=\cdots=x_{l}$. Specialized to our problem, let us choose 
\begin{equation}
\begin{aligned}
    w_{v-1}&=\left(\norm{\overline{\Pi_{s_{v-1}}\Pi_b}A_{s_{v-1}}\cdots A_1\ket{\psi_0}}\mathbf{Cost}\left(A_{s_v}\cdots A_{s_{v-1}+1}\right)\right)^{\frac{2}{3}},\\
    x_{v-1}&=
    \sqrt{\alpha_{s_v-1}}\left(\norm{\overline{\Pi_{s_{v-1}}\Pi_b}A_{s_{v-1}}\cdots A_1\ket{\psi_0}}\mathbf{Cost}\left(A_{s_v}\cdots A_{s_{v-1}+1}\right)\right)^{-\frac{1}{3}},
\end{aligned}
\end{equation}
for $v=2,\ldots,l+1$.
Then, the weighted mean inequality implies that
\begin{equation}
\begin{aligned}
    &\sum_{v=2}^{l+1}\frac{1}{\sqrt{\alpha_{s_{v-1}}}}
    \norm{\overline{\Pi_{s_{v-1}}\Pi_b}A_{s_{v-1}}\cdots A_1\ket{\psi_0}}\mathbf{Cost}\left(A_{s_v}\cdots A_{s_{v-1}+1}\right)\\
    &\geq\left(\sum_{v=2}^{l+1}
    \left(\norm{\overline{\Pi_{s_{v-1}}\Pi_b}A_{s_{v-1}}\cdots A_1\ket{\psi_0}}\mathbf{Cost}\left(A_{s_v}\cdots A_{s_{v-1}+1}\right)\right)^{\frac{2}{3}}\right)^{\frac{3}{2}}
\end{aligned}
\end{equation}
with the lower bound attained when
\begin{equation}
    \alpha_{s_v-1}\propto
    \left(\norm{\overline{\Pi_{s_{v-1}}\Pi_b}A_{s_{v-1}}\cdots A_1\ket{\psi_0}}\mathbf{Cost}\left(A_{s_v}\cdots A_{s_{v-1}+1}\right)\right)^{\frac{2}{3}}.
\end{equation}
We have thus established:

\begin{proposition}[Query complexity of Tunable VTAA]
\label{prop:tunable_query}
Let $\left(\{\Pi_j\}_{j=0}^m,\Pi_b,\{A_j\}_{j=0}^m,\{\alpha_j\}_{j=0}^m,\ket{\psi_0}\right)$ be a Tunable VTAA according to \defn{tunable}. Then,
\begin{enumerate}[label=(\roman*)]
\item Assuming the threshold values satisfy $0\leq\alpha_j\leq1$ and $\sum_{j=1}^m\alpha_j=\mathbf{O}(1)$, nontrivial amplifications happen only at $l$ stages $1\leq s_1\leq\cdots\leq s_{l}\leq m$, where
\begin{equation}
    l=\mathbf{O}\left(\log_3\left(\frac{1}{\sqrt{p_{\text{succ}}}}\right)\right),\qquad
    p_{\text{succ}}=\norm{\overline{\Pi_{m}\Pi_b}A_{m}\cdots A_1\ket{\psi_0}}^2.
\end{equation}
Under the convention that $s_{l+1}=m+1$ and $A_{m+1}=I$, Tunable VTAA has query complexity
\begin{equation}
\hspace{-0.75cm}
\begin{aligned}
    \mathbf{Cost}(\widetilde A_m\ket{\psi_0})
    &=\mathbf{O}\Bigg(\frac{1}{\sqrt{p_{\text{succ}}}}\mathbf{Cost}(A_{s_1}\cdots A_{1}\ket{\psi_0})\\
    &\qquad\quad+\frac{1}{\sqrt{p_{\text{succ}}}}\sum_{v=2}^{l+1}\frac{1}{\sqrt{\alpha_{s_{v-1}}}}
    \norm{\overline{\Pi_{s_{v-1}}\Pi_b}A_{s_{v-1}}\cdots A_1\ket{\psi_0}}\mathbf{Cost}\left(A_{s_v}\cdots A_{s_{v-1}+1}\right)\Bigg).
\end{aligned}
\end{equation}
\item Pre-merging trivial stages and using thresholds
\begin{equation}
    \alpha_{s_{v-1}}\propto
    \left(\norm{\overline{\Pi_{s_{v-1}}\Pi_b}A_{s_{v-1}}\cdots A_1\ket{\psi_0}}\mathbf{Cost}\left(A_{s_v}\cdots A_{s_{v-1}+1}\right)\right)^{\frac{2}{3}},
\end{equation}
the complexity of Tunable VTAA can be minimized to attain the \emph{$\ell_{\frac{2}{3}}$-quasinorm} scaling
\begin{equation}
\hspace{-.75cm}
\begin{aligned}
    \mathbf{Cost}(\widetilde A_m\ket{\psi_0})
    &=
    \mathbf{O}\Bigg(\frac{1}{\sqrt{p_{\text{succ}}}}\mathbf{Cost}(A_{s_1}\cdots A_{1}\ket{\psi_0})\\
    &\qquad\quad+\frac{1}{\sqrt{p_{\text{succ}}}}
    \left(\sum_{v=2}^{l+1}\left(\norm{\overline{\Pi_{s_{v-1}}\Pi_b}A_{s_{v-1}}\cdots A_1\ket{\psi_0}}\mathbf{Cost}\left(A_{s_v}\cdots A_{s_{v-1}+1}\right)\right)^{\frac{2}{3}}\right)^{\frac{3}{2}}\Bigg).
\end{aligned}
\end{equation}
\end{enumerate}
    
\end{proposition}

In our above discussion, we have used precise values of all the norms involved to simplify the presentation, but they can be replaced by their constant multiplicative approximators without affecting the asymptotic analysis. 
If this prior knowledge is not available, one could perform nested amplitude estimation like in~\cite{Ambainis2012VTAA,Chakraborty2018BlockEncoding}, introducing logarithmic overhead and substantially complicating the structure of algorithm.
However, the advantage of using Tunable VTAA is that this nested amplitude estimation may be completely avoided if we choose the threshold values analytically. We will demonstrate this feature in \sec{dinv} for solving the quantum linear system problem.

%% file: dinv.tex
In this section, we introduce a \emph{discretized inverse state} for solving the quantum linear system problem, which can be efficiently prepared by Tunable VTAA. We begin with its construction in \sec{dinv_branch} by performing \emph{Gapped Phase Estimation} (GPE) on quantum walk, taking special care of the issue that eigenphases of the walk operator are split into two branches with opposite $\pm$ signs. We then present a deterministic amplification schedule in \sec{dinv_deterministic} that significantly simplifies the structure of VTAA.

To further reduce the query complexity, we present an improved analysis of VTAA in \sec{dinv_err} by projecting error onto the potentially good subspaces. Finally, we develop a simple solution norm estimation algorithm in \sec{dinv_est} to estimate $p_{\text{succ}}$ when a constant multiplicative approximation of it is unavailable a prior.

%%%%%%%%%%%%%%%%%%%%%%%%%%%%%%%%%%%%%%%%%%%%%%%%%%%%%%%%%%%%%%%%%%%%%%%%%%%%%%
\subsection{Gapped phase estimation with branch marking}
\label{sec:dinv_branch}
In the quantum linear system problem, the coefficient matrix is block encoded by oracle $O_A$ as $A/\alpha_A$ with normalization factor $\alpha_A\geq\norm{A}$, and an upper bound on its inverse $\alpha_{A^{-1}}\geq\norm{A^{-1}}$ is known. We now make a few more assumptions to simplify the discussion without affecting generality. First, we assume that $A$ is Hermitian and $O_A$ is a Hermitian unitary. We can always fulfill this requirement by considering the Hermitian dilation $\ketbra{0}{1}\otimes A+\ketbra{1}{0}\otimes A^\dagger$ block encoded by the Hermitian unitary $\ketbra{0}{1}\otimes O_A+\ketbra{1}{0}\otimes O_A^\dagger$ and two identical isometries, along with the corresponding initial state $\ket{0}\ket{b}$. Second, we assume $\alpha_A\geq2\norm{A}$, which can be achieved by block encoding a constant factor to artificially increase the normalization factor. Third, we assume that both $\alpha_A$ and $\alpha_{A^{-1}}$ are integer powers of $3$. Again the former can be satisfied by rescaling the block encoding, whereas the latter can be realized using a slightly looser bound at most a constant factor $3$ larger.
As an immediate consequence, the upper bound $\kappa=\alpha_A\alpha_{A^{-1}}$ on the condition number is also an integer power of $3$, and thus the number of VTAA stages $m=\log_3(\kappa)$ is an integer. Moreover, for any eigenvalue $\lambda_u$ of $A$, the block encoded operator $\frac{A}{\alpha_A}$ has an eigenvalue $\frac{\lambda}{\alpha_A}$ such that $\abs{\frac{\lambda}{\alpha_A}}\in\left[\frac{1}{3^{m}},1\right)$.

Our algorithm proceeds by applying quantum signal processing to the quantum walk operator. Specifically, assume that the coefficient matrix is block encoded as $\frac{A}{\alpha_A}=G^\dagger O_AG$ for some isometry $G$. Then, the quantum walk operator is defined by $W=\left(2GG^\dagger-I\right)O_A$. In \append{qubitization}, we give a self-contained exposition of the qubitization result~\cite{Low2016Qubitization} linking a block encoding to its walk operator. In particular, we have the following spectrum correspondence: if the input matrix has the spectral decomposition
\begin{equation}
    A=\sum_u\lambda_u\ketbra{\phi_u}{\phi_u},
\end{equation}
then the walk operator has the corresponding spectral decomposition
\begin{equation}
    W
    =\sum_{u}\left(e^{+ i\arccos\left(\frac{\lambda_u}{\alpha_A}\right)}\ketbra{\phi_{u,Y+}}{\phi_{u,Y+}}
    +e^{-i\arccos\left(\frac{\lambda_u}{\alpha_A}\right)}\ketbra{\phi_{u,Y-}}{\phi_{u,Y-}}\right)
\end{equation}
when restricted to $\mathbf{Im}\left(GG^\dagger\right)+\mathbf{Im}\left(O_AGG^\dagger O_A^\dagger\right)$, where
\begin{equation}
    \ket{\phi_{u,0}}=G\ket{\phi_u},\qquad
    \ket{\phi_{u,1}}=\frac{O_AG\ket{\phi_u}-\lambda_uG\ket{\phi_u}}{\sqrt{1-\lambda_u^2}},\qquad
    \ket{\phi_{u,Y\pm}}=\frac{\ket{\phi_{u,0}}\pm i\ket{\phi_{u,1}}}{\sqrt{2}}.
\end{equation}
Note that we have omitted the degenerate $1$-dimensional subspaces in qubitization, because all our eigenvalues satisfy $\abs{\frac{\lambda_u}{\alpha_A}}<1$ following the rescaling assumption at the beginning of this subsection.

Now, suppose that the initial state can be expanded in the eigenbasis of $A$ as $\ket{b}=\sum_u\gamma_u\ket{\phi_u}$. We can then apply $G$ and expand it in the eigenbasis of $W$ as
\begin{equation}
    G\ket{b}=\sum_u\gamma_uG\ket{\phi_u}
    =\sum_u\gamma_u\ket{\phi_{u,0}}
    =\sum_u\gamma_u\frac{\ket{\phi_{u,Y+}}+\ket{\phi_{u,Y-}}}{\sqrt{2}}.
\end{equation}
Next, we apply GPE to label the interval to which every eigenvalue belongs. In earlier work~\cite{Childs2015LinearSystems}, GPE is performed on the time evolution operator $e^{iA}$, which introduces the overhead of Hamiltonian simulation. More recent work~\cite{Chakraborty2023VTAAQLSP} suggested applying GPE directly on the quantum walk operator $W$.  In our notation, this would produce an unnormalized state of the form
\begin{equation}
\label{eq:gpe_no_bm}
\begin{aligned}
    \frac{1}{\sqrt{2}}&\sum_{k=0}^{m-1}\sum_{\abs{\frac{\lambda_u}{\alpha_A}}\in\big[\frac{1}{3^{k+1}},\frac{1}{3^k}\big)}
    \Bigg(\left(\zeta_{k+1,u,+}\frac{3^{k+1}}{3^m}\ket{k}+\zeta_{k,u,+}\frac{3^k}{3^m}\ket{k-1}\right)\gamma_u\ket{\phi_{u,Y+}}\\
    &\qquad\qquad\qquad\qquad\quad+\left(\zeta_{k+1,u,-}\frac{3^{k+1}}{3^m}\ket{k}+\zeta_{k,u,-}\frac{3^k}{3^m}\ket{k-1}\right)\gamma_u\ket{\phi_{u,Y-}}\Bigg).
\end{aligned}
\end{equation}
As the coefficients $\zeta_{k,u,+}$  are different from $\zeta_{k,u,-}$ in general, the original uniform superposition of $\ket{\phi_{u,Y+}}$ and $\ket{\phi_{u,Y-}}$ is ``distorted'' by GPE into a non-uniform superposition. Moreover, this distortion on the $\pm$ branches of quantum walk is not recorded by any ancilla state, and so it is not possible to perform the correct amplitude amplification. This then leads to a highly inaccurate solution state, a critical issue left unattended in~\cite{Chakraborty2023VTAAQLSP}.

We address this issue in \append{gpe_bm} by a gadget called \emph{branch marking}. Up to a controllable error, branch marking realizes the transformation
\begin{equation}
    \ket{+}\ket{\phi_{u,Y\pm}}\mapsto\ket{\pm}\ket{\phi_{u,Y\pm}},
\end{equation}
recording information about the signs of eigenphases of the walk operator into an ancilla register. Given thresholds $\gamma,\frac{\gamma}{\rho}$, we then perform GPE controlled by the branch register to transform
\begin{equation}
    \ket{0}\ket{\pm}\ket{\phi_{u,Y\pm}}\mapsto\ket{\xi_{u}}\ket{\pm}\ket{\phi_{u,Y\pm}},\qquad
    \ket{\xi_{u}}=\xi_{u,0}\ket{0}+\xi_{u,1}\ket{1},
\end{equation}
such that
\begin{equation}
    \ket{\xi_u}\approx
\begin{cases}
    \ket{0},\quad&\frac{\lambda_u}{\alpha_A}\in[\gamma,1),\\
    i\ket{1},&\frac{\lambda_u}{\alpha_A}\in\left[-\frac{\gamma}{\rho},\frac{\gamma}{\rho}\right],\\
    -\ket{0},&\frac{\lambda_u}{\alpha_A}\in\left(-1,-\gamma\right],
\end{cases}
\end{equation}
Importantly, the output state $\ket{\xi_u}$ has \emph{no} dependence on the specific $\pm$ branch of quantum walk, regardless of whether the original eigenvalue $\lambda_u$ is in the passband, stopband or transition band of GPE---this is the primary purpose of introducing branch marking.

We now describe a variable time quantum algorithm to prepare the discretized inverse state. Our algorithm acts on a clock register $\ket{j}$ holding values $j=0,\ldots,m-1$, a two-qubit flag register with possible states
\begin{equation}
    \ket{\mathrm{good}}=\ket{00},\qquad
    \ket{\mathrm{bad}}=\ket{10},\qquad
    \ket{\mathrm{cont'd}}=\ket{01},
\end{equation}
and a system register holding the solution state and ancilla state for block encoding, a single qubit $\ket{\pm}$ for branch marking, and any additional ancillas consumed by the branch marking and GPE. We introduce accuracy parameters $\epsilon_{\text{bm}},\epsilon_{\text{gpe},j}$ to be specified later. The algorithm is then constructed as follows.
\begin{enumerate}
    \item Initialization: 
    \begin{enumerate}
        \item Set the state of clock register to be $\ket{0}$.
        \item Set the state of flag register to be $\ket{\mathrm{cont'd}}$. 
        \item Apply the branch marking unitary from \prop{bm} to implement the transformation $\ket{+}\ket{\phi_{u,Y\pm}}\mapsto\ket{\pm}\ket{\phi_{u,Y\pm}}$ with accuracy $\epsilon_{\text{bm}}$.
    \end{enumerate}
    \item Discretized inversion: for $j=1,\ldots,m-1$,
    \begin{enumerate}
        \item Controlled on the clock state $\ket{j-1}$, flip the second qubit of the flag register and apply GPE from \prop{gpe} with $\gamma=\frac{1}{3^j}$, $\rho=3$ and accuracy $\epsilon_{\text{gpe},j}$, storing outcome back to the second qubit of flag register. Up to error $\epsilon_{\text{gpe},j}$, this implements the transformation $\ket{\mathrm{cont'd}}\ket{\pm}\ket{\phi_{u,Y\pm}}\mapsto\left(\xi_{j,u,0}\ket{\mathrm{good}}+\xi_{j,u,1}\ket{\mathrm{cont'd}}\right)\ket{\pm}\ket{\phi_{u,Y\pm}}$, where $\xi_{j,u,1}=0$ if $\frac{1}{3^j}\leq\abs{\frac{\lambda_u}{\alpha_A}}<\frac{1}{2}$, and $\xi_{j,u,0}=0$ if $0\leq\abs{\frac{\lambda_u}{\alpha_A}}<\frac{1}{3^{j+1}}$. Note that both $\pm$ branches pick up the same coefficients $\xi_{j,u,0}$ and $\xi_{j,u,1}$ in the outcome state.
        \item Controlled on the state $\ket{j-1}$, implement the mapping $\ket{\mathrm{cont'd}}\mapsto\ket{\mathrm{cont'd}},\ket{\mathrm{good}}\mapsto\frac{3^j}{3^m}\ket{\mathrm{good}}+\sqrt{1-\frac{9^j}{9^m}}\ket{\mathrm{bad}}$. This can in turn be achieved by applying a Pauli-$Y$ rotation on the first qubit of flag register, controlled by the second qubit in state $\ket{0}$.
        \item Controlled on the flag state $\ket{\mathrm{cont'd}}$, increment the clock state $\ket{j}\mapsto\ket{j+1}$.
    \end{enumerate}
    \item Finalization:
    \begin{enumerate}
        \item Controlled on the clock state $\ket{m-1}$, apply the transformation $\ket{\mathrm{cont'd}}\mapsto\ket{\mathrm{good}}$. This can be achieved by applying a Pauli-$X$ gate on the second qubit of flag register.
        \item Undo the branch marking by invoking the reversal of \prop{bm}.
    \end{enumerate}
\end{enumerate}

Let us first confirm that the above is indeed a variable time quantum algorithm in the sense of \defn{vta}. To this end, we define the clock projections
\begin{equation}
    \Pi_j=\sum_{0\leq x\leq j-1}\ketbra{x}{x}\otimes\left(\ketbra{\mathrm{good}}{\mathrm{good}}+\ketbra{\mathrm{bad}}{\mathrm{bad}}\right)
\end{equation}
for $j=0,\ldots,m$. Note that $\Pi_m=I$ holds effectively since the last step of our algorithm always maps $\ket{\mathrm{cont'd}}$ to $\ket{\mathrm{good}}$. The flag projection is naturally selected to be
\begin{equation}
    \Pi_b=\ketbra{\mathrm{bad}}{\mathrm{bad}}.
\end{equation}
Finally, we have the following input algorithm at stage $j$ (only showing its nontrivial actions):
\begin{equation}
\begin{aligned}
    C_j&=\ketbra{j}{j-1}\otimes\sum_u\xi_{j,u,1}\ketbra{\mathrm{cont'd}}{\mathrm{cont'd}}\otimes\ketbra{\pm,\phi_{u,Y\pm}}{\pm,\phi_{u,Y\pm}}\\
    &\quad+\ketbra{j-1}{j-1}\otimes\sum_u\xi_{j,u,0}\left(\frac{3^j}{3^m}\ket{\mathrm{good}}+\sqrt{1-\frac{9^j}{9^m}}\ket{\mathrm{bad}}\right)\bra{\mathrm{cont'd}}\otimes\ketbra{\pm,\phi_{u,Y\pm}}{\pm,\phi_{u,Y\pm}}.\\
\end{aligned}
\end{equation}
It is then a routine verification that $C_j\Pi_{j-1}=\Pi_{j-1}$ holds for all $j=1,\ldots,m-1$. This is also true for $j=m$ if we define $C_m=\ketbra{m-1}{m-1}\otimes\ketbra{\mathrm{good}}{\mathrm{cont'd}}$ (corresponding to setting $\xi_{m,u,0}=1$ for all $u$). For notational convenience, we may assume that branch marking and its inverse are incorporated into $C_1$ and $C_m$, respectively.

To simplify the discussion, we consider the case where the branch marking can be performed perfectly, corresponding to input algorithms $B_1,\ldots,B_m$. This produces the output state
\begin{equation}
\begin{aligned}
    \frac{1}{\sqrt{2}}\sum_{k=0}^{m-1}\sum_{\abs{\frac{\lambda_u}{\alpha_A}}\in\big[\frac{1}{3^{k+1}},\frac{1}{3^k}\big)}
    \sum_{j=1}^m
    \zeta_{j,u}
    \ket{j-1}\left(\frac{3^{j}}{3^m}\ket{\mathrm{good}}+\sqrt{1-\frac{9^{j}}{9^m}}\ket{\mathrm{bad}}\right)
    \gamma_u\left(\ket{+,\phi_{u,Y+}}+\ket{-,\phi_{u,Y-}}\right),
\end{aligned}
\end{equation}
where $\zeta_{j,u}$ are the cumulative coefficients at stages $j$
\begin{equation}
    \zeta_{j,u}=\xi_{j,u,0}\prod_{l=1}^{j-1}\xi_{l,u,1}.
\end{equation}
Moreover, we know that GPE has the action:
\begin{enumerate}[label=(\roman*)]
    \item For $\abs{\frac{\lambda_u}{\alpha_A}}\in\left[\frac{1}{3^{k+1}},\frac{1}{3^k}\right)\subseteq\left[\frac{1}{3^{l}},1\right)$, or equivalently, $l\geq k+1$, it holds $\xi_{l,u,1}\approx0$ and $\xi_{l,u,0}\approx\pm1$.
\addtocounter{equation}{1}%
    \item For $\abs{\frac{\lambda_u}{\alpha_A}}\in\left[\frac{1}{3^{k+1}},\frac{1}{3^k}\right)=\left[\frac{1}{3^{l+1}},\frac{1}{3^l}\right)$, or equivalently, $l=k$, it holds $\abs{\xi_{l,u,0}}^2+
        \abs{\xi_{l,u,1}}^2=1$.
\addtocounter{equation}{1}%
    \item For $\abs{\frac{\lambda_u}{\alpha_A}}\in\left[\frac{1}{3^{k+1}},\frac{1}{3^k}\right)\subseteq\left[0,\frac{1}{3^{l+1}}\right)$, or equivalently, $l\leq k-1$, it holds $\xi_{l,u,0}\approx0$ and $\xi_{l,u,1}\approx i$.
\addtocounter{equation}{1}%
\end{enumerate}
This means that the cumulative coefficients satisfy $\zeta_{j,u}\approx 0$ if $j\geq k+2$ or $j\leq k-1$.
So if we further assume that GPE can be performed perfectly in the passband and stopband, corresponding to input algorithms $A_1,\ldots,A_m$, then we get
\begin{equation}
\label{eq:gpe_with_bm}
\begin{aligned}
    \frac{1}{\sqrt{2}}&\sum_{k=0}^{m-1}\sum_{\abs{\frac{\lambda_u}{\alpha_A}}\in\big[\frac{1}{3^{k+1}},\frac{1}{3^k}\big)}
    \Bigg(\zeta_{k+1,u}
    \ket{k}\left(\frac{3^{k+1}}{3^m}\ket{\mathrm{good}}+\sqrt{1-\frac{9^{k+1}}{9^m}}\ket{\mathrm{bad}}\right)\\
    &\qquad\qquad\qquad\qquad\
    +\zeta_{k,u}
    \ket{k-1}\left(\frac{3^{k}}{3^m}\ket{\mathrm{good}}+\sqrt{1-\frac{9^{k}}{9^m}}\ket{\mathrm{bad}}\right)\Bigg)
    \gamma_u\left(\ket{+,\phi_{u,Y+}}+\ket{-,\phi_{u,Y-}}\right).
\end{aligned}
\end{equation}

Compared with \eq{gpe_no_bm}, our \eq{gpe_with_bm} maintains the uniform superposition of eigenstates $\ket{\phi_{u,Y\pm}}$ of the quantum walk operator. At the end, the branch marking register is uncomputed, so the two branches can be merged back to recover the original eigenstate $\ket{\phi_u}$ of the input matrix.
In the next subsection, we will describe a deterministic amplification schedule for Tunable VTAA to amplify such a state.

%%%%%%%%%%%%%%%%%%%%%%%%%%%%%%%%%%%%%%%%%%%%%%%%%%%%%%%%%%%%%%%%%%%%%%%%%%%%%%
\subsection{Deterministic amplification schedule}
\label{sec:dinv_deterministic}
We now describe a deterministic amplification schedule for Tunable VTAA with input algorithms $A_1,\ldots,A_m$, to prepare a state proportional to
\begin{equation}
\begin{aligned}
    \psi_{\text{d-inv}}&=
    \sum_{k=0}^{m-1}\sum_{\abs{\frac{\lambda_u}{\alpha_A}}\in\big[\frac{1}{3^{k+1}},\frac{1}{3^k}\big)}
    \left(\zeta_{k+1,u}\frac{3^{k+1}}{3^m}\ket{k}+\zeta_{k,u}\frac{3^k}{3^m}\ket{k-1}\right)\gamma_u\ket{\phi_u},\\
    \norm{\psi_{\text{d-inv}}}^2&=\sum_{k=0}^{m-1}\sum_{\abs{\frac{\lambda_u}{\alpha_A}}\in\big[\frac{1}{3^{k+1}},\frac{1}{3^k}\big)}\abs{\gamma_u}^2\left(\abs{\zeta_{k+1,u}}^2\frac{9^{k+1}}{9^m}
    +\abs{\zeta_{k,u}}^2\frac{9^k}{9^m}\right)
    =p_{\text{succ,d-inv}}.
\end{aligned}
\end{equation}
Note however that in defining the above state, we have implicitly assumed that GPE and branch marking can both be performed with zero error. If branch marking is perfect but GPE is erroneous, the corresponding input algorithms are denoted by $B_1,\ldots,B_m$ and the resulting amplified state is
\begin{equation*}
\begin{aligned}
        \psi_{\text{d-inv},m}&=
    \sum_{k=0}^{m-1}\sum_{\abs{\frac{\lambda_u}{\alpha_A}}\in\big[\frac{1}{3^{k+1}},\frac{1}{3^k}\big)}
    \sum_{j=1}^{m}
    \zeta_{j,u}\frac{3^{j}}{3^m}
    \ket{j-1}
    \gamma_u\ket{\phi_u},\\
    \norm{\psi_{\text{d-inv},m}}^2&=
    \sum_{k=0}^{m-1}\sum_{\abs{\frac{\lambda_u}{\alpha_A}}\in\big[\frac{1}{3^{k+1}},\frac{1}{3^k}\big)}
    \abs{\gamma_u}^2\sum_{j=1}^{m}\abs{\zeta_{j,u}}^2\frac{9^{j}}{9^m}=p_{\text{succ,d-inv},m}.
\end{aligned}
\end{equation*}
Finally, in the most general case where branch marking and GPE are both erroneous, we denote the input algorithms by $C_1,\ldots,C_m$ and the amplified output state by $\psi_{\text{d-inv,bm}}$, with probability $\norm{\psi_{\text{d-inv,bm}}}^2=p_{\text{succ,d-inv,bm}}$.

For presentational purposes, most of our analysis in this subsection assumes that GPE and branch marking are both perfect, corresponding to the input algorithms $A_1,\ldots,A_m$, and we know the precise value of $p_{\text{succ,d-inv}}=\norm{\psi_{\text{d-inv}}}^2$. This is without loss of generality, because the asymptotic scaling of VTAA is not affected if a constant multiplicative approximation of the success probability is used instead.
Moreover, we can set the accuracy $\epsilon_{\text{bm}}$ and $\epsilon_{\text{gpe},j}$ sufficiently small so that
\begin{equation*}
    p_{\text{succ,d-inv,bm}}=\mathbf{\Theta}\left(p_{\text{succ,d-inv},m}\right)=\mathbf{\Theta}\left(p_{\text{succ,d-inv}}\right)=\mathbf{\Theta}\left(p_{\text{succ}}\right);
\end{equation*}
a similar relation holds for the amplification thresholds. See \append{multi} (in particular \prop{multi_succ} and \prop{multi_amp}) for a detailed explanation.

Following the discussion in \sec{intro_dinv}, we consider the amplification thresholds
\begin{equation}
\label{eq:threshold_def}
    \alpha_j=
    \begin{cases}
        c^29^{j-m+l}\norm{\overline{\Pi_j\Pi_b}A_j\cdots A_1\ket{\psi_0}}^2,\quad&j=m-l+1,\ldots,m,\\
        0,&j=1,\ldots,m-l,
    \end{cases}
\end{equation}
for some constant $c\gtrapprox 1$ (say $c=1.001$) and integer $0\leq l\leq m$ to be determined later.

\paragraph{Analysis of amplification threshold}
From \prop{tunable_universal}, we should aim for $\sum_{j=m-l+1}^m\alpha_j=\mathbf{O}(1)$ to avoid a large loss factor (which would otherwise cause a waste of oracular queries). Here, we can use \eq{pg_amp} and \eq{gpe_with_bm} to evaluate the potentially good probabilities as
\begin{equation}
\begin{aligned}
    &\norm{\overline{\Pi_{j}\Pi_b}A_{j}\cdots A_1\ket{\psi_0}}^2\\
    &=\norm{\overline{\Pi_{j}\Pi_b}A_{m}\cdots A_1\ket{\psi_0}}^2\\
    &=\sum_{k=j}^{m-1}\sum_{\abs{\frac{\lambda_u}{\alpha_A}}\in\big[\frac{1}{3^{k+1}},\frac{1}{3^k}\big)}\abs{\gamma_u}^2\abs{\zeta_{k+1,u}}^2
    +\sum_{k=j+1}^{m-1}\sum_{\abs{\frac{\lambda_u}{\alpha_A}}\in\big[\frac{1}{3^{k+1}},\frac{1}{3^k}\big)}\abs{\gamma_u}^2\abs{\zeta_{k,u}}^2\\
    &\quad+\sum_{k=0}^{j-1}\sum_{\abs{\frac{\lambda_u}{\alpha_A}}\in\big[\frac{1}{3^{k+1}},\frac{1}{3^k}\big)}\abs{\gamma_u}^2\abs{\zeta_{k+1,u}}^2\frac{9^{k+1}}{9^m}
    +\sum_{k=0}^{j}\sum_{\abs{\frac{\lambda_u}{\alpha_A}}\in\big[\frac{1}{3^{k+1}},\frac{1}{3^k}\big)}\abs{\gamma_u}^2\abs{\zeta_{k,u}}^2\frac{9^k}{9^m}.\\
\end{aligned}
\end{equation}
Using this, we compute the sum of threshold values:
\begin{equation}
\begin{aligned}
    &\sum_{j=m-l+1}^m\alpha_j
    =\sum_{j=m-l+1}^{m}c^2\norm{\overline{\Pi_{j}\Pi_b}A_{j}\cdots A_1\ket{\psi_0}}^29^{j-m+l}\\
    &\leq c^2\sum_{j=m-l+1}^m\Bigg(
    \sum_{k=j}^{m-1}\sum_{\abs{\frac{\lambda_u}{\alpha_A}}\in\big[\frac{1}{3^{k+1}},\frac{1}{3^k}\big)}\abs{\gamma_u}^2\abs{\zeta_{k+1,u}}^2
    +\sum_{k=0}^{j-1}\sum_{\abs{\frac{\lambda_u}{\alpha_A}}\in\big[\frac{1}{3^{k+1}},\frac{1}{3^k}\big)}\abs{\gamma_u}^2\abs{\zeta_{k+1,u}}^2\frac{9^{k+1}}{9^m}\\
    &\qquad\qquad\qquad\ +\sum_{k=j+1}^{m-1}\sum_{\abs{\frac{\lambda_u}{\alpha_A}}\in\big[\frac{1}{3^{k+1}},\frac{1}{3^k}\big)}\abs{\gamma_u}^2\abs{\zeta_{k,u}}^2
    +\sum_{k=0}^{j}\sum_{\abs{\frac{\lambda_u}{\alpha_A}}\in\big[\frac{1}{3^{k+1}},\frac{1}{3^k}\big)}\abs{\gamma_u}^2\abs{\zeta_{k,u}}^2\frac{9^k}{9^m}\Bigg)9^{j-m+l}\\
    &=c^2\sum_{j=1}^{l}\Bigg(
    \sum_{k=j+m-l}^{m-1}\sum_{\abs{\frac{\lambda_u}{\alpha_A}}\in\big[\frac{1}{3^{k+1}},\frac{1}{3^k}\big)}\abs{\gamma_u}^2\abs{\zeta_{k+1,u}}^2
    +\sum_{k=0}^{j+m-l-1}\sum_{\abs{\frac{\lambda_u}{\alpha_A}}\in\big[\frac{1}{3^{k+1}},\frac{1}{3^k}\big)}\abs{\gamma_u}^2\abs{\zeta_{k+1,u}}^2\frac{9^{k+1}}{9^m}\\
    &\qquad\qquad\qquad+\sum_{k=j+m-l+1}^{m-1}\sum_{\abs{\frac{\lambda_u}{\alpha_A}}\in\big[\frac{1}{3^{k+1}},\frac{1}{3^k}\big)}\abs{\gamma_u}^2\abs{\zeta_{k,u}}^2
    +\sum_{k=0}^{j+m-l}\sum_{\abs{\frac{\lambda_u}{\alpha_A}}\in\big[\frac{1}{3^{k+1}},\frac{1}{3^k}\big)}\abs{\gamma_u}^2\abs{\zeta_{k,u}}^2\frac{9^k}{9^m}\Bigg)9^{j}.\\
\end{aligned}
\end{equation}

By exchanging the order of summation, we may bound first line of the result as
\begin{small}
\begin{equation}
\newmaketag
\begin{aligned}
    &\sum_{j=1}^{l}\bigg(\sum_{k=j+m-l}^{m-1}\sum_{\abs{\frac{\lambda_u}{\alpha_A}}\in\big[\frac{1}{3^{k+1}},\frac{1}{3^k}\big)}\abs{\gamma_u}^2\abs{\zeta_{k+1,u}}^2
    +\sum_{k=0}^{j+m-l-1}\sum_{\abs{\frac{\lambda_u}{\alpha_A}}\in\big[\frac{1}{3^{k+1}},\frac{1}{3^k}\big)}\abs{\gamma_u}^2\abs{\zeta_{k+1,u}}^2\frac{9^{k+1}}{9^m}\bigg)9^{j}\\
    &=\sum_{k=m-l+1}^{m-1}\sum_{j=1}^{k-m+l}\sum_{\abs{\frac{\lambda_u}{\alpha_A}}\in\big[\frac{1}{3^{k+1}},\frac{1}{3^k}\big)}\abs{\gamma_u}^2\abs{\zeta_{k+1,u}}^29^j
    +\sum_{k=0}^{m-1}\sum_{j=k-m+l+1}^{l}\sum_{\abs{\frac{\lambda_u}{\alpha_A}}\in\big[\frac{1}{3^{k+1}},\frac{1}{3^k}\big)}\abs{\gamma_u}^2\abs{\zeta_{k+1,u}}^2\frac{9^{k+1}}{9^m}9^j\\
    &\leq\frac{9}{8}\sum_{k=m-l+1}^{m-1}\sum_{\abs{\frac{\lambda_u}{\alpha_A}}\in\big[\frac{1}{3^{k+1}},\frac{1}{3^k}\big)}\abs{\gamma_u}^2\abs{\zeta_{k+1,u}}^29^{k-m+l}
    +\frac{9}{8}\sum_{k=0}^{m-1}\sum_{\abs{\frac{\lambda_u}{\alpha_A}}\in\big[\frac{1}{3^{k+1}},\frac{1}{3^k}\big)}\abs{\gamma_u}^2\abs{\zeta_{k+1,u}}^2\frac{9^{k+1}}{9^m}9^l\\
    &\leq\frac{5}{4}\sum_{k=0}^{m-1}\sum_{\abs{\frac{\lambda_u}{\alpha_A}}\in\big[\frac{1}{3^{k+1}},\frac{1}{3^k}\big)}\abs{\gamma_u}^2\abs{\zeta_{k+1,u}}^2\frac{9^{k+1}}{9^m}9^l.
\end{aligned}
\end{equation}
\end{small}
Similarly,
\begin{small}
\begin{equation}
\newmaketag
\begin{aligned}
    &\sum_{j=1}^{l}
    \bigg(\sum_{k=j+m-l+1}^{m-1}\sum_{\abs{\frac{\lambda_u}{\alpha_A}}\in\big[\frac{1}{3^{k+1}},\frac{1}{3^k}\big)}\abs{\gamma_u}^2\abs{\zeta_{k,u}}^2
    +\sum_{k=0}^{j+m-l}\sum_{\abs{\frac{\lambda_u}{\alpha_A}}\in\big[\frac{1}{3^{k+1}},\frac{1}{3^k}\big)}\abs{\gamma_u}^2\abs{\zeta_{k,u}}^2\frac{9^k}{9^m}\bigg)9^{j}\\
    &=\sum_{k=m-l+2}^{m-1}\sum_{j=1}^{k-m+l-1}\sum_{\abs{\frac{\lambda_u}{\alpha_A}}\in\big[\frac{1}{3^{k+1}},\frac{1}{3^k}\big)}\abs{\gamma_u}^2\abs{\zeta_{k,u}}^29^j
    +\sum_{k=0}^{m}\sum_{j=k-m+l}^{l}\sum_{\abs{\frac{\lambda_u}{\alpha_A}}\in\big[\frac{1}{3^{k+1}},\frac{1}{3^k}\big)}\abs{\gamma_u}^2\abs{\zeta_{k,u}}^2\frac{9^k}{9^m}9^{j}\\
    &\leq\frac{9}{8}\sum_{k=m-l+2}^{m-1}\sum_{\abs{\frac{\lambda_u}{\alpha_A}}\in\big[\frac{1}{3^{k+1}},\frac{1}{3^k}\big)}\abs{\gamma_u}^2\abs{\zeta_{k,u}}^29^{k-m+l-1}
    +\frac{9}{8}\sum_{k=0}^{m}\sum_{\abs{\frac{\lambda_u}{\alpha_A}}\in\big[\frac{1}{3^{k+1}},\frac{1}{3^k}\big)}\abs{\gamma_u}^2\abs{\zeta_{k,u}}^2\frac{9^k}{9^m}9^{l}\\
    &\leq\frac{5}{4}\sum_{k=0}^{m-1}\sum_{\abs{\frac{\lambda_u}{\alpha_A}}\in\big[\frac{1}{3^{k+1}},\frac{1}{3^k}\big)}\abs{\gamma_u}^2\abs{\zeta_{k,u}}^2\frac{9^{k}}{9^m}9^l.
\end{aligned}
\end{equation}
\end{small}
Altogether, we obtain
\begin{equation}
\begin{aligned}
    &c^2p_{\text{succ,d-inv}}9^l
    =\alpha_m
    \leq\sum_{j=m-l+1}^m\alpha_j
    =\sum_{j=m-l+1}^{m}c^2\norm{\overline{\Pi_{j}\Pi_b}A_{j}\cdots A_1\ket{\psi_0}}^29^{j-m+l}\\
    &\leq\frac{5c^2}{4}\left(\sum_{k=0}^{m-1}\sum_{\abs{\frac{\lambda_u}{\alpha_A}}\in\big[\frac{1}{3^{k+1}},\frac{1}{3^k}\big)}\abs{\gamma_u}^2\abs{\zeta_{k+1,u}}^2\frac{9^{k+1}}{9^m}
    +\sum_{k=0}^{m-1}\sum_{\abs{\frac{\lambda_u}{\alpha_A}}\in\big[\frac{1}{3^{k+1}},\frac{1}{3^k}\big)}\abs{\gamma_u}^2\abs{\zeta_{k,u}}^2\frac{9^{k}}{9^m}\right)9^l\\
    &=\frac{5c^2}{4}p_{\text{succ,d-inv}}9^l.
\end{aligned}
\end{equation}
We thus have $\sum_{j=m-l+1}^m\alpha_j\leq 1$ as long as
$l\leq\mathbf{Floor}\left(\log_3\left(\frac{2}{\sqrt{5}c\sqrt{p_{\text{succ,d-inv}}}}\right)\right)
    =\mathbf{O}\left(\log_3\left(\frac{1}{\sqrt{p_{\text{succ}}}}\right)\right)$.
In particular, there is no over amplification and the total loss factor of VTAA is $\geq\frac{5}{6}$.

\paragraph{Analysis of amplification schedule}
Next, let us derive the deterministic amplification schedule. At the beginning, $j=m-l+1$, and we compare
\begin{equation}
    \norm{\overline{\Pi_{m-l+1}\Pi_b}A_{m-l+1}\widetilde A_{m-l}\ket{\psi_0}}
    =\norm{\overline{\Pi_{m-l+1}\Pi_b}A_{m-l+1}A_{m-l}\cdots A_1\ket{\psi_0}}
\end{equation}
and
\begin{equation}
    \frac{1}{3}\sqrt{\alpha_{m-l+1}}
    =c\norm{\overline{\Pi_{m-l+1}\Pi_b}A_{m-l+1}A_{m-l}\cdots A_1\ket{\psi_0}}.
\end{equation}
Assuming $c\gtrapprox1$ (say $c=1.001$), we find that
\begin{equation}
\begin{aligned}
    \norm{\overline{\Pi_{m-l+1}\Pi_b}A_{m-l+1}\widetilde A_{m-l}\ket{\psi_0}}&<\frac{1}{3}\sqrt{\alpha_{m-l+1}},\\
    3\norm{\overline{\Pi_{m-l+1}\Pi_b}A_{m-l+1}\widetilde A_{m-l}\ket{\psi_0}}&\geq\frac{1}{3}\sqrt{\alpha_{m-l+1}},
\end{aligned}
\end{equation}
so $3$ amplification steps are needed for $j=m-l+1$. After that, 
\begin{equation}
\begin{aligned}
    &\left(1-\frac{1}{6}9\norm{\overline{\Pi_{m-l+1}\Pi_b}A_{m-l+1}\cdots A_1\ket{\psi_0}}^2\right)3\norm{\overline{\Pi_{m-l+1}\Pi_b}A_{m-l+1}\cdots A_1\ket{\psi_0}}\\
    &\leq\norm{\overline{\Pi_{m-l+1}\Pi_b}\widetilde A_{m-l+1}\ket{\psi_0}}\\
    &\leq3\norm{\overline{\Pi_{m-l+1}\Pi_b}A_{m-l+1}\cdots A_1\ket{\psi_0}}.
\end{aligned}
\end{equation}

By induction, suppose we have
\begin{equation}
\begin{aligned}
    &\prod_{k=m-l+1}^{j}\left(1-\frac{1}{6}9^{k-m+l}\norm{\overline{\Pi_{k}\Pi_b}A_k\cdots A_1\ket{\psi_0}}^2\right)\norm{\overline{\Pi_{j}\Pi_b}A_{j}\cdots A_1\ket{\psi_0}}3^{j-m+l}\\
    &\leq\norm{\overline{\Pi_{j}\Pi_b}\widetilde A_{j}\ket{\psi_0}}\\
    &\leq\norm{\overline{\Pi_{j}\Pi_b}A_{j}\cdots A_1\ket{\psi_0}}3^{j-m+l}
\end{aligned}
\end{equation}
after stage $j$. For $j+1$, we want to compare
\begin{equation}
    \norm{\overline{\Pi_{j+1}\Pi_b}A_{j+1}\widetilde A_{j}\ket{\psi_0}}
\end{equation}
and
\begin{equation}
    \frac{1}{3}\sqrt{\alpha_{j+1}}
    =c3^{j-m+l}\norm{\overline{\Pi_{j+1}\Pi_b}A_{j+1}A_{j}\cdots A_1\ket{\psi_0}}.
\end{equation}
Using the inductive hypothesis and \eq{trans_amp},
\begin{equation}
\begin{aligned}
    &\prod_{k=m-l+1}^{j}\left(1-\frac{1}{6}9^{k-m+l}\norm{\overline{\Pi_{k}\Pi_b}A_k\cdots A_1\ket{\psi_0}}^2\right)\norm{\overline{\Pi_{j+1}\Pi_b}A_{j+1}\cdots A_1\ket{\psi_0}}3^{j-m+l}\\
    &\leq\norm{\overline{\Pi_{j+1}\Pi_b}A_{j+1}\widetilde A_{j}\ket{\psi_0}}\\
    &\leq\norm{\overline{\Pi_{j+1}\Pi_b}A_{j+1}\cdots A_1\ket{\psi_0}}3^{j-m+l}.
\end{aligned}
\end{equation}
Then since the loss factor $\geq\left(\frac{5}{6}\right)^{\sum_{j=1}^m\alpha_j}\geq\frac{5}{6}$, the amplitude satisfies
\begin{equation}
\begin{aligned}
    \norm{\overline{\Pi_{j+1}\Pi_b}A_{j+1}\widetilde A_{j}\ket{\psi_0}}&<\frac{1}{3}\sqrt{\alpha_{j+1}},\\
    3\norm{\overline{\Pi_{j+1}\Pi_b}A_{j+1}\widetilde A_{j}\ket{\psi_0}}&\geq\frac{1}{3}\sqrt{\alpha_{j+1}}.
\end{aligned}
\end{equation}
So again we need $3$ amplification steps for $j+1$. After that, 
\begin{equation}
\begin{aligned}
    &\prod_{k=m-l+1}^{j+1}\left(1-\frac{1}{6}9^{k-m+l}\norm{\overline{\Pi_{k}\Pi_b}A_k\cdots A_1\ket{\psi_0}}^2\right)\norm{\overline{\Pi_{j+1}\Pi_b}A_{j+1}\cdots A_1\ket{\psi_0}}3^{j+1-m+l}\\
    &\leq\left(1-\frac{1}{6}9\norm{\overline{\Pi_{j+1}\Pi_b}A_{j+1}\widetilde A_{j}\ket{\psi_0}}^2\right)3\norm{\overline{\Pi_{j+1}\Pi_b}A_{j+1}\widetilde A_{j}\ket{\psi_0}}\\
    &\leq\norm{\overline{\Pi_{j+1}\Pi_b}\widetilde A_{j+1}\ket{\psi_0}}\\
    &\leq3\norm{\overline{\Pi_{j+1}\Pi_b}A_{j+1}\widetilde A_{j}\ket{\psi_0}}\\
    &\leq\norm{\overline{\Pi_{j+1}\Pi_b}A_{j+1}\cdots A_1\ket{\psi_0}}3^{j+1-m+l}.
\end{aligned}
\end{equation}
Here, the first inequality follows from the inductive hypothesis and the calculation
\begin{equation*}
\begin{aligned}
    9\norm{\overline{\Pi_{j+1}\Pi_b}A_{j+1}\widetilde A_{j}\ket{\psi_0}}^2
    &=9\norm{\overline{\Pi_{j+1}\Pi_b}A_{j+1}\cdots A_1\ket{\psi_0}}^2
    \frac{\norm{\overline{\Pi_{j+1}\Pi_b}A_{j+1}\widetilde A_{j}\ket{\psi_0}}^2}{\norm{\overline{\Pi_{j+1}\Pi_b}A_{j+1}\cdots A_1\ket{\psi_0}}^2}\\
    &=9\norm{\overline{\Pi_{j+1}\Pi_b}A_{j+1}\cdots A_1\ket{\psi_0}}^2
    \frac{\norm{\overline{\Pi_{j}\Pi_b}\widetilde A_{j}\ket{\psi_0}}^2}{\norm{\overline{\Pi_{j}\Pi_b}A_{j}\cdots A_1\ket{\psi_0}}^2}\\
    &\leq 9\norm{\overline{\Pi_{j+1}\Pi_b}A_{j+1}\cdots A_1\ket{\psi_0}}^29^{j-m+l},
\end{aligned}
\end{equation*}
where we have used \eq{trans_amp} in the second equality and \lem{rep_query} in the last inequality (applied from the first nontrivial stage $s_1=m-l+1$ up to stage $s_{j-m+l}=j$ with the amplitudes corresponding to $s_0=0$ canceled out).
The induction is now complete. We obtain the following deterministic amplification schedule as desired:
\begin{equation}
\label{eq:deterministic_schedule2}
    2r_j+1=
    \begin{cases}
        3,\quad &j=m-l+1,\ldots,m,\\
        1,&j=1,\ldots,m-l.
    \end{cases}
\end{equation}

\paragraph{Analysis of query complexity}
We now consider the query complexity. Invoking \prop{tunable_query}, \prop{bm} and \prop{gpe} as well as the choice of thresholds in \eq{threshold_def}, we have the cost
\begin{equation}
\begin{aligned}
    &\mathbf{O}\Bigg(\frac{1}{\sqrt{p_{\text{succ,d-inv}}}}\mathbf{Cost}\left(A_{m-l+1}\cdots A_{1}\ket{\psi_0}\right)\\
    &\qquad+\frac{1}{\sqrt{p_{\text{succ,d-inv}}}}\sum_{j=m-l+2}^{m}\frac{1}{\sqrt{\alpha_{j-1}}}
    \norm{\overline{\Pi_{j-1}\Pi_b}A_{j-1}\cdots A_1\ket{\psi_0}}\mathbf{Cost}\left(A_{j}\right)\Bigg)\\
    &=\mathbf{O}\Bigg(\frac{1}{\sqrt{p_{\text{succ}}}}\mathbf{Cost}\left(O_b\right)
    +\frac{1}{\sqrt{p_{\text{succ}}}}\log\left(\frac{1}{\epsilon_{\text{bm}}}\right)\mathbf{Cost}\left(O_A\right)
    +\frac{1}{\sqrt{p_{\text{succ}}}}\sum_{j=1}^{m-l+1}3^j\log\left(\frac{1}{\epsilon_{\text{gpe},j}}\right)\mathbf{Cost}(O_A)\\
    &\qquad\quad +\frac{1}{\sqrt{p_{\text{succ}}}}\sum_{j=m-l+2}^{m}3^{m-l-j}3^j\log\left(\frac{1}{\epsilon_{\text{gpe},j}}\right)\mathbf{Cost}\left(O_A\right)\Bigg),\\
\end{aligned}
\end{equation}
where we have again used $p_{\text{succ,d-inv}}=\mathbf{\Theta}\left(p_{\text{succ}}\right)$.
Here, the cost of GPE grows exponentially with the stage number $j$ like $\sim3^j$, suggesting a non-uniform error distribution of the form
\begin{equation}
\label{eq:error_schedule}
    \epsilon_{\text{gpe},j}=
    \begin{cases}
        \frac{\epsilon_{\text{gpe}}}{l},\quad&j=m-l+2,\ldots,m,\\
        \frac{\epsilon_{\text{gpe}}}{l\cdot 2^{m-l+2-j}},&j=1,\ldots,m-l+1,
    \end{cases}
    \qquad\qquad
    \sum_{j=1}^m\epsilon_{\text{gpe},j}
    \leq\epsilon_{\text{gpe}}.
\end{equation}
This then leads to the cost
\begin{equation}
\begin{aligned}
    &\mathbf{O}\Bigg(\frac{1}{\sqrt{p_{\text{succ}}}}\mathbf{Cost}\left(O_b\right)
    +\frac{1}{\sqrt{p_{\text{succ}}}}\log\left(\frac{1}{\epsilon_{\text{bm}}}\right)\mathbf{Cost}\left(O_A\right)\\
    &\qquad\quad+\frac{1}{\sqrt{p_{\text{succ}}}}\sum_{j=1}^{m-l+1}3^j(m-l+2-j)\log\left(\frac{l}{\epsilon_{\text{gpe}}}\right)\mathbf{Cost}(O_A)\\
    &\qquad\quad +\frac{1}{\sqrt{p_{\text{succ}}}}\sum_{j=m-l+2}^{m}3^{m-l-j}3^j\log\left(\frac{l}{\epsilon_{\text{gpe}}}\right)\mathbf{Cost}\left(O_A\right)\Bigg)\\
    &=\mathbf{O}\left(\frac{1}{\sqrt{p_{\text{succ}}}}\mathbf{Cost}\left(O_b\right)
    +\frac{1}{\sqrt{p_{\text{succ}}}}\left(\log\left(\frac{1}{\epsilon_{\text{bm}}}\right)+\frac{l3^m}{3^l}\log\left(\frac{l}{\epsilon_{\text{gpe}}}\right)\right)\mathbf{Cost}\left(O_A\right)\right),
\end{aligned}
\end{equation}
which decreases as $l$ increases. So we should choose the largest possible $l$:
\begin{equation}
    l=\mathbf{Floor}\left(\log_3\left(\frac{2}{\sqrt{5}c\sqrt{p_{\text{succ,d-inv}}}}\right)\right)
    =\mathbf{\Theta}\left(\log_3\left(\frac{1}{\sqrt{p_{\text{succ}}}}\right)\right).
\end{equation}

\paragraph{Analysis of success probability}
As an immediate consequence of the deterministic schedule, we obtain
\begin{equation}
\begin{aligned}
    \norm{\overline{\Pi_m\Pi_b}\widetilde A_m\ket{\psi_0}}
    &\geq\frac{5}{6}3^l\norm{\overline{\Pi_m\Pi_b}A_m\cdots A_1\ket{\psi_0}}\\
    &\geq\frac{5}{6}3^{\log_3\left(\frac{2}{\sqrt{5}c\sqrt{p_{\text{succ,d-inv}}}}\right)-1}\sqrt{p_{\text{succ,d-inv}}}\\
    &=\frac{5}{6}\frac{2}{\sqrt{5}c\sqrt{p_{\text{succ,d-inv}}}}\frac{1}{3}\sqrt{p_{\text{succ,d-inv}}}
    =\frac{\sqrt{5}}{9c},
\end{aligned}
\end{equation}
where the first inequality follows from \lem{rep_query} (applied from the first nontrivial stage $s_1=m-l+1$ up to stage $s_l=m$).
So VTAA outputs the normalized version of discretized inverse state $\psi_{\text{d-inv}}$ with a constant success probability.

\paragraph{Analysis of error}
With this deterministic schedule, the initial state $\ket{b}$ and input algorithms $A_1,\ldots,A_{m-l+1}$ are invoked $3^l$ times. After that, algorithm $A_{j}$ is invoked $3^{m-j+1}$ times for $j=m-l+2,\ldots,m$. So the total error of preparing the discretized inverse state is at most
\begin{equation}
    3^l\epsilon_{\text{bm}}
    +3^l\sum_{j=1}^{m-l+1}\frac{\epsilon_{\text{gpe}}}{l\cdot2^{m-l+2-j}}
    +\sum_{j=m-l+2}^m 3^{m-j+1}\frac{\epsilon_{\text{gpe}}}{l}
    =\mathbf{O}\left(3^l\left(\epsilon_{\text{bm}}+\epsilon_{\text{gpe}}\right)\right)
    =\mathbf{O}\left(\frac{\epsilon_{\text{bm}}+\epsilon_{\text{gpe}}}{\sqrt{p_{\text{succ}}}}\right).
\end{equation}
Thus, by setting
\begin{equation}
    \epsilon_{\text{bm}},\epsilon_{\text{gpe}}=\mathbf{\Theta}\left(\sqrt{p_{\text{succ}}}\epsilon\right),
\end{equation}
we can generate the discretized inverse state $\psi_{\text{d-inv}}$ with constant probability, accuracy $\epsilon$ and query complexity
\begin{equation}
    \mathbf{O}\left(\frac{1}{\sqrt{p_{\text{succ}}}}\mathbf{Cost}\left(O_b\right)
    +\kappa
    \log\left(\frac{1}{\sqrt{p_{\text{succ}}}}\right)
    \log\left(\frac{1}{\sqrt{p_{\text{succ}}}\epsilon}\right)
    \mathbf{Cost}\left(O_A\right)\right).
\end{equation}
In the next subsection, we show how to project the error into potentially good subspaces to further improve the query complexity to
\begin{equation}
    \mathbf{O}\left(\frac{1}{\sqrt{p_{\text{succ}}}}\mathbf{Cost}\left(O_b\right)
    +\kappa
    \log\left(\frac{1}{\sqrt{p_{\text{succ}}}}\right)
    \log\left(\frac{\log\left(\frac{1}{\sqrt{p_{\text{succ}}}}\right)}{\epsilon}\right)
    \mathbf{Cost}\left(O_A\right)\right).
\end{equation}

%%%%%%%%%%%%%%%%%%%%%%%%%%%%%%%%%%%%%%%%%%%%%%%%%%%%%%%%%%%%%%%%%%%%%%%%%%%%%%
\subsection{Projecting error toward potentially good subspaces}
\label{sec:dinv_err}
In this subection, we present an improved error analysis for preparing the normalized discretized inverse state by projecting error onto potentially good subspaces.

To begin with, recall that if branch marking were perfectly performed, our discretized inverse state would take the form
\begin{equation}
\begin{aligned}
    \sum_{k=0}^{m-1}\sum_{\abs{\frac{\lambda_u}{\alpha_A}}\in\big[\frac{1}{3^{k+1}},\frac{1}{3^k}\big)}
    \sum_{j=1}^m
    \zeta_{j,u}
    \ket{j-1}\left(\frac{3^{j}}{3^m}\ket{\mathrm{good}}+\sqrt{1-\frac{9^{j}}{9^m}}\ket{\mathrm{bad}}\right)
    \gamma_u\ket{\phi_u},
\end{aligned}
\end{equation}
where $\zeta_{j,u}=\xi_{j,u,0}\prod_{l=1}^{j-1}\xi_{l,u,1}$ are the cumulative coefficients from erroneous GPEs, such that the normalization condition $\abs{\xi_{l,u,0}}^2+\abs{\xi_{l,u,1}}^2=1$ holds at all stages $l$. If GPEs were also performed perfectly, then majority of $\zeta_{j,u}$ would be zero unless $j=k,k+1$ and we get
\begin{equation}
\begin{aligned}
    &\sum_{k=0}^{m-1}\sum_{\abs{\frac{\lambda_u}{\alpha_A}}\in\big[\frac{1}{3^{k+1}},\frac{1}{3^k}\big)}
    \Bigg(\zeta_{k+1,u}
    \ket{k}\left(\frac{3^{k+1}}{3^m}\ket{\mathrm{good}}+\sqrt{1-\frac{9^{k+1}}{9^m}}\ket{\mathrm{bad}}\right)\\
    &\qquad\qquad\qquad\qquad\quad
    +\zeta_{k,u}
    \ket{k-1}\left(\frac{3^{k}}{3^m}\ket{\mathrm{good}}+\sqrt{1-\frac{9^{k}}{9^m}}\ket{\mathrm{bad}}\right)\Bigg)
    \gamma_u\ket{\phi_u}.
\end{aligned}
\end{equation}
Let us compare these two states, but only within the potentially good subspaces. That is, we compare
\begin{equation}
\begin{aligned}
    \psi_{\text{d-inv},m}&=\sum_{k=0}^{m-1}\sum_{\abs{\frac{\lambda_u}{\alpha_A}}\in\big[\frac{1}{3^{k+1}},\frac{1}{3^k}\big)}
    \sum_{j=1}^m
    \zeta_{j,u}
    \ket{j-1}\frac{3^{j}}{3^m}
    \gamma_u\ket{\phi_u},\qquad
    \norm{\psi_{\text{d-inv},m}}^2=p_{\text{succ,d-inv},m},\\
    \psi_{\text{d-inv}}&=\sum_{k=0}^{m-1}\sum_{\abs{\frac{\lambda_u}{\alpha_A}}\in\big[\frac{1}{3^{k+1}},\frac{1}{3^k}\big)}
    \left(\zeta_{k+1,u}
    \ket{k}\frac{3^{k+1}}{3^m}
    +\zeta_{k,u}
    \ket{k-1}\frac{3^{k}}{3^m}\right)
    \gamma_u\ket{\phi_u},\qquad
    \norm{\psi_{\text{d-inv}}}^2=p_{\text{succ,d-inv}}.
\end{aligned}
\end{equation}
Then their squared Euclidean distance is bounded by
\begin{equation}
    \norm{\psi_{\text{d-inv},m}-\psi_{\text{d-inv}}}^2
    \leq\sum_{k=0}^{m-1}\sum_{\abs{\frac{\lambda_u}{\alpha_A}}\in\big[\frac{1}{3^{k+1}},\frac{1}{3^k}\big)}\abs{\gamma_u}^2\sum_{j\neq k,k+1}\abs{\zeta_{j,u}}^2\frac{9^j}{9^m}.
\end{equation}

We divide the analysis into two cases.
\begin{enumerate}[label=(\roman*)]
    \item $j\leq k-1$: In this case, we use the bound 
    \begin{equation}
    \begin{cases}
        \abs{\xi_{l,u,1}}\leq1,\quad&l=1,\ldots,j-1,\\
        \abs{\xi_{j,u,0}}\leq\epsilon_{\text{qpe},j}.&
    \end{cases}
    \end{equation}
    to obtain
    \begin{equation}
        \abs{\zeta_{j,u}}\leq\epsilon_{\text{gpe},j}.
    \end{equation}
    \item $j\geq k+2$: In this case, we use the bound
    \begin{equation}
    \begin{cases}
        \abs{\xi_{l,u,1}}\leq1,\quad&l=1,\ldots,k,\\
        \abs{\xi_{l,u,1}}\leq\epsilon_{\text{gpe},l},\quad&l=k+1,\ldots,j-1,\\
        \abs{\xi_{j,u,0}}\leq1,
    \end{cases}
    \end{equation}
    to obtain
    \begin{equation}
        \abs{\zeta_{j,u}}\leq\prod_{l=k+1}^{j-1}\epsilon_{\text{gpe},l}.
    \end{equation}
\end{enumerate}
This bounds the squared distance by
\begin{equation}
\begin{aligned}
    \norm{\psi_{\text{d-inv},m}-\psi_{\text{d-inv}}}^2&\leq\sum_{k=0}^{m-1}\sum_{\abs{\frac{\lambda_u}{\alpha_A}}\in\big[\frac{1}{3^{k+1}},\frac{1}{3^k}\big)}\abs{\gamma_u}^2
    \left(\sum_{j=1}^{k-1}\epsilon_{\text{gpe},j}^2\frac{9^j}{9^m}
    +\sum_{j=k+2}^{m}\prod_{l=k+1}^{j-1}\epsilon_{\text{gpe},l}^2\frac{9^j}{9^m}\right)\\
    &\leq\sum_{k=0}^{m-1}\sum_{\abs{\frac{\lambda_u}{\alpha_A}}\in\big[\frac{1}{3^{k+1}},\frac{1}{3^k}\big)}\abs{\gamma_u}^2
    \left(\sum_{j=1}^{k-1}\epsilon_{\text{gpe},j}^2\frac{9^{k-1}}{9^m}
    +\sum_{j=k+2}^{m}\epsilon_{\text{gpe},j-1}^2\frac{9^{k+2}}{9^m}\right)\\
    &\leq\left(\sum_{k=0}^{m-1}\sum_{\abs{\frac{\lambda_u}{\alpha_A}}\in\big[\frac{1}{3^{k+1}},\frac{1}{3^k}\big)}\abs{\gamma_u}^2\frac{9^{k+2}}{9^m}\right)
    \left(\sum_{j=1}^{m}\epsilon_{\text{gpe},j}\right)^2
    \leq 9p_{\text{succ,d-inv},1}\epsilon_{\text{gpe}}^2,
\end{aligned}
\end{equation}
where
\begin{equation}
    p_{\text{succ,d-inv},1}=\sum_{k=0}^{m-1}\sum_{\abs{\frac{\lambda_u}{\alpha_A}}\in\big[\frac{1}{3^{k+1}},\frac{1}{3^k}\big)}
    \abs{\gamma_u}^2\frac{9^{k+1}}{9^m}.
\end{equation}
Here, we assume all $\epsilon_{\text{gpe},l}\leq\frac{1}{3}$ in the second inequality and upper bound the $\ell_2$-norm by $\ell_1$-norm in the third inequality.

Now we take the error of branch marking into account. If the actual output state is $\psi_{\text{d-inv,bm}}$ with $\norm{\psi_{\text{d-inv,bm}}}^2=p_{\text{succ,d-inv,bm}}$, then
\begin{equation}
    \norm{\psi_{\text{d-inv,bm}}
    -\psi_{\text{d-inv},m}}
    \leq2\epsilon_{\text{bm}}
    \quad\Rightarrow\quad
    \norm{\psi_{\text{d-inv,bm}}-\psi_{\text{d-inv}}}
    =\mathbf{O}\left(\sqrt{p_{\text{succ,d-inv},1}}\epsilon_{\text{gpe}}+\epsilon_{\text{bm}}\right).
\end{equation}
We will prove in \prop{multi_succ} from \append{multi} that
\begin{equation}
    p_{\text{succ,d-inv,bm}}=\mathbf{\Theta}\left(p_{\text{succ,d-inv},m}\right)=\mathbf{\Theta}\left(p_{\text{succ,d-inv}}\right)=\mathbf{\Theta}(p_{\text{succ,d-inv},1})=\mathbf{\Theta}\left(p_{\text{succ}}\right).
\end{equation}
Hence, by~\cite[Lemma 24]{QEVP},
\begin{equation}
    \norm{\frac{\psi_{\text{d-inv,bm}}}{\sqrt{p_{\text{succ,d-inv,bm}}}}
    -\frac{\psi_{\text{d-inv}}}{\sqrt{p_{\text{succ,d-inv}}}}}
    \leq\frac{2\norm{\psi_{\text{d-inv,bm}}-\psi_{\text{d-inv}}}}{\sqrt{p_{\text{succ,d-inv}}}}
    =\mathbf{O}\left(\epsilon_{\text{gpe}}+\frac{\epsilon_{\text{bm}}}{\sqrt{p_{\text{succ}}}}\right),
\end{equation}
We see that the error of GPE shrinks by a factor of $\mathbf{\Theta}\left(\sqrt{p_{\text{succ}}}\right)$ when projected onto the potentially good subspaces, which cancels with the state normalization factor $\mathbf{\Theta}\left(\frac{1}{\sqrt{p_{\text{succ}}}}\right)$. Hence, we choose 
\begin{equation}
    \epsilon_{\text{gpe}}=\mathbf{\Theta}(\epsilon),\qquad
    \epsilon_{\text{bm}}=\mathbf{\Theta}\left(\sqrt{p_{\text{succ}}}\epsilon\right),
\end{equation}
obtaining the claimed query complexity of preparing discretized inverse state
\begin{equation}
    \mathbf{O}\left(\frac{1}{\sqrt{p_{\text{succ}}}}\mathbf{Cost}\left(O_b\right)
    +\kappa
    \log\left(\frac{1}{\sqrt{p_{\text{succ}}}}\right)
    \log\left(\frac{\log\left(\frac{1}{\sqrt{p_{\text{succ}}}}\right)}{\epsilon}\right)
    \mathbf{Cost}\left(O_A\right)\right).
\end{equation}

%%%%%%%%%%%%%%%%%%%%%%%%%%%%%%%%%%%%%%%%%%%%%%%%%%%%%%%%%%%%%%%%%%%%%%%%%%%%%%
\subsection{Solution norm estimation}
\label{sec:dinv_est}
Up to this point, we have assumed that a constant multiplicative estimate of the solution norm $\norm{A^{-1}\ket{b}}$ is available a prior. This assumption is equivalent to knowing $p_{\text{succ}}$ and $p_{\text{succ,d-inv}}$ to a constant multiplicative accuracy, which is necessary to achieve the $\mathbf{O}\left(\frac{1}{\sqrt{p_{\text{succ}}}}\right)$ scaling for initial state preparation. In this subsection, we show that this strictly linear scaling can be attained even with an unknown solution norm, when $p_{\text{succ}}$ is replaced by its lower bound $\alpha_{p_{\text{succ}}}\leq p_{\text{succ}}$. In fact, we will describe a solution norm estimation algorithm whose query complexity of the initial state oracle has the scaling $\mathbf{O}\left(\frac{1}{\sqrt{\alpha_{p_{\text{succ}}}}}\right)$.

We begin by introducing the amplitude estimation algorithm.
\begin{lemma}[Amplitude estimation]
Let $\ket{\psi_0}$ be a quantum state and $\Pi$ be an orthogonal projection. Then for any $\epsilon,\delta>0$, there exists a quantum algorithm that outputs $y$ with
\begin{equation}
    \mathbf{P}\left(\abs{y-\norm{\Pi \ket{\psi_0}}}\geq\epsilon\right)<\delta
\end{equation}
using
\begin{equation}
    \mathbf{O}\left(\frac{1}{\epsilon}\log\left(\frac{1}{\delta}\right)\right)
\end{equation}
queries to controlled-preparation of $\ket{\psi_0}$, controlled-reflection $I-2\Pi$ and their inverses.
\end{lemma}
\noindent Note that in applications where the target accuracy $\epsilon=\mathbf{\Theta}(1)$ is constant, it is possible to perform amplitude estimation by repeatedly measuring outcome of the algorithm and applying Chernoff bound. Otherwise if a smaller $\epsilon$ is desired, one can use quantum phase estimation~\cite{brassard2002quantum}.

To simplify the discussion, we first consider the discretized inverse state with an unknown $p_{\text{succ,d-inv}}=\norm{\psi_{\text{d-inv}}}^2$ but a known lower bound $\alpha_{p_{\text{succ,d-inv}}}\leq p_{\text{succ,d-inv}}$. We then (mathematically) define
\begin{equation}
    l^*=\mathbf{Floor}\left(\log_3\left(\frac{2}{\sqrt{5}c\sqrt{p_{\text{succ,d-inv}}}}\right)\right).
\end{equation}
If we run Tunable VTAA with $l^*$, our analysis in \sec{dinv_deterministic} then shows that the success amplitude is at least
\begin{equation}
    \norm{\overline{\Pi_m\Pi_b}\widetilde A_m\ket{\psi_0}}
    \geq\frac{\sqrt{5}}{9c}=\frac{5\sqrt{5}}{45c}.
\end{equation}
Our claim is that if $l<l^*$ is sufficiently small, the success amplitude is constant gapped below $\frac{\sqrt{5}}{9c}$. Indeed, if $l\leq l^*-2$, we invoke \lem{rep_query} to get
\begin{equation}
\begin{aligned}
    \norm{\overline{\Pi_m\Pi_b}\widetilde A_m\ket{\psi_0}}
    &\leq 3^l\norm{\overline{\Pi_m\Pi_b} A_m\cdots A_1\ket{\psi_0}}
    \leq3^{l^*-2}\sqrt{p_{\text{succ,d-inv}}}\\
    &\leq3^{\log_3\left(\frac{2}{\sqrt{5}c\sqrt{p_{\text{succ,d-inv}}}}\right)-2}\sqrt{p_{\text{succ,d-inv}}}\\
    &=\frac{2}{\sqrt{5}c\sqrt{p_{\text{succ,d-inv}}}}\frac{1}{9}\sqrt{p_{\text{succ,d-inv}}}=\frac{2\sqrt{5}}{45c}.
\end{aligned}
\end{equation}

Our solution norm estimation algorithm proceeds as follows. We run Tunable VTAA with pre-merging parameter $l=0,1,2,\ldots$, accuracy $\epsilon_{\text{gpe}}=\mathbf{\Theta}(1)$ and $\epsilon_{\text{bm}}=\mathbf{\Theta}\left(\sqrt{p_{\text{succ}}}\right)$, and we perform amplitude estimation with a sufficiently small (yet constant) accuracy and failure probability $\delta_l$. 
When iterating through $l=0,1,2,\ldots$, with high probability we see the estimated amplitude below $\frac{\sqrt{5}}{15c}=\frac{3\sqrt{5}}{45c}$ for all $l=0,1,\ldots,l^*-2$ and amplitude above $\frac{4\sqrt{5}}{45c}$ for either $l=l^*-1,l^*$.
Supposing the first time we see the amplitude above $\frac{4\sqrt{5}}{45c}$ is at some $l$, we consider the value $l+1$. Then, we have
\begin{equation}
    \log_3\left(\frac{2}{\sqrt{5}c\sqrt{p_{\text{succ,d-inv}}}}\right)-1
    \leq l^*
    \leq l+1
    \leq l^*+1
    \leq \log_3\left(\frac{2}{\sqrt{5}c\sqrt{p_{\text{succ,d-inv}}}}\right)+1,
\end{equation}
which implies
\begin{equation}
    \frac{1}{3}\frac{2}{\sqrt{5}c\sqrt{p_{\text{succ,d-inv}}}}
    \leq3^{l+1}
    \leq3\frac{2}{\sqrt{5}c\sqrt{p_{\text{succ,d-inv}}}}.
\end{equation}
Hence, $\frac{2}{\sqrt{5}3^{l+1}c}$ is a $3$-multiplicative approximation of $\sqrt{p_{\text{succ,d-inv}}}$.

Note that for each fixed $l$, Tunable VTAA has query complexity
\begin{equation}
\begin{aligned}
    &3^l\mathbf{Cost}\left(A_{m-l+1}\cdots A_1\ket{\psi_0}\right)
    +\sum_{j=m-l+2}^m3^{m-j+1}\mathbf{Cost}\left(A_j\right)\\
    &=\mathbf{O}\left(3^l\mathbf{Cost}(O_b)
    +3^l\sum_{j=1}^{m-l+1}3^j\log\left(2^{m-l+2-j}l\right)\mathbf{Cost}(O_A)
    +\sum_{j=m-l+2}^m3^{m-j+1}3^j\log(l)\mathbf{Cost}(O_A)\right)\\
    &=\mathbf{O}\left(3^l\mathbf{Cost}(O_b)
    +3^m l\log(l)\mathbf{Cost}(O_A)\right).
\end{aligned}
\end{equation}
Iterating over $l$ and incorporating the cost of amplitude estimation, the total query complexity is thus bounded by
\begin{equation}
    \mathbf{O}\left(\sum_{l=1}^{l^*}3^l\log\left(\frac{1}{\delta_l}\right)\mathbf{Cost}(O_b)
    +\sum_{l=1}^{l^*}3^m l\log(l)\log\left(\frac{1}{\delta_l}\right)\mathbf{Cost}(O_A)\right).
\end{equation}
Let us choose the schedule of failure probabilities to be
\begin{equation}
    \delta_l=\frac{1}{(\alpha_{l^*}-l+3)^2},\qquad
    \alpha_{l^*}=\mathbf{Floor}\left(\log_3\left(\frac{2}{\sqrt{5}c\sqrt{\alpha_{p_{\text{succ,d-inv}}}}}\right)\right)\geq l^*.
\end{equation}
By the union bound, the total failure probability is at most
\begin{equation}
    \sum_{l=1}^{l^*}\frac{1}{(\alpha_{l^*}-l+3)^2}
    \leq\sum_{l=1}^{\alpha_{l^*}}\frac{1}{(\alpha_{l^*}-l+3)^2}
    \leq\frac{1}{3^2}+\frac{1}{4^2}+\cdots
    =\frac{\pi^2}{6}-\frac{5}{4}\approx0.4.
\end{equation}
The query complexity then becomes
\begin{equation}
\begin{aligned}
    &\mathbf{O}\left(\sum_{l=1}^{l^*}3^l\log\left(\alpha_{l^*}-l+3\right)\mathbf{Cost}(O_b)
    +\kappa\sum_{l=1}^{l^*} l\log(l)\log\left(\alpha_{l^*}-l+3\right)\mathbf{Cost}(O_A)\right)\\
    &=\mathbf{O}\left(\frac{1}{\sqrt{\alpha_{p_{\text{succ,d-inv}}}}}\mathbf{Cost}(O_b)
    +\kappa\log^2\left(\frac{1}{\sqrt{\alpha_{p_{\text{succ,d-inv}}}}}\right)\log\log^2\left(\frac{1}{\sqrt{\alpha_{p_{\text{succ,d-inv}}}}}\right)\mathbf{Cost}(O_A)\right).\\
\end{aligned}
\end{equation}

The above analysis handles the discretized inverse state $\psi_{\text{d-inv}}$, and assumes that a lower bound $\alpha_{p_{\text{succ,d-inv}}}\leq p_{\text{succ,d-inv}}=\norm{\psi_{\text{d-inv}}}^2$ is known a prior. 
In practice, the output state is instead $\psi_{\text{d-inv,bm}}$ with $\norm{\psi_{\text{d-inv,bm}}}^2=p_{\text{succ,d-inv,bm}}$, but we would still proceed as above, except we replace all $p_{\text{succ,d-inv}}$ by $p_{\text{succ,d-inv,bm}}$, and analyze VTAA using \prop{multi_succ} and \prop{multi_amp}.
As per \prop{multi_succ}, we can also convert the given lower bound $\alpha_{p_{\text{succ}}}\leq p_{\text{succ}}$ into some $\alpha_{p_{\text{succ,d-inv,bm}}}\leq p_{\text{succ,d-inv,bm}}$.
Consequently, what we actually obtain is a constant multiplicative approximation of $\sqrt{p_{\text{succ,d-inv,bm}}}$. Invoking \prop{multi_succ} once more, we obtain a constant multiplicative approximation of $\sqrt{p_{\text{succ}}}$.

\begin{theorem}[Solution norm estimation with optimal initial state preparation]
\label{thm:sol_est}
Let $A$ be the coefficient matrix such that $A/\alpha_A$ is block encoded by $O_A$ with normalization factor $\alpha_A\geq\norm{A}$. Let $\ket{b}$ be the initial state prepared by oracle $O_b$.
Then the solution norm $\norm{A^{-1}\ket{b}}$, and hence the success amplitude $\sqrt{p_{\text{succ}}}=\frac{\norm{A^{-1}\ket{b}}}{\alpha_{A^{-1}}}$, can be estimated to a constant multiplicative accuracy and success probability $>\frac{1}{2}$ with query complexity
\begin{equation}
    \mathbf{O}\left(\frac{1}{\sqrt{\alpha_{p_{\text{succ}}}}}\mathbf{Cost}(O_b)
    +\kappa\log^2\left(\frac{1}{\sqrt{\alpha_{p_{\text{succ}}}}}\right)\log\log^2\left(\frac{1}{\sqrt{\alpha_{p_{\text{succ}}}}}\right)\mathbf{Cost}(O_A)\right),
\end{equation}
where $\alpha_{A^{-1}}\geq\norm{A^{-1}}$ is a norm upper bound on the inverse matrix, $\kappa=\alpha_A\alpha_{A^{-1}}$ is an upper bound on the spectral condition number, and $\alpha_{p_{\text{succ}}}\leq p_{\text{succ}}$ is a lower bound on the success probability.

The algorithm uses a nested amplitude amplification with deterministic schedule \eq{deterministic_schedule2} and an increasing choice of pre-merging parameter $l=0,1,2,\ldots$
\end{theorem}

%% file: lin.tex
We now consider the quantum linear system problem. In \sec{lin_inv}, we show that the problem can be solved by inverting the input matrix over the discretized inverse state. Combining with the tunable VTAA algorithm from the previous section, we provide a simplified quantum linear system algorithm with optimal queries to the initial state preparation and nearly optimal queries to the coefficient matrix block encoding. We provide a proof of the optimality in \sec{lin_lower}.

%%%%%%%%%%%%%%%%%%%%%%%%%%%%%%%%%%%%%%%%%%%%%%%%%%%%%%%%%%%%%%%%%%%%%%%%%%%%%%
\subsection{Matrix inversion over discretized inverse state}
\label{sec:lin_inv}
Recall from \sec{dinv} that we can use Tunable VTAA to prepare the following state with accuracy $\epsilon$ and constant success probability
\begin{equation}
\begin{aligned}
    \frac{1}{\sqrt{p_{\text{succ,d-inv}}}}&\sum_{k=0}^{m-1}\sum_{\abs{\frac{\lambda_u}{\alpha_A}}\in\big[\frac{1}{3^{k+1}},\frac{1}{3^k}\big)}
    \left(\zeta_{k+1,u}\frac{3^{k+1}}{3^m}
    \ket{k}
    +\zeta_{k,u}\frac{3^{k}}{3^m}
    \ket{k-1}\right)
    \gamma_u\ket{\phi_u},\\
    p_{\text{succ,d-inv}}=&\sum_{k=0}^{m-1}\sum_{\abs{\frac{\lambda_u}{\alpha_A}}\in\big[\frac{1}{3^{k+1}},\frac{1}{3^k}\big)}\abs{\gamma_u}^2\left(\abs{\zeta_{k+1,u}}^2\frac{9^{k+1}}{9^m}
    +\abs{\zeta_{k,u}}^2\frac{9^k}{9^m}\right).
\end{aligned}
\end{equation}
This state differs from the solution state in that the reciprocal of eigenvalues are replaced by the discrete values $\frac{3^k}{3^m}$. To get the actual solution state, we perform a block-encoded matrix inversion using the following gadget.

\begin{lemma}[Block encoding inversion {\cite[Corollary 69]{Gilyen2018singular}}]
\label{lem:inv_block}
     Let $A$ be a matrix such that $A/\alpha_A$ is block encoded by $O_A$ with some normalization factor $\alpha_A\geq\norm{A}$.
     Then the operator
     \begin{equation}
         \frac{A^{-1}}{2\alpha_{A^{-1}}}
     \end{equation}
     can be block encoded with accuracy $\epsilon$ using
     \begin{equation}
        \mathbf{O}\left(\kappa\log\left(\frac{1}{\epsilon}\right)\right)
    \end{equation}
    queries to the controlled-$O_A$ and its inverse, where $\alpha_{A^{-1}}\geq\norm{A^{-1}}$ is a norm upper bound on the inverse matrix, and $\kappa=\alpha_A\alpha_{A^{-1}}$ is an upper bound on the spectral condition number.
\end{lemma}
In our case, we apply the block-encoded inversion of the input matrix with accuracy $\epsilon_{\text{blk}}$, controlled by the clock register:
\begin{equation}
    \sum_{k=0}^{m-1}\ketbra{k}{k}\otimes\frac{\Pi_{[\frac{1}{3^{k+2}},1)}A^{-1}\Pi_{[\frac{1}{3^{k+2}},1)}}{2\cdot \frac{3^{k+2}}{\alpha_A}},\qquad
    \Pi_{\mathcal{S}}=\sum_{\abs{\frac{\lambda_u}{\alpha_A}}\in\mathcal{S}}\ketbra{\phi_u}{\phi_u}.
\end{equation}
Here, the clock state $\ket{k}$ corresponds to eigenvalues of the input matrix within $\abs{\frac{\lambda_u}{\alpha_A}}\in[\frac{1}{3^{k+2}},\frac{1}{3^k})$. Thus to fully cover all possible eigenvalues, we perform the block-encoded matrix inversion over $[\frac{1}{3^{k+2}},1)$. 
This block encoding can be implemented with query complexity $\mathbf{O}\left(\kappa\log(1/\epsilon_{\text{blk}})\right)$, by using the $\ket{k}$ register to control the rotations in the quantum signal processing circuit~\cite{Low2016HamSim}, and the entire operator has accuracy $\epsilon_{\text{blk}}$ due to its block diagonal structure.
With a high probability, we then transform the state into
\begin{equation}
\begin{aligned}
    \frac{1}{\sqrt{p_{\text{succ,blk}}}}&\sum_{k=0}^{m-1}\sum_{\abs{\frac{\lambda_u}{\alpha_A}}\in\big[\frac{1}{3^{k+1}},\frac{1}{3^k}\big)}
    \left(\zeta_{k+1,u}\frac{1}{3^m}
    \frac{\alpha_A}{2\cdot3}
    \frac{1}{\lambda_u}
    \ket{k}
    +\zeta_{k,u}\frac{1}{3^m}
    \frac{\alpha_A}{2\cdot3}
    \frac{1}{\lambda_u}
    \ket{k-1}\right)
    \gamma_u\ket{\phi_u},\\
    p_{\text{succ,blk}}=&\sum_{k=0}^{m-1}\sum_{\abs{\frac{\lambda_u}{\alpha_A}}\in\big[\frac{1}{3^{k+1}},\frac{1}{3^k}\big)}\abs{\frac{1}{3^m}
    \frac{\alpha_A}{2\cdot3}
    \frac{1}{\lambda_u}}^2
    =\frac{p_{\text{succ}}}{36}.
\end{aligned}
\end{equation}
This state can be produced with probability close to unity by another amplitude amplification.
Finally, we uncompute the clock register by reversing the GPE (without performing rotations $\ket{\mathrm{good}}\mapsto\frac{3^j}{3^m}\ket{\mathrm{good}}+\sqrt{1-\frac{9^j}{9^m}}\ket{\mathrm{bad}}$ or VTAA), obtaining the normalized solution state
\begin{equation}
    \frac{\sum_{u}\frac{1}{\lambda_u}\gamma_u\ket{\phi_u}}{\sqrt{\sum_{u}\abs{\frac{\gamma_u}{\lambda_u}}^2}}
    =\frac{A^{-1}\ket{b}}{\norm{A^{-1}\ket{b}}}.
\end{equation}
As per \prop{multi_succ},
\begin{equation*}
    p_{\text{succ,d-inv,bm}}=\mathbf{\Theta}\left(p_{\text{succ,d-inv},m}\right)=\mathbf{\Theta}\left(p_{\text{succ,d-inv}}\right)=\mathbf{\Theta}(p_{\text{succ,d-inv},1})=\mathbf{\Theta}\left(p_{\text{succ}}\right)=\mathbf{\Theta}\left(p_{\text{succ,blk}}\right),
\end{equation*}
so we only need at most constant rounds of amplitude amplifications beyond the preparation of discretized inverse state by Tunable VTAA. Setting $\epsilon_{\text{blk}}=\mathbf{\Theta}(\epsilon)$, we have thus established the main algorithm.

\begin{theorem}[Quantum linear system algorithm with optimal initial state preparation]
\label{thm:qls_opt_init}
Let $A$ be the coefficient matrix such that $A/\alpha_A$ is block encoded by $O_A$ with normalization factor $\alpha_A\geq\norm{A}$. Let $\ket{b}$ be the initial state prepared by oracle $O_b$.
Suppose that the solution norm $\norm{A^{-1}\ket{b}}$, and hence the success amplitude $\sqrt{p_{\text{succ}}}=\frac{\norm{A^{-1}\ket{b}}}{\alpha_{A^{-1}}}$, can be estimated to a constant multiplicative accuracy.
Then the quantum state 
\begin{equation}
    \frac{A^{-1}\ket{b}}{\norm{A^{-1}\ket{b}}}
\end{equation}
can be prepared to accuracy $\epsilon$ and success probability $>\frac{1}{2}$ with query complexity
\begin{equation}
    \mathbf{O}\left(\frac{1}{\sqrt{p_{\text{succ}}}}\mathbf{Cost}\left(O_b\right)
    +\kappa
    \log\left(\frac{1}{\sqrt{p_{\text{succ}}}}\right)
    \log\left(\frac{\log\left(\frac{1}{\sqrt{p_{\text{succ}}}}\right)}{\epsilon}\right)
    \mathbf{Cost}\left(O_A\right)\right),
\end{equation}
where $\alpha_{A^{-1}}\geq\norm{A^{-1}}$ is a norm upper bound on the inverse matrix, and
$\kappa=\alpha_A\alpha_{A^{-1}}$ is an upper bound on the spectral condition number.

The algorithm uses a nested amplitude amplification with deterministic schedule \eq{deterministic_schedule2} and pre-merging parameter $l=\mathbf{\Theta}\left(\log_3\left(\frac{1}{\sqrt{p_{\text{succ}}}}\right)\right)$.
\end{theorem}

%%%%%%%%%%%%%%%%%%%%%%%%%%%%%%%%%%%%%%%%%%%%%%%%%%%%%%%%%%%%%%%%%%%%%%%%%%%%%%
\subsection{Lower bound}
\label{sec:lin_lower}
Now, we prove that our query complexity to the initial state preparation is optimal, by showing how the quantum search problem~\cite{Grover} can be solved by a linear system algorithm. We note that a similar lower bound has been previously derived in~\cite{PRXQuantum.2.010315}.

Specifically, consider the $d$-dimensional initial state
\begin{equation}
    \ket{b}=\frac{1}{\sqrt{d}}\left(\sum_{j\neq w}\ket{j}-\ket{w}\right),
\end{equation}
where $\ket{w}\in\{\ket{0},\ket{1},\ldots,\ket{d-1}\}$ is an arbitrary unknown basis state. This state can be (and only be) prepared by making $1$ query to the Grover oracle with the uniform superposition state. We then perform an orthogonal decomposition of $\ket{b}$ along the uniform superposition state to get
\begin{equation}
    \ket{b}=\frac{d-2}{d}\ket{\phi_{ZX+}}+\frac{2\sqrt{d-1}}{d}\ket{\phi_{ZX-}},
\end{equation}
where using the notation of~\append{qubitization_amp}
\begin{equation}
    \ket{\phi_{ZX+}}=\frac{1}{\sqrt{d}}\sum_{j}\ket{j},\qquad
    \ket{\phi_{ZX-}}=\sum_{j\neq w}\frac{1}{\sqrt{(d-1)d}}\ket{j}+\sqrt{\frac{d-1}{d}}\ket{w}.
\end{equation}

Now, we define the coefficient matrix
\begin{equation}
    A=\frac{1}{\sqrt{d}}\left(I-\ketbra{\phi_{ZX+}}{\phi_{ZX+}}\right)+\ketbra{\phi_{ZX+}}{\phi_{ZX+}}.
\end{equation}
Note that $A$ can be block encoded with normalization factor $\alpha_A=1$ and zero query to the Grover oracle. Then the inverse matrix
\begin{equation}
    A^{-1}=\sqrt{d}\left(I-\ketbra{\phi_{ZX+}}{\phi_{ZX+}}\right)+\ketbra{\phi_{ZX+}}{\phi_{ZX+}}
\end{equation}
has spectral norm
\begin{equation}
    \norm{A^{-1}}=\sqrt{d}=\alpha_{A^{-1}}.
\end{equation}

We now invoke the quantum linear system algorithm, obtaining a state within Euclidean distance $\epsilon_{\text{lin}}$ to the solution $\frac{A^{-1}\ket{b}}{\norm{A^{-1}\ket{b}}}$.
\addtocounter{equation}{1}%
Here,
\begin{equation}
\begin{aligned}
    A^{-1}\ket{b}
    &=2\sqrt{\frac{d-1}{d}}\ket{\phi_{ZX-}}+\frac{d-2}{d}\ket{\phi_{ZX+}}\\
    &=\sum_{j\neq w}\frac{d-2+2\sqrt{d}}{d\sqrt{d}}\ket{j}
    -\frac{2(d-1)\sqrt{d}-d+2}{d\sqrt{d}}\ket{w},
\end{aligned}
\end{equation}
which implies
\begin{equation}
    \norm{A^{-1}\ket{b}}
    =\sqrt{\frac{5d^2-8d+4}{d^2}}.
\end{equation}
So the inverse success amplitude is
\begin{equation}
    \frac{\norm{A^{-1}}}{\norm{A^{-1}\ket{b}}}
    =\sqrt{d}\sqrt{\frac{d^2}{5d^2-8d+4}}
    =\mathbf{\Theta}\left(\sqrt{d}\right).
\end{equation}

Meanwhile, the absolute value of amplitude of the desired $\ket{w}$ is lower bounded by
\begin{equation}
\begin{aligned}
    \abs{\frac{\bra{w}A^{-1}\ket{b}}{\norm{A^{-1}\ket{b}}}}-\epsilon_{\text{lin}}
    &\geq\abs{\frac{\frac{2(d-1)\sqrt{d}+d-2}{d\sqrt{d}}}{\sqrt{\frac{5d^2-8d+4}{d^2}}}}-\epsilon_{\text{lin}}
    \geq\frac{\frac{2(d-1)\sqrt{d}-2}{d\sqrt{d}}}{\sqrt{\frac{5d^2+4}{d^2}}}-\epsilon_{\text{lin}}\\
    &\geq\frac{\frac{2(d-2)\sqrt{d}}{d\sqrt{d}}}{\sqrt{\frac{5(d+1)^2}{d^2}}}-\epsilon_{\text{lin}}
    =\frac{2}{\sqrt{5}}\frac{d-2}{d+1}-\epsilon_{\text{lin}}\\
    &=\frac{2}{\sqrt{5}}\left(1-\frac{3}{d+1}\right)-\epsilon_{\text{lin}}
    \geq\frac{13}{8\sqrt{5}}-\epsilon_{\text{lin}}
\end{aligned}
\end{equation}
when $d\geq 15$. Let us choose $\epsilon_{\text{lin}}=\frac{1}{60}$,
\addtocounter{equation}{1}%
so the probability of getting outcome $\ket{w}$ is at least $\left(\frac{13}{8\sqrt{5}}-\frac{1}{60}\right)^2
    >0.504$.
\addtocounter{equation}{1}%
This thus solves the Grover search problem, and we conclude that the query complexity to the initial state is at least
\begin{equation}
    \mathbf{\Omega}\left(\sqrt{d}\right)
    =\mathbf{\Omega}\left(\frac{\norm{A^{-1}}}{\norm{A^{-1}\ket{b}}}\right).
\end{equation}

\begin{theorem}[Lower bound on initial state preparation cost]
\label{thm:qls_lower}
For any $0<p<1$, there exists a linear system with coefficient matrix $A$ and initial state $\ket{b}$ (prepared by $O_b$) satisfying
\begin{equation}
    \frac{\norm{A^{-1}\ket{b}}^2}{\norm{A^{-1}}^2}=\mathbf{\Theta}(p),
\end{equation}
such that
\begin{equation}
\begin{aligned}
    &\min\bigg\{g\in\mathbb Z_{\geq0}\ \Big|\
    \exists\text{ a quantum algorithm }V\text{ making }g\text{ queries to }O_b,\\
    &\qquad\qquad\qquad\quad\norm{V\ket{0}-\frac{A^{-1}\ket{b}}{\norm{A^{-1}\ket{b}}}}=\mathbf{O}(1)\text{ sufficiently small}\bigg\}
    =\mathbf{\Omega}\left(\frac{1}{\sqrt{p}}\right).
\end{aligned}
\end{equation}
\end{theorem}

%% file: precond.tex
In this section, we introduce the technique of \emph{block preconditioning}, and describe its algorithmic applications. Specifically, we present a simple quantum linear system solver in \sec{precond_qls} with optimal queries to the coefficient matrix block encoding by choosing the initial state itself as the preconditioner. Combining block preconditioning with our main algorithms \thm{sol_est} and \thm{qls_opt_init}, we then show how to reduce the cost of initial state preparation in solving differential equations (\sec{precond_diff_eq}), estimating real eigenvalues of non-normal matrices (\sec{precond_est}), transforming matrices with real spectra and preparing their ground states (\sec{precond_transform}). Additionally, we apply block preconditioning to improve the block encoded version of quantum eigenvalue transformer in \sec{precond_transform_blk}.

%%%%%%%%%%%%%%%%%%%%%%%%%%%%%%%%%%%%%%%%%%%%%%%%%%%%%%%%%%%%%%%%%%%%%%%%%%%%%%
\subsection{Quantum linear system solver}
\label{sec:precond_qls}
Our quantum linear system algorithm in \thm{qls_opt_init} has an optimal query complexity of the initial state preparation, and a nearly optimal cost of the coefficient block encoding. We now show that it is possible to solve a preconditioned quantum linear system problem, producing the same solution state but making an optimal number of queries to the block encoding oracle.

Specifically, we choose the preconditioner to be the initial state $\ket{b}$ itself and define the scaling matrix
\begin{equation}
\label{eq:qls_precond_operator}
    S=s\ketbra{b}{b}+\left(I-\ketbra{b}{b}\right),\qquad
    S^{-1}=\frac{1}{s}\ketbra{b}{b}+\left(I-\ketbra{b}{b}\right),\qquad
    0<s<1.
\end{equation}
Then solving the preconditioned quantum linear system problem would produce the same solution state because
\begin{equation}
    \frac{(SA)^{-1}\ket{b}}{\norm{(SA)^{-1}\ket{b}}}
    =\frac{A^{-1}\ket{b}}{\norm{A^{-1}\ket{b}}}.
\end{equation}
We can block encode $S$ with normalization factor $1$ using $2$ queries to the initial state oracle $O_b$ through the linear combination
\begin{equation}
    S=O_b\left(\frac{1-s}{2}(I-2\ketbra{0}{0})+\frac{1+s}{2}I\right)O_b^\dagger.
\end{equation}
This then yields a block encoding of $SA/\alpha_A$ with normalization factor $\alpha_A$ same as that of the input matrix.
Meanwhile, the inverse matrix has the norm bound
\begin{equation}
\begin{aligned}
    \norm{(SA)^{-1}}&=\norm{A^{-1}S^{-1}}=\norm{\frac{1}{s}A^{-1}\ketbra{b}{b}+A^{-1}\left(I-\ketbra{b}{b}\right)}\\
    &\leq\sqrt{\frac{\norm{A^{-1}\ket{b}}^2}{s^2}+\norm{A^{-1}\left(I-\ketbra{b}{b}\right)}^2}
    \leq\sqrt{\frac{\norm{A^{-1}\ket{b}}^2}{s^2\norm{A^{-1}}^2}+1}\norm{A^{-1}},
\end{aligned}
\end{equation}
whereas the solution norm becomes
\begin{equation}
    \norm{(SA)^{-1}\ket{b}}
    =\frac{1}{s}\norm{A^{-1}\ket{b}}.
\end{equation}

Suppose we have a constant multiplicative approximation of the solution norm $\norm{A^{-1}\ket{b}}$, say
\begin{equation}
\label{eq:qls_precond_multi}
    \frac{t}{c}<\norm{A^{-1}\ket{b}}<ct
    \quad\Leftrightarrow\quad
    \frac{\norm{A^{-1}\ket{b}}}{c}< t<c\norm{A^{-1}\ket{b}}
\end{equation}
with some constant $c>1$. Then we choose
\begin{equation}
\label{eq:qls_precond_parameter}
    s=\frac{t}{c\alpha_{A^{-1}}},\qquad
    \frac{\norm{A^{-1}\ket{b}}}{c^2\alpha_{A^{-1}}}<s<\frac{\norm{A^{-1}\ket{b}}}{\alpha_{A^{-1}}}\leq\frac{\norm{A^{-1}\ket{b}}}{\norm{A^{-1}}}\leq 1.
\end{equation}
This choice gives the norm bound on the preconditioned inverse matrix
\begin{equation}
    \norm{(SA)^{-1}}\leq\sqrt{\frac{\norm{A^{-1}\ket{b}}^2}{s^2}+\norm{A^{-1}}^2}
    \leq\sqrt{c^4+1}\alpha_{A^{-1}},
\end{equation}
and the preconditioned solution norm
\begin{equation}
    \norm{(SA)^{-1}\ket{b}}
    =\frac{1}{s}\norm{A^{-1}\ket{b}}
    =\frac{c\alpha_{A^{-1}}\norm{A^{-1}\ket{b}}}{t},
\end{equation}
which implies the success amplitude after block preconditioning
\begin{equation}
    \sqrt{p_{\text{succ}}}
    =\frac{\norm{(SA)^{-1}\ket{b}}}{\sqrt{c^4+1}\alpha_{A^{-1}}}
    =\frac{c\norm{A^{-1}\ket{b}}}{\sqrt{c^4+1}t}\geq\frac{1}{\sqrt{c^4+1}}.
\end{equation}

Thus, by invoking the QSVT-based linear system solver (\lem{inv_block}), we can obtain the solution state to accuracy $\epsilon$ and constant success probability at least $\frac{1}{c^4+1}-\mathbf{O}(\epsilon)$ with query complexity
\begin{equation}
    \mathbf{O}\left(\kappa\log\left(\frac{1}{\epsilon}\right)\mathbf{Cost}(O_b)
    +\kappa\log\left(\frac{1}{\epsilon}\right)\mathbf{Cost}(O_A)\right).
\end{equation}
Moreover, this approach does not require padding the target linear system, and appears to be conceptually even simpler than the kernel reflection method~\cite{Dalzell2024shortcut} which in turn simplifies the discrete adiabatic method~\cite{Costa2021linearsystems}.

\begin{theorem}[Quantum linear system algorithm with optimal coefficient block encoding]
\label{thm:qls_opt_coeff}
Let $A$ be the coefficient matrix such that $A/\alpha_A$ is block encoded by $O_A$ with normalization factor $\alpha_A\geq\norm{A}$. Let $\ket{b}$ be the initial state prepared by oracle $O_b$.
Suppose that the solution norm $\norm{A^{-1}\ket{b}}$, and hence the success amplitude $\sqrt{p_{\text{succ}}}=\frac{\norm{A^{-1}\ket{b}}}{\alpha_{A^{-1}}}$, can be estimated to a constant multiplicative accuracy.
Then the quantum state 
\begin{equation}
    \frac{A^{-1}\ket{b}}{\norm{A^{-1}\ket{b}}}
\end{equation}
can be prepared to accuracy $\epsilon$ and success probability $>\frac{1}{2}$ with query complexity
\begin{equation}
    \mathbf{O}\left(\kappa\log\left(\frac{1}{\epsilon}\right)\mathbf{Cost}(O_b)
    +\kappa\log\left(\frac{1}{\epsilon}\right)\mathbf{Cost}(O_A)\right),
\end{equation}
where $\alpha_{A^{-1}}\geq\norm{A^{-1}}$ is a norm upper bound on the inverse matrix, and
$\kappa=\alpha_A\alpha_{A^{-1}}$ is an upper bound on the spectral condition number.

The algorithm solves a preconditioned linear system specified by the scaling operator \eq{qls_precond_operator} and parameter \eq{qls_precond_multi} and \eq{qls_precond_parameter}.
\end{theorem}

%%%%%%%%%%%%%%%%%%%%%%%%%%%%%%%%%%%%%%%%%%%%%%%%%%%%%%%%%%%%%%%%%%%%%%%%%%%%%%
\subsection{Quantum differential equation solver}
\label{sec:precond_diff_eq}
Differential equations arise naturally in a broad range of scientific disciplines including engineering, physics, economics, and biology. However, classical differential equation solvers can struggle to handle problems of large dimensions, which motivates the development of quantum algorithms. To be concrete, consider the system of first-order linear differential equations
\begin{equation}
    \frac{\mathrm{d}}{\mathrm{d}t}x(t)=Ax(t),\qquad
    x(0)=b,
\end{equation}
whose solution is given formally by
\begin{equation}
    x(t)=e^{tA}b.
\end{equation}
Here, we assume the coefficient matrix $A/\alpha_A$ is block encoded by $O_A$ with normalization factor $\alpha_A\geq\norm{A}$ and the initial state $\ket{b}$ is prepared through the oracle $O_{b}$.

Many previous quantum differential equation algorithms proceed by recasting the problem as solving a system of linear equations, and then solving the recast problem using a quantum linear system solver. However, such algorithms have query complexity of $O_{b}$ comparable to that of $O_{A}$, and are thus inefficient when preparing initial states incurs a substantial cost. We show that the cost of initial state preparation can be lowered using our \thm{qls_opt_init} and block preconditioning, nearly matching or outperforming the performance of alternative solvers for differential equations, and attaining the query complexity lower bound established in~\cite{Fang2023timemarchingbased}. For illustration purposes, we focus on the differential equation solvers implementing the truncated Taylor series~\cite{Berry2017Differential}, although similar improvements can be achieved for other linear-system-based solvers, such as the one in~\cite{QEVP} based on the truncated Faber series and quantum eigenvalue transformation (\sec{precond_transform}).

In the Taylor-series algorithm, the coefficient matrix is given by
\begin{equation}
\begin{aligned}
    &C_{n,k,p}\left(\frac{A}{\alpha_A}\right)\\
    &=
    \sum_{i=0}^{n-1}\left(\sum_{j=0}^k\ketbra{i(k+1)+j}{i(k+1)+j}\otimes I-\sum_{j=1}^k\ketbra{i(k+1)+j}{i(k+1)+j-1}\otimes\frac{A}{j\alpha_A}\right)\\
    &\quad-\sum_{i=0}^{n-1}\sum_{j=0}^k\ketbra{(i+1)(k+1)}{i(k+1)+j}\otimes I\\
    &\quad+\left(\sum_{j=0}^p\ketbra{n(k+1)+j}{n(k+1)+j}\otimes I-\sum_{j=1}^p\ketbra{n(k+1)+j}{n(k+1)+j-1}\otimes I\right),
\end{aligned}
\end{equation}
whereas the initial state is $\ket{0}\ket{b}$. Here, the parameter $n$ denotes the total number of time steps in the algorithm, $k$ denotes the Taylor truncate order of the evolution operator within each step, $p$ refers to the additional padding steps required to boost the success probability, with the successful outcomes labeled by $n(k+1),\ldots,n(k+1)+p$. We set $p=n=\mathbf{\Theta}(\alpha_At)$. Then after running the quantum linear system solver and measuring the ancilla register of the output state, we get the outcomes $n(k+1),\ldots,n(k+1)+p$ with probability at least 
\begin{equation}
    \mathbf{\Omega}\left(\frac{\norm{e^{tA}\ket{b}}}{\max_{0\leq\tau\leq t}\norm{e^{\tau A}\ket{b}}}\right),
\end{equation}
which can then be boosted to unity by amplitude amplification.
The resulting state has error at most $\mathbf{O}\left(\frac{\alpha_At}{k!}\right)$, which implies the choice of $k=\mathbf{O}\left(\log\left(\frac{\alpha_At}{\epsilon}\right)\right)$ to achieve a target accuracy $\epsilon$.
Meanwhile, the coefficient matrix satisfies 
\begin{equation}
    \norm{C_{n,k,p}\left(\frac{A}{\alpha_A}\right)}=\mathbf{O}\left(\sqrt{k}\right),\qquad
    \norm{C_{n,k,p}^{-1}\left(\frac{A}{\alpha_A}\right)}=\mathbf{O}\left(\max_{0\leq\tau\leq t}\norm{e^{\tau A}}\sqrt{k}n\right).
\end{equation}
With the quantum linear system solver~\cite{Costa2021linearsystems}, this method has a query complexity of
\begin{small}
\begin{equation}
\newmaketag
    \mathbf{O}\left(\frac{\max_{0\leq\tau\leq t}\norm{e^{\tau A}\ket{b}}}{\norm{e^{tA}\ket{b}}}
    \max_{0\leq\tau\leq t}\norm{e^{\tau A}}
    \alpha_At\log\left(\frac{\alpha_At}{\epsilon}\right)
    \log\left(\frac{\max_{0\leq\tau\leq t}\norm{e^{\tau A}\ket{b}}}{\norm{e^{tA}\ket{b}}\epsilon}\right)
    \left(\mathbf{Cost}\left(O_{b}\right)+\mathbf{Cost}\left(O_{A}\right)\right)\right).
\end{equation}
\end{small}%
The state-of-the-art result~\cite{Krovi2023improvedquantum,BerryCosta22} has a slightly better query complexity of $O_{b}$ by shaving off the logarithmic factor $\log\left(\frac{\alpha_At}{\epsilon}\right)$.

To improve over this result, we perform the following block preconditioning. We choose the preconditioner based on the initial ancilla state $\Pi_{\text{cond}}=\ketbra{0}{0}\otimes I$ and set the scaling parameter $s=\frac{1}{\sqrt{kn}}$, defining the scaling operator
\begin{equation}
    S=\frac{1}{\sqrt{kn}}\ketbra{0}{0}\otimes I+\left(I-\ketbra{0}{0}\right)\otimes I,\qquad
    S^{-1}=\sqrt{kn}\ketbra{0}{0}\otimes I+\left(I-\ketbra{0}{0}\right)\otimes I.
\end{equation}
Note in particular that our preconditioner depends only on the ancilla state, and can be implemented without querying the oracle $O_{b}$.
This block preconditioning increases the solution norm of quantum linear system algorithm to
\begin{equation}
    \norm{\left(SC_{n,k,p}\left(\frac{A}{\alpha_A}\right)\right)^{-1}\ket{0}\ket{b}}
    =\frac{1}{s}\norm{C_{n,k,p}^{-1}\left(\frac{A}{\alpha_A}\right)\ket{0}\ket{b}}
    =\sqrt{kn}\norm{C_{n,k,p}^{-1}\left(\frac{A}{\alpha_A}\right)\ket{0}\ket{b}},
\end{equation}
while the norm of the inverse matrix remains asymptotically unaffected
\begin{equation}
\begin{aligned}
    \norm{\left(SC_{n,k,p}\left(\frac{A}{\alpha_A}\right)\right)^{-1}}
    &\leq\sqrt{\frac{\norm{C_{n,k,p}^{-1}\left(\frac{A}{\alpha_A}\right)\ket{0}\otimes I}^2}{s^2}+\norm{C_{n,k,p}^{-1}\left(\frac{A}{\alpha_A}\right)}^2}\\
    &=\mathbf{O}\left(\sqrt{kn}\sqrt{n}\max_{0\leq\tau\leq t}\norm{e^{\tau A}}
    +\max_{0\leq\tau\leq t}\norm{e^{\tau A}}\sqrt{k}n\right)\\
    &=\mathbf{O}\left(\max_{0\leq\tau\leq t}\norm{e^{\tau A}}\sqrt{k}n\right).
\end{aligned}
\end{equation}

Suppose that the norm of solution state of the quantum linear system solver can be estimated to a constant multiplicative accuracy. Then, the query complexity of the initial state preparation is bounded by
\begin{equation}
    \mathbf{O}\left(\frac{\max_{0\leq\tau\leq t}\norm{e^{\tau A}}\sqrt{k}n}{\sqrt{kn}\norm{C_{n,k,p}^{-1}\left(\frac{A}{\alpha_A}\right)\ket{0}\ket{b}}}
    \cdot\frac{\sqrt{kn}\norm{C_{n,k,p}^{-1}\left(\frac{A}{\alpha_A}\right)\ket{0}\ket{b}}}{\sqrt{kn}\sqrt{n}\norm{e^{tA}\ket{b}}}\right)
    =\mathbf{O}\left(\frac{\max_{0\leq\tau\leq t}\norm{e^{\tau A}}}{\norm{e^{tA}\ket{b}}}\right).
\end{equation}
We have thus obtained:
\begin{theorem}[Quantum differential equation solver with optimal initial state preparation]
\label{thm:diff_eq_init}
Let $A$ be a matrix such that $A/\alpha_A$ is block encoded by $O_A$ with normalization factor $\alpha_A\geq\norm{A}$. Let $\ket{b}$ be the initial state prepared by oracle $O_b$.
Then the quantum state
\begin{equation}
    \frac{e^{tA}\ket{b}}{\norm{e^{tA}\ket{b}}}
\end{equation}
can be prepared to accuracy $\epsilon$ and success probability $>\frac{1}{2}$ with query complexity
\begin{equation}
\begin{aligned}
    &\mathbf{O}\Bigg(\frac{\alpha_{\exp}}{\alpha_{\exp,b}}\mathbf{Cost}\left(O_{b}\right)\\
    &\qquad+\frac{\alpha_{\exp,b,\max}}{\alpha_{\exp,b}}
    \alpha_{\exp}\alpha_At
    \log\left(\frac{\alpha_{\exp}}{\alpha_{\exp,b}}\right)
    \log\left(\frac{\log\left(\frac{\alpha_{\exp}}{\alpha_{\exp,b}}\right)}{\epsilon}\right)
    \log\left(\frac{\alpha_At}{\epsilon}\right)
    \mathbf{Cost}\left(O_{A}\right)\Bigg),
\end{aligned}
\end{equation}
where $\alpha_{\exp}\geq\max_{0\leq\tau\leq t}\norm{e^{\tau A}}$ is an upper bound on norm of the evolution operator, $\alpha_{\exp,b,\max}\geq \max_{0\leq\tau\leq t}\norm{e^{\tau A}\ket{b}}$ is a norm upper bound on solution state of the differential equation, and $\alpha_{\exp,b}\leq\norm{e^{tA}\ket{b}}$ is a lower bound on norm of the solution state.
\end{theorem}
\begin{remark}
The query complexity quoted above follows from \thm{qls_opt_init}, assuming that a constant multiplicative estimate of the solution norm is available for the linear system problem. Without this assumption, we can obtain such an estimate using \thm{sol_est}, with the complexity of initial state preparation $\mathbf{O}\left(\frac{\alpha_{\exp}}{\alpha_{\exp,b}}\right)$ remaining the same.
When $\alpha_{\exp}=\mathbf{O}(1)$ as is commonly assumed by recent work on differential equations~\cite{Krovi2023improvedquantum,BerryCosta22,AnChildsLin23}, our complexity of initial state preparation saturates the lower bound of~\cite[Theorem 10]{Fang2023timemarchingbased}.
\end{remark}

%%%%%%%%%%%%%%%%%%%%%%%%%%%%%%%%%%%%%%%%%%%%%%%%%%%%%%%%%%%%%%%%%%%%%%%%%%%%%%
\subsection{Quantum eigenvalue estimator}
\label{sec:precond_est}
The efficient solution of the eigenvalue estimation problem underlies the quantum speedups for factoring integers~\cite{sho_polynomialtime_1997} and elucidating chemical reactions~\cite{vonBurg21,Lee21}. Here, an initial state $\ket{\psi}$ close to an eigenvector of the input matrix $A$ is given, and the goal is to estimate the corresponding eigenvalue.
If $A$ is Hermitian, this can be solved using the quantum singular value estimation algorithm with optimal query complexity~\cite{kerenidis2016quantum,Chakraborty2018BlockEncoding,Gilyen2018singular}. However, the problem becomes considerably more difficult when $A$ is a nonnormal matrix, which is relevant for applications in simulating transcorrelated quantum chemistry~\cite{McArdle20} and non-Hermitian physics~\cite{Bender07,Ashida20}.

To simplify the discussion, we assume $A=S\Lambda S^{-1}$ is diagonalizable with real spectra, and $A/\alpha_A$ is block encoded by $O_A$ with normalization factor $\alpha_A\geq\norm{A}$.
Suppose that oracle $O_\psi\ket{0}=\ket{\psi}$ prepares the initial state close to an eigenstate $\ket{\psi_{j}}$ such that $A\ket{\psi_{j}}=\lambda_j\ket{\psi_{j}}$.
Then recent work provides a linear-system-based quantum algorithm~\cite[Theorem 3]{QEVP} that estimates $\lambda_j$ to accuracy $\epsilon$ and success probability $>\frac{1}{2}$, with query complexity
\begin{equation}
    \mathbf{O}\left(\frac{\alpha_A\kappa_S}{\epsilon}
    \left(\mathbf{Cost}\left(O_A\right)+\mathbf{Cost}\left(O_{\psi}\right)\right)\right),
\end{equation}
where $\kappa_S\geq\norm{S}\norm{S^{-1}}$ is an upper bound on the spectral condition number of the basis transformation.
This improves over previous results~\cite{Shao2019eigenvalues,Shao2020GeneralizedEigenvalue} based on differential equation solvers.
A related result is obtained in~\cite[Theorem 12]{QEVP} for eigenvalue estimation on the unit circle, which is more recently generalized by~\cite{AlaseKaruvade24}.
In terms of the query complexity of $O_A$, the complexity quoted above exactly matches the Heisenberg scaling~\cite{Giovannetti06,Zwierz10} and is provably optimal for eigenvalue estimation.
However, the algorithm uses the same number of queries to the initial state, which underperforms alternative methods~\cite{Zhang2024Nonnormal} in this regard.

In the algorithm of~\cite{QEVP}, the eigenvalue estimation problem is solved by generating a Chebyshev history state, which is in turn realized by solving a linear system. Specifically, we introduce
\begin{equation}
\NiceMatrixOptions
{
    custom-line = 
    {
        letter = I , 
        command = hdashedline , 
        tikz = {dashed,dash phase=3pt} ,
        width = \pgflinewidth
    }
}
\begin{aligned}
    \mathbf{Pad}(A)&=
    \begin{bmatrix}
        A_{11} & 0\\
        A_{21} & A_{22}
    \end{bmatrix}\\
    &=\begin{bNiceArray}{ccccccIccc}
        I & 0 & \cdots & \cdots & \cdots &0   &0 &\cdots &\cdots\\
        -\frac{2A}{\alpha_A} & I & \ddots & \ddots & \ddots & \vdots   &\vdots &\vdots &\vdots \\
        I & -\frac{2A}{\alpha_A} & I & \ddots &\ddots & \vdots   &\vdots &\vdots &\vdots \\
        0 & I & \frac{-2A}{\alpha_A} & I & \ddots & \vdots   &\vdots &\vdots &\vdots \\
        \vdots & \ddots & \ddots & \ddots & \ddots & 0   &\vdots &\vdots &\vdots \\
        0 & \cdots & 0 & I & -\frac{2A}{\alpha_A} & I   &0 &\cdots &\cdots \\
        \hdashedline
        0 & \cdots &\cdots &0 & 0 & -I   &I &0 &\cdots \\
        0 & \cdots & \cdots &\cdots & 0 & 0   &-I &I &\ddots \\
        \vdots & \vdots & \vdots & \vdots & \vdots & \vdots   &\ddots &\ddots &\ddots \\
    \end{bNiceArray},\\
\end{aligned}
\end{equation}
where $A_{11}$ is $n$-by-$n$ and $A_{22}$ is $\eta n$-by-$\eta n$ for some nonnegative integers $\eta$ and $n$. The inverse of this matrix is
\begin{equation}
\NiceMatrixOptions
{
    custom-line = 
    {
        letter = I , 
        command = hdashedline , 
        tikz = {dashed,dash phase=3pt} ,
        width = \pgflinewidth
    }
}
    \mathbf{Pad}(A)^{-1}=
    \begin{bNiceArray}{ccccIccc}
        \mathbf{U}_{0}\left(\frac{A}{\alpha_A}\right) & 0 & \cdots & 0 & 0 & \cdots & \cdots \\
        \mathbf{U}_{1}\left(\frac{A}{\alpha_A}\right) & \mathbf{U}_{0}\left(\frac{A}{\alpha_A}\right) & \ddots & \vdots & \vdots & \vdots & \vdots \\
        \vdots & \ddots & \ddots & 0 & \vdots & \vdots & \vdots \\
        \mathbf{U}_{n-1}\left(\frac{A}{\alpha_A}\right) & \cdots & \mathbf{U}_{1}\left(\frac{A}{\alpha_A}\right) & \mathbf{U}_{0}\left(\frac{A}{\alpha_A}\right) & 0 & \cdots & \cdots \\
        \hdashedline
        \mathbf{U}_{n-1}\left(\frac{A}{\alpha_A}\right) & \cdots & \mathbf{U}_{1}\left(\frac{A}{\alpha_A}\right) & \mathbf{U}_{0}\left(\frac{A}{\alpha_A}\right) & I & 0 & \cdots \\
        \mathbf{U}_{n-1}\left(\frac{A}{\alpha_A}\right) & \cdots & \mathbf{U}_{1}\left(\frac{A}{\alpha_A}\right) & \mathbf{U}_{0}\left(\frac{A}{\alpha_A}\right) & I & I & \ddots \\
        \vdots & \vdots & \vdots & \vdots & \vdots & \vdots & \ddots 
    \end{bNiceArray},
\end{equation}
with $\mathbf{U}_j$ Chebyshev polynomials of the second kind. Here, the coefficent matrix and its inverse have norm bounds
\begin{equation}
    \norm{\mathbf{Pad}(A)}\leq 4,\qquad
    \norm{\mathbf{Pad}^{-1}(A)}=\mathbf{O}(n\kappa_S).
\end{equation}
For the eigenvalue estimation problem, we choose the initial state $\frac{\ket{0}-\ket{2}}{\sqrt{2}}\ket{\psi}$ and let $\eta=0$. Then the solution state becomes
\begin{equation}
    \frac{\mathbf{Pad}(A)^{-1}\frac{\ket{0}-\ket{2}}{\sqrt{2}}\ket{\psi}}{\norm{\mathbf{Pad}(A)^{-1}\frac{\ket{0}-\ket{2}}{\sqrt{2}}\ket{\psi}}}
    =\frac{\sum_{l=0}^{n-1}\ket{l}\widetilde{\mathbf{T}}_l\left(\frac{A}{\alpha_A}\right)\ket{\psi}}{\norm{\sum_{l=0}^{n-1}\ket{l}\widetilde{\mathbf{T}}_l\left(\frac{A}{\alpha_A}\right)\ket{\psi}}},
\end{equation}
with $\widetilde{\mathbf{T}}_l$ rescaled Chebyshev polynomials of the second kind. As long as this state is produced with constant accuracy and $n=\mathbf{\Theta}\left(\frac{\alpha_A}{\epsilon}\right)$, we can estimate the eigenvalue to accuracy $\epsilon$. This gives the query complexity claimed above.

We can improve the cost of initial state preparation using block preconditioning. Specifically, we choose $\Pi_{\text{cond}}=\frac{\ket{0}-\ket{2}}{\sqrt{2}}\frac{\bra{0}-\bra{2}}{\sqrt{2}}\otimes I$, $s=\frac{1}{\sqrt{n}}$, and define the scaling operator
\begin{equation}
\begin{aligned}
    S&=\frac{1}{\sqrt{n}}\frac{\ket{0}-\ket{2}}{\sqrt{2}}\frac{\bra{0}-\bra{2}}{\sqrt{2}}\otimes I
    +\left(I-\frac{\ket{0}-\ket{2}}{\sqrt{2}}\frac{\bra{0}-\bra{2}}{\sqrt{2}}\right)\otimes I,\\
    S^{-1}&=\sqrt{n}\frac{\ket{0}-\ket{2}}{\sqrt{2}}\frac{\bra{0}-\bra{2}}{\sqrt{2}}\otimes I
    +\left(I-\frac{\ket{0}-\ket{2}}{\sqrt{2}}\frac{\bra{0}-\bra{2}}{\sqrt{2}}\right)\otimes I.
\end{aligned}
\end{equation}
Note again that our preconditioner depends only on the ancilla state, and can be implemented without consuming queries to $O_\psi$. After block preconditioning, the solution norm is increased to~\cite[Lemma 17]{QEVP}
\begin{equation}
\begin{aligned}
    \norm{\left(S\mathbf{Pad}(A)\right)^{-1}\frac{\ket{0}-\ket{2}}{\sqrt{2}}\ket{\psi}}
    &=\sqrt{n}\norm{\mathbf{Pad}^{-1}(A)\frac{\ket{0}-\ket{2}}{\sqrt{2}}\ket{\psi}}
    =\mathbf{\Theta}\left(\sqrt{n}\norm{\sum_{l=0}^{n-1}\ket{l}\widetilde{\mathbf{T}}_l\left(\frac{A}{\alpha_A}\right)\ket{\psi}}\right)\\
    &=\mathbf{\Theta}\left(\sqrt{n}\sqrt{\sum_{l=0}^{n-1}\widetilde{\mathbf{T}}_l^2\left(\frac{\lambda_j}{\alpha_A}\right)}\right)
    =\mathbf{\Theta}(n),
\end{aligned}
\end{equation}
while the norm upper bound on the inverse matrix remains asymptotically unchanged:
\begin{equation}
\begin{aligned}
    \norm{\left(S\mathbf{Pad}(A)\right)^{-1}}
    &\leq\sqrt{\frac{\norm{\mathbf{Pad}^{-1}(A)\frac{\ket{0}-\ket{2}}{\sqrt{2}}\otimes I}^2}{s^2}
    +\norm{\mathbf{Pad}^{-1}(A)}^2}\\
    &=\mathbf{O}\left(\sqrt{n}\norm{\sum_{l=0}^{n-1}\ket{l}\widetilde{\mathbf{T}}_l\left(\frac{A}{\alpha_A}\right)}
    +n\kappa_S\right)
    =\mathbf{O}(n\kappa_S).
\end{aligned}
\end{equation}
Invoking \thm{qls_opt_init}, we have:
\begin{theorem}[Quantum eigenvalue estimator with improved initial state preparation]
\label{thm:qeve_init}
Let $A=S\Lambda S^{-1}$ be a diagonalizable matrix with real spectra and upper bound $\kappa_S\geq\norm{S}\norm{S^{-1}}$ on the condition number, such that $A/\alpha_A$ is block encoded by $O_A$ with normalization factor $\alpha_A\geq\norm{A}$.
Suppose that oracle $O_\psi\ket{0}=\ket{\psi}$ prepares an initial state within distance $\norm{\ket{\psi}-\ket{\psi_j}}=\mathbf{O}\left(\frac{1}{\kappa_S}\right)$ to an eigenstate such that $A\ket{\psi_{j}}=\lambda_j\ket{\psi_{j}}$. Then $\lambda_j$ can be estimate to accuracy $\epsilon$ and success probability $>\frac{1}{2}$ with query complexity
\begin{equation}
    \mathbf{O}\left(\kappa_S
    \mathbf{Cost}\left(O_{\psi}\right)
    +\frac{\alpha_A\kappa_S}{\epsilon}
    \log\left(\kappa_S\right)
    \log\log\left(\kappa_S\right)
    \mathbf{Cost}\left(O_A\right)\right).
\end{equation}
\end{theorem}
\begin{remark}
The query complexity quoted above follows from~\thm{qls_opt_init} assuming a constant multiplicative approximation of the solution norm is available. Without this prior knowledge, one can use~\thm{sol_est} to get such an estimate with the same cost of initial state preparation, while the cost of block encoding remains the Heisenberg-limited scaling~\cite{Giovannetti06,Zwierz10}.

When $A$ is Hermitian, $\kappa_S=1$. In this case, we only need a constant number of queries to the initial state oracle and a constant overlap with the target eigenstate. Our result then recovers the performance of standard quantum phase estimation. For estimating real eigenvalues, the above result outperforms the one reported in~\cite{Zhang2024Nonnormal}. Note that the approach of~\cite{Zhang2024Nonnormal} works by implementing projections onto $\mathbf{Ker}(A-\mu I)$ with different shifting values $\mu$, assuming oracular queries to initial states that are $\mu$-dependent. This input model appears to be more demanding than ours.
\end{remark}

%%%%%%%%%%%%%%%%%%%%%%%%%%%%%%%%%%%%%%%%%%%%%%%%%%%%%%%%%%%%%%%%%%%%%%%%%%%%%%
\subsection{Quantum eigenvalue transformer and ground state preparator}
\label{sec:precond_transform}
We obtain an analogous improvement for the quantum eigenvalue transformation of nonnormal matrices. For a diagonalizable input matrix block enocded as $A/\alpha_A$ with normalization factor $\alpha_A\geq\norm{A}$, this means performing a polynomial transformation on the eigenvalues 
\begin{equation}
\label{eq:qevt}
    \frac{A}{\alpha_A}=S\frac{\Lambda}{\alpha_A} S^{-1}\mapsto
    p\left(\frac{A}{\alpha_A}\right)=Sp\left(\frac{\Lambda}{\alpha_A}\right)S^{-1},
\end{equation}
applied to an initial state $\ket{\psi}$.
This covers the special case where $A$ is Hermitian, which is relevant for applications such as Hamiltonian simulation~\cite{Low2016Qubitization}, ground state~\cite{Lin2020nearoptimalground} and thermal state preparation~\cite{Gilyen2018singular}.
More generally, efficient algorithms for the eigenvalue transformation of nonnormal matrices can be applied to solve differential equations and prepare ground states of non-Hermitian matrices with real spectra.

A linear-system-based quantum algorithm was recently developed to transform eigenvalues of nonnormal matrices~\cite[Theorem 4 and 10]{QEVP}. For a diagonalizable input matrix $A=S\Lambda S^{-1}$ with only real eigenvalues, the previous method has query complexity
\begin{equation}
    \mathbf{O}\left(
    \frac{\norm{p}_{\max,\left[-\frac{1}{2},\frac{1}{2}\right]}\kappa_S^2n}{\norm{p\left(\frac{A}{\alpha_A}\right)\ket{\psi}}}
    \log\left(
    \frac{\norm{p}_{\max,\left[-\frac{1}{2},\frac{1}{2}\right]}\kappa_S\log(n)}{\norm{p\left(\frac{A}{\alpha_A}\right)\ket{\psi}}\epsilon}\right)\log\left(n\right)
    \left(\mathbf{Cost}\left(O_A\right)
    +\mathbf{Cost}\left(O_{\psi}\right)\right)\right),
\end{equation}
where $\epsilon$ is the accuracy of the output state, $\kappa_S$ is the condition number of the basis transformation, $p$ is the target polynomial with norm $\norm{p}_{\max,\left[-\frac{1}{2},\frac{1}{2}\right]}=\max_{x\in\left[-\frac{1}{2},\frac{1}{2}\right]}\abs{p(x)}$ and $n-1$ is its degree.
This is achieved by inverting the coefficient matrix $\mathbf{Pad}^{-1}(A)$ introduced in the previous subsection with $\eta=1$, which has the norm bounds $\norm{\mathbf{Pad}(A)}\leq 4$ and $\norm{\mathbf{Pad}^{-1}(A)}=\mathbf{O}(n\kappa_S)$ same as before. For the eigenvalue transformation problem, supposing that the target polynomial is represented under the (rescaled) Chebyshev basis as $p(x)=\sum_{j=0}^{n-1}\widetilde{\beta}_j\widetilde{\mathbf{T}}_j(x)$, we choose the initial state $\ket{0}\ket{\beta}\ket{\psi}$ where
\begin{equation}
\label{eq:beta_state}
    \ket{\beta}=\frac{1}{\alpha_{\widetilde{\beta}}}\sum_{k=0}^{n-1}(\widetilde\beta_k-\widetilde\beta_{k+2})\ket{n-1-k},\qquad
    \alpha_{\widetilde{\beta}}=\sqrt{\sum_{k=0}^{n-1}|\widetilde{\beta}_k-\widetilde{\beta}_{k+2}|^2}
    =\mathbf{\Theta}\left(\norm{p(\cos)\sin}_{2,[-\pi,\pi]}\right),
\end{equation}
where $\norm{p(\cos)\sin}_{2,[-\pi,\pi]}=\sqrt{\int_{-\pi}^{\pi}\mathrm{d}\theta\
    \abs{p(\cos{\theta})\sin{\theta}}^2}$.
Applying the quantum linear system solver then produces the Chebyshev history state
    \begin{equation}
        \frac{\ket{0}\sum_{l=0}^{n-1}\ket{l}
        \sum_{k=n-1-l}^{n-1}\widetilde{\beta}_k\widetilde{\mathbf{T}}_{k+l-n+1}\left(\frac{A}{\alpha_A}\right)\ket{\psi}
        +\ket{1}\sum_{l=0}^{n-1}\ket{l}
        \sum_{k=0}^{n-1}\widetilde{\beta}_k\widetilde{\mathbf{T}}_{k}\left(\frac{A}{\alpha_A}\right)\ket{\psi}}
        {\sqrt{\sum_{l=0}^{n-1}\norm{\sum_{k=n-1-l}^{n-1}\widetilde{\beta}_k\widetilde{\mathbf{T}}_{k+l-n+1}\left(\frac{A}{\alpha_A}\right)\ket{\psi}}^2
        +n\norm{\sum_{k=0}^{n-1}\widetilde{\beta}_k\widetilde{\mathbf{T}}_{k}\left(\frac{A}{\alpha_A}\right)\ket{\psi}}^2}}.
    \end{equation}
Here, the component flagged by the ancilla state $\ket{1}$ is desired, and the success amplitude is at least
\begin{equation}
    \mathbf{\Omega}\left(\frac{\norm{p\left(\frac{A}{\alpha_A}\right)\ket{\psi}}}{\norm{p}_{\max,\left[-\frac{1}{2},\frac{1}{2}\right]}\kappa_S\log(n)}\right).
\end{equation}
This can be boosted close to $1$ by amplitude amplification, leading to the query complexity cited above.
Note that compared to~\cite{QEVP}, we have changed the domain from $[-1,1]$ to $\left[-\frac{1}{2},\frac{1}{2}\right]$ which is without loss of generality by rescaling the input block encoding.

Once again, we can improve the query complexity of the initial state preparation using block preconditioning. To this end, we choose $\Pi_{\text{cond}}=\ketbra{\beta}{\beta}\otimes I$, $s=\frac{\norm{p}_{\max,[-\frac{1}{2},\frac{1}{2}]}}{\sqrt{n}\alpha_{\widetilde{\beta}}}$, and define the scaling operator
\begin{equation}
\begin{aligned}
    S&=\frac{\norm{p}_{\max,\left[-\frac{1}{2},\frac{1}{2}\right]}}{\sqrt{n}\alpha_{\widetilde{\beta}}}\ketbra{0,\beta}{0,\beta}\otimes I+\left(I-\ketbra{0,\beta}{0,\beta}\right)\otimes I,\\
    S^{-1}&=\frac{\sqrt{n}\alpha_{\widetilde{\beta}}}{\norm{p}_{\max,\left[-\frac{1}{2},\frac{1}{2}\right]}}\ketbra{0,\beta}{0,\beta}\otimes I+\left(I-\ketbra{0,\beta}{0,\beta}\right)\otimes I.
\end{aligned}
\end{equation}
Recall that in order for this to be a valid block preconditioning, $S$ must be invertible which further requires $0<s<1$. 
This is confirmed by the following lemma.
% In the following lemma, we show that $\norm{p}_{\max,\left[-\frac{1}{2},\frac{1}{2}\right]}=\mathbf{O}\left(\sqrt{n}\alpha_{\widetilde{\beta}}\right)$ always holds. Thus we can scale down $s$ by at most a constant factor to meet the requirement. For presentational purposes, however, we will continue to use the above concrete expression.
\begin{lemma}
Let $p(x)$ be a polynomial with the expansion $p(x)=\sum_{j=0}^{n-1}\widetilde{\beta}_j\widetilde{\mathbf{T}}_j(x)$ into rescaled Chebyshev polynomials of the first kind. It holds
\begin{equation}
    \norm{p}_{\max,\left[-\frac{1}{2},\frac{1}{2}\right]}<\sqrt{n}\sqrt{\sum_{k=0}^{n-1}|\widetilde{\beta}_k-\widetilde{\beta}_{k+2}|^2}.
\end{equation}
\end{lemma}
\begin{proof}
Extending the definition of coefficients by setting $\widetilde{\beta}_{n}=\widetilde{\beta}_{n+1}=\cdots=0$, we have
\begin{equation}
\begin{aligned}
    p(x)&=\sum_{j=0}^{\infty}\widetilde{\beta}_j\widetilde{\mathbf{T}}_j(x)
    =\sum_{j=2}^{\infty}\widetilde{\beta}_j\frac{\mathbf{U}_j(x)-\mathbf{U}_{j-2}(x)}{2}
    +\widetilde{\beta}_1\widetilde{\mathbf{T}}_1(x)
    +\widetilde{\beta}_0\widetilde{\mathbf{T}}_0(x)\\
    &=\frac{1}{2}\sum_{j=2}^{\infty}\widetilde{\beta}_j\mathbf{U}_j(x)
    -\frac{1}{2}\sum_{j=0}^{\infty}\widetilde{\beta}_{j+2}\mathbf{U}_j(x)
    +\widetilde{\beta}_1\widetilde{\mathbf{T}}_1(x)
    +\widetilde{\beta}_0\widetilde{\mathbf{T}}_0(x)\\
    &=\frac{1}{2}\sum_{j=0}^{\infty}\left(\widetilde{\beta}_j-\widetilde{\beta}_{j+2}\right)\mathbf{U}_j(x)
    -\frac{1}{2}\widetilde{\beta}_1\mathbf{U}_1(x)
    -\frac{1}{2}\widetilde{\beta}_0\mathbf{U}_0(x)
    +\widetilde{\beta}_1\widetilde{\mathbf{T}}_1(x)
    +\widetilde{\beta}_0\widetilde{\mathbf{T}}_0(x)\\
    &=\frac{1}{2}\sum_{j=0}^{\infty}\left(\widetilde{\beta}_j-\widetilde{\beta}_{j+2}\right)\mathbf{U}_j(x)
    =\frac{1}{2}\sum_{j=0}^{n-1}\left(\widetilde{\beta}_j-\widetilde{\beta}_{j+2}\right)\mathbf{U}_j(x).
\end{aligned}
\end{equation}
Now by the Cauchy-Schwarz inequality,
\begin{equation}
    \abs{p(x)}\leq\frac{1}{2}\sqrt{\sum_{j=0}^{n-1}\mathbf{U}_j^2(x)}
    \sqrt{\sum_{k=0}^{n-1}|\widetilde{\beta}_k-\widetilde{\beta}_{k+2}|^2}.
\end{equation}

We now claim that $\sum_{j=0}^{n-1}\mathbf{U}_j^2(x)\leq\frac{2}{3}n+1$ for $-\frac{1}{2}\leq x\leq\frac{1}{2}$.
To this end, let us introduce the angular variable $\phi=\frac{1}{2\pi}\arccos(x)$ and consider
\begin{equation}
\begin{aligned}
    \mathbf{U}_{n-1}^2\left(x\right)+\cdots+\mathbf{U}_{0}^2\left(x\right)
    &=\sum_{j=0}^{n-1}\frac{\sin^2(2\pi (j+1)\phi)}{\sin^2(2\pi\phi)}
    =
    \frac{1}{{\sin^2(2\pi\phi)}}\sum_{j=1}^{n}\frac{1-\cos(4\pi j\phi)}{2}\\
    &=\frac{1}{{\sin^2(2\pi\phi)}}\left(\frac{n}{2}-\frac{1}{4}\sum_{j=0}^{n-1}\left(e^{i4\pi j\phi}+e^{-i4\pi j\phi}\right)\right)\\
    &=\frac{1}{{\sin^2(2\pi\phi)}}\left(\frac{n}{2}-\frac{1}{2}\frac{\sin(2\pi n\phi)}{\sin(2\pi\phi)}\cos((n-1)2\pi\phi)\right).
\end{aligned}
\end{equation}
By our assumption, $\frac{\sqrt{3}}{2}\leq\sin(2\pi\phi)\leq1$ and hence
\begin{equation}
    \frac{n}{2}-\frac{\sqrt{3}}{3}
    \leq\mathbf{U}_{n-1}^2\left(x\right)+\cdots+\mathbf{U}_{0}^2\left(x\right)
    \leq\frac{4}{3}\left(\frac{n}{2}+\frac{\sqrt{3}}{3}\right).
\end{equation}
\end{proof}

Similar as before, our preconditioner is defined by the initial ancilla state, and its implementation uses no query to $O_{\psi}$. After block preconditioning, the solution norm becomes
\begin{equation}
    \norm{\left(S\mathbf{Pad}(A)\right)^{-1}\ket{0}\ket{\beta}\ket{\psi}}
    =\frac{\sqrt{n}\alpha_{\widetilde{\beta}}}{\norm{p}_{\max,\left[-\frac{1}{2},\frac{1}{2}\right]}}
    \norm{\mathbf{Pad}^{-1}(A)\ket{0}\ket{\beta}\ket{\psi}},
\end{equation}
whereas the upper bound on the norm of inverse matrix is asymptotically:
\begin{equation}
\begin{aligned}
    \norm{\left(S\mathbf{Pad}(A)\right)^{-1}}
    &\leq\sqrt{\frac{\norm{\mathbf{Pad}^{-1}(A)\ket{0,\beta}\otimes I}^2}{s^2}
    +\norm{\mathbf{Pad}^{-1}(A)}^2}\\
    &=\mathbf{O}\left(\frac{\sqrt{n}\alpha_{\widetilde{\beta}}}{\norm{p}_{\max,\left[-\frac{1}{2},\frac{1}{2}\right]}}\frac{\sqrt{n}\max_{l}\norm{\sum_{k=l}^{n-1}\widetilde\beta_k\widetilde{\mathbf{T}}_{k-l}\left(\frac{A}{\alpha_A}\right)}}{\alpha_{\widetilde{\beta}}}
    +n\kappa_S\right)\\
    &=\mathbf{O}\left(n\kappa_S\log(n)\right).
\end{aligned}
\end{equation}
This is larger by a factor of $\log(n)$ in the worst case, but one can potentially remove this logarithmic factor when running the algorithm over an average input similar as in~\cite{QEVP}.

Suppose that the norm of solution state from the quantum linear system solver (\thm{qls_opt_init}) can be estimated to a constant multiplicative accuracy. Then the query complexity of the initial state oracle is bounded by
\begin{equation}
\begin{aligned}
    &\mathbf{O}\left(\frac{n\kappa_S\log(n)}{\frac{\sqrt{n}\alpha_{\widetilde{\beta}}}{\norm{p}_{\max,\left[-\frac{1}{2},\frac{1}{2}\right]}}
    \norm{\mathbf{Pad}^{-1}(A)\ket{0}\ket{\beta}\ket{\psi}}}
    \frac{\frac{\sqrt{n}\alpha_{\widetilde{\beta}}}{\norm{p}_{\max,\left[-\frac{1}{2},\frac{1}{2}\right]}}
    \norm{\mathbf{Pad}^{-1}(A)\ket{0}\ket{\beta}\ket{\psi}}}{\frac{\sqrt{n}\alpha_{\widetilde{\beta}}}{\norm{p}_{\max,\left[-\frac{1}{2},\frac{1}{2}\right]}}\frac{\sqrt{n}}{\alpha_{\widetilde{\beta}}}\norm{p\left(\frac{A}{\alpha_A}\right)\ket{\psi}}}\right)\\
    &=\mathbf{O}\left(\frac{\norm{p}_{\max,\left[-\frac{1}{2},\frac{1}{2}\right]}\kappa_S\log(n)}{\norm{p\left(\frac{A}{\alpha_A}\right)\ket{\psi}}}\right).
\end{aligned}
\end{equation}
We thus obtain:
\begin{theorem}[Quantum eigenvalue transformer with improved initial state preparation]
\label{thm:qevt_init}
Let $A=S\Lambda S^{-1}$ be a diagonalizable matrix with real spectra and upper bound $\kappa_S\geq\norm{S}\norm{S^{-1}}$ on the condition number, such that $A/\alpha_A$ is block encoded by $O_A$ with normalization factor $\alpha_A\geq\norm{A}$.
Let $O_\psi\ket{0}=\ket{\psi}$ be the oracle preparing the initial state, and $p(x)=\sum_{k=0}^{n-1}\widetilde\beta_k\widetilde{\mathbf{T}}_{k}(x)=\sum_{k=0}^{n-1}\beta_k{\mathbf{T}}_{k}(x)$ be the Chebyshev expansion of a degree-($n-1$) polynomial $p$.
    Then, the quantum state
    \begin{equation}
        \frac{p\left(\frac{A}{\alpha_A}\right)\ket{\psi}}{\norm{p\left(\frac{A}{\alpha_A}\right)\ket{\psi}}}
    \end{equation}
can be prepared with accuracy $\epsilon$ and success probability $>\frac{1}{2}$ with query complexity
\begin{footnotesize}
\begin{equation}
\newmaketag
\begin{aligned}
    &\mathbf{O}\Bigg(
    \frac{\norm{p}_{\max,[-\frac{1}{2},\frac{1}{2}]}\kappa_S\log\left(n\right)}{\alpha_{p,\psi}}
    \mathbf{Cost}\left(O_{\psi}\right)\\
    &\qquad+\frac{\norm{p}_{\max,[-\frac{1}{2},\frac{1}{2}]}\kappa_S^2n\log\left(n\right)}{\alpha_{p,\psi}}
    \log\left(
    \frac{\norm{p}_{\max,[-\frac{1}{2},\frac{1}{2}]}\kappa_S\log(n)}{\alpha_{p,\psi}}\right)
    \log\left(\frac{\log\left(\frac{\norm{p}_{\max,[-\frac{1}{2},\frac{1}{2}]}\kappa_S\log(n)}{\alpha_{p,\psi}}\right)}{\epsilon}\right)
    \mathbf{Cost}\left(O_{A}\right)\Bigg),
\end{aligned}
\end{equation}
\end{footnotesize}
where $\alpha_{p,\psi}\leq\norm{p\left(\frac{A}{\alpha_A}\right)\ket{\psi}}$ is a norm lower bound on the transformed state.
\end{theorem}

As an immediate application, we also obtain an improved quantum algorithm for ground state preparation.  In the case where the input operator is a Hermitian Hamiltonian, this problem has been extensively studied by previous work~\cite{Poulin09,Ge19}, and can be solved near optimally on a quantum computer~\cite{Lin2020nearoptimalground}.
Here, we consider the general case where inputs are non-normal matrices with real eigenvalues whose ground states are still well defined, which are relevant to applications in non-Hermitian physics and transcorrelated quantum chemistry.
Specifically, let $A=S\Lambda S^{-1}$ be a diagonalizable matrix with only real eigenvalues and an upper bound $\kappa_S$ on the condition number of its basis transformation, such that $A/\alpha_A$ is block encoded by oracle $O_A$ with normalization factor $\alpha_A$.
Suppose that $\lambda_0$ is the smallest eigenvalue of $A$ with the corresponding eigenstate $\ket{\psi_0}$, which is separated from the next eigenvalue $\lambda_1$:
\begin{equation}
    \lambda_0\leq-\frac{\delta_A}{2}<0<\frac{\delta_A}{2}\leq\lambda_1
\end{equation}
for some spectral gap $\delta_A>0$.
Then our goal is to prepare a quantum state that $\epsilon$-approximates the ground state $\ket{\psi_0}$ up to a global phase, given an initial state $\ket{\psi}=\gamma_0\ket{\psi_0}+\sum_{l=1}^{d-1}\gamma_l\ket{\psi_l}$ prepared by oracle $O_{\psi}$.

The best previous quantum ground state preparation algorithm proceeds by implementing a degree
    \begin{equation}
        n=\mathbf{O}\left(\frac{\alpha_A}{\delta_A}
        \log\left(\frac{\kappa_S}{|\gamma_0|\epsilon}\right)\right).
    \end{equation}
polynomial using quantum eigenvalue transformer~\cite[Theorem 8]{QEVP}, which leads to the query complexity
\begin{equation}
    \mathbf{O}\left(\frac{\kappa_S^2}{|\gamma_0|}\frac{\alpha_A}{\delta_A}\log^2\left(\frac{\kappa_S}{|\gamma_0|\epsilon}\right)
    \left(\mathbf{Cost}\left(O_A\right)
    +\mathbf{Cost}\left(O_{\psi}\right)\right)
    \right).
\end{equation}
Using block preconditioning and our quantum linear system solver with optimal initial state preparation, we improve this to:

\begin{theorem}[Quantum ground state preparator with improved initial state preparation]
\label{thm:gs_init}
Let $A=S\Lambda S^{-1}$ be a diagonalizable matrix with real spectra and upper bound $\kappa_S\geq\norm{S}\norm{S^{-1}}$ on the condition number, such that $A/\alpha_A$ is block encoded by $O_A$ with normalization factor $\alpha_A\geq\norm{A}$.
Let eigenvalues of $A$ be ordered nondecreasingly, with $\lambda_0$ the smallest one corresponding to eigenstate $\ket{\psi_0}$, satisfying the condition
    \begin{equation}
        \lambda_0\leq-\frac{\delta_A}{2}<0<\frac{\delta_A}{2}\leq\lambda_1
    \end{equation}
    for some spectral gap $\delta_A>0$. 
Let $O_\psi\ket{0}=\ket{\psi}$ be the oracle preparing the initial state with the eigenbasis expansion
    \begin{equation}
        \ket{\psi}=\gamma_0\ket{\psi_0}+\sum_{l=1}^{d-1}\gamma_l\ket{\psi_l}.
    \end{equation}
Then, the ground state $\ket{\psi_0}$ can be produced with accuracy $\epsilon$, success probability $>\frac{1}{2}$ and global phase factor $\gamma_0/\abs{\gamma_0}$, with query complexity
\begin{equation}
    \mathbf{O}\Bigg(
    \frac{\kappa_S}{|\gamma_0|}\log\left(\frac{\alpha_A}{\delta_A}
        \log\left(\frac{\kappa_S}{|\gamma_0|\epsilon}\right)\right)
    \mathbf{Cost}\left(O_{\psi}\right)
    +\frac{\kappa_S^2}{|\gamma_0|}
    \frac{\alpha_A}{\delta_A}
    \polylog\left(\frac{\alpha_A}{\delta_A}
        \log\left(\frac{\kappa_S}{|\gamma_0|\epsilon}\right)\right)
    \mathbf{Cost}\left(O_{A}\right)\Bigg).
\end{equation}
\end{theorem}

%%%%%%%%%%%%%%%%%%%%%%%%%%%%%%%%%%%%%%%%%%%%%%%%%%%%%%%%%%%%%%%%%%%%%%%%%%%%%%
\subsection{Block encoded quantum eigenvalue transformer}
\label{sec:precond_transform_blk}
We report another improvement to the block-encoding version of the quantum eigenvalue transformation algorithm. This is similar to the application from the previous subsection, except we construct the block encoding of the transformed matrix, rather than applying it to an initial state, making it more versatile when quantum eigenvalue transformation is used as a subroutine in constructing other quantum algorithms.

Specifically, we set $\eta=1$ and use \lem{inv_block} to construct a block encoding of
    \begin{equation}
        \frac{\mathbf{Pad}^{-1}(A)}{2\alpha_{\scriptscriptstyle \mathbf{Pad^{-1}(A)}}}
    \end{equation}
with a normalization factor $\alpha_{\scriptscriptstyle \mathbf{Pad^{-1}(A)}}=\mathbf{O}\left(n\kappa_S\right)$. 
Together with the preparation of initial state $\ket{0}\ket{\beta}$ and unpreparation of $\ket{1}\frac{1}{\sqrt{n}}\sum_{k=0}^{n-1}\ket{k}$ where $\ket{\beta}$ is given by \eq{beta_state}, we obtain the block encoding
    \begin{equation}
        \left(\bra{1}\frac{1}{\sqrt{n}}\sum_{k=0}^{n-1}\bra{k}\otimes I\right)
        \frac{\mathbf{Pad}^{-1}(A)}{2\alpha_{\scriptscriptstyle \mathbf{Pad^{-1}(A)}}}
        \left(\ket{0}\ket{\beta}\otimes I\right)
        =\frac{p\left(\frac{A}{\alpha_A}\right)}{\alpha_{p,\text{pre}}}
    \end{equation}
    with
    \begin{equation}
        \alpha_{p,\text{pre}}=\frac{\alpha_{\scriptscriptstyle \mathbf{Pad(A)}^{-1}}\alpha_{\widetilde\beta}}{\sqrt{n}}
        =\mathbf{O}\left(\sqrt{n}\kappa_S \alpha_{\widetilde\beta}\right),
    \end{equation}
Therefore, we have block encoded the target polynomial but with a larger normalization factor. In prior art~\cite[Theorem 5]{QEVP}, further amplification is performed to reduce the normalization factor, giving
    \begin{equation}
        \frac{p\left(\frac{A}{\alpha_A}\right)}{2\norm{p\left(\frac{A}{\alpha_A}\right)}}.
    \end{equation}
The overall query complexity is then
    \begin{equation}
        \mathbf{O}\left(\frac{\norm{p(\cos)\sin}_{2,[-\pi,\pi]}n^{\frac{3}{2}}\kappa_S^2}{\norm{p\left(\frac{A}{\alpha_A}\right)}}\log\left(\frac{\norm{p(\cos)\sin}_{2,[-\pi,\pi]}\sqrt{n}\kappa_S}{\norm{p\left(\frac{A}{\alpha_A}\right)}\epsilon}\right)\log\left(\frac{1}{\epsilon}\right)\right).
    \end{equation}

We improve this block encoding cost using the block preconditioning technique. Specifically, we define the scaling operator
\begin{equation}
\begin{aligned}
    S&=\frac{\norm{p}_{\max,\left[-\frac{1}{2},\frac{1}{2}\right]}}{\sqrt{n}\alpha_{\widetilde{\beta}}}\ketbra{0,\beta}{0,\beta}\otimes I+\left(I-\ketbra{0,\beta}{0,\beta}\right)\otimes I,\\
    S^{-1}&=\frac{\sqrt{n}\alpha_{\widetilde{\beta}}}{\norm{p}_{\max,\left[-\frac{1}{2},\frac{1}{2}\right]}}\ketbra{0,\beta}{0,\beta}\otimes I+\left(I-\ketbra{0,\beta}{0,\beta}\right)\otimes I
\end{aligned}
\end{equation}
same as in the previous subsection. Recall that this does not increase the asymptotic scaling of the condition number. However, the block encoding now becomes
    \begin{equation}
        \left(\bra{1}\frac{1}{\sqrt{n}}\sum_{k=0}^{n-1}\bra{k}\otimes I\right)
        \frac{\mathbf{Pad}^{-1}(A)S^{-1}}{2\alpha_{\scriptscriptstyle \mathbf{Pad^{-1}(A)}}}
        \left(\ket{0}\ket{\beta}\otimes I\right)
        =\frac{\sqrt{n}\alpha_{\widetilde{\beta}}}{\norm{p}_{\max,\left[-\frac{1}{2},\frac{1}{2}\right]}}\frac{p\left(\frac{A}{\alpha_A}\right)}{\alpha_{p,\text{pre}}}
        =\frac{p\left(\frac{A}{\alpha_A}\right)}{\alpha_{p,\text{cond}}}
    \end{equation}
for
\begin{equation}
    \alpha_{p,\text{cond}}
    =\mathbf{O}\left(\sqrt{n}\kappa_S \alpha_{\widetilde\beta}
    \frac{\norm{p}_{\max,\left[-\frac{1}{2},\frac{1}{2}\right]}}{\sqrt{n}\alpha_{\widetilde{\beta}}}\right)
    =\mathbf{O}\left(\norm{p}_{\max,\left[-\frac{1}{2},\frac{1}{2}\right]}\kappa_S\right).
\end{equation}
So the query complexity to the block encoding is improved from $\mathbf{O}(n^{1.5})$ to $\mathbf{O}(n)$.

\begin{theorem}[Block encoded quantum eigenvalue transformer with linear degree cost]
\label{thm:qevt_blk}
Let $A=S\Lambda S^{-1}$ be a diagonalizable matrix with real spectra and upper bound $\kappa_S\geq\norm{S}\norm{S^{-1}}$ on the condition number, such that $A/\alpha_A$ is block encoded by $O_A$ with normalization factor $\alpha_A\geq\norm{A}$.
Let $p(x)=\sum_{k=0}^{n-1}\widetilde\beta_k\widetilde{\mathbf{T}}_{k}(x)=\sum_{k=0}^{n-1}\beta_k{\mathbf{T}}_{k}(x)$ be the Chebyshev expansion of a degree-($n-1$) polynomial $p$.
Then for any $\alpha_p\geq\norm{p\left(\frac{A}{\alpha_A}\right)}$, the operator
    \begin{equation}
        \frac{p\left(\frac{A}{\alpha_A}\right)}{2\alpha_p}
    \end{equation}
    can be block encoded with accuracy $\epsilon$ using
    \begin{equation}
        \mathbf{O}\left(\frac{\norm{p}_{\max,\left[-\frac{1}{2},\frac{1}{2}\right]}n\kappa_S^2}{\alpha_p}\log\left(\frac{\norm{p}_{\max,\left[-\frac{1}{2},\frac{1}{2}\right]}\kappa_S}{\alpha_p\epsilon}\right)\log\left(\frac{1}{\epsilon}\right)\right)
    \end{equation}
    queries to $O_A$.
\end{theorem}

%% file: discuss.tex
In this work, we have developed a quantum linear system algorithm that achieves an optimal query complexity of initial state preparation, while maintaining a nearly optimal cost scaling of the coefficient block encoding. This outperforms recent linear system solvers that make optimal uses of one resource, but far too many queries to the other. Our algorithm employs a nested amplitude amplification with a deterministic amplification schedule, which substantially simplifies prior approaches based on VTAA while improving their asymptotic query complexity. We also present a quantum algorithm that produces a constant multiplicative approximation of norm of the solution state, improving over all previous approaches where finding such an estimate incurs polylogarithmic overhead.

Our state preparation cost scales strictly linear in the inverse success amplitude $\mathbf{O}\left(\frac{1}{\sqrt{p_{\text{succ}}}}\right)$. We further develop a block preconditioning technique that boosts $p_{\text{succ}}$ in practical applications of quantum linear system solvers. As a result, we obtain quantum algorithms for solving differential equations, for preparing ground states of non-Hermitian operators, and for processing eigenvalues of non-normal matrices, all with reduced query complexity of the initial state preparation, nearly matching or outperforming alternative methods for the same tasks. Of independent interest, our block preconditioning technique also allows for improvement of other scaling parameters. We present an extremely simple quantum linear system solver with an optimal block encoding cost by choosing the initial state as the preconditioner. We also give a preconditioned quantum eigenvalue transformer that implements a degree-$n$ polynomial using $\mathbf{O}(n)$ queries to the block encoding oracle, whereas the best prior scaling was $\mathbf{O}(n^{1.5})$.

We have examined variable time amplitude amplifications with tunable threshold values that capture the power of a generic nested amplification. We have shown that Tunable VTAA makes $\mathbf{O}\left(\frac{1}{\sqrt{p_{\text{succ}}}}\right)$ queries to the initial state matching the performance of standard amplitude amplification, while only contains at most $\mathbf{O}\left(\log\left(\frac{1}{\sqrt{p_{\text{succ}}}}\right)\right)$ nontrivial amplification stages allowing majority of the algorithms to be pre-merged.
Specialized to the quantum linear system problem, we have designed a discretized inverse state which can be prepared by Tunable VTAA with a complexity scaling with $\ell_2$-norm of the input costs. We show that this scheme translates naturally to a deterministic amplification schedule, avoiding the significant setup overhead of prior VTAA approaches. However, we prove that one can in principle optimize the amplification thresholds so that the cost of Tunable VTAA attains $\ell_{\frac{2}{3}}$-quasinorm of the input costs. This improves over the $\ell_1$-norm result of Ambainis and the more commonly used $\ell_2$-norm result. 
With more prior knowledge about the target linear system, this result can be utilized to further improve the query complexity of quantum solvers. For instance, consider the coefficient matrix
\begin{equation}
A=
\begin{bmatrix}
    \frac{1}{3} & & &\\
    & \frac{1}{9} & &\\
    & & \ddots &\\
    & & & \frac{1}{3^m}
\end{bmatrix}
\end{equation}
and initial state $\ket{b}=\ket{m-l-1}$, such that $A\ket{b}=\frac{1}{3^{m-l}}\ket{b}$. Here $A$ can be block encoded with normalization factor $\alpha_A=1$, whereas $\norm{A^{-1}}\leq 3^m=\alpha_{A^{-1}}=\kappa$. Meanwhile, we have $\norm{A^{-1}\ket{b}}=3^{m-l}$, so $\sqrt{p_{\text{succ}}}=\frac{\norm{A^{-1}\ket{b}}}{\alpha_{A^{-1}}}=\frac{1}{3^l}$. Choosing a non-uniform schedule of GPE accuracies $\epsilon_{\text{gpe},j}=\frac{\epsilon}{2^{m-j+2}}$ for $j=m-l+1,\ldots,m$ and $\epsilon_{\text{gpe},j}=\frac{\epsilon}{2^{m-l-j+2}}$ for $j=1,\ldots,m-l$, one can attain the improved query complexity $\mathbf{O}\left(\frac{1}{\sqrt{p_{\text{succ}}}}\mathbf{Cost}(O_b)
    +\kappa\log\left(\frac{1}{\epsilon}\right)\mathbf{Cost}\left(O_A\right)\right)$.
It would be interesting to identify more algorithmic applications where the $\ell_{\frac{2}{3}}$-quasinorm can be tightly bounded without introducing additional scaling factors.

Related to the above discussion, an obvious open question is whether it is possible to have a quantum linear system algorithm with query complexity
\begin{equation}
    \mathbf{O}\left(\frac{1}{\sqrt{p_{\text{succ}}}}\mathbf{Cost}(O_b)
    +\kappa\log\left(\frac{1}{\epsilon}\right)\mathbf{Cost}\left(O_A\right)\right)
\end{equation}
for a generic linear system problem.
From an algorithmic perspective, our Tunable VTAA matches the above query scaling of $O_b$ and, theoretically, the cost of $O_A$ scales strictly linear with the $\ell_{\frac{2}{3}}$-quasinorm of input costs, although upper bounding it in terms of $\kappa$ without logarithmic overhead can be difficult in general. On the other hand, there may also be room to further improve the lower bound when one considers the two oracles $O_b$ and $O_A$ simultaneously.

We have mainly focused on applications of quantum linear system solvers where the target outputs are quantum states or block encoded matrices. In this case, we have introduced a scaling operator such that, when inverted from the right side of coefficient matrix, the success amplitude can be boosted up and complexity of the algorithm can be reduced. However, there also exist other linear-system related problems such as computing Green's functions, where the goal is to estimate the expectation value of an observable. We leave it as a subject for future work to explore whether algorithms for such problems can be improved by implementing generalized versions of block preconditioning on both sides of the coefficient matrix.

A number of other questions call for more investigations. Our preconditioned quantum linear system algorithm has an optimal query complexity of block encoding and appears to be conceptually simpler than recent optimal methods based on the discrete adiabatic evolution and kernel reflection. It may be fruitful to assess the resources required to implement this algorithm, for both near-term and fault-tolerant quantum devices.
It is also possible to improve the constant prefactor of the complexity of VTAA: for instance, one may use the exact trigonometric identity $\frac{\sin(3\theta)}{3\sin(\theta)}=1-\frac{4}{3}\sin^2(\theta)$ as opposed to its looser bounds when the amplification step number is exactly $3$; one may also analyze the nested amplitude amplification directly without resorting to~\prop{tunable_universal} which introduces constant-factor overhead; or one may replace the first $m-l$ stages of GPE by a single GPE covering a larger interval of eigenvalues (although this modification could significantly increase the success probability and result in over amplification).
Additionally, our solution norm estimation algorithm introduces extra logarithmic factors to the complexity of $O_A$, which may be improved by solving an augmented linear system similar to~\cite{Dalzell2024shortcut}. 
Its query complexity of $O_b$ scales with a known lower bound $\alpha_{p_{\text{succ}}}$ worse than the actual success probability $p_{\text{succ}}$. However, we show how to tightly bound $\alpha_{p_{\text{succ}}}$ in applications to differential equations and eigenvalue processing.
Finally, it would be of interest to explore other algorithmic applications of quantum linear system solvers beyond those studied here.

%% file: axiom.tex
In this appendix, we collect mathematical results useful for formulating the axiomatic definition of variable time amplification in \sec{prelim_amp}.
We assume throughout this appendix that an underlying Hilbert space is fixed, on which all operators act.
We denote the complementary operator of $\Pi$ by $\overline{\Pi}=I-\Pi$.

%%%%%%%%%%%%%%%%%%%%%%%%%%%%%%%%%%%%%%%%%%%%%%%%%%%%%%%%%%%%%%%%%%%%%%%%%%%%%%
\subsection{Clock projections}
\label{append:axiom_clock}

\begin{lemma}[Partially ordered projections]
\label{lem:partial_order}
    Let $\Pi_1^2=\Pi_1=\Pi_1^\dagger$ and $\Pi_2^2=\Pi_2=\Pi_2^\dagger$ be orthogonal projections. The following statements are all equivalent:
    \begin{enumerate}
        \item $\mathbf{Im}(\Pi_1)\subseteq\mathbf{Im}(\Pi_2)$.
        \item 
        \begin{enumerate}
            \item $\Pi_2\Pi_1=\Pi_1$.
            \item $\Pi_1\Pi_2=\Pi_1$.
            \item $(\Pi_2-\Pi_1)^2=\Pi_2-\Pi_1$ is an orthogonal projection.
            \item $\Pi_1\Pi_2\Pi_1=\Pi_1$.
        \end{enumerate}
        \item 
        \begin{enumerate}
            \item $\Pi_1\leq\Pi_2$ as Hermitian operators.
            \item $\norm{\Pi_1\ket{\psi}}\leq\norm{\Pi_2\ket{\psi}}$ for all states $\ket{\psi}$.
        \end{enumerate}
    \end{enumerate}
\end{lemma}
\begin{proof}
We begin with the equivalence of 3(a) and 3(b) which follows from a direct verification:
\begin{equation}
\begin{aligned}
    \norm{\Pi_1\ket{\psi}}\leq\norm{\Pi_2\ket{\psi}}
    \quad&\Leftrightarrow\quad
    \norm{\Pi_1\ket{\psi}}^2\leq\norm{\Pi_2\ket{\psi}}^2
    \quad&&\Leftrightarrow\quad
    \bra{\psi}\Pi_1^\dagger\Pi_1\ket{\psi}\leq\bra{\psi}\Pi_2^\dagger\Pi_2\ket{\psi}\\
    &\Leftrightarrow\quad\bra{\psi}\Pi_1\ket{\psi}\leq\bra{\psi}\Pi_2\ket{\psi}
    &&\Leftrightarrow\quad\Pi_1\leq\Pi_2.
\end{aligned}
\end{equation}
The equivalence of 2(a) and 2(b) is also immediate:
\begin{equation}
    \Pi_1=\Pi_2\Pi_1
    \quad\Leftrightarrow\quad
    \Pi_1^\dagger=\left(\Pi_2\Pi_1\right)^\dagger=\Pi_1^\dagger\Pi_2^\dagger
    \quad\Leftrightarrow\quad
    \Pi_1=\Pi_1\Pi_2.
\end{equation}
Assuming 2(b) (and hence 2(a)) holds, 2(c) follows from the calculation
\begin{equation}
    (\Pi_2-\Pi_1)^2
    =\Pi_2^2+\Pi_1^2-\Pi_1\Pi_2-\Pi_2\Pi_1
    =\Pi_2+\Pi_1-\Pi_1-\Pi_1=\Pi_2-\Pi_1.
\end{equation}
Now if 2(c) is true, we have
\begin{equation}
\begin{aligned}
    &\Pi_2+\Pi_1-\Pi_1\Pi_2-\Pi_2\Pi_1
    =(\Pi_2-\Pi_1)^2=\Pi_2-\Pi_1\\
    \quad\Leftrightarrow\quad&
    2\Pi_1=\Pi_1\Pi_2+\Pi_2\Pi_1
    \quad\Rightarrow\quad
    \Pi_1=\Pi_1\Pi_2\Pi_1,
\end{aligned}
\end{equation}
where in the last step we perform $\Pi_1$ on both sides of the operators. This proves 2(d).

The implication 2(d) $\Rightarrow$ 2(a) can be established by a trace argument. Note that 
\begin{equation}
    \mathbf{Tr}\left((I-\Pi_2)\Pi_1\right)
    =\mathbf{Tr}\left(\Pi_1-\Pi_2\Pi_1\right)
    =\mathbf{Tr}\left(\Pi_1-\Pi_1\Pi_2\Pi_1\right)
    =0,
\end{equation}
where the second equality follows from cyclic property of the trace function. Since $I-\Pi_2$ and $\Pi_1$ are both positive semidefinite, this necessarily means that $(I-\Pi_2)\Pi_1=0$ which is equivalent to 2(a).

Suppose that $\mathbf{Im}(\Pi_1)\subseteq\mathbf{Im}(\Pi_2)$. For any state $\ket{\psi}$, we have $\Pi_1\ket{\psi}\in\mathbf{Im}(\Pi_1)\subseteq\mathbf{Im}(\Pi_2)$, which gives $\Pi_2\Pi_1\ket{\psi}=\Pi_1\ket{\psi}$ and thus $\Pi_1=\Pi_2\Pi_1$. This shows 1 $\Rightarrow$ 2(a). The implication 2(c) $\Rightarrow$ 3(a) is trivial. Finally, assume that 3(b) is true. For any $\ket{\psi}\in\mathbf{Ker}(\Pi_2)$, we have $0\leq\norm{\Pi_1\ket{\psi}}\leq\norm{\Pi_2\ket{\psi}}=0$, which forces $\norm{\Pi_1\ket{\psi}}=0$ and hence $\Pi_1\ket{\psi}=0$. This means $\mathbf{Ker}(\Pi_2)\subseteq\mathbf{Ker}(\Pi_1)$, which is equivalent to Claim 1 after taking the orthogonal complement.
\end{proof}
The above list can be expanded to include other equivalent characterizations, but this is not needed for our work.
See~\cite[Section 2.5]{kadison1983fundamentals} and~\cite[Section 9.6]{kreyszig1991introductory} for further discussions about the partial ordering of orthogonal projections.

\begin{proposition}
Let $\{\Pi_j\}_{j=0}^m$ be orthogonal projections partially ordered as  $0=\Pi_0\leq\Pi_1\leq\cdots\leq\Pi_m=I$. Then, we have $I=\overline{\Pi_0}\geq\overline{\Pi_1}\geq\cdots\geq\overline{\Pi_m}=0$ and
\begin{equation}
    \Pi_j=\Pi_k\Pi_j=\Pi_j\Pi_k,\qquad 
    \overline{\Pi_k}=\overline{\Pi_j}\cdot\overline{\Pi_k}=\overline{\Pi_k}\cdot\overline{\Pi_j},\qquad
    0\leq j\leq k\leq m,
\end{equation}
whereas $\Pi_k-\Pi_j$ are all orthogonal projections.
\end{proposition}

%%%%%%%%%%%%%%%%%%%%%%%%%%%%%%%%%%%%%%%%%%%%%%%%%%%%%%%%%%%%%%%%%%%%%%%%%%%%%%
\subsection{Flag projections}
\label{append:axiom_flag}

\begin{proposition}
Let $\{\Pi_j\}_{j=0}^m$ be orthogonal projections partially ordered as  $0=\Pi_0\leq\Pi_1\leq\cdots\leq\Pi_m=I$. Let $\Pi_b$ be an orthogonal projection commuting with all $\Pi_j$: $\Pi_b\Pi_j=\Pi_j\Pi_b$ for $j=0,\ldots, m$. Then, the following two resolutions of identity hold
\begin{equation}
    I=\sum_{j=1}^{m}\left(\Pi_j-\Pi_{j-1}\right)
    =\Pi_b+\overline{\Pi_b},
\end{equation}
where all $\{\Pi_j-\Pi_{j-1}\}_{j=1}^m$ and $\{\Pi_b,\overline{\Pi_b}\}$ are orthogonal projections and pairwise commute.
Moreover, we have
\begin{equation}
    0=\Pi_0\Pi_b\leq\Pi_1\Pi_b\leq\cdots\leq\Pi_m\Pi_b=\Pi_b,\qquad
    I=\overline{\Pi_0\Pi_b}\geq\overline{\Pi_1\Pi_b}\geq\cdots\geq\overline{\Pi_m\Pi_b}=\overline{\Pi_b},
\end{equation}
with
\begin{equation}
    \Pi_j\Pi_b=\Pi_k\Pi_b\cdot\Pi_j\Pi_b=\Pi_j\Pi_b\cdot\Pi_k\Pi_b,\qquad
    \overline{\Pi_k\Pi_b}=\overline{\Pi_j\Pi_b}\cdot\overline{\Pi_k\Pi_b}=\overline{\Pi_k\Pi_b}\cdot\overline{\Pi_j\Pi_b},\qquad
    0\leq j\leq k\leq m.
\end{equation}
\end{proposition}

The above proposition shows that the two sets of orthogonal projections $\{\Pi_j-\Pi_{j-1}\}_{j=1}^m$ and $\{\Pi_b,I-\Pi_b\}$ are complete and pairwise commute, so they can be simultaneously measured in a quantum computation. We tabulate the meaning of different outcomes from such a simultaneous measurement in \tab{vtaa_proj}.
Note that logical operations correspond directly to arithmetics of the orthogonal projections. For instance, the statement ``the algorithm does not fail up to stage $j$'' means ``the algorithm can potentially succeed at stage $j$'' and is logically equivalent to ``the algorithm succeeds before or at stage $j$, or it is still running'', which can be represented in terms of operators as
\begin{equation}
    \overline{\Pi_j\Pi_b}
    =I-\Pi_j\Pi_b
    =\Pi_j(I-\Pi_b)+(I-\Pi_j)
    =\Pi_j\overline{\Pi_b}+\overline{\Pi_j}.
\end{equation}

\begin{table}[t]
    \centering
    \begin{tabular}{c|c}
        Projection & Meaning of the image space \\
        \hline
        $\Pi_j$ & Algorithm stops before or at stage $j$\\
        $\Pi_j\overline{\Pi_b}$ & Algorithm succeeds before or at stage $j$\\
        $\Pi_j\Pi_b$ & Algorithm fails before or at stage $j$\\
        $\Pi_j-\Pi_{j-1}$ & Algorithm stops exactly at stage $j$\\
        $\left(\Pi_j-\Pi_{j-1}\right)\overline{\Pi_b}$ & Algorithm succeeds exactly at stage $j$\\
        $\left(\Pi_j-\Pi_{j-1}\right)\Pi_b$ & Algorithm fails exactly at stage $j$\\
        $\overline{\Pi_j}$ & Algorithm is still running at stage $j$\\
        $\overline{\Pi_j\overline{\Pi_b}}$ & Algorithm does not succeed up to stage $j$\\
        $\overline{\Pi_j\Pi_b}$ & Algorithm does not fail up to stage $j$\\
    \end{tabular}
    \caption{Orthogonal projections and meaning of their image spaces in the description of variable time algorithms.}
    \label{tab:vtaa_proj}
\end{table}

%%%%%%%%%%%%%%%%%%%%%%%%%%%%%%%%%%%%%%%%%%%%%%%%%%%%%%%%%%%%%%%%%%%%%%%%%%%%%%
\subsection{Input algorithms}
\label{append:axiom_input}

\begin{lemma}[Controlled unitaries]
\label{lem:control_unitary}
    Let $U^\dagger U=UU^\dagger=I$ be a unitary operator and $\Pi^2=\Pi=\Pi^\dagger$ be an orthogonal projection. The following statements are all equivalent:
    \begin{enumerate}
        \item 
        \begin{enumerate}
            \item $U\Pi=\Pi$.
            \item $\Pi U=\Pi$.
            \item $\Pi U\Pi=\Pi$.
        \end{enumerate}
        \item $U=\Pi+(I-\Pi)U(I-\Pi)$.
    \end{enumerate}
\end{lemma}
\begin{proof}
We start with the equivalence 1(a) $\Leftrightarrow$ 1(b) which follows from a direct verification
\begin{equation}
    U\Pi=\Pi
    \quad\Leftrightarrow\quad
    \Pi U^\dagger=\Pi
    \quad\Leftrightarrow\quad
    \Pi=\Pi U.
\end{equation}
Now assuming either 1(a) or 1(b), proving 1(c) is trivial.

The implication 1(c) $\Rightarrow$ 1(a) can be established by a trace argument. Note that
\begin{equation}
\begin{aligned}
    \mathbf{Tr}\left(\Pi(U-I)^\dagger(U-I)\Pi\right)
    &=\mathbf{Tr}\left(\Pi\left(2I-U-U^\dagger\right)\Pi\right)
    =\mathbf{Tr}\left(2\Pi-\Pi U\Pi-\left(\Pi U\Pi\right)^\dagger\right)\\
    &=\mathbf{Tr}\left(2\Pi-\Pi-\Pi\right)
    =0.
\end{aligned}
\end{equation}
This means that $(U-I)\Pi=0$, which is equivalent to 1(a).

Finally, if 1(a) (hence also 1(b) and 1(c)) is true,
\begin{equation}
\begin{aligned}
    U&=\left(\Pi+(I-\Pi)\right)U\left(\Pi+(I-\Pi)\right)\\
    &=\Pi U\Pi+\Pi U(I-\Pi)+(I-\Pi)U\Pi+(I-\Pi)U(I-\Pi)\\
    &=\Pi+\Pi(I-\Pi)+(I-\Pi)\Pi+(I-\Pi)U\Pi+(I-\Pi)U(I-\Pi)\\
    &=\Pi+(I-\Pi)U(I-\Pi),
\end{aligned}
\end{equation}
which proves Claim 2. The reverse direction follows from a direct calculation.
\end{proof}

\begin{proposition}
Let $\{\Pi_j\}_{j=0}^m$ be orthogonal projections partially ordered as  $0=\Pi_0\leq\Pi_1\leq\cdots\leq\Pi_m=I$.
Let $\Pi_b$ be an orthogonal projection commuting with all $\Pi_j$: $\Pi_b\Pi_j=\Pi_j\Pi_b$ for $j=0,\ldots, m$.
Let $A_j$ be unitaries such that $A_j\Pi_{j-1}=\Pi_{j-1}$ for all $j=1,\ldots, m$, and $A_0=I$. Then,
\begin{equation}
\begin{aligned}
    A_j\Pi_l&=\Pi_l=\Pi_lA_j,\qquad &&A_j\Pi_l\Pi_b=\Pi_l\Pi_bA_j,\\
    A_j\overline{\Pi_l}&=\overline{\Pi_l}=\overline{\Pi_l}A_j,
    &&A_j\overline{\Pi_l\Pi_b}=\overline{\Pi_l\Pi_b}A_j,\qquad
    0\leq l<j\leq m.
\end{aligned}
\end{equation}
Moreover, for any quantum state $\ket{\psi}$,
\begin{equation}
    \norm{\overline{\Pi_{j}\Pi_b}\ket{\psi}}=\norm{\overline{\Pi_{j}\Pi_b}A_{k}\cdots A_{j+1}\ket{\psi}},\qquad
    0\leq j\leq k\leq m,
\end{equation}
and for any state $\ket{\psi_0}$,
\begin{equation}
    1=\norm{\overline{\Pi_{0}\Pi_b}\ket{\psi_0}}\geq\norm{\overline{\Pi_{1}\Pi_b}A_1\ket{\psi_0}}\geq\cdots\geq\norm{\overline{\Pi_{m}\Pi_b}A_{m}\cdots A_1\ket{\psi_0}}=\norm{\overline{\Pi_b}A_{m}\cdots A_1\ket{\psi_0}}.
\end{equation}
\end{proposition}
\begin{proof}
We will only prove $A_j\Pi_l=\Pi_l=\Pi_lA_j$ and $A_j\overline{\Pi_l\Pi_b}=\overline{\Pi_l\Pi_b}A_j$ as verifications of other operator identities proceed in a similar way. They follow from
\begin{equation}
    A_j\Pi_l=A_j\Pi_{j-1}\Pi_l=\Pi_{j-1}\Pi_l
    =\Pi_l=\Pi_l\Pi_{j-1}=\Pi_l\Pi_{j-1}A_j=\Pi_lA_j
\end{equation}
and 
\begin{equation}
    A_j\overline{\Pi_l\Pi_b}=A_j(1-\Pi_l\Pi_b)
    =A_j-\Pi_l\Pi_b=A_j-\Pi_b\Pi_l
    =(1-\Pi_b\Pi_l)A_j=\overline{\Pi_l\Pi_b}A_j.
\end{equation}
The norm identity then follows directly as
\begin{equation}
    \norm{\overline{\Pi_{j}\Pi_b}\ket{\psi}}
    =\norm{A_{k}\cdots A_{j+1}\overline{\Pi_{j}\Pi_b}\ket{\psi}}
    =\norm{\overline{\Pi_{j}\Pi_b}A_{k}\cdots A_{j+1}\ket{\psi}}.
\end{equation}
Using the above identity, the claim about monotonicity is established:
    \begin{equation}
    \begin{aligned}
        \norm{\overline{\Pi_{j}\Pi_b}A_{j}\cdots A_1\ket{\psi_0}}
        &=\norm{\overline{\Pi_{j}\Pi_b}A_{m}\cdots A_1\ket{\psi_0}}\\
        &\geq\norm{\overline{\Pi_{j+1}\Pi_b}A_{m}\cdots A_1\ket{\psi_0}}
        =\norm{\overline{\Pi_{j+1}\Pi_b}A_{j+1}\cdots A_1\ket{\psi_0}}.
    \end{aligned}
    \end{equation}
\end{proof}

%%%%%%%%%%%%%%%%%%%%%%%%%%%%%%%%%%%%%%%%%%%%%%%%%%%%%%%%%%%%%%%%%%%%%%%%%%%%%%
\subsection{Amplified algorithms}
\label{append:axiom_amp}
\begin{proposition}
Let $\{\Pi_j\}_{j=0}^m$ be orthogonal projections partially ordered as  $0=\Pi_0\leq\Pi_1\leq\cdots\leq\Pi_m=I$.
Let $\Pi_b$ be an orthogonal projection commuting with all $\Pi_j$: $\Pi_b\Pi_j=\Pi_j\Pi_b$ for $j=0,\ldots, m$.
Let $A_j$ be unitaries such that $A_j\Pi_{j-1}=\Pi_{j-1}$ for all $j=1,\ldots, m$, and $A_0=I$.
Finally, let $\widetilde A_j$ be unitaries such that $\frac{\overline{\Pi_j\Pi_b}\widetilde{A}_j\ket{\psi_0}}{\norm{\overline{\Pi_j\Pi_b}\widetilde{A}_j\ket{\psi_0}}}
    =\frac{\overline{\Pi_j\Pi_b}A_j\widetilde{A}_{j-1}\ket{\psi_0}}{\norm{\overline{\Pi_j\Pi_b}A_j\widetilde{A}_{j-1}\ket{\psi_0}}}$ for all $j=1,\ldots, m$, and $\widetilde A_0=I$.
Then, for any quantum state $\ket{\psi_0}$,
\begin{equation}
    \frac{\overline{\Pi_{k}\Pi_b}A_k\cdots A_{j+1}\widetilde{A}_{j}\ket{\psi_0}}{\norm{\overline{\Pi_{k}\Pi_b}A_k\cdots A_{j+1}\widetilde{A}_{j}\ket{\psi_0}}}
        =\frac{\overline{\Pi_{k}\Pi_b}A_k\cdots A_{l+1}\widetilde{A}_l\ket{\psi_0}}{\norm{\overline{\Pi_{k}\Pi_b}A_k\cdots A_{l+1}\widetilde{A}_l\ket{\psi_0}}},\qquad 0\leq l,j\leq k\leq m,
\end{equation}
and
\begin{equation}
    \frac{\norm{\overline{\Pi_{h}\Pi_b}A_{h}\cdots A_{j+1}\widetilde{A}_{j}\ket{\psi_0}}}{\norm{\overline{\Pi_{k}\Pi_b}A_{k}\cdots A_{j+1}\widetilde{A}_{j}\ket{\psi_0}}}
    =\frac{\norm{\overline{\Pi_{h}\Pi_b}A_{h}\cdots A_{l+1}\widetilde{A}_{l}\ket{\psi_0}}}{\norm{\overline{\Pi_{k}\Pi_b}A_{k}\cdots A_{l+1}\widetilde{A}_{l}\ket{\psi_0}}},\qquad
    0\leq l,j\leq k,h\leq m.
\end{equation}
\end{proposition}
\begin{proof}
    By exchanging the two sides of both equations and taking reciprocal of the second one, we may assume without loss of generality that $0\leq l\leq j\leq k\leq h\leq m$.
    Recall that the defining property of $\widetilde{A}_j$ can be succinctly represented as $\overline{\Pi_j\Pi_b}\widetilde{A}_j\ket{\psi_0}\propto\overline{\Pi_j\Pi_b}A_j\widetilde{A}_{j-1}\ket{\psi_0}$.
    We then have
    \begin{equation}
    \begin{aligned}
        \overline{\Pi_{j}\Pi_b}\widetilde{A}_{j}\ket{\psi_0}
        &\propto\overline{\Pi_j\Pi_b}A_j\widetilde{A}_{j-1}\ket{\psi_0}
        =\overline{\Pi_j\Pi_b}\cdot\overline{\Pi_{j-1}\Pi_b}A_j\widetilde{A}_{j-1}\ket{\psi_0}
        =\overline{\Pi_j\Pi_b}A_j\overline{\Pi_{j-1}\Pi_b}\widetilde{A}_{j-1}\ket{\psi_0}\\
        &\propto\overline{\Pi_j\Pi_b}A_j\overline{\Pi_{j-1}\Pi_b}A_{j-1}\widetilde{A}_{j-2}\ket{\psi_0}
        \propto\cdots\\
        &\propto\overline{\Pi_{j}\Pi_b}A_j\cdots \overline{\Pi_{l+1}\Pi_b}A_{l+1}
        \overline{\Pi_{l}\Pi_b}\widetilde{A}_l\ket{\psi_0}
        =\overline{\Pi_{j}\Pi_b}\cdots\overline{\Pi_{l+1}\Pi_b}A_j\cdots A_{l+1}\widetilde{A}_l\ket{\psi_0}\\
        &=\overline{\Pi_{j}\Pi_b}A_j\cdots A_{l+1}\widetilde{A}_l\ket{\psi_0},
    \end{aligned}
    \end{equation}
    which proves the first claim through
    \begin{equation}
    \begin{aligned}
        \overline{\Pi_{k}\Pi_b}A_k\cdots A_{j+1}\widetilde{A}_{j}\ket{\psi_0}
        &=\overline{\Pi_{k}\Pi_b}A_k\cdots A_{j+1}\overline{\Pi_{j}\Pi_b}\widetilde{A}_{j}\ket{\psi_0}\\
        &\propto\overline{\Pi_{k}\Pi_b}A_k\cdots A_{j+1}\overline{\Pi_{j}\Pi_b}A_j\cdots A_{l+1}\widetilde{A}_l\ket{\psi_0}
        =\overline{\Pi_{k}\Pi_b}A_k\cdots A_{l+1}\widetilde{A}_l\ket{\psi_0}.
    \end{aligned}
    \end{equation}
As for the second claim, we use the first one to rewrite
\begin{equation}
\begin{aligned}
    \frac{\overline{\Pi_{h}\Pi_b}A_{h}\cdots A_{j+1}\widetilde{A}_{j}\ket{\psi_0}}{\norm{\overline{\Pi_{k}\Pi_b}A_{k}\cdots A_{j+1}\widetilde{A}_{j}\ket{\psi_0}}}
    &=\overline{\Pi_{h}\Pi_b}A_{h}\cdots A_{k+1}\frac{\overline{\Pi_{k}\Pi_b}A_{k}\cdots A_{j+1}\widetilde{A}_{j}\ket{\psi_0}}{\norm{\overline{\Pi_{k}\Pi_b}A_{k}\cdots A_{j+1}\widetilde{A}_{j}\ket{\psi_0}}}\\
    &=\overline{\Pi_{h}\Pi_b}A_{h}\cdots A_{k+1}\frac{\overline{\Pi_{k}\Pi_b}A_k\cdots A_{l+1}\widetilde{A}_l\ket{\psi_0}}{\norm{\overline{\Pi_{k}\Pi_b}A_k\cdots A_{l+1}\widetilde{A}_l\ket{\psi_0}}}\\
    &=\frac{\overline{\Pi_{h}\Pi_b}A_h\cdots A_{l+1}\widetilde{A}_l\ket{\psi_0}}{\norm{\overline{\Pi_{k}\Pi_b}A_k\cdots A_{l+1}\widetilde{A}_l\ket{\psi_0}}}.
\end{aligned}
\end{equation}
The claim is then justified by taking the norm on both sides of the equation.
\end{proof}

%% file: dirichlet.tex
In this appendix, we prove the following bounds useful for analyzing the loss factor of amplitude amplification in \sec{tunable_universal}. Up to a change of variables, they are essentially bounds on the Dirichlet kernel that arises in the Fourier analysis.

\begin{proposition}[Tight bounds on the Dirichlet kernel]
\label{prop:dirichlet}
For any integer $\rho\geq3$ and real number $\theta\geq0$ such that $0\leq\rho\theta\leq\frac{\pi}{2}$, we have
    \begin{equation}
        1-\frac{1}{6}\rho^2\sin^2(\theta)
    \leq\frac{\sin(\rho\theta)}{\rho\sin(\theta)}
    \leq1-\frac{4\pi-8}{\pi^3}\rho^2\sin^2(\theta).
    \end{equation}
The constant factors are tight in the sense that
    \begin{equation}
    \begin{aligned}
        \frac{1}{6}&=\inf\left\{c>0\ \bigg|\ 1-c\rho^2\sin^2(\theta)
    \leq\frac{\sin(\rho\theta)}{\rho\sin(\theta)},\ 0\leq\rho\theta\leq\frac{\pi}{2}\right\},\\
    \frac{4\pi-8}{\pi^3}&=\sup\left\{c>0\ \bigg|\ \frac{\sin(\rho\theta)}{\rho\sin(\theta)}\leq1-c\rho^2\sin^2(\theta)
    ,\ 0\leq\rho\theta\leq\frac{\pi}{2}\right\}.\\
    \end{aligned}
    \end{equation}
In particular, when $\rho=3$,
\begin{equation}
    \frac{\sin(3\theta)}{3\sin(\theta)}
    =1-\frac{4}{3}\sin^2(\theta).
\end{equation}
\end{proposition}

%%%%%%%%%%%%%%%%%%%%%%%%%%%%%%%%%%%%%%%%%%%%%%%%%%%%%%%%%%%%%%%%%%%%%%%%%%%%%%
\subsection{Proof of lower bound}
Let us consider the lower bound first. We start by reorganizing it as
    \begin{equation}
        \left(1-\frac{1}{6}\rho^2\sin^2(\theta)\right)
    \overset{?}{\leq}\frac{\sin(\rho\theta)}{\rho\sin(\theta)}
    \qquad\Leftrightarrow\qquad
    \frac{\rho\sin(\theta)-\sin(\rho\theta)}{\rho^3\sin^3(\theta)}\overset{?}{\leq}\frac{1}{6}.
    \end{equation}
    Note that 
    \begin{equation}
        \lim_{\theta\rightarrow0}\frac{\rho\sin(\theta)-\sin(\rho\theta)}{\rho^3\sin^3(\theta)}
        =\lim_{\theta\rightarrow0}\frac{-\frac{\rho}{6}\theta^3+\frac{\rho^3}{6}\theta^3}{\rho^3\theta^3}
        =\frac{1}{6}-\frac{1}{6\rho^2}\nearrow\frac{1}{6},
    \end{equation}
    which shows that the lower bound is tight.
    So it remains to prove that the function on the left hand side is monotonically decreasing over $0\leq\theta\leq\frac{\pi}{2\rho}$.

    To this end, we differentiate it with respect to $\theta$
    \begin{equation}
        \left(\frac{\rho\sin(\theta)-\sin(\rho\theta)}{\rho^3\sin^3(\theta)}\right)'
        \overset{?}{\leq}0,
    \end{equation}
    which simplifies to
    \begin{equation}
        3\sin(\rho\theta)\cos(\theta)
        \overset{?}{\leq}2\rho\sin(\theta)\cos(\theta)
        +\rho\sin(\theta)\cos(\rho\theta).
    \end{equation}
    Using the trigonometric inequalities
    \begin{equation}
    \begin{aligned}
        \theta-\frac{\theta^3}{6}+\frac{\theta^5}{120}-\frac{\theta^7}{5040}
        &\leq\sin(\theta)
        \leq\theta-\frac{\theta^3}{6}+\frac{\theta^5}{120},\\
        1-\frac{\theta^2}{2}+\frac{\theta^4}{24}-\frac{\theta^6}{720}
        &\leq\cos(\theta)
        \leq1-\frac{\theta^2}{2}+\frac{\theta^4}{24},
    \end{aligned}
    \end{equation}
    justified via Taylor's theorem with Lagrange's remainder, we want
    \begin{equation}
    \begin{aligned}
        3\left(\rho\theta-\frac{\rho^3\theta^3}{6}+\frac{\rho^5\theta^5}{120}\right)\left(1-\frac{\theta^2}{2}+\frac{\theta^4}{24}\right)
        \overset{?}{\leq}&\
        2\rho\left(\theta-\frac{\theta^3}{6}+\frac{\theta^5}{120}-\frac{\theta^7}{5040}\right)\left(1-\frac{\theta^2}{2}+\frac{\theta^4}{24}-\frac{\theta^6}{720}\right)\\
        &\ +\rho\left(\theta-\frac{\theta^3}{6}+\frac{\theta^5}{120}-\frac{\theta^7}{5040}\right)\left(1-\frac{\rho^2\theta^2}{2}+\frac{\rho^4\theta^4}{24}-\frac{\rho^6\theta^6}{720}\right)
    \end{aligned}
    \end{equation}
    or equivalently
    \begin{equation}
    \begin{aligned}
        3\left(1-\frac{\rho^2\theta^2}{6}+\frac{\rho^4\theta^4}{120}\right)\left(1-\frac{\theta^2}{2}+\frac{\theta^4}{24}\right)
        \overset{?}{\leq}&\
        2\left(1-\frac{\theta^2}{6}+\frac{\theta^4}{120}-\frac{\theta^6}{5040}\right)\left(1-\frac{\theta^2}{2}+\frac{\theta^4}{24}-\frac{\theta^6}{720}\right)\\
        &\ +\left(1-\frac{\theta^2}{6}+\frac{\theta^4}{120}-\frac{\theta^6}{5040}\right)\left(1-\frac{\rho^2\theta^2}{2}+\frac{\rho^4\theta^4}{24}-\frac{\rho^6\theta^6}{720}\right).
    \end{aligned}
    \end{equation}

    The left hand side of the above equation expands to
    \begin{equation}
        3
        +\frac{-\rho^2-3}{2}\theta^2
        +\frac{\rho^4+10\rho^2+5}{40}\theta^4
        +\frac{-\rho^2(3\rho^2+5)}{240}\theta^6
        +\frac{\rho^4}{960}\theta^8,
    \end{equation}
    whereas the right hand side expands to
    \begin{equation}
    \begin{aligned}
        &3+\frac{-\rho^2-3}{2}\theta^2
        +\frac{5\rho^4+10\rho^2+33}{120}\theta^4\\
        &+\frac{-7\rho^6-35\rho^4-21\rho^2-129}{5040}\theta^6
        +\frac{14\rho^6+21\rho^4+6\rho^2+82}{60480}\theta^8\\
        &+\frac{-7\rho^6-5\rho^4-24}{604800}\theta^{10}
        +\frac{\rho^6+2}{3628800}\theta^{12}.
    \end{aligned}
    \end{equation}
    It thus suffices to show that
    \begin{equation}
    \label{eq:vtaa_taylor_ineq}
    \begin{aligned}
        \frac{\rho^4+10\rho^2+5}{40}\theta^4
        +\frac{\rho^4}{960}\theta^8
        \overset{?}{\leq}&\ \frac{5\rho^4+10\rho^2+33}{120}\theta^4\\
        &+\frac{-7\rho^6-35\rho^4-21\rho^2-129}{5040}\theta^6
        +\frac{-7\rho^6-5\rho^4-24}{604800}\theta^{10}.
    \end{aligned}
    \end{equation}
    
    Note that $\theta^2\leq\frac{\pi^2}{4\rho^2}<\frac{5}{2\rho^2}$.
    This upper bounds the left hand side of \eq{vtaa_taylor_ineq} by
    \begin{equation}
        \left(\frac{\rho^4}{40}+\frac{\rho^2}{4}+\frac{101}{768}\right)\theta^4.
    \end{equation}
    Meanwhile, using $\theta^2<\frac{5}{2\rho^2}\leq\frac{5}{32}$ for $\rho\geq 4$, we lower bound the right hand side of \eq{vtaa_taylor_ineq} by
    \begin{equation}
    \begin{aligned}
        &\frac{5\rho^4+10\rho^2+33}{120}\theta^4
        +\frac{-7\rho^6-35\rho^4-21\rho^2-129}{5040}\theta^6
        +\frac{-7\cdot\frac{25}{4}\rho^2-5\cdot\frac{25}{4}-24\cdot\frac{25}{1024}}{604800}\theta^{6}\\
        &\geq\frac{5\rho^4+10\rho^2+33}{120}\theta^4
        +\frac{-8\rho^6-36\rho^4-22\rho^2-130}{5040}\theta^6\\
        &\geq\frac{5\rho^4+10\rho^2+33}{120}\theta^4
        +\frac{-8\cdot\frac{5}{2}\rho^4-36\cdot\frac{5}{2}\rho^2-22\cdot\frac{5}{2}-130\cdot\frac{5}{32}}{5040}\theta^4\\
        &\geq\frac{190\rho^4+330\rho^2+1331-130\cdot\frac{5}{32}}{5040}\theta^4
        \geq\frac{190\rho^4+330\rho^2+1310}{5040}\theta^4.
    \end{aligned}
    \end{equation}
    We thus want to prove that
    \begin{equation}
        \frac{\rho^4}{40}+\frac{\rho^2}{4}+\frac{101}{768}
        \overset{?}{\leq}\frac{19\rho^4+33\rho^2+131}{504}
    \end{equation}
    which can be directly verified to hold for $\rho\geq4$.

%%%%%%%%%%%%%%%%%%%%%%%%%%%%%%%%%%%%%%%%%%%%%%%%%%%%%%%%%%%%%%%%%%%%%%%%%%%%%%
\subsection{Proof of upper bound}
Recall that in the proof of lower bound, we have actually shown that when $\rho\geq 4$,
$\frac{\rho\sin(\theta)-\sin(\rho\theta)}{\rho^3\sin^3(\theta)}$
monotonically decreases as a function of $\theta$ over $0\leq\theta\leq\frac{\pi}{2\rho}$. Hence,
\begin{equation}
    \frac{\rho\sin(\theta)-\sin(\rho\theta)}{\rho^3\sin^3(\theta)}
    \geq\frac{\rho\sin\left(\frac{\pi}{2\rho}\right)-1}{\rho^3\sin^3\left(\frac{\pi}{2\rho}\right)}.
\end{equation}
Denoting $x=\frac{1}{\rho}$, our goal is to lower bound
\begin{equation}
    \frac{\frac{\sin\left(\frac{\pi}{2}x\right)}{x}-1}{\frac{\sin^3\left(\frac{\pi}{2}x\right)}{x^3}}=\frac{x^2\sin\left(\frac{\pi}{2}x\right)-x^3}{\sin^3\left(\frac{\pi}{2}x\right)}
\end{equation}
for $0<x\leq\frac{1}{3}$.

Note that 
\begin{equation}
    \lim_{x\rightarrow0}\frac{x^2\sin\left(\frac{\pi}{2}x\right)-x^3}{\sin^3\left(\frac{\pi}{2}x\right)}
    =\frac{\frac{\pi}{2}-1}{\frac{\pi^3}{2^3}}
    =\frac{4\pi-8}{\pi^3},
\end{equation}
which shows that the claimed constant factor is tight. It remains to prove that the function on the left hand side is monotonically increasing over $0<x\leq\frac{1}{3}$.

To this end, we differentiate it with respect to $x$:
\begin{equation}
    \left(\frac{x^2\sin\left(\frac{\pi}{2}x\right)-x^3}{\sin^3\left(\frac{\pi}{2}x\right)}\right)'
    =\frac{x}{2}
    \left(\pi x\cot\left(\frac{\pi}{2}x\right)-2\right)
    \csc^2\left(\frac{\pi}{2}x\right)
    \left(3x\csc\left(\frac{\pi}{2}x\right)-2\right).
\end{equation}
Then, it suffices to show that
\begin{equation}
    \pi x\cot\left(\frac{\pi}{2}x\right)-2\overset{?}{\leq}0,\qquad
    3x\csc\left(\frac{\pi}{2}x\right)-2\overset{?}{\leq}0.
\end{equation}
The first inequality is equivalent to $\frac{\pi}{2}x\overset{?}{\leq}\tan\left(\frac{\pi}{2}x\right)$ and becomes trivial. The second inequality is equivalent to $\frac{3}{2}x\overset{?}{\leq}\sin\left(\frac{\pi}{2}x\right)$ which also holds trivially when $x\leq\frac{1}{3}$.

Altogether, we have shown that
\begin{equation}
    \frac{\rho\sin(\theta)-\sin(\rho\theta)}{\rho^3\sin^3(\theta)}
    \geq\frac{4\pi-8}{\pi^3}.
\end{equation}
This is equivalent to the claimed bound.

%% file: qubitization.tex
In this appendix, we review results on Hermitian qubitization~\cite{Low2016Qubitization} useful for constructing the branch marking and gapped phase estimation algorithm in \sec{dinv_branch}. We also give a self-contained exposition of amplitude amplification based on qubitization. We use $\mathcal{G}$ and $\mathcal{H}$ to represent finite-dimensional Hilbert spaces, on which all operators act.

%%%%%%%%%%%%%%%%%%%%%%%%%%%%%%%%%%%%%%%%%%%%%%%%%%%%%%%%%%%%%%%%%%%%%%%%%%%%%%
\subsection{\texorpdfstring{$U$}{U}-cyclic subspaces and qubitization}
\label{append:qubitization_subspace}
We begin by discussing how subspaces can be invariant under repeated applications of a unitary operator.
\begin{lemma}[$U$-cyclic subspaces]
\label{lem:subspace}
Let $U:\mathcal{H}\rightarrow\mathcal{H}$ be a unitary operator and $G:\mathcal{G}\rightarrow\mathcal{H}$ be an isometry, such that $G^\dagger UG$ is Hermitian. Let $\ket{\phi_u}$ be eigenvectors of $G^\dagger UG$ with corresponding eigenvalues $\lambda_u$, and define the cyclic subspaces
\begin{equation}
    \mathcal{H}_u=\mathbf{Span}\left\{\ldots,U^{\dagger2} G\ket{\phi_u},U^\dagger G\ket{\phi_u},G\ket{\phi_u},UG\ket{\phi_u},U^2G\ket{\phi_u},\ldots\right\}.
\end{equation}
Then:
\begin{enumerate}
    \item The following conditions are equivalent:
    \begin{enumerate}
        \item $U^2G=G$;
        \item $G^\dagger U^2G=I$.
    \end{enumerate}
    When any of the conditions is satisfied, all $\mathcal{H}_u=\mathbf{Span}\left\{G\ket{\phi_u},UG\ket{\phi_u}\right\}$ have dimensions $\dim(\mathcal{H}_u)=1,2$ and are invariant under $U$, $U^\dagger$ and $GG^\dagger$.
    \item The following conditions are equivalent for 1D subspaces:
    \begin{enumerate}
        \item $\dim(\mathcal{H}_u)=1$;
        \item $\left\{G\ket{\phi_u}\right\}$ is a basis for $\mathcal{H}_u$;
        \item $\lambda_u=\pm1$;
    \end{enumerate}
    in which case we have the matrix representation
    \begin{equation}
        U=\begin{bmatrix}
            \lambda_u
        \end{bmatrix},\qquad
        GG^\dagger=\begin{bmatrix}
            1
        \end{bmatrix}.
    \end{equation}
    \item The following conditions are equivalent  for 2D subspaces:
    \begin{enumerate}
        \item $\dim(\mathcal{H}_u)=2$;
        \item $\left\{G\ket{\phi_u},UG\ket{\phi_u}\right\}$ is a basis for $\mathcal{H}_u$;
        \item $-1<\lambda_u<1$;
    \end{enumerate}
    in which case we have the matrix representation
    \begin{equation}
        U=\begin{bmatrix}
            0 & 1\\
            1 & 0
        \end{bmatrix},\qquad
        GG^\dagger=\begin{bmatrix}
            1 & \lambda_u\\
            0 & 0
        \end{bmatrix}.
    \end{equation}
\end{enumerate}
\end{lemma}
\begin{proof}
Applying \lem{control_unitary} to unitary $U^2$ and orthogonal projection $GG^\dagger$, we have
\begin{equation}
    U^2G=G
    \quad\Leftrightarrow\quad
    U^2GG^\dagger=GG^\dagger
    \quad\Leftrightarrow\quad
    GG^\dagger U^2 GG^\dagger=GG^\dagger
    \quad\Leftrightarrow\quad
    G^\dagger U^2G=I.
\end{equation}
One concludes from $\norm{G^\dagger UG}\leq\norm{G^\dagger}\norm{U}\norm{G}=1$ that the real eigenvalues of $G^\dagger UG$ must satisfy $-1\leq\lambda_u\leq 1$. As a cyclic subspace, $\mathcal{H}_u$ is clearly invariant under $U$ and $U^\dagger$. Furthermore, due to the equivalence $G^\dagger U^2G=I\Leftrightarrow U^2G=G\Leftrightarrow U^{\dagger2}G=G$, all $U^lG\ket{\phi_u}$ can be reduced to either $G\ket{\phi_u}$ or $UG\ket{\phi_u}$, implying $\mathcal{H}_u=\mathbf{Span}\left\{G\ket{\phi_u},UG\ket{\phi_u}\right\}$ and $\dim(\mathcal{H}_u)=1,2$.

Note that $G\ket{\phi_u}$ is a vector of unit length, so it must be a basis of $\mathcal{H}_u$ if $\dim(\mathcal{H}_u)=1$. When this happens, $UG\ket{\phi_u}=\mu_uG\ket{\phi_u}$ for some complex number $\mu_u$. But this means $\lambda_uG\ket{\phi_u}=GG^\dagger UG\ket{\phi_u}=\mu_uGG^\dagger G\ket{\phi_u}=\mu_uG\ket{\phi_u}$ and hence $\abs{\lambda_u}=\norm{\lambda_uG\ket{\phi_u}}=\norm{\mu_uG\ket{\phi_u}}=\norm{UG\ket{\phi_u}}=1$, giving $\lambda_u=\pm 1$. Now assuming $\lambda_u=\pm 1$, we consider the orthogonal decomposition $UG\ket{\phi_u}=\left(GG^\dagger\right) UG\ket{\phi_u}+\left(I-GG^\dagger\right)UG\ket{\phi_u}=\lambda_uG\ket{\phi_u}+\left(I-GG^\dagger\right)UG\ket{\phi_u}$. By the  Pythagorean theorem, $\norm{\left(I-GG^\dagger\right)UG\ket{\phi_u}}=0$, which implies that $\left(I-GG^\dagger\right)UG\ket{\phi_u}=0$ and $UG\ket{\phi_u}=\lambda_uG\ket{\phi_u}$, establishing the equivalence of three conditions. The matrix representation of $U$ and $GG^\dagger$ then follows from a direct calculation.

When $\dim(\mathcal{H}_u)=2$, the span set $\left\{G\ket{\phi_u},UG\ket{\phi_u}\right\}$ is naturally a basis for $\mathcal{H}_u$. In this case, we know that $-1\leq\lambda_u\leq 1$ and $\lambda_u\neq\pm1$ simultaneously hold, so it must be that $-1<\lambda_u<1$. Similarly, assuming $-1<\lambda_u<1$, then it simultaneously holds that $\dim(\mathcal{H}_u)=1,2$ and $\dim(\mathcal{H}_u)\neq1$, implying $\dim(\mathcal{H}_u)=2$. The matrix representation of $U$ and $GG^\dagger$ follows again from a direct calculation.
\end{proof}

Now, we perform the orthogonal decompositions
\begin{equation}
\begin{aligned}
    \mathcal{H}&=\mathbf{Im}\left(GG^\dagger+UGG^\dagger U^\dagger\right)\obot\mathbf{Ker}\left(GG^\dagger+UGG^\dagger U^\dagger\right)\\
    &=\left(\mathbf{Im}\left(GG^\dagger\right)+\mathbf{Im}\left(UGG^\dagger U^\dagger\right)\right)
    \obot\left(\mathbf{Ker}\left(GG^\dagger\right)\cap\mathbf{Ker}\left(UGG^\dagger U^\dagger\right)\right),
\end{aligned}
\end{equation}
which are equivalent since for any two positive semidefinite operators $P$, $Q$ and state $\ket{\phi}$,
\begin{equation}
    (P+Q)\ket{\phi}=0
    \quad\Leftrightarrow\quad
    \bra{\phi}(P+Q)\ket{\phi}=0
    \quad\Leftrightarrow\quad
    \bra{\phi}P\ket{\phi}=\bra{\phi}Q\ket{\phi}=0
    \quad\Leftrightarrow\quad
    P\ket{\phi}=Q\ket{\phi}=0.
\end{equation}
For the purpose of qubitization, we further decompose the first term into
\begin{equation}
\begin{aligned}
    \mathbf{Im}\left(GG^\dagger\right)&=\mathbf{Im}(G)=G\bigobot_{u}\mathbf{Span}\left\{\ket{\phi_u}\right\}=\bigobot_{u}\mathbf{Span}\left\{G\ket{\phi_u}\right\},\\
    \mathbf{Im}\left(UGG^\dagger U^\dagger\right)
    &=\mathbf{Im}(UG)=UG\bigobot_{u}\mathbf{Span}\left\{\ket{\phi_u}\right\}=\bigobot_{u}\mathbf{Span}\left\{UG\ket{\phi_u}\right\}.
\end{aligned}
\end{equation}
This gives:

\begin{proposition}[Hermitian qubitization]
\label{prop:qubitization}
Let $U:\mathcal{H}\rightarrow\mathcal{H}$ be a unitary operator and $G:\mathcal{G}\rightarrow\mathcal{H}$ be an isometry, such that $G^\dagger UG$ is Hermitian and $G^\dagger U^2G=I$. Let $\ket{\phi_u}$ be eigenvectors of $G^\dagger UG$ with corresponding eigenvalues $\lambda_u$, and define the cyclic subspaces
\begin{equation}
    \mathcal{H}_u=\mathbf{Span}\left\{\ldots,U^{\dagger2} G\ket{\phi_u},U^\dagger G\ket{\phi_u},G\ket{\phi_u},UG\ket{\phi_u},U^2G\ket{\phi_u},\ldots\right\}.
\end{equation}
Then:
\begin{enumerate}
    \item $\mathcal{H}$ admits the orthogonal decomposition
    \begin{equation}
        \mathcal{H}=\bigobot_{u}\mathcal{H}_u
        \obot\mathcal{H}_\bot,
    \end{equation}
    where $\mathcal{H}_u=\mathbf{Span}\left\{G\ket{\phi_u},UG\ket{\phi_u}\right\}$ sum up to 
    $
    \bigobot_{u}\mathcal{H}_u=\mathbf{Im}\left(GG^\dagger\right)+\mathbf{Im}\left(UGG^\dagger U^\dagger\right)$, and $\mathcal{H}_\bot=\ker\left(GG^\dagger\right)\cap\ker\left(UGG^\dagger U^\dagger\right)$.
    All subspaces are invariant under $U$, $U^\dagger$ and $GG^\dagger$.
    \item When $\lambda_u=\pm1$, $\left\{G\ket{\phi_u}\right\}$ is an orthonormal basis for $\mathcal{H}_u$, under which
    \begin{equation}
        U=\begin{bmatrix}
            \lambda_u
        \end{bmatrix},\qquad
        GG^\dagger=\begin{bmatrix}
            1
        \end{bmatrix}.
    \end{equation}
    \item When $-1<\lambda_u<1$, $\left\{G\ket{\phi_u},\frac{UG\ket{\phi_u}-\lambda_uG\ket{\phi_u}}{\sqrt{1-\lambda_u^2}}\right\}$ is an orthonormal basis for $\mathcal{H}_u$, under which
    \begin{equation}
        U=\begin{bmatrix}
            \lambda_u & \sqrt{1-\lambda_u^2}\\
            \sqrt{1-\lambda_u^2} & -\lambda_u
        \end{bmatrix},\qquad
        GG^\dagger=
        \begin{bmatrix}
            1 & 0\\
            0 & 0
        \end{bmatrix}.
    \end{equation}
    \item Restricted to $\mathcal{H}_\bot$, $U$ is still a unitary and $GG^\dagger=0$.
\end{enumerate}
\end{proposition}
\begin{proof}
We start by checking the pairwise orthogonality of $\mathcal{H}_u$ and $\mathcal{H}_v$ for $u\neq v$:
\begin{equation}
\begin{aligned}
    \bra{\phi_v}G^\dagger G\ket{\phi_u}=0,\qquad
    \bra{\phi_v}G^\dagger UG\ket{\phi_u}=\lambda_u\langle\phi_v|\phi_u\rangle=0,\\
    \bra{\phi_v}G^\dagger U^\dagger G\ket{\phi_u}=\lambda_v\langle\phi_v|\phi_u\rangle=0,\qquad
    \bra{\phi_v}G^\dagger U^\dagger UG\ket{\phi_u}=\langle\phi_v|\phi_u\rangle=0.
\end{aligned}
\end{equation}
Combining with the analysis proceeding this theorem, we have established the claimed orthogonal decomposition. Since $\bigobot_{u}\mathcal{H}_u$ is invariant under the normal operators $U$, $U^\dagger$ and $GG^\dagger$, its orthogonal complement $\mathcal{H}_\bot$ is also invariant under $U$, $U^\dagger$ and $GG^\dagger$.

The statement about 1D subspaces is already proved in \lem{subspace}. 
When $\dim(\mathcal{H}_u)=2$, we can construct an orthonormal basis by applying the Gram-Schmidt process to $\left\{G\ket{\phi_u},UG\ket{\phi_u}\right\}$. This produces the unit basis vector
\begin{equation}
\begin{aligned}
    \frac{UG\ket{\phi_u}-G\ketbra{\phi_u}{\phi_u}G^\dagger UG\ket{\phi_u}}{\norm{UG\ket{\phi_u}-G\ketbra{\phi_u}{\phi_u}G^\dagger UG\ket{\phi_u}}}
    &=\frac{UG\ket{\phi_u}-\lambda_uG\ket{\phi_u}}{\norm{UG\ket{\phi_u}-\lambda_uG\ket{\phi_u}}}\\
    &=\frac{UG\ket{\phi_u}-\lambda_uG\ket{\phi_u}}{\sqrt{1+\lambda_u^2-\lambda_u\bra{\phi_u}G^\dagger UG\ket{\phi_u}-\lambda_u\bra{\phi_u}G^\dagger U^\dagger G\ket{\phi_u}}}\\
    &=\frac{UG\ket{\phi_u}-\lambda_uG\ket{\phi_u}}{\sqrt{1-\lambda_u^2}}
\end{aligned}
\end{equation}
orthogonal to $G\ket{\phi_u}$. The matrix representation then follows from a direct calculation.
\end{proof}

%%%%%%%%%%%%%%%%%%%%%%%%%%%%%%%%%%%%%%%%%%%%%%%%%%%%%%%%%%%%%%%%%%%%%%%%%%%%%%
\subsection{Quantum walk}
\label{append:qubitization_walk}
For a 2D subspace $\mathcal{H}_u$ where $-1<\lambda_u<1$, we see that 
\begin{equation*}
    \ket{\phi_{u,0}}=G\ket{\phi_u},\qquad
    \ket{\phi_{u,1}}=\frac{UG\ket{\phi_u}-\lambda_uG\ket{\phi_u}}{\sqrt{1-\lambda_u^2}}
\end{equation*}
is an orthonormal basis, under which
    \begin{equation}
        U=\begin{bmatrix}
            \lambda_u & \sqrt{1-\lambda_u^2}\\
            \sqrt{1-\lambda_u^2} & -\lambda_u
        \end{bmatrix}
        =\lambda_uZ+\sqrt{1-\lambda_u^2}X,\qquad
        GG^\dagger=
        \begin{bmatrix}
            1 & 0\\
            0 & 0
        \end{bmatrix}
        =\frac{1}{2}I+\frac{1}{2}Z.
    \end{equation}
This means that the quantum walk operator defined by $W=\left(2GG^\dagger-I\right)U$ has the matrix representation
\begin{equation}
    W=\left(2GG^\dagger-I\right)U=
    \begin{bmatrix}
        \lambda_u & \sqrt{1-\lambda_u^2}\\
        -\sqrt{1-\lambda_u^2} & \lambda_u
    \end{bmatrix}
    =\lambda_uI+i\sqrt{1-\lambda_u^2}Y.
\end{equation}
given by a linear combination of Pauli operators.
Therefore, for each $u$, $W$ has two eigenvalues
\begin{equation}
    \lambda_{u,\pm}=\lambda_u\pm i\sqrt{1-\lambda_u^2}
    =e^{\pm i\arccos(\lambda_u)}
\end{equation}
with the associated eigenvectors
\begin{equation}
    \ket{\phi_{u,Y\pm}}=\frac{\ket{\phi_{u,0}}\pm i\ket{\phi_{u,1}}}{\sqrt{2}}.
\end{equation}
We thus obtain the following spectral decomposition of the walk operator.

\begin{corollary}[Quantum walk]
\label{cor:walk}
Let $U$ be a unitary operator and $G$ be an isometry, such that $G^\dagger UG$ is Hermitian and $G^\dagger U^2G=I$.
If $G^\dagger UG$ has the spectral decomposition
\begin{equation}
    G^\dagger UG=\sum_u\lambda_u\ketbra{\phi_u}{\phi_u},
\end{equation}
then the quantum walk operator $W=\left(2GG^\dagger-I\right)U$ has the spectral decomposition
\begin{equation}
\begin{aligned}
    W&=\left(2GG^\dagger-I\right)U\\
    &=\sum_{\abs{\lambda_u}=1}\lambda_u\ketbra{\phi_{u,0}}{\phi_{u,0}}
    +\sum_{\abs{\lambda_u}<1}\left(e^{+ i\arccos(\lambda_u)}\ketbra{\phi_{u,Y+}}{\phi_{u,+}}
    +e^{-i\arccos(\lambda_u)}\ketbra{\phi_{u,Y-}}{\phi_{u,-}}\right)
\end{aligned}
\end{equation}
within $\mathbf{Im}\left(GG^\dagger\right)+\mathbf{Im}\left(UGG^\dagger U^\dagger\right)$, where
\begin{equation}
    \ket{\phi_{u,0}}=G\ket{\phi_u},\qquad
    \ket{\phi_{u,1}}=\frac{UG\ket{\phi_u}-\lambda_uG\ket{\phi_u}}{\sqrt{1-\lambda_u^2}},\qquad
    \ket{\phi_{u,Y\pm}}=\frac{\ket{\phi_{u,0}}\pm i\ket{\phi_{u,1}}}{\sqrt{2}}.
\end{equation}

\end{corollary}

A similar analysis applies to $U$. Recall that $U$ is represented as a $2$-by-$2$ matrix which can be expanded as a linear combination of Pauli operators $Z$ and $X$. Through diagonalization, we find that $U$ has eigenvectors
\begin{equation*}
    \ket{\phi_{u,ZX\pm}}=\frac{\sqrt{1\pm\lambda_u}\ket{\phi_{u,0}}\pm\sqrt{1\mp\lambda_u}\ket{\phi_{u,1}}}{\sqrt{2}}
\end{equation*}
with corresponding eigenvalues $\pm1$. We summarize the spectral properties of operators $2GG^\dagger-I$, $W$ and $U$ in \tab{qubitization_eigen}.

\begin{table}
    \centering
    \begin{tabular}{c|c|c|c}
        Operator & Matrix representation & Eigenvalues & Eigenvectors\\
        \hline
        $G^\dagger UG$ & $\left[\begin{smallmatrix}
            \lambda_u
        \end{smallmatrix}\right]$ & $\lambda_u$ & $\ket{\phi_u}$\\
        $2GG^\dagger-I$ & $\left[\begin{smallmatrix}
            1 & 0 \\
            0 & -1
        \end{smallmatrix}\right]$ & $\pm1$ & $\ket{\phi_{u,0}}=G\ket{\phi_u},\ \ket{\phi_{u,1}}=\frac{UG\ket{\phi_u}-\lambda_uG\ket{\phi_u}}{\sqrt{1-\lambda_u^2}}$ \\
    $U$ & $\left[\begin{smallmatrix}
            \lambda_u & \sqrt{1-\lambda_u^2}\\
            \sqrt{1-\lambda_u^2} & -\lambda_u
        \end{smallmatrix}\right]$ & $\pm1$ & $\ket{\phi_{u,ZX\pm}}=\frac{\sqrt{1\pm\lambda_u}\ket{\phi_{u,0}}\pm\sqrt{1\mp\lambda_u}\ket{\phi_{u,1}}}{\sqrt{2}}$\\
        $\left(2GG^\dagger-I\right)U$ & $\left[\begin{smallmatrix}
        \lambda_u & \sqrt{1-\lambda_u^2}\\
        -\sqrt{1-\lambda_u^2} & \lambda_u
    \end{smallmatrix}\right]$ & $e^{\pm i\arccos(\lambda_u)}$ & $\ket{\phi_{u,Y\pm}}=\frac{\ket{\phi_{u,0}}\pm i\ket{\phi_{u,1}}}{\sqrt{2}}$\\
    $U\left(2GG^\dagger-I\right)$ & $\left[\begin{smallmatrix}
        \lambda_u & -\sqrt{1-\lambda_u^2}\\
        \sqrt{1-\lambda_u^2} & \lambda_u
    \end{smallmatrix}\right]$ & $e^{\mp i\arccos(\lambda_u)}$ & $\ket{\phi_{u,Y\pm}}=\frac{\ket{\phi_{u,0}}\pm i\ket{\phi_{u,1}}}{\sqrt{2}}$\\
    \end{tabular}
    \caption{Spectral properties of common operators involved in Hermitian qubitization. Here, the matrix representation of $G^\dagger UG$ is defined with respect to the basis $\{\ket{\phi_u}\}$, while the matrix representation of $2GG^\dagger-I$, $\left(2GG^\dagger-I\right)U$, $U\left(2GG^\dagger-I\right)$, and $U$ is defined relative to $\{\ket{\phi_{u,0}},\ket{\phi_{u,1}}\}$.}
    \label{tab:qubitization_eigen}
\end{table}

%%%%%%%%%%%%%%%%%%%%%%%%%%%%%%%%%%%%%%%%%%%%%%%%%%%%%%%%%%%%%%%%%%%%%%%%%%%%%%
\subsection{Quantum amplitude amplification}
\label{append:qubitization_amp}
We now give a self-contained description of quantum amplitude amplification based on qubitization.
Let $\ket{\psi_0}$ be a quantum state and $\Pi$ be an orthogonal projection, so that directly measuring $\ket{\psi_0}$ yields the outcome $\Pi$ with success amplitude
\begin{equation*}
    \norm{\Pi\ket{\psi_0}}=\sin(\theta).
\end{equation*}
Here, amplitude amplification is infeasible if $\theta=0$, whereas the amplitude is already maximally amplified when $\theta=\frac{\pi}{2}$. So we assume $0<\theta<\frac{\pi}{2}$ in the following analysis.

We let
\begin{equation*}
    G=\ket{\psi_0},\qquad
    U=I-2\Pi,
\end{equation*}
and verify that $G$ is an isometry and $U$ is a Hermitian unitary. Then,
\begin{equation*}
    G^\dagger UG=\bra{\psi_0}(I-2\Pi)\ket{\psi_0}=1-2\norm{\Pi\ket{\psi_0}}^2=\cos(2\theta)=\lambda,
\end{equation*}
where we have dropped the subscript $u$ since $G^\dagger UG$ has only one eigenvalue.
This implies
\begin{equation*}
    \ket{\phi_0}=\ket{\psi_0},\qquad
    \ket{\phi_1}=\frac{-2\Pi\ket{\psi_0}+2\sin^2(\theta)\ket{\psi_0}}{\sin(2\theta)}.
\end{equation*}

Within the 2D subspace $\mathbf{Span}\left\{\ket{\phi_0},\ket{\phi_1}\right\}$, $U$ and $2GG^\dagger-I$ has the matrix representation
\begin{equation*}
    U=\begin{bmatrix}
        \cos(2\theta) & \sin(2\theta)\\
        \sin(2\theta) & -\cos(2\theta)
    \end{bmatrix},\qquad
    2GG^\dagger-I=\begin{bmatrix}
        1 & 0\\
        0 & -1
    \end{bmatrix}.
\end{equation*}
Hence, for any nonnegative integer $r$,
\begin{equation*}
    \left(-(I-2\ketbra{\psi_0}{\psi_0})(I-2\Pi)\right)^r
   =\left(\left(2GG^\dagger-I\right)U\right)^r=\begin{bmatrix}
        \cos(2r\theta) & \sin(2r\theta)\\
        -\sin(2r\theta) & \cos(2r\theta)
    \end{bmatrix}.
\end{equation*}
Applied to the initial state, this gives the amplified state
\begin{equation*}
    \left(-(I-2\ketbra{\psi_0}{\psi_0})(I-2\Pi)\right)^r\ket{\psi_0}
    =\left(\left(2GG^\dagger-I\right)U\right)^r\ket{\phi_0}
    =\cos(2r\theta)\ket{\phi_0}-\sin(2r\theta)\ket{\phi_1}.
\end{equation*}

We now rewrite the initial and amplified states with respect to the eigenvectors $\ket{\phi_{ZX\pm}}$ of $U$
with eigenvalues $\pm1$. For the initial state, we have
\begin{equation*}
    \ket{\phi_0}=\cos(\theta)\ket{\phi_{ZX+}}+\sin(\theta)\ket{\phi_{ZX-}},
\end{equation*}
whereas
\begin{equation*}
\begin{aligned}
    &\left(\left(2GG^\dagger-I\right)U\right)^r\ket{\phi_0}=\cos(2r\theta)\ket{\phi_0}-\sin(2r\theta)\ket{\phi_1}\\
    &=\cos(2r\theta)\left(\cos(\theta)\ket{\phi_{ZX+}}+\sin(\theta)\ket{\phi_{ZX-}}\right)
    -\sin(2r\theta)\left(\sin(\theta)\ket{\phi_{ZX+}}-\cos(\theta)\ket{\phi_{ZX-}}\right)\\
    &=\cos((2r+1)\theta)\ket{\phi_{ZX+}}+\sin((2r+1)\theta)\ket{\phi_{ZX-}}.
\end{aligned}
\end{equation*}
Finally, since the eigenvalues $+1$, $-1$ of $U$ are mapped to the eigenvalues $0$, $1$ of $\Pi=\frac{I-U}{2}$,
\begin{equation*}
    \Pi\ket{\phi_{ZX+}}=0,\qquad
    \Pi\ket{\phi_{ZX-}}=\ket{\phi_{ZX-}}.
\end{equation*}
This gives the following version of amplitude amplification.

\begin{lemma*}[Quantum amplitude amplification]
Let $\ket{\psi_0}$ be a quantum state and $\Pi$ be an orthogonal projection. There exist $\ket{\phi_{ZX-}}\in\mathbf{Im}(\Pi)$ and $\ket{\phi_{ZX+}}\in\mathbf{Ker}(\Pi)$ such that for any nonnegative integer $r=0,1,2,\ldots$,
\begin{equation*}
\begin{aligned}
    &\left(-(I-2\ketbra{\psi_0}{\psi_0})(I-2\Pi)\right)^r\ket{\psi_0}\\
    &=\sin\left((2r+1)\arcsin(\norm{\Pi\ket{\psi_0}})\right)\ket{\phi_{ZX-}}
    +\cos\left((2r+1)\arcsin(\norm{\Pi\ket{\psi_0}})\right)\ket{\phi_{ZX+}}.
\end{aligned}
\end{equation*}
Moreover, 
\begin{equation*}
    \ket{\phi_{ZX-}}=\frac{\Pi\ket{\psi_0}}{\norm{\Pi\ket{\psi_0}}},\qquad
    \ket{\phi_{ZX+}}=\frac{(I-\Pi)\ket{\psi_0}}{\norm{(I-\Pi)\ket{\psi_0}}},
\end{equation*}
are uniquely determined if $\ket{\psi_0}\notin\mathbf{Ker}(\Pi)$ and $\ket{\psi_0}\notin\mathbf{Im}(\Pi)$ respectively.

\end{lemma*}

%% file: gpe_bm.tex
In this appendix, we provide details on the construction of the gapped phase estimation algorithm with branch marking, which is used in \sec{dinv_branch} for preparing the discretized inverse state.

%%%%%%%%%%%%%%%%%%%%%%%%%%%%%%%%%%%%%%%%%%%%%%%%%%%%%%%%%%%%%%%%%%%%%%%%%%%%%%
\subsection{Simultaneous Fourier approximation of even and odd functions}
\label{append:gpe_bm_fourier}
We begin by considering the general problem of applying functions to the eigenphases of an input unitary. Specifically, suppose that the given unitary has the spectral decomposition $U=\sum_ve^{i\theta_v}\ketbra{\phi_v}{\phi_v}$. Our goal is to obtain an operator $V$ close to $\sum_v \left(f_a(\theta_v)I+if_c(\theta_v)X\right)\otimes\ketbra{\phi_v}{\phi_v}$ acting jointly on the input system and an ancilla qubit, for some desired even/odd periodic functions $f_a$ and $f_c$, using oracular queries to the unitary operator $U$ and its inverse. This problem is solved by the quantum signal processing technique~\cite{Low2016HamSim,MethodologyLYC}, which we review below.

\begin{lemma}[Quantum signal processing]
\label{lem:qsp}
Let $n>0$ be an integer, $f_a(\theta)=\sum_{k=0}^{n}a_j\cos(k\theta)$ and $f_c(\theta)=\sum_{k=1}^{n}a_j\sin(k\theta)$ be real Fourier series, such that $f_a^2(\theta)+f_c^2(\theta)\leq 1$ for all $\theta\in[-\pi,\pi]$ and $f_a(0)=1$. Given $U$ with the spectral decomposition $U=\sum_ve^{i\theta_v}\ketbra{\phi_v}{\phi_v}$, there exists a quantum circuit $V$ acting on the system register and a single-qubit ancilla as
\begin{equation}
    V=\sum_v \left(f_{a}(\theta_v)I+if_{b}(\theta_v)Z+if_{c}(\theta_v)X+if_{d}(\theta_v)Y\right)\otimes\ketbra{\phi_v}{\phi_v},
\end{equation}
with some even/odd real $2\pi$-periodic functions $f_b$/$f_d$ satisfying $f_{a}^2(\theta)+f_{b}^2(\theta)+f_{c}^2(\theta)+f_{d}^2(\theta)=1$ for all $\theta\in[-\pi,\pi]$, using at most $2n$ queries to $U$ and $U^\dagger$.
\end{lemma}

Note that this is slightly different from the setting of QSVT where we are often interested in a single function with a fixed parity---here we want to construct a unitary $V$ that simultaneously implements both even and odd functions.
In practice, the desired $f_a$ and $f_c$ can be obtained by truncating an infinite Fourier series. We then slightly modify the functions to satisfy the prerequisites $f_a^2(\theta)+f_c^2(\theta)\leq 1$ and $f_a(0)=1$. We will return to this point momentarily.

Now we describe how to construct the Fourier approximation. Our starting point is the following Chebyshev approximation of the sign function~\cite{Low2017USA,Wan22}
\begin{equation}
    \sum\limits_{\substack{j=0\\j\text{ odd}}}^{n}\beta_j\mathbf{T}_j(x)\in
    % \begin{cases}
    %     [-1,-1+\epsilon],\quad&x\in(-\infty,-\nu],\\
    %     [-1,1],&x\in[-\nu,\nu],\\
    %     [1-\epsilon,1],&x\in[\nu,+\infty),
    % \end{cases}
    \begin{cases}
        [-1,-1+\epsilon],\quad&x\in[-1,-\nu],\\
        [-1,1],&x\in[-\nu,\nu],\\
        [1-\epsilon,1],&x\in[\nu,1],
    \end{cases}
\end{equation}
for arbitrary margin $0<\nu<1$, accuracy $\epsilon>0$ and some coefficients $\beta_j$. Here, one can choose the polynomial degree to scale like
\begin{equation}
\label{eq:sgn_degree}
    n=\mathbf{O}\left(\frac{1}{\nu}\log\left(\frac{1}{\epsilon}\right)\right).
\end{equation}
From this, we then construct a Fourier approximation of the periodic sign function by substituting $x=\sin(\theta)$ and $\nu=\sin(\varphi)$, assuming $0<\varphi\leq\frac{\pi}{2}$. Within the period of $[-\pi,\pi]$, it has the following behavior
\begin{equation}
    g(\theta)=\sum\limits_{\substack{j=0\\j\text{ odd}}}^{n}\beta_j\mathbf{T}_j(\sin(\theta))\in
    \begin{cases}
        [-1,1],\quad&x\in[-\pi,-\pi+\varphi],\\
        [-1,-1+\epsilon],&x\in[-\pi+\varphi,-\varphi],\\
        [-1,1],&x\in[-\varphi,\varphi],\\
        [1-\epsilon,1],&x\in[\varphi,\pi-\varphi],\\
        [-1,1],&x\in[\pi-\varphi,\pi].
    \end{cases}
\end{equation}

Now assuming $0<\varphi\leq\theta_0\leq\frac{\pi}{2}$, we reflect and shift the periodic sign function  to construct the threshold functions for GPE. Specifically, we define
\begin{equation}
\begin{aligned}
    f_c(\theta)&=\frac{g(\theta-\theta_0)-g(-\theta-\theta_0)}{2}
    =\frac{\sum\limits_{\substack{j=0\\j\text{ odd}}}^{n}\beta_j\mathbf{T}_j(\sin(\theta-\theta_0))
    -\sum\limits_{\substack{j=0\\j\text{ odd}}}^{n}\beta_j\mathbf{T}_j(\sin(-\theta-\theta_0))}{2}\\
    &=\frac{\sum\limits_{\substack{j=0\\j\text{ odd}}}^{n}\beta_j\sin(j\theta-j\theta_0)(-1)^{\frac{j-1}{2}}
    -\sum\limits_{\substack{j=0\\j\text{ odd}}}^{n}\beta_j\sin(-j\theta-j\theta_0)(-1)^{\frac{j-1}{2}}}{2}\\
    &=\sum\limits_{\substack{j=0\\j\text{ odd}}}^{n}\beta_j\sin(j\theta)\cos(j\theta_0)(-1)^{\frac{j-1}{2}},
\end{aligned}
\end{equation}
and
\begin{equation}
\begin{aligned}
    f_a(\theta)&=\frac{-g(\theta-\theta_0)-g(-\theta-\theta_0)}{2}
    =\frac{-\sum\limits_{\substack{j=0\\j\text{ odd}}}^{n}\beta_j\mathbf{T}_j(\sin(\theta-\theta_0))
    -\sum\limits_{\substack{j=0\\j\text{ odd}}}^{n}\beta_j\mathbf{T}_j(\sin(-\theta-\theta_0))}{2}\\
    &=\frac{-\sum\limits_{\substack{j=0\\j\text{ odd}}}^{n}\beta_j\sin(j\theta-j\theta_0)(-1)^{\frac{j-1}{2}}
    -\sum\limits_{\substack{j=0\\j\text{ odd}}}^{n}\beta_j\sin(-j\theta-j\theta_0)(-1)^{\frac{j-1}{2}}}{2}\\
    &=\sum\limits_{\substack{j=0\\j\text{ odd}}}^{n}\beta_j\cos(j\theta)\sin(j\theta_0)(-1)^{\frac{j-1}{2}}.
\end{aligned}
\end{equation}
The above functions actually depend on various parameters such as $\theta_0$, $\varphi$ and $\epsilon$, but when their values are clear from the context, we will suppress these dependences for presentational purposes.
The qualitative behaviors of these functions are shown in \tab{fourier_gpe} and \fig{qsp_poly}.

%%%%%%%%%%%%%%%%%%%%%%%%%%%%%%%%%%%%%%%%%%%%%%%%%%%%%%%%%%%%%%%%%%%%%%%%%%%%%%
\begin{table}[t]
\begin{center}
\begin{tabular}{c|cc|cc}
     $\theta$ & $g(\theta-\theta_0)$ & $g(-\theta-\theta_0)$ & $f_c(\theta)$ & $f_a(\theta)$  \\
     \hline
     $[-\pi,-\pi+\theta_0-\varphi]$ & $[1-\epsilon,1]$ & $[1-\epsilon,1]$ & $[-\frac{\epsilon}{2},\frac{\epsilon}{2}]$ & $[-1,-1+\epsilon]$\\
     $[-\pi+\theta_0-\varphi,-\pi+\theta_0+\varphi]$ & $[-1,1]$ & $[1-\epsilon,1]$ & $[-1,\frac{\epsilon}{2}]$ & $[-1,\frac{\epsilon}{2}]$\\
     $[-\pi+\theta_0+\varphi,-\theta_0-\varphi]$ & $[-1,-1+\epsilon]$ & $[1-\epsilon,1]$ & $[-1,-1+\epsilon]$ & $[-\frac{\epsilon}{2},\frac{\epsilon}{2}]$\\
     $[-\theta_0-\varphi,-\theta_0+\varphi]$ & $[-1,-1+\epsilon]$ & $[-1,1]$ & $[-1,\frac{\epsilon}{2}]$ & $[-\frac{\epsilon}{2},1]$\\
     $[-\theta_0+\varphi,\theta_0-\varphi]$ & $[-1,-1+\epsilon]$ & $[-1,-1+\epsilon]$ & $[-\frac{\epsilon}{2},\frac{\epsilon}{2}]$ & $[1-\epsilon,1]$\\
     $[\theta_0-\varphi,\theta_0+\varphi]$ & $[-1,1]$ & $[-1,-1+\epsilon]$ & $[-\frac{\epsilon}{2},1]$ & $[-\frac{\epsilon}{2},1]$\\
     $[\theta_0+\varphi,\pi-\theta_0-\varphi]$ & $[1-\epsilon,1]$ & $[-1,-1+\epsilon]$ & $[1-\epsilon,1]$ & $[-\frac{\epsilon}{2},\frac{\epsilon}{2}]$\\
     $[\pi-\theta_0-\varphi,\pi-\theta_0+\varphi]$ & $[1-\epsilon,1]$ & $[-1,1]$ & $[-\frac{\epsilon}{2},1]$ & $[-1,\frac{\epsilon}{2}]$\\
     $[\pi-\theta_0+\varphi,\pi]$ & $[1-\epsilon,1]$ & $[1-\epsilon,1]$ & $[-\frac{\epsilon}{2},\frac{\epsilon}{2}]$ & $[-1,-1+\epsilon]$\\
\end{tabular}
\end{center}
    \caption{Qualitative behavior of the Fourier approximation of threshold functions used in branch marking and GPE.}
    \label{tab:fourier_gpe}
\end{table}
%%%%%%%%%%%%%%%%%%%%%%%%%%%%%%%%%%%%%%%%%%%%%%%%%%%%%%%%%%%%%%%%%%%%%%%%%%%%%%

%%%%%%%%%%%%%%%%%%%%%%%%%%%%%%%%%%%%%%%%%%%%%%%%%%%%%%%%%%%%%%%%%%%%%%%%%%%%%%
\begin{figure}[t]
    \centering
    \begin{subfigure}[b]{0.475\textwidth}
        \centering
        \includegraphics[width=.975\textwidth]{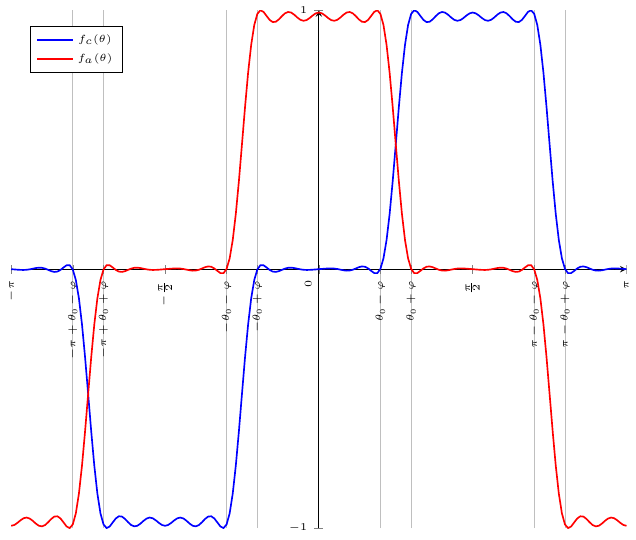}
        \caption{}
    \end{subfigure}%
    \hspace{0.5cm}
    \begin{subfigure}[b]{0.475\textwidth}
        \centering
        \includegraphics[width=.975\textwidth]{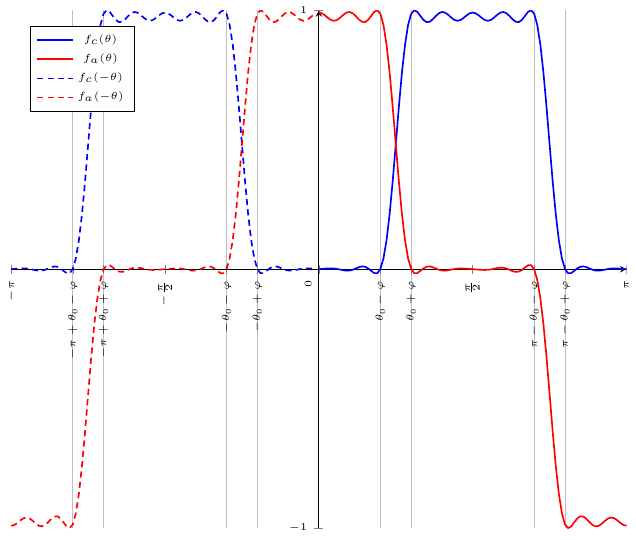}
        \caption{}
    \end{subfigure}
    \caption{Illustration of the qualitative behavior of functions used in branch marking (left panel) and GPE (right panel).}
    \label{fig:qsp_poly}
\end{figure}
%%%%%%%%%%%%%%%%%%%%%%%%%%%%%%%%%%%%%%%%%%%%%%%%%%%%%%%%%%%%%%%%%%%%%%%%%%%%%%

%%%%%%%%%%%%%%%%%%%%%%%%%%%%%%%%%%%%%%%%%%%%%%%%%%%%%%%%%%%%%%%%%%%%%%%%%%%%%%
\subsection{Branch marking}
\label{append:gpe_bm_bm}
Recall from~\sec{dinv_branch} that our preparation of the discretized inverse state uses the quantum walk operator. When the input matrix has the spectral decomposition
\begin{equation}
    A=\sum_u\lambda_u\ketbra{\phi_u}{\phi_u},
\end{equation}
the walk operator has the corresponding spectral decomposition
\begin{equation}
    W
    =\sum_{u}\left(e^{i\theta_{u,+}}\ketbra{\phi_{u,Y+}}{\phi_{u,Y+}}
    +e^{i\theta_{u,-}}\ketbra{\phi_{u,Y-}}{\phi_{u,Y-}}\right),\qquad
    \theta_{u,\pm}=\pm\arccos\left(\frac{\lambda_u}{\alpha_A}\right).
\end{equation}
Here, each eigenvector $\ket{\phi_u}$ of the coefficient matrix splits into $2$ eigenvectors $\ket{\phi_{u,Y\pm}}$ of the quantum walk opearator, and our primary goal is to ensure that these two branches undergo the exact same computation in GPE. This way, they can be merged back to recover $\ket{\phi_{u,0}}$, on which we perform the eigenvalue inversion $\frac{1}{\lambda_u}$.

To this end, we choose the block encoding normalization factor $\alpha_A\geq2\norm{A}$, so that
\begin{equation}
    \frac{\lambda_u}{\alpha_A}\in\left[-\frac{1}{2},\frac{1}{2}\right].
\end{equation}
Then in the quantum walk operator, this eigenvalue corresponds to two eigenphases
\begin{equation}
    \arccos\left(\frac{\lambda_u}{\alpha_A}\right)\in\left[\frac{\pi}{3},\frac{2\pi}{3}\right],\qquad
    -\arccos\left(\frac{\lambda_u}{\alpha_A}\right)\in\left[-\frac{2\pi}{3},-\frac{\pi}{3}\right].
\end{equation}
To distinguish between these cases, we can implement the periodic threshold functions with $\theta_0=\varphi=\frac{\pi}{6}$. This gives
\begin{equation}
    f_c\left(\pm\arccos\left(\frac{\lambda_u}{\alpha_A}\right)\right)\approx\pm1,\qquad
    f_a\left(\pm\arccos\left(\frac{\lambda_u}{\alpha_A}\right)\right)\approx0.
\end{equation}

We now carefully analyze the error. Specifically, we consider the two periodic threshold functions $f_c(\theta)$ and $f_a(\theta)$ with $\theta_0=\varphi=\frac{\pi}{6}$ whose actions are given by \tab{fourier_gpe} and \fig{qsp_poly}. By symmetry, it suffices to focus on the interval $\left[\frac{\pi}{3},\frac{2\pi}{3}\right]$. 
Recall that $f_c(\theta)=\frac{g(\theta-\theta_0)
    -g(-\theta-\theta_0)}{2}$ and $f_a(\theta)=\frac{-g(\theta-\theta_0)
    -g(-\theta-\theta_0)}{2}$, which implies
\begin{equation}
    f_c^2(\theta)+f_a^2(\theta)=\frac{g^2(\theta-\theta_0)+g^2(-\theta-\theta_0)}{2}.
\end{equation}
For all $-\pi\leq\theta\leq\pi$, we have from \tab{fourier_gpe} that $\abs{g(\theta-\theta_0)},\abs{g(-\theta-\theta_0)}\leq1$, so $f_c^2(\theta)+f_a^2(\theta)\leq1$ always holds. Also we know from \tab{fourier_gpe} that $1-\epsilon\leq f_a(0)\leq1$. Furthermore, for the target interval where $\frac{\pi}{3}\leq\theta\leq\frac{2\pi}{3}$, we have $1-\epsilon\leq g(\theta-\theta_0)\leq 1$ and $-1\leq g(-\theta-\theta_0)\leq-1+\epsilon$, which implies $1-\epsilon\leq f_c(\theta)\leq1$. To summarize,
\begin{equation}
\begin{cases}
    f_a^2(\theta)+f_c^2(\theta)\leq1,\quad &\forall \theta\in[-\pi,\pi],\\
    1-\epsilon\leq f_a(0)\leq1,\\
    1-\epsilon\leq f_c(\theta)\leq1,&\forall\theta\in\left[\frac{\pi}{3},\frac{2\pi}{3}\right],\\
    -1\leq f_c(\theta)\leq-1+\epsilon,&\forall\theta\in\left[-\frac{2\pi}{3},-\frac{\pi}{3}\right].
\end{cases}
\end{equation}

We now explain how to satisfy the requirements of \lem{qsp}.
We first run the sum-of-squares method~\cite{MethodologyLYC} to obtain functions $f_{b,1}(\theta)$, and $f_{d,1}(\theta)$ from
\begin{equation}
    f_{a,1}(\theta)=f_a(\theta),\qquad
    f_{c,1}(\theta)=f_c(\theta),
\end{equation}
where
\begin{equation}
\begin{cases}
    f_{a,1}^2(\theta)+f_{b,1}^2(\theta)+f_{c,1}^2(\theta)+f_{d,1}^2(\theta)=1,\quad &\forall \theta\in[-\pi,\pi],\\
    f_{a,1}^2(0)+f_{b,1}^2(0)=1.
\end{cases}
\end{equation}
Then, we let
\begin{equation}
    f_{a,2}(\theta)=f_{a,1}(\theta)f_{a,1}(0)+f_{b,1}(\theta)f_{b,1}(0),\qquad
    f_{c,2}(\theta)=f_{c,1}(\theta).
\end{equation}
Note that by the Cauchy-Schwarz inequality,
\begin{equation}
\begin{aligned}
    f_{a,2}^2(\theta)+f_{c,2}^2(\theta)
    &=\left(f_{a,1}(\theta)f_{a,1}(0)+f_{b,1}(\theta)f_{b,1}(0)\right)^2+f_{c,1}^2(\theta)\\
    &\leq\left(f_{a,1}^2(\theta)+f_{b,1}^2(\theta)\right)\left(f_{a,1}^2(0)+f_{b,1}^2(0)\right)+f_{c,1}^2(\theta)\\
    &=f_{a,1}^2(\theta)+f_{b,1}^2(\theta)+f_{c,1}^2(\theta)\leq1.
\end{aligned}
\end{equation}
Thus, we can rerun the sum-of-squares method to get functions $f_{b,2}(\theta)$, and $f_{d,2}(\theta)$. These functions now satisfy
\begin{equation}
\begin{cases}
    f_{a,2}^2(\theta)+f_{b,2}^2(\theta)+f_{c,2}^2(\theta)+f_{d,2}^2(\theta)=1,\quad &\forall \theta\in[-\pi,\pi],\\
    f_{a,2}(0)=1,
\end{cases}
\end{equation}
as desired, resulting in the QSP operator
\begin{equation}
    \sum_v \left(f_{a,2}(\theta_v)I+if_{b,2}(\theta_v)Z+if_{c,2}(\theta_v)X+if_{d,2}(\theta_v)Y\right)\otimes\ketbra{\phi_v}{\phi_v},
\end{equation}
where $v$ goes through all 2-tuple $(u,\pm)$.

It is clear that $f_{c,2}=f_{c,1}=f_{c}$, so we introduce no additional error to the function $f_c$:
\begin{equation}
\begin{cases}
    1-\epsilon\leq f_{c,2}(\theta)\leq1,\quad &\forall \theta\in\left[\frac{\pi}{3},\frac{2\pi}{3}\right],\\
    -1\leq f_{c,2}(\theta)\leq-1+\epsilon,&\forall\theta\in\left[-\frac{2\pi}{3},-\frac{\pi}{3}\right].
\end{cases}
\end{equation}
As for the remaining components,
\begin{equation}
\begin{aligned}
    f_{a,2}^2(\theta)+f_{b,2}^2(\theta)+f_{d,2}^2(\theta)=1-f_{c,2}^2(\theta)\leq1-(1-\epsilon)^2=2\epsilon-\epsilon^2,\qquad
    \forall\theta\in\left[-\frac{2\pi}{3},-\frac{\pi}{3}\right]\cup\left[\frac{\pi}{3},\frac{2\pi}{3}\right].
\end{aligned}
\end{equation}
To proceed, we need the following distance formula for matrices expanded in the Pauli basis.
\begin{lemma}[Matrix distance in the Pauli basis]
For real vectors $\beta=\left[
    \beta_a\ \beta_b\ \beta_c\ \beta_d
\right]$ and $\gamma=\left[
    \gamma_a\ \gamma_b\ \gamma_c\ \gamma_d \right]$,
\begin{equation}
    \norm{\left(\beta_aI+i\beta_bZ+i\beta_cX+i\beta_dY\right)
    -\left(\gamma_aI+i\gamma_bZ+i\gamma_cX+i\gamma_dY\right)}
    =\norm{\beta-\gamma}.
\end{equation}
\end{lemma}
\begin{proof}
We know that the Hermitian matrix
\begin{equation}
    (\beta_b-\gamma_b)Z+(\beta_c-\gamma_c)X+(\beta_d-\gamma_d)Y
\end{equation}
has eigenvalues
\begin{equation}
    \pm\sqrt{(\beta_b-\gamma_b)^2+(\beta_c-\gamma_c)^2+(\beta_d-\gamma_d)^2}.
\end{equation}
The claim then follows by rescaling by $i$ and shifting by $\beta_a-\gamma_a$.
\end{proof}

When $\theta\in\left[\frac{\pi}{3},\frac{2\pi}{3}\right]$, our ideal operator is $iX$, so the error is bounded by
\begin{equation}
\begin{aligned}
    &\ \norm{iX-\left(\left(f_{a,2}(\theta)I+if_{b,2}(\theta)Z+if_{c,2}(\theta)X+if_{d,2}(\theta)Y\right)\right)}\\
    &\leq\ \sqrt{
    (1-f_{c,2}(\theta))^2+f_{a,2}^2(\theta)+f_{b,2}^2(\theta)+f_{d,2}^2(\theta)}\\
    &\leq\ \sqrt{\epsilon^2+2\epsilon-\epsilon^2}
    =\sqrt{2\epsilon}.
\end{aligned}
\end{equation}
Similarly, when $\theta\in\left[-\frac{2\pi}{3},-\frac{\pi}{3}\right]$, our ideal operator is $-iX$, so the error is again bounded by
\begin{equation}
\begin{aligned}
    &\ \norm{-iX-\left(\left(f_{a,2}(\theta)I+if_{b,2}(\theta)Z+if_{c,2}(\theta)X+if_{d,2}(\theta)Y\right)\right)}\\
    &\leq\ \sqrt{
    (-1-f_{c,2}(\theta))^2+f_{a,2}^2(\theta)+f_{b,2}^2(\theta)+f_{d,2}^2(\theta)}\\
    &\leq\ \sqrt{\epsilon^2+2\epsilon-\epsilon^2}
    =\sqrt{2\epsilon}.
\end{aligned}
\end{equation}
This establishes the following:
\begin{proposition}[Branch marking]
\label{prop:bm}
Let $A/\alpha_A$ be block encoded by $O_A$ with normalization factor $\alpha_A\geq2\norm{A}$. 
Let $A\ket{\phi_u}=\lambda_u\ket{\phi_u}$ and $W\ket{\phi_{u,Y\pm}}= e^{\pm i\arccos\left(\frac{\lambda_u}{\alpha_A}\right)}\ket{\phi_{u,Y\pm}}$ be the corresponding eigenpairs of $A$ and the quantum walk operator $W$.
For any $\epsilon>0$, the isometry
    \begin{equation}
        \ket{+}\ket{+}\ket{\phi_{u,Y\pm}}\mapsto
        \ket{\xi_{u,\pm}}\ket{\phi_{u,Y\pm}},\qquad
        \norm{\ket{\xi_{u,\pm}}-\ket{+}\ket{\pm}}\leq\epsilon
    \end{equation}
can be implemented using
    \begin{equation}
        \mathbf{O}\left(\log\left(\frac{1}{\epsilon}\right)\right)
    \end{equation}
    queries to $O_A$.
\end{proposition}
\begin{proof}
By \lem{qsp} and the proceeding analysis, we can implement the operator
\begin{equation}
    V=\sum_{u} \left(F(\theta_{u,+})\otimes\ketbra{\phi_{u,Y+}}{\phi_{u,Y+}}
    +F(\theta_{u,-})\otimes\ketbra{\phi_{u,Y-}}{\phi_{u,Y-}}\right),
\end{equation}
such that $\norm{F(\theta_{u,\pm})-\left(\pm iX\right)}\leq\sqrt{2\epsilon}$ for all $u$.
Then, the controlled unitary
\begin{equation}
    \ketbra{0}{0}\otimes I+\ketbra{1}{1}\otimes (-iV)
\end{equation}
implements the mapping
\begin{equation}
    \ket{+}\ket{+}\ket{\phi_{u,Y\pm}}\mapsto
        \ket{\pm}\ket{+}\ket{\phi_{u,Y\pm}}
\end{equation}
to accuracy $\sqrt{2\epsilon}$. The query complexity follows now from \eq{sgn_degree} and the rescaling $\sqrt{2\epsilon}\mapsto\epsilon$.
\end{proof}

%%%%%%%%%%%%%%%%%%%%%%%%%%%%%%%%%%%%%%%%%%%%%%%%%%%%%%%%%%%%%%%%%%%%%%%%%%%%%%
\subsection{Branch marked gapped phase estimation}
\label{append:gpe_bm_gpe}
In GPE, our goal is to distinguish between eigenvalues of the input matrix within intervals
\begin{equation}
    \frac{\lambda_u}{\alpha_A}\in\left[-\frac{\gamma}{\rho},\frac{\gamma}{\rho}\right]\quad\text{and}\quad
    \frac{\lambda_u}{\alpha_A}\in\left(-1,-\gamma\right]\cup\left[\gamma,1\right).
\end{equation}
When we construct the quantum walk operator, each eigenvalue is split into two eigenphases. For the positive branch, our goal is to distinguish between
\begin{equation}
    \arccos\left(\frac{\lambda_u}{\alpha_A}\right)\in\left[\arccos\left(\frac{\gamma}{\rho}\right),\pi-\arccos\left(\frac{\gamma}{\rho}\right)\right]
\end{equation}
and
\begin{equation}
    \arccos\left(\frac{\lambda_u}{\alpha_A}\right)\in\left(0,\arccos(\gamma)\right]\cup\left[\pi-\arccos(\gamma),\pi\right),
\end{equation}
whereas for the negative branch, we need to differentiate 
\begin{equation}
    -\arccos\left(\frac{\lambda_u}{\alpha_A}\right)\in\left[-\pi+\arccos\left(\frac{\gamma}{\rho}\right),-\arccos\left(\frac{\gamma}{\rho}\right)\right]
\end{equation}
and
\begin{equation}
    -\arccos\left(\frac{\lambda_u}{\alpha_A}\right)\in\left(-\pi,-\pi+\arccos(\gamma)\right]\cup\left[-\arccos(\gamma),0\right).
\end{equation}
Additionally, we need to ensure that GPE has exactly the same behavior on eigenphases with opposite signs.

Let us consider the positive branch first (corresponding to $\arccos\left(\frac{\lambda_u}{\alpha_A}\right)$). To this end, we choose
\begin{equation}
    \theta_0=\frac{\arccos\left(\gamma\right)+\arccos\left(\frac{\gamma}{\rho}\right)}{2},\qquad
    \varphi=\frac{\arccos\left(\frac{\gamma}{\rho}\right)-\arccos\left(\gamma\right)}{2}.
\end{equation}
Then, we implement
\begin{equation}
    \sum_u \left(f_{a}(\theta_{u,+})I+if_{b}(\theta_{u,+})Z+if_{c}(\theta_{u,+})X+if_{d}(\theta_{u,+})Y\right)\otimes\ketbra{\phi_{u,Y+}}{\phi_{u,Y+}},
\end{equation}
where
\begin{equation}
    f_c(\theta)\approx
    \begin{cases}
        0,\quad&\theta\in\left(0,\arccos(\gamma)\right],\\
        1,\quad&\theta\in\left[\arccos\left(\frac{\gamma}{\rho}\right),\pi-\arccos\left(\frac{\gamma}{\rho}\right)\right],\\
        0,\quad&\theta\in\left[\pi-\arccos(\gamma),\pi\right),\\
    \end{cases}
\end{equation}
and
\begin{equation}
    f_a(\theta)\approx
    \begin{cases}
        1,\quad&\theta\in\left(0,\arccos(\gamma)\right],\\
        0,\quad&\theta\in\left[\arccos\left(\frac{\gamma}{\rho}\right),\pi-\arccos\left(\frac{\gamma}{\rho}\right)\right],\\
        -1,\quad&\theta\in\left[\pi-\arccos(\gamma),\pi\right).\\
    \end{cases}.
\end{equation}
Note that $f_a(\theta_{u,+})\approx\pm1$ has different signs for eigenvalue of the input matrix $\lambda_u\gtrless0$ with opposite signs, and the function values are also quite different in the transition band. But this is acceptable since this auxiliary information will be uncomputed in the end.

However, this choice of functions will have a different behavior on the negative branch. To address this issue, we instead aim to perform
\begin{equation}
    \sum_v \left(f_{a}(-\theta_{u,-})I+if_{b}(-\theta_{u,-})Z+if_{c}(-\theta_{u,-})X+if_{d}(-\theta_{u,-})Y\right)\otimes\ketbra{\phi_{u,Y-}}{\phi_{u,Y-}}.
\end{equation}
That is, the action of this new operator on negative eigenphases is exactly the same as that of the above operator on positive eigenphases. Moreover, the new polynomials automatically satisfy the requirements. In particular, it holds that $f_c(-\theta)=-f_c(\theta)$, $f_a(-\theta)=f_a(\theta)$.

Now we discuss error analysis. Without loss of generality, we focus on the distinguishment of $\left[0,\arccos\left(\gamma\right)\right]$ and $\left[\arccos\left(\frac{\gamma}{\rho}\right),\frac{\pi}{2}\right]$, since the other cases have exactly the same error by symmetry. Proceeding as in the previous subsection, we consider $f_c(\theta)=\frac{g(\theta-\theta_0)
    -g(-\theta-\theta_0)}{2}$ and $f_a(\theta)=\frac{-g(\theta-\theta_0)
    -g(-\theta-\theta_0)}{2}$ where
\begin{equation}
    f_c^2(\theta)+f_a^2(\theta)=\frac{g^2(\theta-\theta_0)+g^2(-\theta-\theta_0)}{2}.
\end{equation}
We also know from \tab{fourier_gpe} that $f_c^2(\theta)+f_a^2(\theta)\leq1$ and $1-\epsilon\leq f_a(0)\leq1$. Furthermore, we have $-\frac{\epsilon}{2}\leq f_c(\theta)\leq\frac{\epsilon}{2}$ and $1-\epsilon\leq f_a(\theta)\leq1$ for the target interval $0\leq\theta\leq\arccos\left(\gamma\right)$, and $1-\epsilon\leq f_c(\theta)\leq1$ and $-\frac{\epsilon}{2}\leq f_a(\theta)\leq\frac{\epsilon}{2}$ for the target interval $\arccos\left(\frac{\gamma}{\rho}\right)\leq\theta\leq\frac{\pi}{2}$.
To summarize,
\begin{equation}
\begin{cases}
    f_a^2(\theta)+f_c^2(\theta)\leq1,\quad &\forall \theta\in[-\pi,\pi],\\
    1-\epsilon\leq f_a(0)\leq1,\\
    \abs{f_c(\theta)}\leq\frac{\epsilon}{2},\ 1-\epsilon\leq f_a(\theta)\leq1,&\forall\theta\in\left[0,\arccos\left(\gamma\right)\right],\\
    1-\epsilon\leq f_c(\theta)\leq1,\ \abs{f_a(\theta)}\leq\frac{\epsilon}{2},&\forall\theta\in\left[\arccos\left(\frac{\gamma}{\rho}\right),\frac{\pi}{2}\right].
\end{cases}
\end{equation}
When $\theta\in\left[\arccos\left(\frac{\gamma}{\rho}\right),\frac{\pi}{2}\right]$, our ideal operator is $iX$, so the error can be bounded in a similar way as in the previous subsection
\begin{equation}
\begin{aligned}
    &\ \norm{iX-\left(\left(f_{a,2}(\theta)I+if_{b,2}(\theta)Z+if_{c,2}(\theta)X+if_{d,2}(\theta)Y\right)\right)}\\
    &\leq\ \sqrt{
    (1-f_{c,2}(\theta))^2+f_{a,2}^2(\theta)+f_{b,2}^2(\theta)+f_{d,2}^2(\theta)}\\
    &\leq\ \sqrt{\epsilon^2+2\epsilon-\epsilon^2}
    =\sqrt{2\epsilon}.
\end{aligned}
\end{equation}

Let us now consider $\theta\in\left[0,\arccos\left(\gamma\right)\right]$. We let
\begin{equation}
    f_{a,1}(\theta)=f_a(\theta),\qquad
    f_{c,1}(\theta)=f_c(\theta),
\end{equation}
and run the sum-of-squares method~\cite{MethodologyLYC} to obtain
\begin{equation}
\begin{cases}
    f_{a,1}^2(\theta)+f_{b,1}^2(\theta)+f_{c,1}^2(\theta)+f_{d,1}^2(\theta)=1,\quad &\forall \theta\in[-\pi,\pi],\\
    f_{a,1}^2(0)+f_{b,1}^2(0)=1.
\end{cases}
\end{equation}
Continuing, we define
\begin{equation}
    f_{a,2}(\theta)=f_{a,1}(\theta)f_{a,1}(0)+f_{b,1}(\theta)f_{b,1}(0),\qquad
    f_{c,2}(\theta)=f_{c,1}(\theta).
\end{equation}
and rerun the sum-of-squares method to get
\begin{equation}
\begin{cases}
    f_{a,2}^2(\theta)+f_{b,2}^2(\theta)+f_{c,2}^2(\theta)+f_{d,2}^2(\theta)=1,\quad &\forall \theta\in[-\pi,\pi],\\
    f_{a,2}(0)=1.
\end{cases}
\end{equation}
This results in
\begin{equation}
    \sum_u \left(f_{a,2}(\theta_{u,+})I+if_{b,2}(\theta_{u,+})Z+if_{c,2}(\theta_{u,+})X+if_{d,2}(\theta_{u,+})Y\right)\otimes\ketbra{\phi_{u,Y+}}{\phi_{u,Y+}}.
\end{equation}

Now we have
\begin{equation}
\begin{aligned}
    \abs{f_{a,2}(\theta)-f_{a,1}(\theta)}
    &=\abs{f_{a,1}(\theta)\left(f_{a,1}(0)-1\right)+f_{b,1}(\theta)f_{b,1}(0)}\\
    &\leq\abs{f_{a,1}(\theta)}\abs{f_{a,1}(0)-1}+\sqrt{1-f_{a,1}^2(\theta)}\sqrt{1-f_{a,1}^2(0)}\\
    &\leq1\cdot\epsilon+\sqrt{1-(1-\epsilon)^2}\sqrt{1-(1-\epsilon)^2}
    =3\epsilon-\epsilon^2,
\end{aligned}
\end{equation}
which implies
\begin{equation}
    \abs{f_{a,2}(\theta)-1}\leq4\epsilon-\epsilon^2.
\end{equation}
Combining with the obvious requirement that $\abs{f_{a,2}}\leq1$, we arrive at
\begin{equation}
    1-4\epsilon\leq1-4\epsilon+\epsilon^2\leq f_{a,2}(\theta)\leq1.
\end{equation}
When $\theta\in\left[0,\arccos\left(\gamma\right)\right]$, our ideal operator is $I$, so the error can be bounded as
\begin{equation}
\begin{aligned}
    &\ \norm{I-\left(\left(f_{a,2}(\theta)I+if_{b,2}(\theta)Z+if_{c,2}(\theta)X+if_{d,2}(\theta)Y\right)\right)}\\
    &\leq\ \sqrt{
    (1-f_{a,2}(\theta))^2+f_{b,2}^2(\theta)+f_{c,2}^2(\theta)+f_{d,2}^2(\theta)}\\
    &\leq\ \sqrt{(4\epsilon)^2+1-(1-4\epsilon)^2}
    =2\sqrt{2\epsilon}.
\end{aligned}
\end{equation}
Altogether, we have established the following proposition. Note it is important that the output state $\ket{\xi_u}$ of gapped phase estimation has \emph{no} dependence on the $\pm$ branches of eigenphases. This is the primary goal of branch marking. 
\begin{proposition}[Branch marked gapped phase estimation]
\label{prop:gpe}
Let $A/\alpha_A$ be block encoded by $O_A$ with normalization factor $\alpha_A\geq2\norm{A}$. 
Let $A\ket{\phi_u}=\lambda_u\ket{\phi_u}$ and $W\ket{\phi_{u,Y\pm}}= e^{\pm i\arccos\left(\frac{\lambda_u}{\alpha_A}\right)}\ket{\phi_{u,Y\pm}}$ be the corresponding eigenpairs of $A$ and the quantum walk operator $W$.
For any $\epsilon>0$, $0<\gamma<\frac{\alpha_A}{2}$ and constant $\rho>0$, the isometry
\begin{equation}
    \ket{\pm}\ket{0}\ket{\phi_{u,Y\pm}}\mapsto\ket{\pm}\ket{\xi_{u}}\ket{\phi_{u,Y\pm}},\qquad
\begin{cases}
    \norm{\ket{\xi_{u}}-\ket{0}}\leq\epsilon,\quad&\frac{\lambda_u}{\alpha_A}\in[\gamma,1),\\
    \norm{\ket{\xi_{u}}-i\ket{1}}\leq\epsilon,&\frac{\lambda_u}{\alpha_A}\in\left[-\frac{\gamma}{\rho},\frac{\gamma}{\rho}\right],\\
    \norm{\ket{\xi_{u}}-(-\ket{0})}\leq\epsilon,&\frac{\lambda_u}{\alpha_A}\in\left(-1,-\gamma\right],
\end{cases}
\end{equation}
can be implemented 
using
    \begin{equation}
        \mathbf{O}\left(\frac{1}{\gamma}\log\left(\frac{1}{\epsilon}\right)\right)
    \end{equation}
    queries to $O_A$.
\end{proposition}
\begin{proof}
By \lem{qsp} and the proceeding analysis, we can implement the operators
\begin{equation}
    V_+=\sum_u F(\theta_{u,+})\otimes\ketbra{\phi_{u,Y+}}{\phi_{u,Y+}},\qquad
    V_-=\sum_u F(-\theta_{u,-})\otimes\ketbra{\phi_{u,Y-}}{\phi_{u,Y-}}
\end{equation}
respectively, such that
\begin{equation}
\begin{cases}
    \norm{F(\theta)-I}\leq2\sqrt{2\epsilon},\quad&\theta\in\left(0,\arccos\left(\gamma\right)\right],\\
    \norm{F(\theta)-iX}\leq\sqrt{2\epsilon},&\theta\in\left[\arccos\left(\frac{\gamma}{\rho}\right),\pi-\arccos\left(\frac{\gamma}{\rho}\right)\right],\\
    \norm{F(\theta)-(-I)}\leq2\sqrt{2\epsilon},&\theta\in\left[\pi-\arccos(\gamma),\pi\right).\\
\end{cases}
\end{equation}
Then, the controlled unitary
\begin{equation}
    \ketbra{+}{+}\otimes V_++\ketbra{-}{-}\otimes V_-
\end{equation}
implements the mapping
\begin{equation}
    \ket{\pm}\ket{0}\ket{\phi_{u,Y\pm}}\mapsto
    \begin{cases}
        \ket{\pm}\ket{0}\ket{\phi_{u,Y\pm}},\quad&\frac{\lambda_u}{\alpha_A}\in[\gamma,1),\\
        i\ket{\pm}\ket{1}\ket{\phi_{u,Y\pm}},&\frac{\lambda_u}{\alpha_A}\in\left[-\frac{\gamma}{\rho},\frac{\gamma}{\rho}\right],\\
        -\ket{\pm}\ket{0}\ket{\phi_{u,Y\pm}},&\frac{\lambda_u}{\alpha_A}\in\left(-1,-\gamma\right],\\
    \end{cases}
\end{equation}
to accuracy $2\sqrt{2\epsilon}$. The query complexity follows now from \eq{sgn_degree} and the rescaling $2\sqrt{2\epsilon}\mapsto\epsilon$.
\end{proof}

%% file: multi.tex
In this appendix, we will consider a number of success probabilities involved in the analysis of our quantum linear system algorithm in \sec{dinv}. Specifically, we will analyze five different probabilities and show that they are in fact all constant multiplicative approximations of one another.

%%%%%%%%%%%%%%%%%%%%%%%%%%%%%%%%%%%%%%%%%%%%%%%%%%%%%%%%%%%%%%%%%%%%%%%%%%%%%%
\subsection{Definition of success probabilities}
\label{append:multi_def}

The first probability corresponds to the success probability of applying a block-encoded inverse of the input matrix to the initial state. Suppose that our input matrix is block encoded as $A/\alpha_A$ with normalization factor $\alpha_A\geq\norm{A}$ and upper bound $\alpha_{A^{-1}}\geq\norm{A^{-1}}$ on norm of the inverse matrix. Without loss of generality, we assume that both $\alpha_A$ and $\alpha_{A^{-1}}$ are powers of $3$, so that $\alpha_A\alpha_{A^{-1}}=3^m$ for some positive integer $m$.
By switching to the Hermitian dilation $\ketbra{0}{1}\otimes A+\ketbra{1}{0}\otimes A^\dagger$, we further assume that $A$ is Hermitian with the spectral decomposition $A=\sum_u\lambda_u\ketbra{\phi_u}{\phi_u}$ where $\frac{1}{\alpha_{A^{-1}}}\leq\abs{\lambda_u}\leq\alpha_A$, while the initial state can be expanded in the eigenbasis of $A$ as $\ket{b}=\sum_u\gamma_u\ket{\phi_u}$. 
Then, we define
\begin{equation}
    \psi_{\text{inv}}=\frac{A^{-1}}{\alpha_{A^{-1}}}\ket{b}
    =\sum_u\frac{1}{\alpha_{A^{-1}}\lambda_u}\gamma_u\ket{\phi_u},\qquad
    p_{\text{succ}}=\frac{\norm{A^{-1}\ket{b}}^2}{\alpha_{A^{-1}}^2}
    =\frac{1}{\alpha_{A^{-1}}^2}\left(\sum_u\abs{\frac{\gamma_u}{\lambda_u}}^2\right).
\end{equation}

The second probability is an auxiliary probability for a version of discretized inverse state supported only on one basis state of the time register. We let
\begin{equation}
    \psi_{\text{d-inv},1}=
    \sum_{k=0}^{m-1}\frac{3^{k+1}}{3^m}\ket{k}\sum_{\abs{\frac{\lambda_u}{\alpha_A}}\in\big[\frac{1}{3^{k+1}},\frac{1}{3^k}\big)}
    \gamma_u\ket{\phi_u},\qquad
    p_{\text{succ,d-inv},1}=\sum_{k=0}^{m-1}\sum_{\abs{\frac{\lambda_u}{\alpha_A}}\in\big[\frac{1}{3^{k+1}},\frac{1}{3^k}\big)}
    \abs{\gamma_u}^2\frac{9^{k+1}}{9^m}.
\end{equation}

The third probability corresponds to the probability of the discretized inverse state supported on two basis states of the time register. Concretely,
\begin{equation}
\begin{aligned}
    \psi_{\text{d-inv}}&=
    \sum_{k=0}^{m-1}\sum_{\abs{\frac{\lambda_u}{\alpha_A}}\in\big[\frac{1}{3^{k+1}},\frac{1}{3^k}\big)}
    \left(\zeta_{k+1,u}\frac{3^{k+1}}{3^m}\ket{k}+\zeta_{k,u}\frac{3^k}{3^m}\ket{k-1}\right)\gamma_u\ket{\phi_u},\\
    p_{\text{succ,d-inv}}&=\sum_{k=0}^{m-1}\sum_{\abs{\frac{\lambda_u}{\alpha_A}}\in\big[\frac{1}{3^{k+1}},\frac{1}{3^k}\big)}
    \abs{\gamma_u}^2
    \sum_{j=k,k+1}\abs{\zeta_{j,u}}^2\frac{9^{j}}{9^m}.
\end{aligned}
\end{equation}

The fourth probability is also an auxiliary probability for a version of discretized inverse state with full support on all basis states of the time register, assuming branch marking is performed perfectly:
\begin{equation}
\begin{aligned}
    \psi_{\text{d-inv},m}&=
    \sum_{k=0}^{m-1}\sum_{\abs{\frac{\lambda_u}{\alpha_A}}\in\big[\frac{1}{3^{k+1}},\frac{1}{3^k}\big)}
    \sum_{j=1}^{m}
    \zeta_{j,u}\frac{3^{j}}{3^m}
    \ket{j-1}
    \gamma_u\ket{\phi_u},\\
    p_{\text{succ,d-inv},m}&=
    \sum_{k=0}^{m-1}\sum_{\abs{\frac{\lambda_u}{\alpha_A}}\in\big[\frac{1}{3^{k+1}},\frac{1}{3^k}\big)}
    \abs{\gamma_u}^2\sum_{j=1}^{m}\abs{\zeta_{j,u}}^2\frac{9^{j}}{9^m}.
\end{aligned}
\end{equation}

Finally, the fifth probability is for the actual state $\psi_{\text{d-inv,bm}}$ prepared by the erroneous GPE and branch marking. We denote this probability by $p_{\text{succ,d-inv,bm}}$.

%%%%%%%%%%%%%%%%%%%%%%%%%%%%%%%%%%%%%%%%%%%%%%%%%%%%%%%%%%%%%%%%%%%%%%%%%%%%%%
\subsection{Multiplicative bounds on success probabilities}
\label{append:multi_bounds}

We now give multiplicative bounds on the success probabilities introduced in the previous subsection. Our analysis uses the following lemma.
\begin{lemma}
\label{lem:multi_lem}
The following estimates hold:
\begin{equation}
\begin{aligned}
    \sum_{k=0}^{m-1}\sum_{\abs{\frac{\lambda_u}{\alpha_A}}\in\big[\frac{1}{3^{k+1}},\frac{1}{3^k}\big)}
    \abs{\gamma_u}^2
    \sum_{j\neq k,k+1}\abs{\zeta_{j,u}}^2\frac{9^{j}}{9^m}
    &\leq\left(\sum_{k=0}^{m-1}\sum_{\abs{\frac{\lambda_u}{\alpha_A}}\in\big[\frac{1}{3^{k+1}},\frac{1}{3^k}\big)}
    \abs{\gamma_u}^2
    \frac{9^{k+2}}{9^m}\right)\epsilon_{\text{gpe}}^2,\\
    \sum_{k=0}^{m-1}\sum_{\abs{\frac{\lambda_u}{\alpha_A}}\in\big[\frac{1}{3^{k+1}},\frac{1}{3^k}\big)}
    \abs{\gamma_u}^2
    \sum_{j\neq k,k+1}\abs{\zeta_{j,u}}^2\frac{9^{k}}{9^m}
    &\leq\left(\sum_{k=0}^{m-1}\sum_{\abs{\frac{\lambda_u}{\alpha_A}}\in\big[\frac{1}{3^{k+1}},\frac{1}{3^k}\big)}
    \abs{\gamma_u}^2
    \frac{9^{k}}{9^m}\right)\epsilon_{\text{gpe}}^2.
\end{aligned}
\end{equation}
\end{lemma}
\begin{proof}
Recall that the cumulative coefficients $\zeta_{j,u}$ have the following bound
\begin{equation}
    \abs{\zeta_{j,u}}\leq
    \begin{cases}
        \epsilon_{\text{gpe},j},\qquad&j\leq k-1,\\
        \prod\limits_{l=k+1}^{j-1}\epsilon_{\text{gpe},l},&j\geq k+2.
    \end{cases}
\end{equation}
Therefore, we can prove the first claim as
\begin{equation}
\begin{aligned}
    &\sum_{k=0}^{m-1}\sum_{\abs{\frac{\lambda_u}{\alpha_A}}\in\big[\frac{1}{3^{k+1}},\frac{1}{3^k}\big)}
    \abs{\gamma_u}^2
    \left(\sum_{j=1}^{k-1}\epsilon_{\text{gpe},j}^2\frac{9^{j}}{9^m}
    +\sum_{j=k+2}^{m}\prod_{l=k+1}^{j-1}\epsilon_{\text{gpe},l}^2\frac{9^{j}}{9^m}\right)\\
    &\leq\sum_{k=0}^{m-1}\sum_{\abs{\frac{\lambda_u}{\alpha_A}}\in\big[\frac{1}{3^{k+1}},\frac{1}{3^k}\big)}
    \abs{\gamma_u}^2
    \left(\sum_{j=1}^{k-1}\epsilon_{\text{gpe},j}^2\frac{9^{k-1}}{9^m}
    +\sum_{j=k+2}^{m}\epsilon_{\text{gpe},j-1}^2\frac{9^{k+2}}{9^m}\right)\\
    &\leq\sum_{k=0}^{m-1}\sum_{\abs{\frac{\lambda_u}{\alpha_A}}\in\big[\frac{1}{3^{k+1}},\frac{1}{3^k}\big)}
    \abs{\gamma_u}^2
    \sum_{j=1}^{m-1}\epsilon_{\text{gpe},j}^2
    \frac{9^{k+2}}{9^m}
    \leq\sum_{k=0}^{m-1}\sum_{\abs{\frac{\lambda_u}{\alpha_A}}\in\big[\frac{1}{3^{k+1}},\frac{1}{3^k}\big)}
    \abs{\gamma_u}^2
    \frac{9^{k+2}}{9^m}
    \underbrace{\left(\sum_{j=1}^{m-1}\epsilon_{\text{gpe},j}\right)^2}_{\epsilon_{\text{gpe}}^2},
\end{aligned}
\end{equation}
where we have assumed all $\epsilon_{\text{gpe},l}\leq\frac{1}{3}$.
Similarly for the second claim,
\begin{equation}
\begin{aligned}
    &\sum_{k=0}^{m-1}\sum_{\abs{\frac{\lambda_u}{\alpha_A}}\in\big[\frac{1}{3^{k+1}},\frac{1}{3^k}\big)}
    \abs{\gamma_u}^2
    \left(\sum_{j=1}^{k-1}\epsilon_{\text{gpe},j}^2\frac{9^{k}}{9^m}
    +\sum_{j=k+2}^{m}\prod_{l=k+1}^{j-1}\epsilon_{\text{gpe},l}^2\frac{9^{k}}{9^m}\right)\\
    &\leq\sum_{k=0}^{m-1}\sum_{\abs{\frac{\lambda_u}{\alpha_A}}\in\big[\frac{1}{3^{k+1}},\frac{1}{3^k}\big)}
    \abs{\gamma_u}^2
    \left(\sum_{j=1}^{k-1}\epsilon_{\text{gpe},j}^2\frac{9^{k}}{9^m}
    +\sum_{j=k+2}^{m}\epsilon_{\text{gpe},j-1}^2\frac{9^{k}}{9^m}\right)\\
    &\leq\sum_{k=0}^{m-1}\sum_{\abs{\frac{\lambda_u}{\alpha_A}}\in\big[\frac{1}{3^{k+1}},\frac{1}{3^k}\big)}
    \abs{\gamma_u}^2
    \sum_{j=1}^{m-1}\epsilon_{\text{gpe},j}^2
    \frac{9^{k}}{9^m}
    \leq\sum_{k=0}^{m-1}\sum_{\abs{\frac{\lambda_u}{\alpha_A}}\in\big[\frac{1}{3^{k+1}},\frac{1}{3^k}\big)}
    \abs{\gamma_u}^2
    \frac{9^{k}}{9^m}
    \underbrace{\left(\sum_{j=1}^{m-1}\epsilon_{\text{gpe},j}\right)^2}_{\epsilon_{\text{gpe}}^2}.
\end{aligned}
\end{equation}
\end{proof}

Let us begin with the relation between $p_{\text{succ}}$ and $p_{\text{succ,d-inv},1}$ which is easy to describe. Following
\begin{small}
\begin{equation}
\newmaketag
\begin{aligned}
    \frac{1}{\alpha_{A^{-1}}^2}\sum_{k=0}^{m-1}\sum_{\abs{\frac{\lambda_u}{\alpha_A}}\in\big[\frac{1}{3^{k+1}},\frac{1}{3^k}\big)}\abs{\frac{\gamma_u}{\lambda_u}}^2
    &\leq\frac{1}{\alpha_{A^{-1}}^2}\sum_{k=0}^{m-1}\sum_{\abs{\frac{\lambda_u}{\alpha_A}}\in\big[\frac{1}{3^{k+1}},\frac{1}{3^k}\big)}\abs{\gamma_u}^2\frac{9^{k+1}}{\alpha_A^2}
    =\sum_{k=0}^{m-1}\sum_{\abs{\frac{\lambda_u}{\alpha_A}}\in\big[\frac{1}{3^{k+1}},\frac{1}{3^k}\big)}\abs{\gamma_u}^2\frac{9^{k+1}}{9^m},\\
    \frac{1}{\alpha_{A^{-1}}^2}\sum_{k=0}^{m-1}\sum_{\abs{\frac{\lambda_u}{\alpha_A}}\in\big[\frac{1}{3^{k+1}},\frac{1}{3^k}\big)}\abs{\frac{\gamma_u}{\lambda_u}}^2
    &\geq\frac{1}{\alpha_{A^{-1}}^2}\sum_{k=0}^{m-1}\sum_{\abs{\frac{\lambda_u}{\alpha_A}}\in\big[\frac{1}{3^{k+1}},\frac{1}{3^k}\big)}\abs{\gamma_u}^2\frac{9^{k}}{\alpha_A^2}
    =\sum_{k=0}^{m-1}\sum_{\abs{\frac{\lambda_u}{\alpha_A}}\in\big[\frac{1}{3^{k+1}},\frac{1}{3^k}\big)}\abs{\gamma_u}^2\frac{9^{k}}{9^m},
\end{aligned}
\end{equation}
\end{small}%
we conclude that
\begin{equation}
    \frac{p_{\text{succ,d-inv},1}}{9}\leq p_{\text{succ}}\leq p_{\text{succ,d-inv},1}.
\end{equation}

Then, we consider $p_{\text{succ,d-inv},1}$ and $p_{\text{succ,d-inv}}$. It is fairly straightforward to derive one side of the bound
\begin{equation}
    \sum_{k=0}^{m-1}\sum_{\abs{\frac{\lambda_u}{\alpha_A}}\in\big[\frac{1}{3^{k+1}},\frac{1}{3^k}\big)}
    \abs{\gamma_u}^2
    \sum_{j= k,k+1}\abs{\zeta_{j,u}}^2\frac{9^{j}}{9^m}
    \leq\sum_{k=0}^{m-1}\sum_{\abs{\frac{\lambda_u}{\alpha_A}}\in\big[\frac{1}{3^{k+1}},\frac{1}{3^k}\big)}
    \abs{\gamma_u}^2
    \frac{9^{k+1}}{9^m}.
\end{equation}
Analysis of the other side however is more involved and relies on \lem{multi_lem}:
\begin{equation}
\begin{aligned}
    &\sum_{k=0}^{m-1}\sum_{\abs{\frac{\lambda_u}{\alpha_A}}\in\big[\frac{1}{3^{k+1}},\frac{1}{3^k}\big)}
    \abs{\gamma_u}^2
    \sum_{j= k,k+1}\abs{\zeta_{j,u}}^2\frac{9^{j}}{9^m}
    \geq\sum_{k=0}^{m-1}\sum_{\abs{\frac{\lambda_u}{\alpha_A}}\in\big[\frac{1}{3^{k+1}},\frac{1}{3^k}\big)}
    \abs{\gamma_u}^2
    \sum_{j= k,k+1}\abs{\zeta_{j,u}}^2
    \frac{9^{k}}{9^m}\\
    &=\sum_{k=0}^{m-1}\sum_{\abs{\frac{\lambda_u}{\alpha_A}}\in\big[\frac{1}{3^{k+1}},\frac{1}{3^k}\big)}
    \abs{\gamma_u}^2
    \sum_{j= 1}^m\abs{\zeta_{j,u}}^2
    \frac{9^{k}}{9^m}
    -\sum_{k=0}^{m-1}\sum_{\abs{\frac{\lambda_u}{\alpha_A}}\in\big[\frac{1}{3^{k+1}},\frac{1}{3^k}\big)}
    \abs{\gamma_u}^2
    \sum_{j\neq k,k+1}\abs{\zeta_{j,u}}^2
    \frac{9^{k}}{9^m}\\
    &\geq \frac{p_{\text{succ,d-inv},1}}{9}-\frac{p_{\text{succ,d-inv},1}}{9}\epsilon_{\text{gpe}}^2.
\end{aligned}
\end{equation}
This gives
\begin{equation}
    \frac{p_{\text{succ,d-inv},1}}{9}\left(1-\epsilon_{\text{gpe}}^2\right)
    \leq p_{\text{succ,d-inv}}
    \leq p_{\text{succ,d-inv},1}.
\end{equation}

Next, we discuss $p_{\text{succ,d-inv}}$ and $p_{\text{succ,d-inv},m}$. Again, one side of the bound follows trivially from the definition
\begin{equation}
    \sum_{k=0}^{m-1}\sum_{\abs{\frac{\lambda_u}{\alpha_A}}\in\big[\frac{1}{3^{k+1}},\frac{1}{3^k}\big)}
    \abs{\gamma_u}^2
    \sum_{j= k,k+1}\abs{\zeta_{j,u}}^2\frac{9^{j}}{9^m}
    \leq\sum_{k=0}^{m-1}\sum_{\abs{\frac{\lambda_u}{\alpha_A}}\in\big[\frac{1}{3^{k+1}},\frac{1}{3^k}\big)}
    \abs{\gamma_u}^2
    \sum_{j=1}^m\abs{\zeta_{j,u}}^2\frac{9^{j}}{9^m}.
\end{equation}
As for the other side, we apply \lem{multi_lem} to compute
\begin{equation}
\begin{aligned}
    &\sum_{k=0}^{m-1}\sum_{\abs{\frac{\lambda_u}{\alpha_A}}\in\big[\frac{1}{3^{k+1}},\frac{1}{3^k}\big)}
    \abs{\gamma_u}^2
    \sum_{j= k,k+1}\abs{\zeta_{j,u}}^2\frac{9^{j}}{9^m}\\
    &=\sum_{k=0}^{m-1}\sum_{\abs{\frac{\lambda_u}{\alpha_A}}\in\big[\frac{1}{3^{k+1}},\frac{1}{3^k}\big)}
    \abs{\gamma_u}^2
    \sum_{j=1}^m\abs{\zeta_{j,u}}^2\frac{9^{j}}{9^m}
    -\sum_{k=0}^{m-1}\sum_{\abs{\frac{\lambda_u}{\alpha_A}}\in\big[\frac{1}{3^{k+1}},\frac{1}{3^k}\big)}
    \abs{\gamma_u}^2
    \sum_{j\neq k,k+1}\abs{\zeta_{j,u}}^2\frac{9^{j}}{9^m}\\
    &\geq p_{\text{succ,d-inv},m}-9^2p_{\text{succ,d-inv},1}\epsilon_{\text{gpe}}^2
    \geq p_{\text{succ,d-inv},m}-9^3p_{\text{succ,d-inv}}\left(\frac{\epsilon_{\text{gpe}}}{1-\epsilon_{\text{gpe}}}\right)^2.
\end{aligned}
\end{equation}
We thus obtain
\begin{equation}
    p_{\text{succ,d-inv}}
    \leq p_{\text{succ,d-inv},m}
    \leq p_{\text{succ,d-inv}}\left(1+9^3\frac{\epsilon_{\text{gpe}}^2}{\left(1-\epsilon_{\text{gpe}}\right)^2}\right).
\end{equation}

Finally, we examine $p_{\text{succ,d-inv},m}$ and $p_{\text{succ,d-inv,bm}}$. As branch marking is invoked twice, we have the additive approximation
\begin{equation}
    \abs{p_{\text{succ,d-inv},m}-p_{\text{succ,d-inv,bm}}}
    \leq2\abs{\sqrt{p_{\text{succ,d-inv},m}}-\sqrt{p_{\text{succ,d-inv,bm}}}}
    \leq4\epsilon_{\text{bm}}.
\end{equation}
Let us summarize these estimates as follows.
\begin{proposition}[Multiplicative bounds on success probabilities]
\label{prop:multi_succ}
Let the five success probabilities $p_{\text{succ}}$, $p_{\text{succ,d-inv},1}$, $p_{\text{succ,d-inv}}$, $p_{\text{succ,d-inv},m}$ and $p_{\text{succ,d-inv,bm}}$ be defined as in \append{multi_def}. Then,
\begin{equation}
\begin{aligned}
    \frac{p_{\text{succ,d-inv},1}}{9}\leq &\ p_{\text{succ}}\leq p_{\text{succ,d-inv},1},\\
    \frac{p_{\text{succ,d-inv},1}}{9}\left(1-\epsilon_{\text{gpe}}^2\right)
    \leq &\ p_{\text{succ,d-inv}}
    \leq p_{\text{succ,d-inv},1},\\
    p_{\text{succ,d-inv}}
    \leq &\ p_{\text{succ,d-inv},m}
    \leq p_{\text{succ,d-inv}}\left(1+9^3\frac{\epsilon_{\text{gpe}}^2}{\left(1-\epsilon_{\text{gpe}}\right)^2}\right),\\
    &\abs{p_{\text{succ,d-inv},m}-p_{\text{succ,d-inv,bm}}}
    \leq4\epsilon_{\text{bm}}.
\end{aligned}
\end{equation}
Hence, as long as
\begin{equation}
    \epsilon_{\text{gpe}}=\mathbf{O}(1),\qquad
    \epsilon_{\text{bm}}=\mathbf{O}\left(p_{\text{succ}}\right),
\end{equation}
all success probabilities are constant multiplicative approximations of one another
\begin{equation}
    p_{\text{succ,d-inv},1},\
    p_{\text{succ,d-inv}},\
    p_{\text{succ,d-inv},m},\
    p_{\text{succ,d-inv,bm}}
    =\mathbf{\Theta}\left(p_{\text{succ}}\right).
\end{equation}
\end{proposition}

%%%%%%%%%%%%%%%%%%%%%%%%%%%%%%%%%%%%%%%%%%%%%%%%%%%%%%%%%%%%%%%%%%%%%%%%%%%%%%
\subsection{Multiplicative bounds on sum of amplification thresholds}
\label{append:multi_threshold}
In \sec{dinv_deterministic}, we have derived the following multiplicative bounds on sum of amplification thresholds
\begin{equation}
    p_{\text{succ,d-inv}}9^l
    \leq\sum_{j=m-l+1}^{m}\norm{\overline{\Pi_{j}\Pi_b}A_{j}\cdots A_1\ket{\psi_0}}^29^{j-m+l}
    \leq\frac{5}{4}p_{\text{succ,d-inv}}9^l,
\end{equation}
where we assume VTAA with input algorithms $A_1,\ldots,A_m$ produces $\psi_{\text{d-inv}}$. We now incorporate the error of GPE and branch marking into the analysis.

When we consider erroneous GPE and perfect branch marking, VTAA with input algorithms $B_1,\ldots,B_m$ produces $\psi_{\text{d-inv},m}$ and the potentially good probabilities take the form
\begin{equation}
\begin{aligned}
    \norm{\overline{\Pi_{j}\Pi_b}B_{j}\cdots B_1\ket{\psi_0}}^2
    &=
    \sum_{k=0}^{m-1}\sum_{\abs{\frac{\lambda_u}{\alpha_A}}\in\big[\frac{1}{3^{k+1}},\frac{1}{3^k}\big)}
    \abs{\gamma_u}^2\sum_{h=j+1}^{m}\abs{\zeta_{h,u}}^2\\
    &\quad+\sum_{k=0}^{m-1}\sum_{\abs{\frac{\lambda_u}{\alpha_A}}\in\big[\frac{1}{3^{k+1}},\frac{1}{3^k}\big)}
    \abs{\gamma_u}^2\sum_{h=1}^{j}\abs{\zeta_{h,u}}^2\frac{9^{h}}{9^m}.
\end{aligned}
\end{equation}
Correspondingly, the sum of amplification thresholds are split into two terms:
\begin{equation}
\begin{aligned}
    \sum_{j=m-l+1}^{m}\norm{\overline{\Pi_{j}\Pi_b}B_{j}\cdots B_1\ket{\psi_0}}^29^{j-m+l}
    &=
    \sum_{j=m-l+1}^{m}\sum_{k=0}^{m-1}\sum_{\abs{\frac{\lambda_u}{\alpha_A}}\in\big[\frac{1}{3^{k+1}},\frac{1}{3^k}\big)}
    \abs{\gamma_u}^2\sum_{h=j+1}^{m}\abs{\zeta_{h,u}}^29^{j-m+l}\\
    &\quad+\sum_{j=m-l+1}^{m}\sum_{k=0}^{m-1}\sum_{\abs{\frac{\lambda_u}{\alpha_A}}\in\big[\frac{1}{3^{k+1}},\frac{1}{3^k}\big)}
    \abs{\gamma_u}^2\sum_{h=1}^{j}\abs{\zeta_{h,u}}^2\frac{9^{h}}{9^m}9^{j-m+l}.
\end{aligned}
\end{equation}
We exchange the summation order in the first term to get
\begin{equation}
    \sum_{j=m-l+1}^{m}\sum_{h=j+1}^{m}9^j
    =\sum_{h=m-l+2}^{m}\sum_{j=m-l+1}^{h-1}9^j
    =\frac{9}{8}\sum_{h=m-l+2}^{m}9^{h-1},
\end{equation}
whereas for the second term
\begin{equation}
    \sum_{j=m-l+1}^{m}\sum_{h=1}^{j}9^j
    =\sum_{h=1}^{m}\sum_{j=h}^{m}9^j
    =\frac{9}{8}\sum_{h=1}^{m}9^m.
\end{equation}
Altogether,
\begin{equation}
\begin{aligned}
    p_{\text{succ,d-inv},m}9^{l}
    &\leq\sum_{j=m-l+1}^{m}\norm{\overline{\Pi_{j}\Pi_b}B_{j}\cdots B_1\ket{\psi_0}}^29^{j-m+l}\\
    &\leq\frac{5}{4}\sum_{k=0}^{m-1}\sum_{\abs{\frac{\lambda_u}{\alpha_A}}\in\big[\frac{1}{3^{k+1}},\frac{1}{3^k}\big)}\abs{\gamma_u}^2
    \sum_{h=1}^{m}\abs{\zeta_{h,u}}^29^{h-m+l}
    \leq \frac{5}{4}p_{\text{succ,d-inv},m}9^{l}.
\end{aligned}
\end{equation}

Finally, we consider the general case where both GPE and branch marking are imperfect. In this case, VTAA with input algorithms $C_1,\ldots,C_m$ produces $\psi_{\text{d-inv,bm}}$ with $\norm{\psi_{\text{d-inv,bm}}}^2=p_{\text{succ,d-inv,bm}}$. Then, the sum of amplification thresholds has error
\begin{equation}
    \abs{\sum_{j=m-l+1}^{m}\norm{\overline{\Pi_{j}\Pi_b}B_{j}\cdots B_1\ket{\psi_0}}^29^{j-m+l}
    -\sum_{j=m-l+1}^{m}\norm{\overline{\Pi_{j}\Pi_b}C_{j}\cdots C_1\ket{\psi_0}}^29^{j-m+l}}
    \leq\frac{9}{8}4\epsilon_{\text{bm}}9^l.
\end{equation}
We thus obtain:
\begin{proposition}[Multiplicative bounds on sum of amplification thresholds]
\label{prop:multi_amp}
Let the input algorithms $A_1,\ldots,A_m$, $B_1,\ldots,B_m$ and $C_1,\ldots,C_m$ be defined as in \append{multi_threshold}. Then,
\begin{equation}
\begin{aligned}
    p_{\text{succ,d-inv}}9^l
    \leq\sum_{j=m-l+1}^{m}\norm{\overline{\Pi_{j}\Pi_b}A_{j}\cdots A_1\ket{\psi_0}}^29^{j-m+l}
    &\leq\frac{5}{4}p_{\text{succ,d-inv}}9^l,\\
    p_{\text{succ,d-inv},m}9^{l}
    \leq\sum_{j=m-l+1}^{m}\norm{\overline{\Pi_{j}\Pi_b}B_{j}\cdots B_1\ket{\psi_0}}^29^{j-m+l}
    &\leq \frac{5}{4}p_{\text{succ,d-inv},m}9^{l},\\
    \abs{\sum_{j=m-l+1}^{m}\norm{\overline{\Pi_{j}\Pi_b}B_{j}\cdots B_1\ket{\psi_0}}^29^{j-m+l}
    -\sum_{j=m-l+1}^{m}\norm{\overline{\Pi_{j}\Pi_b}C_{j}\cdots C_1\ket{\psi_0}}^29^{j-m+l}}
    &\leq\frac{9}{8}4\epsilon_{\text{bm}}9^l.
\end{aligned}
\end{equation}
Hence, as long as
\begin{equation}
    \epsilon_{\text{gpe}}=\mathbf{O}(1),\qquad
    \epsilon_{\text{bm}}=\mathbf{O}\left(p_{\text{succ}}\right),
\end{equation}
all sums of amplification thresholds have the scaling
\begin{equation}
\begin{aligned}
    \sum_{j=m-l+1}^{m}\norm{\overline{\Pi_{j}\Pi_b}A_{j}\cdots A_1\ket{\psi_0}}^29^{j-m+l}
    &=\mathbf{\Theta}\left(p_{\text{succ}}9^l\right),\\
    \sum_{j=m-l+1}^{m}\norm{\overline{\Pi_{j}\Pi_b}B_{j}\cdots B_1\ket{\psi_0}}^29^{j-m+l}
    &=\mathbf{\Theta}\left(p_{\text{succ}}9^l\right),\\
    \sum_{j=m-l+1}^{m}\norm{\overline{\Pi_{j}\Pi_b}C_{j}\cdots C_1\ket{\psi_0}}^29^{j-m+l}
    &=\mathbf{\Theta}\left(p_{\text{succ}}9^l\right).\\
\end{aligned}
\end{equation}
\end{proposition}